\documentclass[aps,prl, twocolumn,amsmath,amssymb,longbibliography]{revtex4-2}
\usepackage{amsmath,amssymb,amsfonts,bm}
\usepackage{amsthm}
\usepackage{adjustbox}
\usepackage{xcolor}
\usepackage{bbold}
\usepackage{soul}
\usepackage{graphicx}
\usepackage{dcolumn}
\usepackage{epstopdf}
\usepackage{epsfig}
\usepackage{url}
\usepackage{hyperref}
\makeatletter
\let\origaddcontentsline\addcontentsline
\newcommand{\suspendtoc}{\let\addcontentsline\@gobblethree}
\newcommand{\resumetoc}{\let\addcontentsline\origaddcontentsline}
\makeatother

\usepackage{verbatim}
\usepackage{braket}
\usepackage[normalem]{ulem}
\usepackage{IEEEtrantools}
\usepackage{mathrsfs}

\usepackage{enumitem}

\usepackage{tensor}

\usepackage{adjustbox}
\usepackage{float}

\usepackage{amsfonts}
\usepackage{booktabs}
\usepackage{siunitx}

\usepackage[flushleft]{threeparttable}
\usepackage{makecell,booktabs}

\newtheorem{theorem}{Theorem}[section]   
\newtheorem{lemma}[theorem]{Lemma}       
\newtheorem{proposition}[theorem]{Proposition}

\theoremstyle{definition}
\newtheorem{definition}[theorem]{Definition}

\theoremstyle{remark}

\newcommand{\gauss}{{\mathrm G}}
\newcommand{\ngauss}{{\mathrm{nG}}}
\newcommand{\piedge}{\pi_{\mathrm e}}
\newcommand{\pibulk}{\pi_{\mathrm b}}
\newcommand{\zeroedge}{0_{\mathrm e}}
\newcommand{\zerobulk}{0_{\mathrm b}}
\newcommand{\EP}{\mathrm{EP}}

\newcommand{\Peter}{\mathcal{P}}
\newcommand{\Overlap}{\mathcal{O}}
\newcommand{\CH}{\mathrm{CH}}
\newcommand{\PTsym}{\mathcal{PT}}

\newcommand{\Trb}{\mathsf{t}}
\newcommand{\Detb}{\mathsf{d}}
\newcommand{\boltz}{\mathcal{B}}
\newcommand{\Leff}{L_\mathrm{e}}

\newcommand{\Knormalized}{\mathcal{K}}

\newcommand{\DeltaDisc}{\Delta}
\newcommand{\Ttilde}{\widetilde{\boltz}}
\newcommand{\Ncal}{\mathcal{N}}

\newcommand{\spatdef}{\mathrm{top}}
\newcommand{\atanh}{\operatorname{atanh}}

\newcommand{\Asf}{\mathsf{A}}
\newcommand{\Bsf}{\mathsf{B}}
\newcommand{\Csf}{\mathsf{C}}
\newcommand{\Ksf}{\mathsf{K}}

\newcommand{\TT}{\mathcal{T}}
\newcommand{\tth}{t_\mathrm{Th}}

\newcommand{\uu}{\mathcal{U}}

\newcommand{\Tr}{\operatorname{Tr}}

\newcommand{\tr}{\operatorname{tr}}

\newcommand{\be}{\begin{equation}}
\newcommand{\ee}{\end{equation}}
\newcommand{\ba}{\begin{aligned}}
\newcommand{\ea}{\end{aligned}}
\newcommand{\PT}{\mathcal{PT}}
\newcommand{\RPM}{\mathrm{RPM}}

\newcommand{\wg}{\mathrm{Wg}}
\newcommand{\Wg}{\mathrm{Wg}}

\newcommand{\Th}{\mathrm{Th}}

\newcommand{\iden}{\mathbb{1}}

\newcommand{\topo}{\text{top}}

\newcommand{\Lth}{L_{\mathrm{Th}} }

\usepackage{mathtools}

\graphicspath{
  {Figures/}
  {Figures/Daniel_figures/figures_for_tsff/Parity_dual/}
  {Figures/Daniel_figures/figures_for_tsff/Parity_dual/classification_plots/}
  {Figures/Daniel_figures/figures_for_tsff/Parity_dual/analysis_plots/}
  {Figures/Daniel_figures/figures_for_tsff/combined/}
  {Figures_upload/}
}

\begin{document}
  \title{Topological spectral form factor reveals emergent   
  non-Hermitian  \\
   single-particle $\mathcal{PT}$ transitions from many-body quantum chaos}

\newcommand{\titleinfo}{Topological spectral form factor reveals emergent   
  non-Hermitian  \\
   single-particle $\mathcal{PT}$ transitions from many-body quantum chaos}

\author{Daniel Harkin}
\email{daniel.harkin@lancaster.ac.uk}
\address{Department of Physics, Lancaster University, Lancaster LA1 4YB, United Kingdom}

\author{Chun Y. Leung}
\email{c.leung@northumbria.ac.uk}
\address{School of Engineering, Physics and Mathematics, Northumbria
University, Newcastle upon Tyne NE1 8ST, United Kingdom}
\address{Department of Physics, Lancaster University, Lancaster LA1 4YB, United Kingdom}

\author{Amos Chan}
\email{amos.chan@warwick.ac.uk}
\affiliation{Department of Physics, University of Warwick, Coventry, CV4 7AL, United Kingdom}
\address{Department of Physics, Lancaster University, Lancaster LA1 4YB, United Kingdom}
\date{\today}

\begin{abstract}
In equilibrium, topological defect insertions in quantum and classical partition functions provide non-perturbative probes of phase transitions beyond local observables. 
In non-equilibrium physics, the spectral form factor provides a minimal probe of universal quantum dynamics and is a product of two partition functions at imaginary inverse temperature. 
We define the topological spectral form factor (TopSFF) by inserting topological defects acting non-trivially on the doubled partition functions, producing mismatched spacetime world-sheet topologies. 
For the minimal \(\mathbb{Z}_2\) spatially extended defect, implemented by the global swap operator, we derive an exact mapping of the TopSFF of a generic \((1+1)\)D  many-body chaotic system to an emergent \((3+1)\)D non-Hermitian single-particle problem describing a temporal domain wall (tDW).  We show analytically that the effective tDW dynamics undergoes a \(\mathcal{PT}\) symmetry breaking transition at finite interaction strength \(\epsilon_{\EP}\): below \(\epsilon_{\EP}\), the leading modes are polarized into Gaussian or non-Gaussian tDW sectors and the TopSFF varies monotonically and exponentially with system size;
above \(\epsilon_{\EP}\), the tDW sectors hybridize and the TopSFF oscillates with system size; at the exceptional point \(\epsilon_{\EP}\), Jordan non-diagonality produces a linear-in-system-size enhancement.
For minimal temporally extended topological defects associated with global swap or time translation, we derive exact universal TopSFF free-energy scaling forms. We verify these signatures numerically in independent models.

\end{abstract}
\maketitle

\suspendtoc

\begin{figure}[ht]
    \centering
    \includegraphics[width=0.48\textwidth]{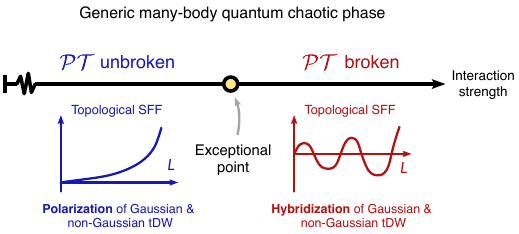}
    \caption{
    \textbf{Topological SFF and $\mathcal{PT}$ transitions.} In generic quantum many-body chaotic systems, TopSFF with minimal  spatially extended \(\mathbb{Z}_2\) defect, implemented by a global swap insertion, exhibits an emergent $\mathcal{PT}$ symmetry breaking transition at a finite interaction strength. 
    In the \(\mathcal{PT}\) unbroken phase, the leading modes are polarised into Gaussian or non-Gaussian temporal domain walls (tDW), producing monotonic exponential scaling in system size \(L\) at sufficiently large \(L\). 
    In the \(\mathcal{PT}\) broken phase, the two tDW species hybridise, producing oscillations in \(L\) with an exponential envelope.
    At the exceptional point, mode coalescence renders the effective tDW dynamics non-diagonalizable, producing a Jordan-block contribution with an additional polynomial prefactor in \(L\). 
    }
    \label{fig:trans}
\end{figure}

\begin{figure}[ht]
    \centering
\includegraphics[width=0.34\textwidth]{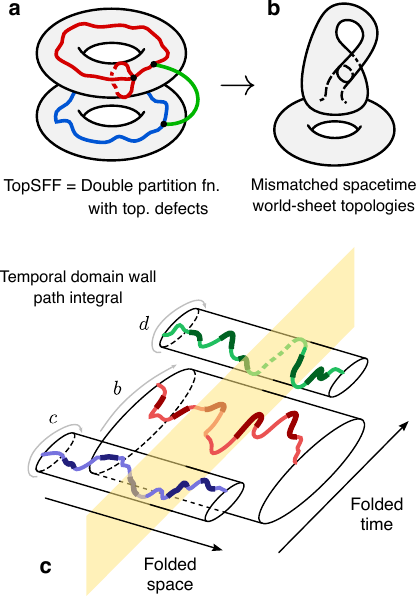}
    \caption{\textbf{Topological SFF as path integral of temporal domain walls.} (a) TopSFF is constructed by inserting non-trivial topological defects onto SFF as double partition functions, leading to (b) mismatched spacetime world-sheet topologies. 
(c) For a \((1+1)\mathrm{D}\) generic quantum many-body chaotic system,
the ensemble-averaged TopSFF with a spatially extended global swap defect
exactly maps to an effective \((3+1)\mathrm{D}\) path integral of emergent
tDW (yellow slice).  The tDW has two species, Gaussian and non-Gaussian (lighter and darker colours), and is specified by three folded-time loop coordinates while propagating along space as a non-Hermitian single particle.
    }
    \label{fig:schematic}
\end{figure}

Topological defects are robust non-perturbative structures that arise when locally ordered degrees of freedom cannot be made globally consistent, providing probes of global properties beyond the reach of local observables~\cite{kibble1976,mermin1979}.
In equilibrium statistical mechanics and quantum field theory, inserting a topological defect \(\hat{\mathcal D}\) into the partition function, \(Z_{\mathcal D}(\beta):=
\Tr(\hat{\mathcal D}e^{-\beta\hat H})\), with Hamiltonian 
\(\hat H\) and inverse temperature  \(\beta\),
diagnoses symmetry breaking, dualities, anomalies and topological
phases~\cite{kadanoff1971, Petkova_2001, conformaldefects2004frohlich, Gaiotto2015gensym, aasen2020topological}.
 Their key feature is topological invariance: subject to local commutation relations, topological defects can be deformed without changing the partition function. 

In non-equilibrium quantum dynamics, the partition function viewpoint has a natural analogue: The analytically continued partition function, \(Z(\beta= it):=\Tr e^{-i\hat H t}= \Tr\hat U(t)\), is the trace of the propagator and probes the model-dependent one-point spectral density.
The minimal universal probe is instead the spectral form factor (SFF), \(K(t)=|Z(it)|^2\), the Fourier
transform of the two-level spectral correlation function~\cite{Mehta, berry1985semiclassical}. Equivalently,
the SFF is a doubled partition function at imaginary inverse temperature, or a double sum over Feynman paths.
In generic many-body quantum chaotic systems~\cite{deutsch1991quantum, Srednicki, Rigol2008}, the connected SFF displays the bump-ramp-plateau behaviour~\cite{chan2018spectral}. The ramp and plateau are well-known universal signatures of quantum chaos~\cite{bohigas1984characterization, berry1985semiclassical, Sieber_2001, muller_2004, Cotler_2017, Kos_2018, chan2018solution, bertini2018exact, chan2018spectral, friedman2019, saad2019semiclassical, li2021spectral, bertini2021random}, whereas the bump is an intrinsically many-body phenomenon: it arises from spatial domain walls between locally distinct pairings of Feynman paths~\cite{chan2018spectral}, can be probed as topological defects via twisted boundary conditions~\cite{Garratt2021prx}, and exhibits
universality beyond random matrix theory (RMT)~\cite{friedman2019, chan2021trans, Shivam_2023}. Related deviations from RMT in the SFF bump regime have since been identified across a range of many-body systems~\cite{Gharibyan_2018, winer2022hydrodynamic,chaoschallengeMBL, chan2020lyap, roy2020random, yoshimura2025operator, altland2026path, kieler2026semiclassical}.

\begin{figure*}[ht]
    \centering
\includegraphics[width=0.98\textwidth]{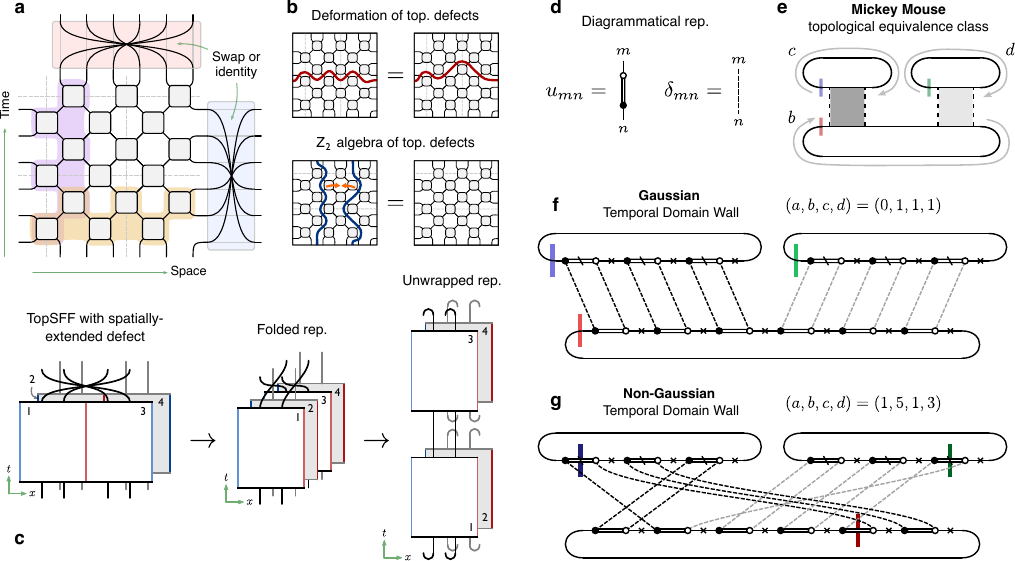}
\caption{
\textbf{Temporal domain walls and Mickey Mouse diagrams.}
(a) A global swap defect or the identity operator may be inserted along the spatial or temporal direction of the worldsheet describing a quantum many-body system evolved by local unitary gates. (b) Topological invariance allows local deformations of the defect line after folding described in (c), while two swap defects fuse to the identity, realizing the \(\mathbb{Z}_2\) algebra.
(c) For the TopSFF with a spatially extended global swap defect, worldsheets 3 and 4 are folded onto worldsheets 1 and 2 along the parity inversion axes (red and blue), so that parity-related degrees of freedom are stacked. Unwrapping each folded site in the time direction yields three loops at each folded site: one large loop formed by worldlines in \(\Tr[\hat{S}\hat{U}]\) passing through sheets 1 and 3, and two smaller loops formed by worldlines in \(\Tr[\hat{U}^*]\) passing through sheets 2 and 4.
(d) Diagrammatic notation for a local Haar-random unitary gate element \(u_{mn}\) and for the contraction \(\delta_{mn}\) generated from ensemble averaging.
(e) We prove that at large \(q\), the leading local diagrams form the ``Mickey Mouse'' topological equivalence class, where dashed ribbons denote a series of parallel contractions in (d)~\cite{supplementary}.
(f,g) Mickey Mouse diagrams contain temporal domain walls (tDW).  The red, blue and green vertical lines mark the tDW positions on the three folded-time loops.  Each tDW is labelled by \((a,b,c,d)\), where \((b,c,d)\) specify the folded-time loop coordinates, while \(a=0,1\) distinguishes the Gaussian and non-Gaussian tDW. A Gaussian tDW crosses a physical worldline represented by a single line, whereas a non-Gaussian  tDW crosses a unitary double line.
}
\label{fig:tdw_mickey_mouse}
\end{figure*}

Exploiting the doubled structure of the SFF, we define the topological
spectral form factor (TopSFF) by inserting a topological defect that acts non-trivially on the doubled partition function [Fig.~\ref{fig:schematic}],
\begin{equation}
\label{eq:K_top}
K_{\mathrm{top}}\!:=\! 
\Tr\!\left[
\hat{\mathcal D}
\big(\hat U_{p}\otimes \hat U_{p}^{*} \big)
\right],
\,\,\,
\left[
f_{\mathrm d}\big(\hat{\mathcal D}\big),
f_{\mathrm q}\big( \hat U_{p}\otimes \hat U_{p}^{*}\big)
\right] \!=0 .
\end{equation}
Here \(\hat{\mathcal D}\) is a defect operator and
\(\hat U_{p}\otimes \hat U_{p}^{*}\) is the doubled many-body evolution
operator, both acting on the doubled Hilbert space. In general, extended topological defects may be inserted at arbitrary locations, orientations and slices in spacetime. We focus on two representative orientations: spatially extended defects with \(p=\mathrm{t}\), which probe the interplay between topological defects
and unitary time evolution; temporally extended defects with  \(p=\mathrm{s}\), which probe the generally non-unitary spatial propagation.  
The commutation relation encodes the local pulling-through
condition responsible for the topological invariance of the TopSFF.  The model-dependent maps \(f_{\mathrm d}\) and \(f_{\mathrm q}\) select, respectively, the defect component and the corresponding component of the doubled evolution that commute locally.
In this manuscript, we focus on the \textit{minimal} non-trivial \(\mathbb{Z}_2\) defect \(\hat{\mathcal D}=\hat O\otimes\hat\iden\), implemented by the global swap operator \(\hat O = \hat S\), inserted into the doubled evolution of generic quantum many-body chaotic circuits \( \hat{U}(t,L)\)  [Fig.~\ref{fig:tdw_mickey_mouse}(a)]. In one dimension, acting on \(L\) qudits,
\(\hat S\ket{a_1,a_2,\ldots,a_L}=\ket{a_L,\ldots,a_2,a_1}\). \(f_{\mathrm{d}}\) is the identity map, while \(f_{\mathrm{q}}(\hat{U}_p (t,L) \otimes \hat{U}_p^{*}(t,L)) \) selects the time evolution operator (orange) for \(p=\mathrm{t}\), or the dual transfer matrix (purple) for \(p=\mathrm{s}\) in Fig.~\ref{fig:tdw_mickey_mouse}(a). The commutation relation is local in both space and time after the folding procedure discussed below in Fig.~\ref{fig:tdw_mickey_mouse}(b). 
We emphasize that the asymmetric and non-trivial insertion into the doubled partition function is needed to create mismatched world-sheet topologies -- a non-orientable Klein bottle and an orientable torus [Fig.~\ref{fig:schematic}(b)] -- thereby inducing a topological defect in the pairing structure of Feynman paths; see also \cite{Garratt2021prx}.
Algebraically, if \([\hat S,\hat U]=0\), defining
\(\hat\Pi_\pm=(\hat\iden\pm\hat S)/2\) and
\(Z_\pm=\Tr[\hat\Pi_\pm\hat U^t]\), the ordinary SFF gives
\(K=|Z_+|^2+|Z_-|^2 + \mathrm{Re}(Z_+Z_-^*) \), while the TopSFF with \((\hat D=\hat S\otimes\hat 1)\) gives
\(K_{\topo}=(Z_+-Z_-)(Z_+^*+Z_-^*)
=|Z_+|^2-|Z_-|^2+2i \mathrm{Im}(Z_+Z_-^*)\).
Thus, TopSFF accesses the imaginary cross sector interference, invisible to the regular SFF. In contrast, with the matched insertion \(\hat{\mathcal D}=\hat S\otimes\hat S\), the TopSFF reduces to the ordinary SFF of a system with a twisted boundary condition, and does not generate mismatched pairing topologies and emergent single-particle defect dynamics described below, see \cite{supplementary} and \cite{nakai2025discretetimecrystalsdetected}. 
Equation~\eqref{eq:K_top} generalizes directly to higher replica defect insertions and observables beyond the SFF; more general topological defects will be discussed in forthcoming work~\cite{upcoming}.
After the Heisenberg scale, set by the Hilbert space dimension, the TopSFF in chaotic systems displays RMT plateau behaviour, whose explicit solution is given in~\cite{supplementary}. From now on, we focus instead on universal many-body effects beyond RMT in TopSFF in the pre-Heisenberg regime.

\begin{figure*}[ht!]
    \centering
    \includegraphics[width=1\linewidth]{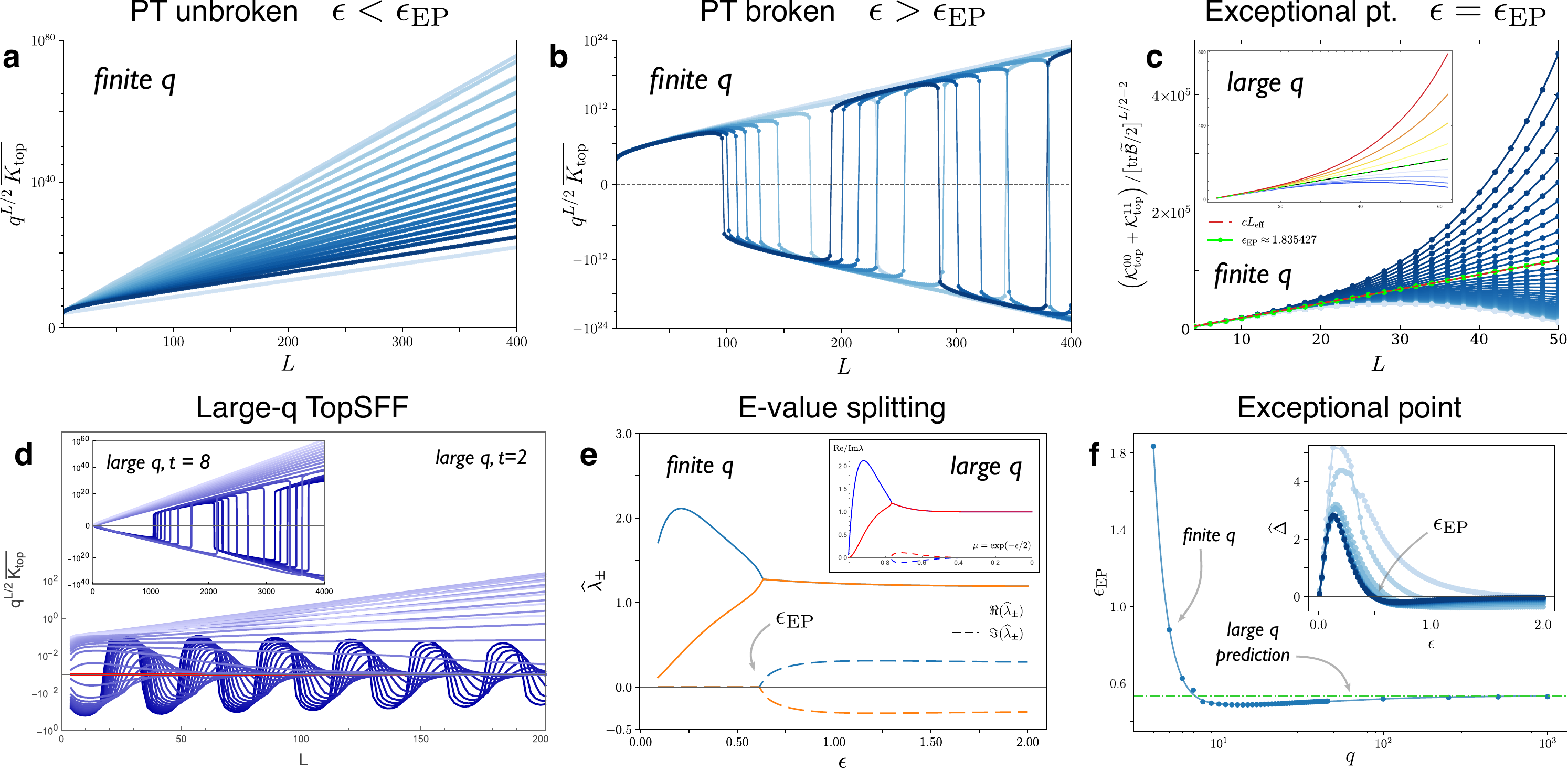}
    \caption{\textbf{Finite-\(q\) and large-\(q\) signatures of TopSFF \(\mathcal{PT}\) transition.} We exactly compute TopSFF in the \(k=(\pi,\pi,0)\) sector, with finite-\(q\) data at \(t=4\) and with large-\(q\) results for arbitrary \(t\) in the pre-Heisenberg regime.
    Darker lines denote stronger interaction strength. 
(a) In the \(\mathcal{PT}\) unbroken phase, the TopSFF at \(q=4\) displays exponential growth with system size \(L\).  
(b) In the \(\mathcal{PT}\) broken phase, TopSFF at \(q=4\) exhibits oscillatory behaviour with an exponential envelope. 
(c) At the exceptional point, eigenvector coalescence produces a Jordan-block contribution, giving the boundary-resolved TopSFF a  prefactor linear in system size at \(q=4\). The inset shows the corresponding large-\(q\) result. 
(d) Large-\(q\) exact TopSFF displays \(\mathcal{PT}\) transition at $t=2$ (main panel) and generally larger $t$, e.g. 
\(t=8\) (inset).
(e) Eigenvalue splitting across the transition at \(q=6\). The real eigenvalues in the \(\mathcal{PT}\) unbroken phase coalesce at the exceptional point and become a complex-conjugate pair in the \(\mathcal{PT}\) broken phase. 
Solid and dashed curves denote real and imaginary parts, respectively. 
Remarkably, the finite-$q$ spectrum is already qualitatively consistent with the large-$q$ spectrum shown in the inset and, with increasing $q$, converges toward the large-$q$ prediction~\cite{supplementary}.
(f) Finite-\(q\) exceptional points extracted from the Cayley-Hamilton discriminant converge towards the analytic large-\(q\) prediction as \(q\) increases. The inset shows the finite-\(q\) discriminant crossing used to locate the transition. 
Lastly, $q=2$ simulations of an independent model with non-diagonal coupling gates exhibit the oscillatory and monotonic TopSFF behaviours predicted by the large-$q$ RPM analysis~\cite{supplementary}.
}
    \label{fig:panel_finite_q_data}
\end{figure*}

\textbf{Path integral of temporal domain walls.}
We first consider the TopSFF of a spatially extended global swap defect
\(\hat{\mathcal D}=\hat S\otimes \hat\iden\) acting in generic quantum
many-body chaotic systems \(\hat U_{\mathrm{t}}(t,L)= \hat U^t\) with discrete time translational symmetry, realised by the parity inversion symmetric
Random Phase Model (RPM) with local Hilbert space dimension \(q\) (see Methods). 
The tensor network for \(K_{\mathrm{top}}=\Tr[\hat S\hat U^t]\,\Tr[\hat U^{*t}]\) contains two mismatched world-sheets [Fig.~\ref{fig:tdw_mickey_mouse}(c)]. Folding along the parity inversion axes stacks parity-related degrees of freedom, while unwrapping each folded site in the time direction exposes three folded-time loops: one long \(b\)-loop associated with
\(\Tr[\hat S\hat U^{t}]\), and two shorter loops, denoted the \(c\)- and \(d\)-loops, associated with \(\Tr[\hat U^{*t}]\).
We represent each local Haar-random unitary gate element \(u_{mn}\) of the RPM with two dots connected by a double line, representing the incoming and outgoing degrees of freedom \(m\) and \(n\) of \(u_{mn}\) [Fig.~\ref{fig:tdw_mickey_mouse}(d)]. Upon Haar averaging, the indices of local unitary gates and their complex conjugations are contracted with weights given by Weingarten functions~\cite{weingarten1978asymptotic, Brouwer1996, chan2018solution} [Fig.~\ref{fig:tdw_mickey_mouse}(d) right]. 
In the large-\(q\) limit, we prove that the leading local diagrams belong to the ``Mickey Mouse'' topological equivalence class
[Fig.~\ref{fig:tdw_mickey_mouse}(e)].  Each Mickey Mouse diagram contains a temporal domain wall (tDW), separating two locally distinct pairing domains along folded time: one with parallel contractions between the \(b\)- and \(c\)-loops, and the other with contractions between the \(b\)- and \(d\)-loops. The equivalence class is defined modulo independent cyclic rotations along the three folded-time loops, labelled \(b\), \(c\) and \(d\).  
There are two types of tDW.  A Gaussian tDW crosses the physical
worldline of the long \(b\)-loop, whereas a non-Gaussian tDW crosses the unitary double line, i.e. the two matrix indices of the associated unitary are contracted separately with \(c\)- and \(d\)-loops [Fig.~\ref{fig:tdw_mickey_mouse}(f,g)].  
A local tDW configuration is labelled by
\(\rho=(a,b,c,d)\), with \(a=0,1\) distinguishing the Gaussian and non-Gaussian species, and \((b,c,d)\in\{1,\ldots,2t\}\times\{1,\ldots,t\}^{2}\) specifying the three folded-time loop coordinates, giving \(4t^{3}\) configurations in total.
Importantly, the corresponding leading Weingarten weights differ by a
sign: \(+1\) for a Gaussian tDW and \(-1\) for a non-Gaussian tDW.
The interacting gates generate a generalized Boltzmann factor \(\mathcal B(\rho_1,\rho_2)\) between neighbouring folded sites. 
 Translational invariance along the three folded-time loops allows \(\mathcal B\) to be Fourier transformed  into  momentum (or frequency) sectors \(k=(k_b,k_c,k_d)\) as \(\widetilde{\mathcal B}\). Upon ensemble-averaging, the TopSFF can be evaluated as a discrete path integral of an emergent tDW propagating along folded space as a non-Hermitian single particle~\cite{supplementary},
\begin{equation}
\begin{aligned}
\label{eq:Ttilde_matrix}
\overline{\Knormalized_{\mathrm{top}}}= \sum_{k} \phi(-k)^T\Ttilde(k)^{\Leff}\phi(k), \,\,\,\, 
\Ttilde= \, 
\begin{pmatrix}
\Ttilde_{00} & \Ttilde_{01}\\
\Ttilde_{10} & \Ttilde_{11}
\end{pmatrix}
.
\end{aligned}
\end{equation}
Non-trivially, \(\Ttilde_{00},\Ttilde_{11}\in\mathbb R\), 
\(\Ttilde_{01},\Ttilde_{10}\in i \mathbb{R}\). 
The complex-valued generalized Boltzmann factor \(\Ttilde\), equivalently the transfer matrix, is written in the
two-species tDW basis, with \(\ket{\gauss}\) and \(\ket{\ngauss}\) denoting respectively the Gaussian and non-Gaussian state for each momentum sector. The matrix element \(\Ttilde_{a a'}(k)\) gives the transition amplitude from species \(a\) to species \(a'\).
We use the normalization \(\Knormalized_{\mathrm{top}}:=q^{L/2}K_{\mathrm{top}}\), and the effective propagation length is \(\Leff=L/2-1\).
\(\phi(k)\) is the physical boundary vector generated by the couplings along the parity-inversion axes. For period-one driving along the axes, the boundary vector is \(\phi=(1,i)^T\), an equal-weight superposition of Gaussian and non-Gaussian tDW sectors. The factor of \(i\) originates from the relative Weingarten sign. 
The trace and determinant of \(\Ttilde\) are always real, and therefore, the two eigenvalues are either real or a complex conjugate pair. Further, \(\Ttilde\) exhibits an emergent antiunitary \((\mathcal{PT}) \Ttilde(k) (\mathcal{PT})^{-1}=\Ttilde(-k)\), 
with complex conjugation \(\mathcal T\)  and \(\mathcal P=\sigma_z\) acting on the Gaussian/non-Gaussian tDW index \(a=0,1\). We prove that this antiunitary symmetry becomes a \(\mathcal{PT}\) symmetry~\cite{BenderBoettcher1998} for \(k_b+k_c+k_d=0\pmod{\pi}\)~\cite{supplementary}. The emergence of antiunitary symmetry intrinsically follows from the relative Weingarten sign of Gaussian and non-Gaussian tDWs: conversion between the two types of tDW carries a factor of \(i\).

In the large-\(q\) limit, for each momentum sector \(k\), the generalized Boltzmann factor is an exactly solvable two-by-two transfer matrix, with eigenvalues and eigenvectors obtained explicitly in~\cite{supplementary}. The eigenvalues are \(\lambda_{\pm}(k)
=
[\Ttilde_{00}(k)+\Ttilde_{11}(k)]/2
\pm
\sqrt{\Delta(k)}\), and real or complex conjugate pairs depending on the sign of discriminant
\begin{equation}
\label{eq:Ttilde_eigenvalues}
\Delta(k,\epsilon)
=
\underbrace{
\left[
\frac{\Ttilde_{00}-\Ttilde_{11}}{2}
\right]^2
}_{\text{detuning}}
+
\underbrace{
\Ttilde_{01}\Ttilde_{10}
}_{\text{conversion}}.
\end{equation}
Thus the spectral bifurcation, corresponding to the \(\PT\) transition in \(\PT\) symmetric sectors, is controlled by the competition between the Gaussian/non-Gaussian detuning \(|[\Ttilde_{00}-\Ttilde_{11})/2|^2\) and the negative-valued inter-species conversion term \(\Ttilde_{01}\Ttilde_{10}\).
Numerous momentum sectors of \(\Ttilde\) exhibit \(\PT\) transitions. Hereafter, we focus on the sector \(k=(\pi,\pi,0)\) for even \(t\) and \(k=(2\pi/t,-2\pi/t,0)\) for odd \(t\), which are distinguished by the coexistence of two features: (i) leading eigenvalues have moduli exceeding unity, giving the exponential growth with system size discussed below, and (ii) an exceptional point (EP) exactly solvable at large-\(q\) for both even and odd \(t\) as
\begin{equation}
\label{eq:epsilon_EP}
\epsilon_{\EP}
=
\log\!\left[
\frac{\sqrt{(t-1)^2+8t}-(t-1)}{2}
\right] \xrightarrow{t\gg 1}
\log 2,
\end{equation}
which remains \textit{finite} even at large \(t\), within the pre-Heisenberg regime of the large-\(q\) analytics. 
The special role of \((\pi,\pi,0)\) follows from the momentum structure of the exact large-\(q\) transfer matrix. This sector satisfies the \(\mathcal{PT}\) symmetric condition \(k_b+k_c+k_d=0 \pmod{\pi}\), while combining a zero momentum on one folded-time loop with alternating \(\pi\) phase factors on the other two loops. The zero momentum gives a non-oscillatory enhancement, whereas the alternating signs reduce the Gaussian--non-Gaussian detuning without eliminating their off-diagonal mixing. This balance allows \(\Delta(k,\epsilon)\) to change sign at finite \(\epsilon\) unlike, for example, the zero-momentum sector \((0,0,0)\) where the \(\pi\) phase factors are absent. 
The \((\pi,\pi,0)\) sector is accessible, for example, in ordinary time evolution by computing TopSFF with period-two driving specifically along the parity inversion axes~\cite{supplementary}. More directly, we gain access to  the \((\pi,\pi,0)\) and \((2\pi/t,-2\pi/t,0)\) sectors via non-unitary spatial propagation under the dual transfer matrix for small \(t\) and large \(L\)~\cite{chan2018spectral, chan2020lyap, Akila_2016, bertini2018exact} (Methods).

\textbf{Emergent \(\mathcal{PT}\) transition.} In the \(\mathcal{PT}\) unbroken phase, with
\(\epsilon<\epsilon_{\EP}\) and \(\Delta>0\), detuning dominates conversion in \eqref{eq:Ttilde_eigenvalues}, and the
eigenvalues of \(\Ttilde\) are real. For sufficiently large \(L\), the TopSFF is controlled by the eigenvalue of largest modulus and therefore grows or decays monotonically and exponentially according to a single dominant eigenmode,
\begin{equation}
\label{eq:Ktop_PT_unbroken}
\overline{\Knormalized_{\mathrm{top}}}
\propto
\lambda_{+}^{\Leff},
\qquad
\epsilon < \epsilon_{\EP},
\end{equation}
where \(|\lambda_{+}|>|\lambda_{-}|\). This leading eigenmode is generally a superposition of Gaussian tDW and non-Gaussian tDW, polarized predominantly towards one of them, i.e., \(\ket{\mathrm{G}}\) or \(\ket{\mathrm{nG}}\). We emphasize that this
\(\mathcal{PT}\) unbroken phase lies within the strongly
interacting chaotic regime. 

In the \(\mathcal{PT}\) broken phase, with
\(\epsilon>\epsilon_{\EP}\) and \(\Delta<0\), Gaussian--non-Gaussian conversion dominates detuning. The eigenvalues
form a complex conjugate pair, \(\lambda_{\pm}=|\lambda|e^{\pm i\theta}\), and the propagating tDW eigenmodes are hybridized Gaussian and non-Gaussian tDW superpositions, \(\ket{\pm} \propto (\ket{\mathrm{G}} \pm \ket{\mathrm{nG}})\).
The TopSFF is controlled by a coherent superposition of the two leading tDW eigenmodes. Under spatial propagation, these two modes acquire opposite phases, and their linear combination produces constructive and destructive interference as \(L\) is varied. Consequently, the TopSFF acquires an oscillatory dependence on system size with an exponential envelope, 
\begin{equation}
\label{eq:Ktop_PT_broken}
\overline{\Knormalized_{\mathrm{top}}}
\propto
|\lambda|^{\Leff}\sin(\Leff \theta+\varphi),
\qquad 
\epsilon> \epsilon_{\EP},
\end{equation}
with a phase offset \(\varphi\) fixed by the boundary vectors.

At the exceptional point, \(\epsilon=\epsilon_{\EP}\) and \(\Delta=0\), the eigenvalues become degenerate, \(\lambda_{+}=\lambda_{-}\equiv\lambda_{\EP}\), and the corresponding eigenvectors coalesce, rendering \(\Ttilde\) non-diagonalizable. Writing \(\Ttilde(k)|_{\EP}=\lambda_{\EP}\iden+N\), with \(N^2=0\), gives \(\Ttilde(k)^{\Leff}|_{\EP} = \lambda_{\EP}^{\Leff}\iden + \Leff\lambda_{\EP}^{\Leff-1}N\). The second term is the characteristic Jordan-block contribution, producing a polynomial enhancement in \(\Leff\).
For the boundary vector above, \(\phi=(1,i)^T\), an equal weight superposition of Gaussian and non-Gaussian tDW sectors, the Jordan contribution has vanishing overlap at the exceptional point, and therefore does not directly reveal the non-diagonalizability. We instead resolve the TopSFF into Gaussian and non-Gaussian boundary components via \(\overline{\Knormalized_{\spatdef}^{mn}} := \bigl[\phi^{(m)}\bigr]^T \Ttilde(k)^{\Leff} \phi^{(n)}\), with \(m,n\in\{0,1\}\), \(\phi^{(0)}=(1,0)^T\) and \(\phi^{(1)}=(0,i)^T\). The exceptional point and its Jordan non-diagonality are then directly exposed through
\begin{equation}\label{eq:Ktop_EP}
\frac{\overline{\Knormalized_{\spatdef}^{00}}+ \overline{\Knormalized_{\spatdef}^{11}}}
    {[\tr \widetilde{\boltz}/2]^{\Leff -1}}
    \propto   \Leff \, ,
    \qquad \epsilon = \epsilon_{\EP}.
\end{equation}
At the EP, dividing by the degenerate eigenvalue removes the exponential envelope and isolates the Jordan prefactor.

\textbf{Evidence from simulations.}
We now present analytical and numerical evidence for \(\mathcal{PT}\) transition signatures in the large-\(q\) theory, and show that signatures persist at finite \(q\). In the large-\(q\) theory, Fig.~\ref{fig:panel_finite_q_data}  displays the exact TopSFF across the \(\PT\) transition [(d)], confirming \eqref{eq:Ktop_PT_unbroken} and \eqref{eq:Ktop_PT_broken}; at the EP  [(c) inset], confirming \eqref{eq:epsilon_EP} and \eqref{eq:Ktop_EP}; and the associated eigenvalue splitting [(e) inset and (f)]. 
 The EP is further accompanied by divergent eigenvector overlap diagnostics, including the Petermann factor~\cite{petermann1979}, reflecting the non-orthogonality of the coalescing eigenvectors~\cite{supplementary}. 
At finite \(q\), we devise exact simulations of TopSFF  at \(t=4\), varying large \(L\) (Methods).
 In the strongly interacting regime, the Cayley-Hamilton (CH) extraction (Methods) reveals the same three finite-\(q\) signatures in Fig.~\ref{fig:panel_finite_q_data}: exponential dependence in the \(\mathcal{PT}\) unbroken regime [(a)], oscillations in the \(\mathcal{PT}\) broken regime [(b)], and a linear prefactor in \(L\) from Jordan non-diagonality at the EP [(c)]. 
Remarkably, the extracted \(\epsilon_{\EP}\) converges well to the large-\(q\) prediction as \(q\) increases [(f)], while the eigenvalue splitting shows qualitative agreement with the large-\(q\) behaviour [(e)].
Although the finite-\(q\) numerical simulations here are restricted to \(t=4\), their non-trivial agreement with the large-\(q\) theory supports the same underlying \(\PT\) transition mechanism at finite \(q\).
Crucially, the exact large-\(q\) analysis shows that these signatures persist to arbitrarily large \(t\) within the pre-Heisenberg regime, e.g. \(t=8\) in Fig.~\ref{fig:panel_finite_q_data}(d) inset and~\cite{supplementary}, with a finite interaction strength transition and an exponentially increasing TopSFF.
Importantly, $q=2$ simulations of an independent model, the Brick Wall Model (BWM)~\cite{supplementary}, which, unlike the RPM, is composed of non-diagonal gates, qualitatively reproduce the exact large-$q$ RPM predictions, exhibiting the predicted momentum-sector-dependent oscillatory and monotonic TopSFF behaviours. A tunable $\PT$ transition cannot be accessed in the simplest form of the BWM, since it lacks an interaction-strength tuning parameter. Nevertheless, this nontrivial agreement demonstrates that the TopSFF signatures obtained in the large-$q$ RPM are not specific to the RPM or to the large-$q$ limit.

\begin{figure}
    \centering
    \includegraphics[width=0.99\linewidth]{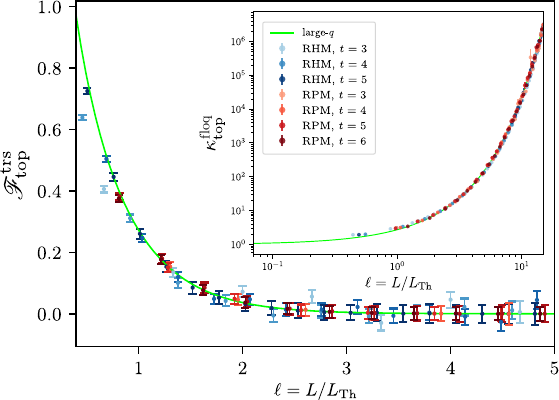}
\caption{
\textbf{Finite-\(q\) TopSFF collapse for temporally extended defects.}
Topological free energy \(\mathscr F^{\mathrm{trs}}_{\topo}\) of a temporally extended global swap defect in the TRS RPM at \(q=3\) (red), and the TRS Random Hamiltonian Model (RHM)~\cite{supplementary} at \(q=2\) (blue), compared with the large-\(q\) scaling form (green). 
Inset: finite-\(q\) collapse of the TopSFF \(\mathcal K^{\mathrm{floq}}_{\topo}\) with temporally extended translation defect for the corresponding Floquet models (without TRS).
}
\label{fig:topSFF_temporal_defect_scaling}
\end{figure}

\textbf{Temporally extended defects.} 
We next study TopSFFs with temporally extended defects \(\hat{\mathcal{D}}\) inserted in the spatial evolution \(\hat U_{\mathrm{s}}\otimes\hat U_{\mathrm{s}}^{*}\), where
\(\hat U_{\mathrm{s}}(t,L)=\prod_{i=1}^{L}\hat V_i(t)\) is generated by the
space-time dual transfer matrices of the RPM [Fig.~\ref{fig:tdw_mickey_mouse}(a), purple]. 
We consider \(\hat{\mathcal D}=\hat O\otimes\hat\iden\), with \(\hat O=\hat S\) the
global swap defect of the time reversal symmetric (TRS) RPM, or \(\hat O=\hat T\)
the time translation defect of the Floquet RPM with
\(\hat T\ket{a_1,a_2,\ldots,a_t}
=\ket{a_t,a_1,\ldots,a_{t-1}}\).
For the global swap case, at large \(q\), the averaged TopSFF is equivalent to an Ising partition function with a spin-flip defect. This defect forces a spatial domain wall (sDW), i.e. a kink in the pairing or Ising variable along the spatial direction, in the sense of Ref.~\cite{Garratt2021prx}.
The partition function evaluates to 
\(
\overline{K^{\mathrm{trs}}_{\mathrm{top}}}
=
\bigl[1+e^{-\epsilon\alpha(t)}\bigr]^L
-
\bigl[1-e^{-\epsilon\alpha(t)}\bigr]^L\), with 
\(\alpha(t):=t-1-\delta_{0,t\,{\rm mod}\,2}\). 
The ordinary SFF is obtained by replacing the minus sign by a plus sign~\cite{wu2025}. 
In generic many-body chaotic systems, the Thouless length \(\Lth\) is the crossover scale separating the 0D RMT regime, \(L\ll \Lth\), from the many-body regime, \(L\gtrsim \Lth\).
In the Thouless double scaling limit~\cite{chan2021trans} where \(t,L \to \infty\) at  fixed \(\ell:=L/\Lth = L e^{-\epsilon t}\) fixed,  an exact universal scaling form of TopSFF can be derived at large \(q\),
\begin{equation}
\kappa^{\mathrm{trs}}_{\topo}
:=
\lim_{\substack{L,t\to\infty\\ \ell=L/\Lth}}
\overline{K^{\mathrm{trs}}_{\topo}}
=
2\sinh\ell .
\label{eq:tsff_gtrs_scaling}
\end{equation}
The analogous scaling form of SFF is derived in \cite{wu2025}. The effective free energy cost of the spatial domain wall (sDW) is then
\(
\Delta F^{\mathrm{trs}}_{\topo}
:=
-\frac{1}{t}
\log
\big(\overline{K^{\mathrm{trs}}_{\mathrm{top}}}/ \overline{K^{\mathrm{trs}}}\big)
=
-\frac{1}{t}
\log
\!\left[
\tanh\!\left(L\,\operatorname{atanh}e^{-\epsilon\alpha(t)}\right)
\right].
\)
In the Thouless scaling limit, 
the TopSFF free energy scaling form at large \(q\) is
\be
\mathscr{F}^{\mathrm{trs}}_{\topo}\equiv \lim_{\substack{L,t \to \infty\\ \ell \equiv  L/\Lth}}  t\, \Delta F^{\mathrm{trs}}_{\topo}= - \ln \tanh \ell \,.
\label{eq:trs_scaling}
\ee

For the time translation case, analogously, the averaged TopSFF is the \(t\)-state Potts partition function~\cite{chan2018spectral} with color-flip defect,
forcing an sDW in the Potts color or pairing structure.
It evaluates as \(\overline{K^{\mathrm{floq}}_{\topo}} = 
\left[1+ (t-1) e^{-\epsilon t}\right]^L
-(1- e^{-\epsilon t})^L\). Taking the Thouless scaling limit gives the TopSFF universal scaling form 
\be\label{eq:floq_tsff_scaling}
\kappa^{\mathrm{floq}}_{\topo} :=  \lim_{\substack{L,t \to \infty\\ \ell \equiv  L/\Lth}}  \! \left(\, \overline{K^{\mathrm{floq}}_\topo}  +1 \right) = e^\ell  .
\ee
with  \(\Lth =e^{\epsilon t}/t \). The analogous scaling form for the ordinary SFF is derived in \cite{chan2021trans}. Again, comparing TopSFF and SFF gives the topological free energy difference
\(
\Delta F^{\mathrm{floq}}_{\topo}
:=
-\frac{1}{t}
\log
\big(\overline{K^{\mathrm{floq}}_{\mathrm{top}}}/ \overline{K^{\mathrm{floq}}}\big)
 =  -\frac{1}{t} \ln \, \frac{
1-\omega^L 
}{
1 +  (t-1) \omega^L} \) with \(\omega =  (1- e^{-\epsilon t})/[1+ (t-1) e^{-\epsilon t}]\). In the large-\(q\) limit, the universal scaling
form of TopSFF with time translation defect is
\begin{equation}\label{eq:time_trans_scaling}
\mathscr{F}^{\mathrm{floq}}_{\topo}\equiv 
\lim_{\substack{L,t\to\infty\\ \ell \equiv L / \Lth }} \!
\left(
t\Delta F^{\mathrm{floq}}_{\topo}-\ln t
\right)
=
-\ln(e^\ell-1),
\end{equation}
 capturing the free energy cost of sDW \cite{chan2018spectral, Garratt2021prx}.
The convergence to Eq.~\eqref{eq:time_trans_scaling} requires \(\ell\ll\ln t\). Since \(\ln t\) grows only logarithmically, this window is narrow in the numerically accessible regime, and we test the TopSFF scaling form~\eqref{eq:floq_tsff_scaling}, rather than the free energy scaling form~\eqref{eq:time_trans_scaling}. In Fig.~\ref{fig:topSFF_temporal_defect_scaling}, finite-\(q\) simulations of two independent models (Methods) show excellent agreement with the universal free energy scaling form for the TRS case~\eqref{eq:trs_scaling} and with the universal TopSFF scaling form for the time-translation case~\eqref{eq:floq_tsff_scaling}.

The TopSFF free energy scaling forms capture the universal sDW physics of generic many-body chaotic systems  beyond the 0D RMT regime~\cite{chan2018spectral, Garratt2021prx, wu2025}.
\eqref{eq:trs_scaling} identifies Thouless length \(L_{\Th}\) as the sole remaining length scale, acting as the effective correlation length of the sDW theory. For \(L\ll L_{\Th}\), 
\(\mathscr{F}^{\mathrm{trs}}_{\topo}\simeq \ln(L_{\Th}/L)\), i.e., a forced sDW is rare and costly, and the system behaves as a 0D
chaotic system described by a single random matrix. For \(L\gg L_{\Th}\), proliferating sDW screen the defect insertion, causing the free energy to vanish exponentially and producing universal many-body deviations from the 0D RMT behaviour. The time translation defect has the same interface free energy interpretation, but the sDW also carries a \(t\)-fold entropy, giving the \(\ln t\) subtraction in
Eq.~\eqref{eq:time_trans_scaling}.

\textbf{Discussion.} 
Generic strongly interacting quantum many-body systems are notoriously difficult to analyse. The naturalness of the TopSFF formulation as a defect-decorated doubled partition function enables exact analytic access to universal non-perturbative signatures that are not resolved by conventional observables such as the SFF.
These signatures include an emergent \(\mathcal{PT}\) transition within the generic many-body chaotic regime and universal scaling forms of topological-defect free energy, whose generality is corroborated across multiple models and for both temporally and spatially extended defects.
More broadly, our results expose an interplay between unitary evolution in time and non-unitary propagation in space. In particular, the tDW path integral provides an unexpected bridge between Hermitian many-body quantum chaos and emergent non-Hermitian single-particle physics.

The universal TopSFF signatures are intrinsically \textit{many-body} in two complementary senses. First, they arise from collective temporal and spatial DW degrees of freedom before the system reduces to the zero-dimensional RMT regime beyond the Thouless scale. Second, the Gaussian--non-Gaussian interference underlying the spatially extended TopSFF has no non-trivial zero-dimensional RMT counterpart. Although the same Mickey Mouse contraction topology occurs in the CUE, its leading Gaussian and non-Gaussian contributions cancel, with
$\overline{\Tr U^{2t}(\Tr U^{*t})^2}=0$ 
exactly for matrix size $N\geq 2t$~\cite{supplementary}. Many-body interactions lift this cancellation, generating distinct propagation amplitudes in the non-Hermitian tDW path integral, and inducing a finite-interaction-strength $\mathcal{PT}$ transition and exceptional point driven by many-body competition.
As a complementary test, a classically simulable Clifford/stabilizer TopSFF counterpart exhibits only an asymptotic \(\mathcal{PT}\) transition, lacking the finite-interaction-strength Gaussian--non-Gaussian competition found here (Methods).

This work opens many directions for future studies~\cite{upcoming}. Beyond conventional symmetries such as \(U(1)\), generalized symmetry defects and defect junctions define a rich class of defect-resolved observables~\cite{Gaiotto2015gensym,aasen2020topological}. The emergent non-Hermitian generalized Boltzmann factors also call for field-theoretic descriptions, tDW trajectory analysis, and a topological classification of exceptional points and topological invariants. TopSFF in the Sachdev-Ye-Kitaev model~\cite{Sachdev_1993,kitaev2015simple} could potentially connect defect-resolved observables to dual gravitational descriptions. 
Finally, TopSFF may be experimentally accessible by adapting existing measurement protocols~\cite{Poulin_2003,ZollerSFF2021,sff_exp_2025} to incorporate non-trivial defect insertions in the doubled evolution.

\section*{Methods}

\textbf{Models.} The Random Phase Model~\cite{chan2018spectral} is a one-dimensional random quantum circuit~\cite{Nahum2017} acting on \(L\) qudits of local Hilbert-space dimension \(q\). A single RPM bilayer is \(\hat{U}_{\RPM}=\hat{U}_2 \hat{U}_1\), where \(\hat{U}_1=\bigotimes_{r=1}^{L}u_r\) is a layer of independent Haar-random one-site gates \(u_r\in\mathrm{CUE}(q)\), and \(\hat{U}_2=\prod_r\Theta_{r,r+1}\) is a layer of commuting two-site diagonal phase gates, containing Gaussianly distributed random phases. In the computational basis, we have
\(
[\Theta_{r,r+1}]^{cd}_{ab}
=
\delta_{ac}\delta_{bd}e^{i\phi^{(r)}_{ab}}\) and 
\(\phi^{(r)}_{ab}\in \mathcal N(0,\epsilon).
\)
The variance \(\epsilon\) controls the interaction strength. To verify the generality of our results, we also use the Random Hamiltonian Model (RHM)~\cite{Bensa2022phantom, huang2023outoftimeorder} and the Brick Wall Model (BWM)~\cite{chan2018solution, Nahum2017} with the corresponding symmetry constraints, defined explicitly in \cite{supplementary}. 

To impose parity inversion symmetry, gates related by \(r\leftrightarrow L-r+1\) are identified. The one-site gates are mirrored across the inversion axis, and the two-site phase gates are chosen in mirror-related pairs. Gates straddling the inversion axis are locally parity symmetric, which for the diagonal RPM gate means \(\phi_{ab}=\phi_{ba}\). 
The TRS RPM~\cite{wu2025} is obtained by imposing the corresponding reflection constraint in the time direction. We work in the \(\mathcal T^2=1\) class and choose a basis such that  \(\mathcal T\) is the complex conjugation. A random half-circuit \(w(t,L)\) is sampled from the above RPM ensemble and the TRS circuit is generated via transposition as \(\hat U(t,L)=w(t,L)^T w(t,L)\), so that \(\hat U^T=\hat U\) and \(\mathcal T \hat U(t)\mathcal T^{-1}=\hat U^\dagger(t)\). See explicit definition in \cite{supplementary}.

\textbf{Large-\(q\) exact TopSFF.} Using the Cayley-Hamilton relation, for the boundary vector 
\(\varphi=(\varphi_1,\varphi_2)^T\) supported in momentum sector \(k\), the TopSFF is exactly evaluated at large \(q\) as
\begin{equation}\label{eq:Kvarphi_general}
\overline{\Knormalized_{\topo}}
=
\varphi^T\widetilde{\boltz}(k)^{\Leff}\varphi
=
g \, u_{\Leff}
-
h\det\widetilde{\boltz}\,u_{\Leff-1},
\end{equation}
where
\(
h=\varphi_1^2+\varphi_2^2\),
and
\(g(k)=
\Ttilde_{00}\varphi_1^2
+
\bigl[\Ttilde_{01}+\Ttilde_{10}\bigr]\varphi_1\varphi_2
+
\Ttilde_{11}\varphi_2^2\).
The system size dependence is encoded in
\begin{equation}
u_{\Leff}(k)=
\begin{cases}
\dfrac{\lambda_+^{\Leff}(k)-\lambda_-^{\Leff}(k)}
{\lambda_+(k)-\lambda_-(k)},
& \DeltaDisc(k)\neq 0,\\
\Leff \! \left[\dfrac{\tr \Ttilde(k)}{2}\right]^{\! \Leff-1},
& \DeltaDisc(k)=0.
\end{cases}
\end{equation}
where \(\lambda_\pm\) and \(\DeltaDisc(k)\) are defined around \eqref{eq:Ttilde_eigenvalues}. For \(\DeltaDisc(k)\neq0\), the eigenvalues \(\lambda_\pm(k)\) are distinct. At the EP, \(\DeltaDisc(k)=0\), they coalesce, and the second line gives the corresponding Jordan-block contribution. Explicit exact large-\(q\) expressions for the generalized Boltzmann factor \(\Ttilde(k)\) in all momentum sectors are given in~\cite{supplementary}.

\textbf{\((\pi,\pi,0)\) momentum sector.}
For the momentum sectors \(\piedge \equiv(\pi,\pi,0)\) in even \(t\) and \((2\pi/t,-2\pi/t,0)\) in odd \(t\), the large-\(q\) generalized Boltzmann factor is a \(\PTsym\) dimer, and displays a \(\PTsym\) phase transition at finite interaction strength. For concreteness, we supply here the exact Boltzmann factor for \(\piedge\equiv(\pi,\pi,0)\) for even \(t\). The odd \(t\) case can be  obtained analogously~\cite{supplementary}.
To write the generalized Boltzmann factor at \(\pi,\pi,0\), define 
\(S_2(n):=(1-\mu^{-2n})/(1-\mu^{-2})\),
\(S_4(n):=(1-\mu^{-4n})/(1-\mu^{-4})\),
\(A_n:=\sum_{r=1}^n(\mu^{-2r}-1)^2
=\mu^{-4}S_4(n)-2\mu^{-2}S_2(n)+n\), and
\(C_n:=A_n-(\mu^{-2}-1)^2\). Then, the exact solution of \(\Ttilde\) is given by
\begin{equation} \nonumber
\begin{aligned}
\widetilde{\mathcal B}^{\piedge}_{00}={}&
1+\mu^{4t-4}
+(t-2)(\mu^{2t}+\mu^{4t-2})
\\
& 
-2(t-1)\mu^{4t+2}
+2t\mu^{4t-4}(t-1) S_2(t-1) \\
&
-2t(t-1)\mu^{4t}
+2\mu^{4t}C_t
+2\mu^{4t}(1-\mu^2),
\\[3pt]
\widetilde{\mathcal B}^{\piedge}_{01}
=
\widetilde{\mathcal B}^{\piedge}_{10}
={}&
i\mu^{4t}\Big[
(\mu^{-2t}-1)^2+(\mu^{-2}-1)^2
 \\
&
+
t\{(\mu^{-2t}-1)+(\mu^{-2}-1)
\\
& 
+(\mu^{-2}+\mu^{-4})(t-1) S_2(t-1)
-2(t-1)\} 
\\
&
+A_{t-1}+C_t
\Big],
\\[3pt]
\widetilde{\mathcal B}^{\piedge}_{11}={}&
-\mu^{4t}\Big[
(\mu^{-2t}-1)^2+2A_{t-1}
\\
&
+
t\{(\mu^{-2t}-1)
\\
& +2\mu^{-2}(t-1) S_2(t-1)-2(t-1)\}
\Big],
\end{aligned}
\end{equation}
with interaction strength parametrised as \(\mu \equiv e^{-\epsilon/2}\).

\textbf{Finite-\(q\) simulations and analysis.} At finite \(q\), we compute the TopSFF by two complementary methods. First, we estimate the TopSFF by sampling over random circuit realizations. Second, by constructing a transfer matrix after the ensemble average of the RPM tensor network, we exactly compute TopSFF for small \(t\) and arbitrary \(L\), free from sample-to-sample fluctuations~\cite{supplementary}.
Motivated by the exact large-\(q\) two-mode structure and the persistence of the \(b\)-, \(c\)-, and \(d\)-loop translational symmetries at finite \(q\), we diagnose the \(\mathcal{PT}\) transition using the large-\(q\) Cayley-Hamilton (CH) relation as an ansatz for finite-\(q\), \(\overline{\mathcal{K}_{\topo}}(\Leff+2)=\widehat \Trb \, \, \overline{\mathcal{K}_{\topo}}(\Leff+1)-\widehat \Detb \, \,\overline{\mathcal{K}_{\topo}}(\Leff)\). From finite-\(q\) TopSFF data, we extract the effective trace \(\widehat t \), determinant \(\widehat d \), discriminant
\(\widehat\Delta:=\widehat \Trb^{\,2}-4\widehat \Detb\), and eigenvalues
\(\widehat\lambda_{\pm}:=(\widehat \Trb \pm\sqrt{\widehat\Delta})/2\). 
In finite \(q\), the CH relation is overdetermined over a range of \(\Leff\). In practice, the extracted quantities \(\widehat \Trb= \widehat \Trb_{\Leff}\) and \(\widehat\Detb=\widehat\Detb_{\Leff}\) depend on \(\Leff\), and are interpreted as effective trace and determinant only when they are stable upon varying \(\Leff\). In SI~\cite{supplementary}, we show empirically that, for \(q\geq 4\), the effective two-mode CH relation is accurate over an extended region of the \(\epsilon\)-\(L\) plane including across the finite interaction strength EP \(\epsilon_\EP\) and in both chaotic and non-chaotic regions. Note conversely that two-mode CH validity does not imply chaos.
For the accessible even (odd) \(t\) cases, the TopSFF requires more than two modes when \(q<4\) (\(q<6\), respectively), a regime we leave for future study.
Lastly, remarkably, $q=2$ simulations of an independent model with non-diagonal coupling gates, the BWM, already exhibit the oscillatory and monotonic TopSFF behaviours predicted by the large-$q$ RPM analysis~\cite{supplementary}.

Can the \(\mathcal{PT}\) transition in the TopSFF be classically simulated? To test this, we exactly compute the TopSFF with a spatially extended swap defect in an RPM without time translation symmetry, where each random local unitary appears only twice in the time evolution. Although the TopSFF of this model can be reproduced by classically simulable Clifford/stabilizer circuits~\cite{aaronson2004improved}, it exhibits only an asymptotic \(\mathcal{PT}\) transition~\cite{supplementary}, lacking the finite-interaction-strength competition between Gaussian--non-Gaussian detuning and conversion found in the time translation symmetric case in Eq.~\eqref{eq:epsilon_EP}. 

\textbf{The second Thouless double scaling limit.} 
In addition to the exact universal TopSFF scaling obtained in the Thouless double-scaling limit~\cite{chan2021trans}, where \(t,L\to\infty\) with \(\ell:=L/L_{\Th}=Le^{-\epsilon t}\) held fixed, a complementary limit with \(\tau:=t/t_{\Th}=\epsilon t/\ln L\) held fixed yields the exact scaling function
\be\label{eq:trs_2thoulessform}
\mathscr F^{'\mathrm{trs}}_{\topo}\equiv \lim_{\substack{L,t \to \infty\\ \tau  \equiv  t/\tth}}  \Delta F_{\topo} = 
\begin{cases}
0 \quad  & \tau  \leq 1 \,,
\\
\epsilon \left(1- \tfrac{1}{\tau } \right) & \tau  > 1\,.
\end{cases}
\ee
 Thus, in this Thouless scaling limit, the universal scaling form has  a sharp crossover at $\tau =1$ or equivalently at $t = \tth$.
Finite-size convergence is parametrically slower in this second Thouless scaling limit than in the first~\cite{supplementary}. At fixed \(\ell=Le^{-\epsilon t}\), the first scaling form receives power-law corrections, \(t\Delta F_{\topo}=\mathscr F^{\mathrm{trs}}_{\topo}(\ell)+O(L^{-2})\). In contrast, at fixed \(\tau=\epsilon t/\ln L\), one has \(Le^{-\epsilon t}=L^{1-\tau}\), giving 
\(\Delta F_{\topo}=\mathscr F^{'\mathrm{trs}}_{\topo}(\tau)+O(1/\ln L)\) with a crossover width \(\delta\tau\sim1/\ln L\).

The two branches in Eq.~\eqref{eq:trs_2thoulessform} distinguish the intrinsically many-body regime from the RMT regime. For \(\tau<1\) with \(t<t_{\Th}\) and \(L > L_{\Th}(t)\), spatial domain walls proliferate and screen the defect insertion, giving \(\Delta F_{\topo}\to0\). This is the pre-Thouless regime responsible for the many-body deviations from RMT that produce the SFF bump. For \(\tau>1\) with \(t>t_{\Th}\) and \(L < L_{\Th}(t)\), spatial domain walls are dilute and the defect carries a finite interface free energy cost. This corresponds to the onset of RMT behaviour and the ramp, with \(\Delta F_{\topo}\to\epsilon\) at late times.

\section*{Acknowledgments.}
We are grateful to Subhajyoti Bid, Weitao Chen, Ilyoun Na, Hisanori Oshima, Rémy Dubertrand and, particularly, David Huse for helpful discussions. AC thanks John Chalker, Andrea De Luca, and Saumya Shivam for previous collaboration on related works. CYL acknowledges financial support by the Engineering and Physical Sciences Research Council [EP/Y023005/1] (PI: Remy Dubertrand). The numerical simulation was carried out using both Lancaster University's High End Computing facility and the Higgs HPC system at Northumbria University.

\bibliography{biblio.bib}

\onecolumngrid
\newpage

\clearpage

\appendix

\setcounter{equation}{0}
\setcounter{figure}{0}
\renewcommand{\thetable}{S\arabic{table}}
\renewcommand{\theequation}{S\thesection.\arabic{equation}}
\renewcommand{\thefigure}{S\arabic{figure}}
\setcounter{secnumdepth}{3}

\begin{center}
{\Large Supplementary Information\par}

\vspace{0.35cm}

{\Large \titleinfo \par}

\vspace{0.45cm}

{\large
Daniel Harkin\(^{1}\), Chun Y. Leung\(^{2,1}\), and Amos Chan\(^{3,1}\)
\par}

\vspace{0.25cm}

{\small
\(^{1}\)Department of Physics, Lancaster University, Lancaster LA1 4YB, United Kingdom\\
\(^{2}\)Department of Mathematics, Physics and Electrical Engineering,
Northumbria University, Newcastle upon Tyne NE1 8ST, United Kingdom\\
\(^{3}\)Department of Physics, University of Warwick, Coventry CV4 7AL, United Kingdom
\par}

\end{center}

\vspace{0.5cm}

\section*{Overview}

This Supplementary Information is organized as follows. 
Appendix~\ref{app:model} defines the random circuit models and their parity inversion, time reversal, and time translation symmetric variants. Appendices~\ref{app:tsff_rmt} to \ref{app:sff_twisted} collect exact results for TopSFF in RMT, SFF without topological defects, and SFFs with twisted boundary conditions. 
Appendix~\ref{app:tsff_toolbox} derives the exact solution of generalized Boltzmann factor for the TopSFF with a spatially extended swap defect at large \(q\), including momentum-space formulation, Lemma~\ref{app_lemma:leadingtsff}  that leading local diagrams  of temporal domain wall belong to the Mickey Mouse topological equivalence class, Proposition~\ref{prop:PT_classification_full} on \(\mathcal{PT}\) symmetry, and boundary states. 
Appendix~\ref{app:tsff_toolbox} derives the exact large-\(q\) generalized Boltzmann factor for the TopSFF with a spatially extended swap defect. It includes Lemma~\ref{app_lemma:leadingtsff} showing that the leading local temporal domain wall diagrams belong to the Mickey Mouse topological equivalence class, and Proposition~\ref{prop:PT_classification_full} classifying the \(\mathcal{PT}\) symmetric momentum sectors.
Appendix~\ref{app:tsff_exact_q_spatial} gives the exact large-\(q\) evaluation of the spatially extended TopSFF. 
Appendix~\ref{app:PT_mapping} interprets the TopSFF in terms of emergent \(\mathcal{PT}\) symmetric dimers. 
Appendix~\ref{app:tsff_ep} analyses the exceptional points and their signatures across the \(\mathcal{PT}\) transition appearing in TopSFF at large \(q\). 
Appendix~\ref{app:finite_q_exact} presents the exact finite-\(q\) evaluation of the TopSFF with a spatially extended swap defect, including an analytic calculation at \(t=1\). 
Appendix~\ref{app:tsff_tmat} studies TopSFF with spatially extended defects in systems without discrete time-translation symmetry. 
Appendix~\ref{app:temp_ext} derives the corresponding results for temporally extended defects.
Appendix~\ref{app:num_sim} presents finite-\(q\) numerical simulations that support the \(\PT\) transition in the TopSFF with spatially extended defects, as well as the universal scaling forms for temporally extended defects.


\setcounter{tocdepth}{3}
\tableofcontents

\resumetoc

\vspace{1cm}

\begin{figure}[ht!]
    \centering
    \includegraphics[width=0.9 \textwidth]{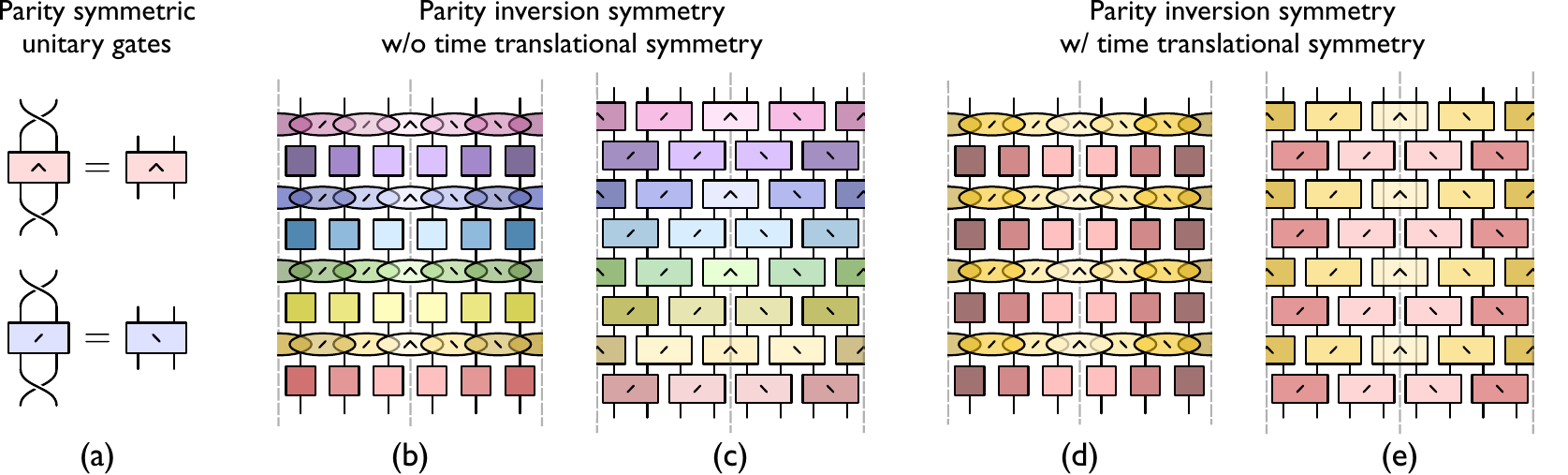}
    \caption{\textbf{Parity inversion symmetric models.} (a) Transformation property of gates that straddle (top) or do not straddle (bottom) across the inversion axis.
(b,c) Many-body quantum circuits with parity inversion symmetry and without time translational symmetry for the random phase model (b) and the random Hamiltonian model (c).
(d,e) Many-body quantum circuits with parity inversion symmetry and discrete time translational symmetry for the random phase model (d) and the random Hamiltonian model (e).
The grey dashed lines denote the inversion axis. Gates of the same color are identical. Periodic boundary conditions are imposed.   }
    \label{fig:parity_models}
\end{figure}

\section{Models}\label{app:model}
In this appendix, we define the one-dimensional random circuit ensembles used throughout this work. We first specify the temporal structure of the dynamics -- temporally random, globally time-reversal symmetric, and Floquet -- and then introduce the parity inversion, time reversal, and time translation symmetric Random Phase Model (RPM), Random Hamiltonian Model (RHM), and Brick Wall Model (BWM).
\subsection{Dynamics in time}\label{app:time_dyn}
We are interested in three classes of dynamics in time: (i) temporal random models, (ii) global time reversal symmetric models, and (iii) discrete time translational symmetric  (Floquet) models. Specifically, they are defined by the following unitary time evolution operators at time $t$,
\begin{IEEEeqnarray}{rrll} \label{app_eq:time_dyn}
\text{temporal random:}  \qquad  &  \uu_{ \mathrm{r}}(t;u) =& \,\prod_{t'=1}^t u(t') \, , 
\qquad  &  
\\
\text{Global TRS:} \qquad   & \uu_{ \mathrm{trs}}(t; w ) =& \, \,  w^{T} (t)  \, w(t) \, ,   & 
\\
\text{Discrete time translational invariant (Floquet):} \qquad   & \uu_{\mathrm{f}} (t; u) =& \, \,  u^t \,,   & 
 \end{IEEEeqnarray}
where $u$ and $w$ denote certain choices or distributions of unitary operators.

\subsection{Parity inversion symmetric models}\label{app:model_par_inv}
We consider one-dimensional quantum circuits $U$ that act on the Hilbert space of $L$ qudits with Hilbert space dimension $q^L$. These models are parity inversion symmetric in the sense that the quantum circuit commutes with the global swap operator $S$, i.e., $[U,S]=0$, where ${S}$ is defined by the action ${S}\ket{a_1, a_2 ,\dots , a_{L-1}, a_L} = \ket{a_L, a_{L-1} ,\dots , a_2, a_1}$. The universal signatures of parity inversion symmetry in quantum many-body chaotic systems are discussed in detail in a forthcoming manuscript~\cite{upcoming}.

 \subsubsection{Parity inversion symmetric Random Phase Model}\label{app:parity_ptRPM}
 The random phase model (RPM)~\cite{chan2018spectral} is a one-dimensional quantum circuit that acts on the Hilbert space of $L$ qudits with Hilbert space dimension $q^L$. The RPM without symmetries is known to be chaotic for $q \geq 3$, and not chaotic for $q=2$ \cite{chan2018spectral}. 
 In the parity inversion symmetric RPM, the parity inversion axis lies on the bond $(L/2, L/2+1)$ for the open boundary condition (obc), and on the bonds $(L/2, L/2+1)$  and $(L,1)$ for the periodic boundary condition (pbc). 
Specifically, the parity inversion symmetric RPM (p-RPM) $U_{\text{p-RPM}}$ [Fig.~\ref{fig:parity_models}] is defined by
\begin{IEEEeqnarray}{rrl} \label{app_eq:par_rpm_temprand}
\text{Temporal random:}  \qquad  & U_{\text{r-p-RPM}} = & \,\,  \uu_{ \mathrm{r}}(t;\mathcal{U}_{\text{p-RPM}})   \, ,
\\
\label{app_eq:par_rpm_floq}
\text{Discrete time translational invariant (Floquet):} \qquad   & U_{\text{f-p-RPM}} =& \,\, \uu_{\mathrm{f}} (t; \mathcal{U}_{\text{p-RPM}})\,,   
 \end{IEEEeqnarray}
where a single bi-layer is given by
 \begin{equation}
\label{app_eq:RPM_geom}
\begin{aligned}
\mathcal{U}_{\text{p-RPM}}(t,L) =\,  \mathcal{U}_2(t,L) \,  \mathcal{U}_1(t,L) 
\;, \qquad \qquad 
\mathcal{U}_1(t,L) = \bigotimes_{r =1}^L u(t,r)  \;,
\qquad 
\mathcal{U}_2(t,L) =  \bigotimes_{r =1} w(t,r)  \;.
\end{aligned}
\end{equation}
The gates within $\mathcal{U}_1(t,L)$ consist of  
\begin{equation}
\label{app_eq:gp_RPM_def}
\begin{aligned}
& u(t, r) \in \mathrm{CUE} \quad &\text{if } r\in [1,2,\dots, L/2] \;,
\\
& 
u(t, r) = S \,  u (t, L - r +1 )  \, S  \quad &\text{if } r\in [ L/2+1,L/2+2, \dots L] \;.
\end{aligned}
\end{equation}
where $u(t,r)$ acts on site $r$, and is randomly and independently drawn from the Circular Unitary Ensemble (CUE). The gates within $\mathcal{U}_2(t,L)$ consist of  
\begin{equation}
\label{app_eq:pRPM_def3}
\begin{aligned}
& [w(t, r)]_{ab}^{cd} =\delta_{ac} \delta_{bd}\exp[i \phi_{ab}(t,r)], \quad 
\quad  \phi_{ab}(t,r) \in \mathcal{N}(0, \epsilon)
\quad &\text{if } r\in [1,2,\dots, L/2-1] \;,
\\
&
[w(t, r)]_{ab}^{cd} =\delta_{ac} \delta_{bd}\exp[i \phi_{ab}(t,r)], \quad 
\quad  \phi_{ab}(t,r) \in \mathcal{N}(0, \epsilon) \text{ and }\phi_{ab} = \phi_{ba}\, ,
\quad  &\text{if } r= L/2, L  \;, 
\\
& 
w(t, r) = S \,  w(t, L - r )  \, S  \quad &\text{if } r\in [ L/2+1,
\dots L-1] \;,
\end{aligned}
\end{equation}
where $\phi_{\mu\nu}(t,r)$ are independent real Gaussian variables with mean zero and variance $\epsilon$. 
The diagonal random phase gates $w(t,r)$ with $r= L/2$ and $L$ are chosen such that $S w(t,r)S^{-1} = w(t,r)$. 
Again,  $S$ is the global swap operator defined by the action $S\ket{a_1, a_2 ,\dots , a_{L-1}, a_L} = \ket{a_L, a_{L-1} ,\dots , a_2, a_1}$, and $\ket{a_1, a_2 ,\dots , a_{L-1}, a_L}$ is the computational basis state with $a_i = 0, 1, \dots, q-1$. Note that for obc (pbc), the 2-site gate straddling across sites $L$ and $1$, namely $w(t,r=L)$, is removed (kept).

\subsubsection{Parity inversion symmetric Random Hamiltonian Model}
The random Hamiltonian model (RHM) is a one-dimensional quantum circuit that acts on the Hilbert space of $L$ qudits with Hilbert space dimension $q^L$. The RHM is utilized to numerically simulate generic quantum many-body chaotic dynamics at $q=2$. The RHM is defined similarly to the RPM, except that the random phase gate is replaced with a coupling gate with more independent parameters, leading to quantum chaotic dynamics. 
Specifically, the  parity inversion symmetric RHM (p-RHM)  is defined by
\begin{IEEEeqnarray}{rrl} 
\label{app_eq:rpm_parity_temp_rand}
\text{Temporal random:}  \qquad  & U_{\text{r-p-RHM}} = & \,\,  \uu_{ \mathrm{r}}(t;\mathcal{U}_{\text{p-RHM}})   \, ,
\\
\label{app_eq:rpm_parity_floq}
\text{Discrete time translational invariant (Floquet):} \qquad   & U_{\text{f-p-RHM}} =& \,\, \uu_{\mathrm{f}} (t; \mathcal{U}_{\text{p-RHM}})\,,   
 \end{IEEEeqnarray}
 where $\mathcal{U}_{\text{p-RHM}}$ is a bilayer unitary operator with the brickwall geometry given by
\be \label{app_eq:bilayer}
\uu_{\text{p-RHM}}(t,L)= \uu_2(t,L) \, \uu_1(t,L) \, , \qquad 
\qquad 
	\uu_m(t,L)=\bigotimes_{\substack{r \in 2 \mathbb{Z} + m\, \text{mod}2}}w(t,r) \;.
\ee
$w(t,r)$ is a two-site unitary quantum gate acting on sites $r$ and $r+1$. For periodic boundary condition (pbc), the product is over $2\mathbb{Z}+1 \equiv \{1,3,5,\dots, L-1\}$ or $2\mathbb{Z} \equiv \{2,4,6,\dots, L\}$. 
For open boundary condition (obc), the 2-site gate straddling across sites $L$ and $1$, namely $u(t,r=L)$, is removed, and we have a product over  $2\mathbb{Z} \equiv \{2,4,6,\dots, L-2\}$. For simplicity, we only consider quantum circuits of even system sizes, i.e. $L\in 2\mathbb{Z}$. $w(t,r)$ is defined by 
\begin{equation}
\label{app_eq:pRHM_def3}
\begin{aligned}
& w(t, r) =\exp{ \left[ i\sum_{\mu, \nu=0, 1, 2, 3} a_{\mu \nu} \sigma_\mu\otimes\sigma_\nu \right] }  \big[ u(t,r)\otimes  u(t,r+1) \big]
, \quad 
\quad &\text{if } r\in [1,2,\dots, L/2-1] \,, 
\\
& 
\hspace{2.5cm}  \text{with }
a_{\mu \nu} \in \mathcal{N}(0, \eta)\, , \quad 
u(t, r) \in \mathrm{CUE}  \,,
\\
&
w(t, r) =\exp{ \left[ i\sum_{\mu, \nu=0, 1, 2, 3} b_{\mu \nu} \sigma_\mu\otimes\sigma_\nu \right] }  \big[ u(t,r)\otimes  u(t,r) \big]
\, , \quad 
 &\text{if } r= L/2, L  \;, 
 \\
&  \hspace{2.5cm} \text{with }
b_{\mu \nu} \in \mathcal{N}(0, \eta)\, \,  \text{ and }\, \,  b_{\mu \nu}= b_{ \nu \mu}\, ,
\quad 
u(t, r) \in \mathrm{CUE}  \,,
\\
& 
w(t, r) = S \,  w(t, L - r )  \, S  \quad &\text{if } r\in [ L/2+1,
\dots L-1] \;,
\end{aligned}
\end{equation}
where $a_{\mu \nu}$ and $b_{\mu \nu}$ are independent real Gaussian variables with mean zero and variance $\eta$.
For computational efficiency, we adopt RHM models in which, for each realization of $U_{\text{r-p-RHM}}$ and $U_{\text{f-p-RHM}}$, only a single set of $a_{\mu \nu}$ and  a single set of $b_{\mu \nu}$ are sampled and used across both space and time.
Again,  note that for obc (pbc), the 2-site gate straddling across sites $L$ and $1$, namely $w(t,r=L)$, is removed (kept). Note that for sites without local parity inversion symmetry, i.e. for $r \notin L/2$ and $L$, the generator of $w(t,r)$ can equivalently be written as the Gaussian Unitary Ensemble (GUE).

\subsubsection{Parity inversion symmetric Brick Wall Model}\label{app:model_bwm}
The parity inversion symmetric Brick-Wall Model (BWM) has the same brick wall geometry, boundary conditions, and temporal ensembles as the parity inversion symmetric RHM, but its two-site gates are sampled directly from the Haar measure. It is defined by
\begin{IEEEeqnarray}{rrl}
\text{Temporal random:}\qquad
& U_{\mathrm{r\text{-}p\text{-}BWM}}
=&\,\uu_{\mathrm r}(t;\mathcal U_{\mathrm{p\text{-}BWM}}),
\\
\text{Discrete time translational invariant (Floquet):}\qquad
& U_{\mathrm{f\text{-}p\text{-}BWM}}
=&\,\uu_{\mathrm f}(t;\mathcal U_{\mathrm{p\text{-}BWM}}),
\end{IEEEeqnarray}
where
\begin{equation}
\mathcal U_{\mathrm{p\text{-}BWM}}(t,L)
=
\mathcal U_2(t,L)\mathcal U_1(t,L),
\qquad
\mathcal U_m(t,L)
=
\bigotimes_{\substack{r\in 2\mathbb Z+m\,\mathrm{mod}\,2}}
w(t,r).
\end{equation}
For \(r=1,\ldots,L/2-1\), the gates \(w(t,r)\in\mathrm{CUE}(q^2)\) are independent Haar-random two-site unitaries. Parity inversion fixes the remaining gates through
\begin{equation}
w(t,r)=S\,w(t,L-r)\,S,
\qquad r=L/2+1,\ldots,L-1,
\end{equation}
while gates centred on an inversion axis, \(r=L/2,L\), are sampled from the locally parity-symmetric Haar ensemble satisfying
\begin{equation}
[w(t,r),S]=0.
\end{equation}
As in the p-RHM, the gate at \(r=L\) is retained for periodic boundary conditions and removed for open boundary conditions.
For an analytical treatment of this model, see Ref.~\cite{upcomingharkin}.

\subsection{Time reversal symmetric models}\label{model:trs}
We consider one-dimensional time reversal symmetric quantum circuits $U$ that act on the Hilbert space of $L$ qubits with Hilbert space dimension $q^L$ with $q=2$. We adopt the framework and models of Ref.~\cite{wu2025} to analyze quantum many-body chaos in  time reversal symmetric systems. 
A quantum system described by a Hamiltonian $H(t)$ at time $t$ is {time reversal symmetric} (TRS) if $H(t) = \TT H(-t) \TT^{-1}$, where $\TT$ is an antiunitary operator. A quantum system described by a unitary operator $U(t)$ is TRS if $\TT U(t) \TT^{-1} = U^\dagger(t)$.
It can be shown that the antiunitary operator satisfies $\TT^2=\pm 1$, where the plus and minus signs correspond to the cases of systems with TRS with integer and half-integer spins respectively. Here we focus on dynamical and spectral properties of TRS quantum many-body systems with $\TT^2 = 1$. 
For any unitary $u$ satisfying equation $\mathcal{T} u \mathcal{T}^{-1}= u^\dagger$, in the basis where $\mathcal{T}$ is complex conjugation operator,  $u= u^T$ is symmetric. Therefore, $u$ always has a spectral decomposition $u = \mathcal{S} \Lambda \mathcal{S}^{T}$ where the superscript $T$ denotes transposition, $\Lambda$ is the diagonal matrix of eigenvalues, and the similarity matrix $\mathcal{S}$ satisfies $\mathcal{S}^{T} \mathcal{S} = \mathbb{1}$. In other words,  we can always decompose $u = v v^\mathrm{T} $ and we call $v = \mathcal{S} \sqrt{\Lambda} $ a half gate of $u$. As an example, if $u$ is drawn from the Circular Orthogonal Ensemble (COE), then it has a natural decomposition $u = v v^{T} $ with $v$ drawn from the CUE, and so $v$ from CUE is a half gate of $u$ from COE. In the following, we will use half gates at the first and last time steps of quantum circuits.
Lastly, as discussed in~\cite{wu2025}, if the many-body quantum circuit $U$ is TRS, then $V$, the spacetime dual of $U$, commutes with the global swap operator $S$ in the spatial direction, i.e. $[V,S]=0$, where ${S}$ acts on the dual Hilbert space.

\subsubsection{TRS Random Phase Model}\label{app_sec:trs_rpm}
\begin{figure*}[htbp]
    \centering
    \includegraphics[width=0.7\textwidth]{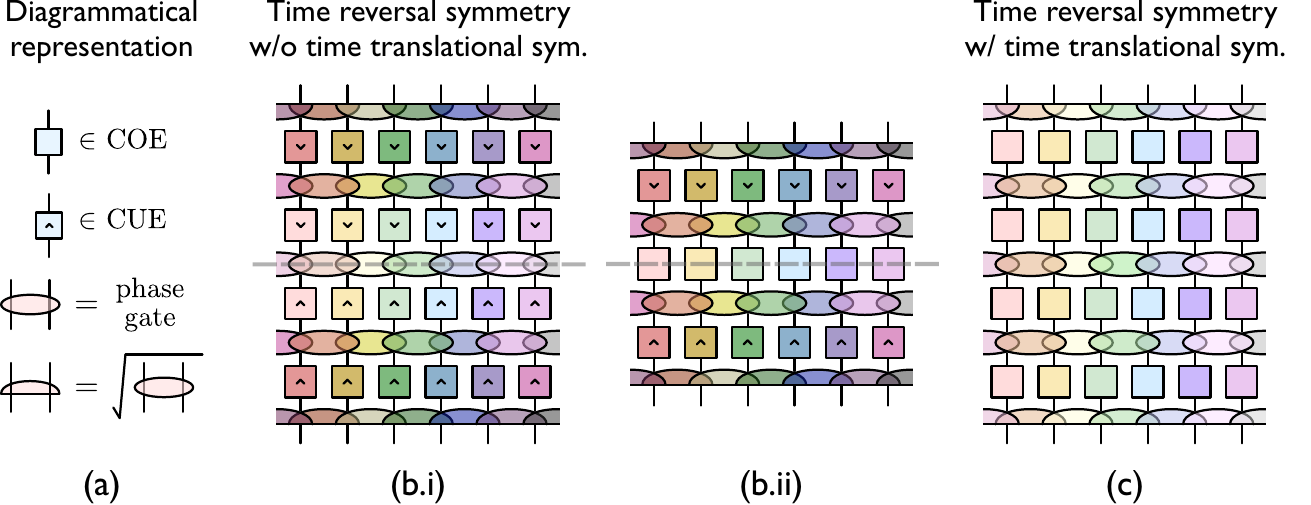}
    \caption{\textbf{TRS Random Phase Models.} Gates with the same color are identical. 
    (a) Dictionary of the diagrammatic representation. (b) Illustrations of the t-RPM without time translational symmetry at (b.i) $t=4$ and (b.ii) $t=3$. 
    Note that gates reflected across the time reversal axes (dashed lines) are identified under transposition. (c) Illustrations of t-RPM with discrete time translational symmetry at $t=4$. 
    }
    \label{sm_fig:model_rpm}
\end{figure*}
The random phase model (RPM)~\cite{chan2018spectral} is a one-dimensional quantum circuit that acts on the Hilbert space of $L$ qudits with Hilbert space dimension $q^L$. The model is composed of  $q$-by-$q$  one-site gates
   \begin{align}
 	u^{(r)}(s) \quad  \in \text{CUE} \text{ or } \text{COE}\, ,
\end{align} 
acting on site $r$, where CUE and COE stand for Circular Unitary Ensemble and the Circular Orthogonal Ensemble respectively; and $q^2$-by-$q^2$  two-site gates, 
\begin{align}
    [\Theta^{(r,r+1)}(s, \epsilon)]_{a_r a_{r+1}, b_r b_{r+1}}=\delta_{a_r, b_r} \delta_{a_{r+1}, b_{r+1}}\exp[ i \varphi^{(r)}_{a_r,a_{r+1}}(s,\epsilon )] \,,
\end{align}
coupling neighbouring sites via a diagonal random phase ($a_r = 1,2\ldots, q$). Each coefficient $\varphi_{a_r,a_{r+1}}^{(r)}(s)$ is an independent Gaussian random real variable with mean zero and variance $\epsilon$, which controls the coupling strength between neighbouring spins.

Then, we can define monolayers of one-site gates, of two-site gates and of two-site half gates  as 
 \begin{IEEEeqnarray}{rrl} \label{app_eq:rqc_cue_og}
\text{1-site gate monolayer:} \qquad   & w_{\text{h}}(s;u) \, &  =  \bigotimes_{r=1}^L u^{(r)}(s) \, , \, 
\\
\text{2-site gate monolayer:}  \qquad  &   w_{\text{ph}}(s;\epsilon) \, &= \prod_{r=1}^L \Theta^{(r, r+1)}(s,\epsilon)  \,,  
\\
\text{RPM bilayer:}  \qquad  &   w_{\text{bl}}(s, \epsilon;u) \, &= 
 w_{\text{ph}}(s;\epsilon ) w_{\text{h}}(s;u)
 \\
\text{Shifted RPM bilayer:}  \qquad  &   w_{\text{sbl}}(s, \epsilon;u) \, &= 
 w_{\text{ph}}(s;\epsilon/4 ) w_{\text{h}}(s;u)  w_{\text{ph}}(s;\epsilon/4 ) 
 \end{IEEEeqnarray}
 where the variable $s$ allows the freedom for different bilayer labelled by $s$ to be independent of each other.

 Note that in the shifted TRS, the first and the third layers are drawn from the same ensemble. The variances in these layers are chosen such that $w_{\text{ph}}(s;\epsilon/4 )^2 = w_{\text{ph}}(s;\epsilon)$.

Now we define RPM with various types of TRS. 
For the RPM with global TRS without discrete time translational symmetry, we define
\begin{IEEEeqnarray}{rl}
w_{\text{g-trs}}(t,L) \; &=
\begin{cases} 
w_{\text{h}}((t-1)/2+1;u_{\mathrm{CUE}})
\left[ \prod_{s=1}^{(t-1)/2} 
w_{\mathrm{bl}}(s, \epsilon; u_{\text{CUE}}) \right] 
w_{\text{ph}}((t-1)/2+1;\epsilon/4 ) \,, \quad \qquad \qquad & t \text{ odd,} 
\\
 w_{\text{ph}}(t/2;\epsilon/4 )
  w_{\text{h}}(t/2;u_{\mathrm{CUE}})
\left[ \prod_{s=1}^{t/2-1} 
w_{\mathrm{bl}}(s, \epsilon; u_{\text{CUE}})  \right]  
w_{\text{ph}}(0;\epsilon/4 )   \,,  \qquad \qquad  & t \text{ even,}
\end{cases}
 \end{IEEEeqnarray}
 which governs the first half of the global TRS circuits. The entire global TRS circuits are then generated by extending the above circuits with their transpositions, as defined below.
The Floquet TRS RPM can be defined by repeated action of a shifted RPM bilayer. Together, we define 
\begin{IEEEeqnarray}{rrl} \label{app_eq:rqc_coe}
\text{Global TRS:}  \qquad  &  U_{\text{g-trs-RPM}}(t,L) \; & = \uu_{\text{trs}} [w_{\text{g-trs}} (t , L)] \, ,
\label{app_eq:global_trs_rpm}
\\
\text{Floquet TRS:}  \qquad  &  U_{\text{f-trs-RPM}}(t,L) \; & = \uu_{\text{f}} [t; w_{\text{sbl}}( 1 ,\epsilon ; u_{\mathrm{COE}})]\, . \label{app_eq:floq_trs_rpm} 
 \end{IEEEeqnarray}
See Fig.~\ref{sm_fig:model_rpm} for illustrations of the TRS RPMs.

\subsubsection{TRS Random Hamiltonian Model}
\begin{figure*}[htbp]
    \centering
    \includegraphics[width=0.7\textwidth]{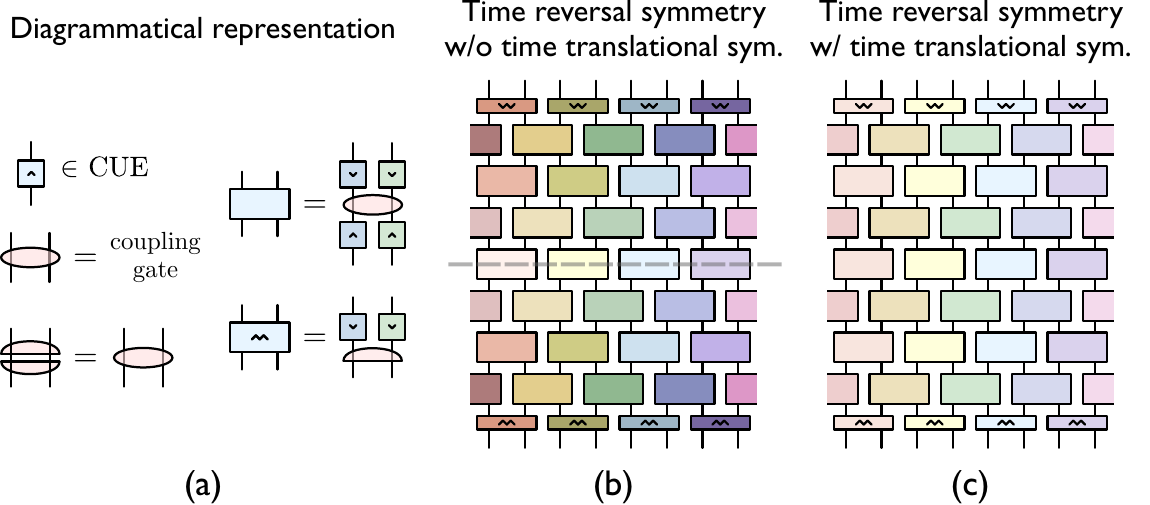}
    \caption{\textbf{TRS Random Hamiltonian Models.} Gates with the same color are identical. 
    (a) Dictionary of the diagrammatic representation. (b) Illustrations of the t-RHM (b)  without and (c) with discrete time translational symmetry at $t=4$. 
    Note that gates reflected across the time reversal axes (dashed lines) are identified under transposition.
    }
    \label{sm_fig:model_rhm}
\end{figure*}
The random Hamiltonian model (RHM) is a one-dimensional quantum circuit that acts on the Hilbert space of $L$ qudits with Hilbert space dimension $q^L$. The RHM is utilized to numerically simulate generic quantum many-body chaotic dynamics at $q=2$. The TRS RHM is defined similarly to the t-RPM, except that the random phase gate is replaced with a coupling gate with more independent parameters, leading quantum chaotic dynamics at $q=2$. 
Due to the non-commuting interaction among the coupling gates, we construct RHM with a brick-wall geometry composed of two-site gates acting on the $r$-th and $(r+1)$-th qudits, 
\begin{equation}
\label{app_eq:tRHM_def3}
\begin{aligned}
& u_{\text{RHM}}(t, r) =
\bigg( \left[u^{(r)}_1(t) \right]^T \otimes \left[u^{(r)}_2(t)\right]^T \bigg)
\exp{ \left( i\sum_{\mu, \nu=0, 1, 2, 3} a_{\mu \nu} \sigma_\mu\otimes\sigma_\nu \right) }  
\bigg( u^{(r)}_1(t) \otimes u^{(r)}_2(t) \bigg)
\\
& \quad u_1^{(r)}(t),  u_2^{(r)}(t) \in \text{CUE}\;,  \quad \quad 
 a_{\mu \nu} \in \mathcal{N}(0, \eta)
\;, \quad \quad 
 a_{\rho 2}= a_{2 \rho}=0 \text{ for }\rho =0,1,3
\;.
\end{aligned}
\end{equation}
where $a_{\mu \nu}$ are independent real Gaussian variables with mean zero and variance $\eta$.  $u_1$ and $u_2$ are $q$-by-$q$ unitaries drawn from the CUE. 
The constraints on $a_{\rho 2}$ and $a_{2\rho}$ arise as follows: Since we are focusing on $\mathcal{T}^2=1$, we can choose a basis such that the anti-unitary symmetry operator can be written as $\mathcal{T}=\mathcal{K}$, the complex conjugation operator. In this basis, the generators of the exponential in Eq.~\eqref{app_eq:tRHM_def3}  are real. Since \(\sigma_{2}\) is purely imaginary, \(\mathcal{T}\) invariance requires the coefficient for the basis elements with an odd number of \(\sigma_{2}\) to vanish, leading to $a_{\rho 2}=a_{2\rho}=0$ for  $\rho=0,1,3$.
As before, for computational efficiency, we adopt RHM models in which, for each realization in the ensemble, only a single set of $a_{\mu \nu}$ is sampled and used across both space and time.

As discussed in the beginning of the section, for any unitary $u$ satisfying equation $\mathcal{T} u \mathcal{T}^{-1}= u^\dagger$, there always exists a basis such that $u$ is symmetric, and we can decompose $u = v v^\mathrm{T}$ with a half gate $v$. In the definitions below, for each $u_{\text{RHM}}(t, r)$, we will use its half gate $v_{\text{RHM}}(t, r)$  satisfying $u_{\text{RHM}}(t, r)= v_{\text{RHM}} v^T_{\text{RHM}}$.
 Specifically, we respectively define a monolayer and a bilayer of two-site gates in the brick-wall geometry with even system sizes as
 \begin{IEEEeqnarray}{rrl} \label{app_eq:rqc_cue}
\text{Brick-wall monolayer:} \qquad   & w_{m}(s;u) \, &  =  \bigotimes_{\substack{j \in 2 \mathbb{Z} + m\, \text{mod}2}}u^{(r,r+1)}(s) \, , \, 
\\ \label{app_eq:rqc_cue2}
\text{Brick-wall bilayer:}  \qquad  &   w_{\text{bl}}(s;u_\mathrm{RHM}) \, &= w_2(s;u_\mathrm{RHM}) \, w_1(s;u_\mathrm{RHM}) \,,  
\\ \label{app_eq:rqc_cue3}
\text{Shifted TRS brick-wall bilayer:}  \qquad  &   w_{\text{sbl}}(s) \, &= w_1(s, v_{\mathrm{RHM}})^T \, w_2(s;u_\mathrm{RHM}) \, w_1(s;v_\mathrm{RHM}) \,,  
 \end{IEEEeqnarray}
 where the variable $s$ allows the freedom for different bilayer labelled by $s$ to be independent of each other. Note that in the shifted TRS, by construction, the first and the third layers are drawn from the same ensemble.

 Now we define the TRS RHM with and without discrete time translational symmetry. To define the latter, we define 
\begin{IEEEeqnarray}{rl}
w_{\text{trs}}(t,L) \; &= 
w_{t}(t; v_{\text{RHM}})
\left[ \prod_{s=2}^{t-1} 
w_{s}(s; u_{\text{RHM}}) 
\right]
w_{1}(1; v_{\text{RHM}})  \,,
 \end{IEEEeqnarray}
 which governs the first half of the global TRS circuits. The entire global TRS circuits are then generated by extending the above circuits with their transpositions, as defined below. 
 Note that the first and last layers of $w_{\text{g-TRS}}$ are drawn from half gates, such that, upon the extension with transposition, we can construct a full gate across the time reversal axis (dashed lines in Fig.~\ref{sm_fig:model_rhm}), using the construction, $u_{\mathrm{RHM}}= v_{\mathrm{RHM}} v^T_{\mathrm{RHM}}$.
To define the Floquet TRS RHM, we can simply repeatedly act a shifted brick-wall bilayer. Together, we define for both odd and even $t$,
\begin{IEEEeqnarray}{rrl} 
\text{Global TRS RHM:}  \qquad  &  U_{\text{g-trs-RHM}}(t,L) \; & = \uu_{\text{trs}} [w_{\text{TRS}} (t , L)] \, , \label{app_eq:rhm_gtrs}
\\
\text{Floquet TRS RHM:}  \qquad  &  U_{\text{f-trs-RHM}}^{\text{F}}(t,L) \; & = \uu_\text{f} [t; w_{\text{sbl}} ( 1 )] \, . 
 \end{IEEEeqnarray}
Again,  note that for obc (pbc), the 2-site gate straddling across sites $L$ and $1$, namely $w(t,r=L)$, is removed (kept). See Fig.~\ref{sm_fig:model_rhm} for illustrations of the TRS RHMs.

\subsection{Time translational symmetric only models}\label{app:model_time_trans}
We now define quantum circuits with discrete time translational symmetry, without imposing parity inversion symmetry or time reversal symmetry. In these models, a unitary circuit layer, which may itself consist of multiple sublayers, is sampled once and then repeated periodically in time.

\subsubsection{Random Phase Model}
The random phase model (RPM)~\cite{chan2018spectral} is a one-dimensional quantum circuit that acts on the Hilbert space of $L$ qudits with Hilbert space dimension $q^L$. The RPM without additional symmetries is known to be quantum chaotic for $q\ge 3$, and not quantum chaotic for $q=2$~\cite{chan2018spectral}. 
The Floquet RPM is defined by repeated action of a single bilayer unitary operator,
\begin{IEEEeqnarray}{rrl}
\label{app_eq:floqonly_rpm}
\text{Discrete time translational invariant (Floquet):} \qquad
& U_{\text{f-RPM}}(t,L) =& \,\, \uu_{\mathrm{f}}(t;\mathcal{U}_{\text{RPM}})\,,
\end{IEEEeqnarray}
where a single bilayer is given by
\begin{equation}
\label{app_eq:floqonly_rpm_geom}
\begin{aligned}
\mathcal{U}_{\text{RPM}}(L)
=
\mathcal{U}_2(L)\,\mathcal{U}_1(L)\,,
\qquad
\mathcal{U}_1(L)=\bigotimes_{r=1}^{L} u(r)\,,
\qquad
\mathcal{U}_2(L)=\prod_{r=1}^{L} w(r)\, .
\end{aligned}
\end{equation}
Here $u(r)$ is a one-site unitary gate acting on site $r$, drawn independently from the Circular Unitary Ensemble (CUE),
\begin{equation}
\label{app_eq:floqonly_rpm_u}
u(r)\in \mathrm{CUE}\,,
\qquad r=1,2,\dots,L\,,
\end{equation}
and $w(r)$ is a two-site diagonal random phase gate acting on sites $(r,r+1)$, defined by
\begin{equation}
\label{app_eq:floqonly_rpm_w}
[w(r)]_{ab}^{cd}
=
\delta_{ac}\delta_{bd}\exp\!\big[i\phi_{ab}(r)\big]\,,
\qquad
\phi_{ab}(r)\in \mathcal{N}(0,\epsilon)\,.
\end{equation}
The real Gaussian variables $\phi_{ab}(r)$ are independent with mean zero and variance $\epsilon$, which controls the coupling strength between neighbouring sites.
For pbc, the product in $\mathcal{U}_2(L)$ includes the gate straddling sites $(L,1)$. For obc, this gate is removed. 

\subsubsection{Random Hamiltonian Model}
The random Hamiltonian model (RHM) is a one-dimensional quantum circuit that acts on the Hilbert space of $L$ qubits with Hilbert space dimension $q^L$ and $q=2$. The RHM is used to numerically simulate generic quantum many-body chaotic dynamics at $q=2$. In the absence of additional symmetries, the Floquet RHM is defined by repeated action of a brick-wall bilayer,
\begin{IEEEeqnarray}{rrl}
\label{app_eq:floqonly_rhm}
\text{Discrete time translational invariant (Floquet):} \qquad
& U_{\text{f-RHM}}(t,L) =& \,\, \uu_{\mathrm{f}}(t;\mathcal{U}_{\text{RHM}})\,,
\end{IEEEeqnarray}
where
\begin{equation}
\label{app_eq:floqonly_rhm_geom}
\mathcal{U}_{\text{RHM}}(L)=\uu_2(L)\,\uu_1(L)\,,
\qquad
\uu_m(L)=\bigotimes_{\substack{r\in 2\mathbb{Z}+m\,\mathrm{mod}\,2}} w(r)\,.
\end{equation}
Here $w(r)$ is a two-site unitary gate acting on sites $(r,r+1)$, of the form
\begin{equation}
\label{app_eq:floqonly_rhm_w}
w(r)
=
\exp\!\left[
i\sum_{\mu,\nu=0}^{3} a_{\mu\nu}(r)\,\sigma_\mu\otimes \sigma_\nu
\right]
\big[u(r)\otimes u(r+1)\big]\,,
\end{equation}
with
\begin{equation}
\label{app_eq:floqonly_rhm_params}
a_{\mu\nu}(r)\in \mathcal{N}(0,\eta)\,,
\qquad
u(r)\in \mathrm{CUE}\,.
\end{equation}
The coefficients $a_{\mu\nu}(r)$ are independent real Gaussian variables with mean zero and variance $\eta$, while the one-site gates $u(r)$ are independently drawn from the CUE. The parameter $\eta$ controls the interaction strength between neighbouring qubits.
For pbc, the products in $\uu_1(L)$ and $\uu_2(L)$ are over $2\mathbb{Z}+1\equiv \{1,3,5,\dots,L-1\}$ and $2\mathbb{Z}\equiv \{2,4,6,\dots,L\}$ respectively. For obc, the two-site gate straddling sites $(L,1)$ is removed, and the product over even bonds becomes $2\mathbb{Z}\equiv \{2,4,6,\dots,L-2\}$. For simplicity, we consider even system sizes, i.e. $L\in 2\mathbb{Z}$.
As in the Floquet RPM, the same sampled bilayer $\mathcal{U}_{\text{RHM}}(L)$ is repeated at each Floquet period.

\section{TopSFF for RMT }\label{app:tsff_rmt}
In this appendix, we compute the TopSFF and related observables in random matrix theory (RMT) for two purposes: 
(i) to establish zero-dimensional RMT benchmarks for quantum many-body chaotic systems in the late-time or large-system-size regimes, and 
(ii) to elucidate the intrinsically many-body character of the nontrivial TopSFF behaviour discussed in the main text.
We first derive the ensemble-averaged TopSFF for a general defect operator acting on the doubled Hilbert space for unitary invariant random matrix ensembles. We then specialize this result to the Gaussian unitary ensemble (GUE) and circular unitary ensemble (CUE), which describe universal spectral correlations of quantum chaotic systems at late times~\cite{bohigas1984characterization}, and to the non-Hermitian Ginibre unitary ensemble (GinUE), which emerges as a universal description of many-body chaotic systems in the large-system-size limit~\cite{Shivam_2023}.
We note that the RMT TopSFF considered here is not ``topological'' in the sense of the main text, since the random matrices are not required to commute with the defect operator. Nevertheless, it provides a useful zero-dimensional benchmark for the corresponding many-body observables.

We finally compare the Mickey Mouse contraction topology underlying the spatially extended TopSFF with its zero-dimensional CUE counterpart. We show that the leading Gaussian and non-Gaussian Mickey Mouse contributions cancel at large $N$, and derive the corresponding exact finite-$N$ result. Thus, although the same contraction topology is present in RMT, its nontrivial Gaussian--non-Gaussian dynamics is absent in the zero-dimensional problem. This contrast highlights the intrinsically many-body origin of the TopSFF behaviour associated with propagating temporal domain walls, including the emergent $\mathcal{PT}$ transition.

\subsection{TopSFF for Unitary invariant ensemble}
We first formulate the RMT TopSFF for a general defect operator $\mathcal{D}$ acting on the doubled Hilbert space. For an $N$-dimensional Hilbert space, define
\be
\label{app_eq:general_doubled_defect_rmt}
K^{\mathrm{top}}_{\mathrm{RMT}}(\mathcal{D};X)
:=
\Tr \left[
\mathcal{D}\left(X\otimes X^*\right)
\right],
\ee
where $\mathcal{D}$ is an arbitrary $N^2$-by-$N^2$ operator.
Let the ensemble of $X$ is invariant under unitary conjugation,
$X\mapsto VXV^\dagger$.
Consequently, $\overline{X\otimes X^*}$ commutes with $V\otimes V^*$ and decomposes into the singlet and adjoint sectors as
\be
\overline{X\otimes X^*}
=
s_X P_0+r_X P_{\mathrm{ad}},
\qquad
P_0:=\frac{|\Omega\rangle\langle\Omega|}{N},
\qquad
P_{\mathrm{ad}}:= \iden -P_0,
\qquad
|\Omega\rangle:=\sum_{a=1}^N |a\rangle\otimes |a\rangle .
\ee
Define the two invariants as
\be
\overline{K_X}
:=
\overline{\Tr X \Tr X^\dagger}
= \overline{\Tr X \Tr X^\dagger},,
\qquad
\overline{R_X}
:=
\overline{\Tr(XX^\dagger)},
\ee
where $\overline{K_X}$ is the conventional SFF.
Using
$\Tr\overline{(X\otimes X^*)}=\overline{K_X}$
and
$\Tr[P_0\overline{(X\otimes X^*)}]=\overline{R_X}/N$,
we obtain
\be
s_X=\frac{\overline{R_X}}{N},
\qquad
r_X=
\frac{\overline{K_X}-\overline{R_X}/N}{N^2-1}.
\ee
The RMT TopSFF for an arbitrary doubled-space defect is therefore
\be
\label{app_eq:rmt_master_doubled_defect}
\overline{K^{\mathrm{top}}_{\mathrm{RMT}}(\mathcal{D};X)}
=
\frac{\overline{R_X}}{N}
\Tr(\mathcal{D}P_0)
+
\frac{\overline{K_X}-\overline{R_X}/N}{N^2-1}
\Tr\left[\mathcal{D}P_{\mathrm{ad}}\right].
\ee
Equivalently,
\be
\label{app_eq:rmt_master_doubled_defect_expanded}
\overline{K^{\mathrm{top}}_{\mathrm{RMT}}(\mathcal{D};X)}
=
\frac{N\overline{K_X}-\overline{R_X}}
{N(N^2-1)}
\Tr\mathcal{D}
+
\frac{N\overline{R_X}-\overline{K_X}}
{N(N^2-1)}
\langle\Omega|\mathcal{D}|\Omega\rangle .
\ee
Even though $\mathcal{D}$ may be a general non-factorizable operator on the doubled Hilbert space, the ensemble-averaged RMT TopSFF depends only on the two defect invariants
$\Tr\mathcal{D}$ and $\langle\Omega|\mathcal{D}|\Omega\rangle$.

\subsection{TopSFF for Gaussian and circular ensembles}

We now apply Eq.~\eqref{app_eq:rmt_master_doubled_defect} to the Gaussian and circular unitary ensembles.
The time evolution is defined by
\be
\ba
U_{\text{r-h-RMT}}(t)
:=&
\prod_{\tau=1}^{t}u(\tau),
\qquad & \text{Random},
\\
U_{\text{f-h-RMT}}(t)
:=& \, 
u^t,
\qquad & \text{Floquet},
\ea
\ee
where $u(\tau)$ and $u$ may be drawn from the circular unitary ensemble (CUE), or
$u(\tau)=\exp[iH(\tau)]$ and $u=\exp(iH)$ with $H$ drawn from the Gaussian unitary ensemble (GUE).
We use ``h-RMT'' to denote ensembles of unitary or Hermitian matrices, as opposed to the non-Hermitian ensembles considered below.
Since $U(t)$ is unitary, we have 
$\overline{R_{\mathrm{h-RMT}}(t)}
=
\overline{\Tr[U(t)U^\dagger(t)]}
=N$ and the regular SFF
$\overline{K_{\mathrm{h-RMT}}(t)}
:=
\overline{\Tr[U(t)]\Tr[U^\dagger(t)]}$.
Eq.~\eqref{app_eq:rmt_master_doubled_defect_expanded} gives
\be
\overline{K^{\mathrm{top}}_{\mathrm{h-RMT}}(\mathcal{D};t)}
=
\frac{\overline{K_{\mathrm{h-RMT}}(t)}-1}{N^2-1}
\Tr\mathcal{D}
+
\frac{N^2-\overline{K_{\mathrm{h-RMT}}(t)}}
{N(N^2-1)}
\langle\Omega|\mathcal{D}|\Omega\rangle .
\label{app_eq:general_tsff_h_rmt}
\ee
In particular, if the defect has a factorized form
$
\mathcal{D}
=
D_1^\alpha\otimes(D_2^\beta)^T$, with
$\alpha,\beta\in\{0,1\}$,  then Eq.~\eqref{app_eq:general_tsff_h_rmt} reduces to
\be
\label{app_eq:factorized_tsff_h_rmt}
\begin{aligned}
\overline{
K^{\mathrm{top}}_{\mathrm{h-RMT}}
(D_1,D_2;\alpha,\beta;t)}
=
\frac{
\tr D_1^\alpha\,\tr D_2^\beta
-
N^{-1}\tr(D_1^\alpha D_2^\beta)}
{N^2-1}
\, \overline{K_{\mathrm{h-RMT}}(t)}
+
\frac{
N\tr(D_1^\alpha D_2^\beta)
-
\tr D_1^\alpha\,\tr D_2^\beta}
{N^2-1}.
\end{aligned}
\ee
For a Haar-random CUE matrix, including a product of independent CUE matrices,
$\overline{K_{\mathrm{h-RMT}}}=1$, and hence
\be
\overline{K^{\mathrm{top}}_{\mathrm{r-CUE}}(\mathcal{D})}
=
\frac{1}{N}
\langle\Omega|\mathcal{D}|\Omega\rangle,
\ee
which for the factorized defect becomes
\be
\overline{
K^{\mathrm{top}}_{\mathrm{r-CUE}}
(D_1,D_2;\alpha,\beta)}
=
\frac{1}{N}
\tr(D_1^\alpha D_2^\beta).
\ee
For Floquet CUE and GUE evolution, $\overline{K_{\mathrm{h-RMT}}(t)}$ in
Eq.~\eqref{app_eq:factorized_tsff_h_rmt} is respectively the ordinary CUE or GUE SFF.

\subsection{TopSFF for Ginibre ensembles}

We next consider the non-Hermitian Ginibre unitary ensemble.
Let $G$ denote an $N$-by-$N$ complex Ginibre matrix with probability distribution
$P(G)\propto\exp[-N\Tr(GG^\dagger)]$, such that
\be
\overline{G_{ab}G^*_{cd}}
=
\frac{1}{N}\delta_{ac}\delta_{bd},
\qquad
\overline{G_{ab}}=0.
\ee
We define
\be
\ba
G_{\text{r-nh-RMT}}(L)
:=&
\prod_{\ell=1}^L G(\ell),
\qquad & \text{Random},
\\
G_{\text{f-nh-RMT}}(L)
:=&
G^L,
\qquad & \text{Floquet},
\ea
\ee
where the matrices $G(\ell)$ are independently drawn in the random case.
We use ``nh-RMT'' to denote non-Hermitian random matrix ensembles.

For the Floquet Ginibre ensemble, define
\be
\overline{K_{\mathrm{nh-RMT}}(N,L)}
:=
\overline{\left|\Tr G^L\right|^2},
\qquad
\overline{R_{\mathrm{nh-RMT}}(N,L)}
:=
\overline{\Tr\left[G^L(G^\dagger)^L\right]}.
\ee
Applying the master formula~\eqref{app_eq:rmt_master_doubled_defect_expanded}, we obtain
\be
\begin{aligned}
\label{app_eq:general_tsff_ginibre}
\overline{
K^{\mathrm{top}}_{\mathrm{nh-RMT}}(\mathcal{D};L)}
={}&
\frac{
N\overline{K_{\mathrm{nh-RMT}}(N,L)}
-
\overline{R_{\mathrm{nh-RMT}}(N,L)}
}
{N(N^2-1)}
\Tr\mathcal{D}
\\
&+
\frac{
N\overline{R_{\mathrm{nh-RMT}}(N,L)}
-
\overline{K_{\mathrm{nh-RMT}}(N,L)}
}
{N(N^2-1)}
\langle\Omega|\mathcal{D}|\Omega\rangle .
\end{aligned}
\ee
For the factorized defect
$\mathcal{D}=D_1^\alpha\otimes(D_2^\beta)^T$,
this becomes
\be
\label{app_eq:factorized_tsff_ginibre}
\begin{aligned}
\overline{
K^{\mathrm{top}}_{\mathrm{nh-RMT}}
(D_1,D_2;\alpha,\beta;L)}
={}&
\frac{
N\tr D_1^\alpha\,\tr D_2^\beta
-
\tr(D_1^\alpha D_2^\beta)}
{N(N^2-1)}
\overline{K_{\mathrm{nh-RMT}}(N,L)}
\\
&+
\frac{
N\tr(D_1^\alpha D_2^\beta)
-
\tr D_1^\alpha\,\tr D_2^\beta}
{N(N^2-1)}
\overline{R_{\mathrm{nh-RMT}}(N,L)}.
\end{aligned}
\ee
As a consistency check, setting $D_1=D_2=\iden$ gives
\be
\overline{
K^{\mathrm{top}}_{\mathrm{nh-RMT}}
(\iden, \iden ;1,1;L)}
=
\overline{K_{\mathrm{nh-RMT}}(N,L)},
\ee
as required.
For the random Ginibre product, we have
\be
\overline{
K^{\mathrm{top}}_{\mathrm{r-nh-RMT}}(\mathcal{D};L)}
=
\Tr(\mathcal{D}P_0)
=
\frac{1}{N}\langle\Omega|\mathcal{D}|\Omega\rangle.
\ee
For a factorized defect, this reduces to
\be
\overline{
K^{\mathrm{top}}_{\mathrm{r-nh-RMT}}
(D_1,D_2;\alpha,\beta;L)}
=
\frac{1}{N}
\tr(D_1^\alpha D_2^\beta).
\ee

\subsection{Mickey Mouse observable $\Tr U^{2t}(\Tr U^{*t})^2$ vanishes in large-$N$ RMT}
To elucidate the intrinsically many-body character of the TopSFF with a spatially extended topological defect involving Mickey Mouse diagrams, we compare it with its zero-dimensional random-matrix counterpart, defined through the correlator
\begin{equation}
M(t,N)
:=
\Tr U^{2t}
\left(\Tr U^{*t}\right)^2,
\qquad
U\in\mathrm{CUE}(N).
\label{eq:RMT_three_trace}
\end{equation}
Since $U$ is unitary, $\Tr U^{*t}=\Tr U^{-t}$. Diagrammatically, Eq.~\eqref{eq:RMT_three_trace} contains one temporal loop of length $2t$ and two loops of length $t$, and therefore has the same temporal topology as the Mickey Mouse configurations discussed in the main text. We first show that, at large $N$, the leading Gaussian and non-Gaussian Mickey Mouse contributions discussed in the main text cancel. We then derive the exact finite-$N$ expression for $\overline{M(t,N)}$ in the CUE. These results highlight the intrinsically many-body origin of the nontrivial TopSFF behaviour, including the emergent $\mathcal{PT}$ transition, which is absent in its zero-dimensional RMT counterpart.

\subsubsection{Large-$N$ Gaussian--non-Gaussian cancellation}
\label{app:circular_unitary_mickey_mouse}

Let $\alpha=(1\,2\,\cdots\,2t)$ and $\beta=(1\,2\,\cdots\,t)(t+1\,t+2\,\cdots\,2t)$ denote the cycle structures associated with the long and short temporal loops, respectively. The correlator can be expanded as
\begin{equation}
\overline{M(t,N)}
=
\sum_{\sigma,\pi\in S_{2t}}
N^{C\left(
\alpha^{-1}\sigma^{-1}\beta\sigma\pi
\right)}
\Wg(\pi; N),
\label{eq:RMT_three_trace_Wg}
\end{equation}
where $C(\rho)$ denotes the number of cycles of permutation $\rho$, while $\pi=\sigma^{-1}\tau$ specifies the Weingarten class.

The Gaussian sector corresponds to $\pi=e$, where $e$ denotes the identity permutation, or equivalently $\sigma=\tau$. Its Weingarten function is
\(
\Wg(\{1^{2t}\};N)=N^{-2t}+O(N^{-2t-2}),
\)
where $\{c_i^{n_i}\}$ denotes a cycle type containing $n_i$ cycles of length $c_i$. The leading Gaussian Mickey Mouse contractions contain $2t-1$ index loops and hence contribute at order $N^{-1}$. There are $2t^3$ such contractions. The leading non-Gaussian sector has $\pi$ equal to a single transposition, with $\Wg(\{2,1^{2t-2}\};N)=-N^{-(2t+1)}+O(N^{-(2t+3)})$. Although its Weingarten factor is suppressed by one additional power of $N^{-1}$, the contraction contains one additional index loop and therefore also contributes at order $N^{-1}$. As discussed in the main text, there are $2t^3$ such diagrams. Thus
\begin{equation}
\overline{M^{\mathrm{G}}}
=
\frac{2t^3}{N}+O(N^{-3}),
\qquad
\overline{M^{\mathrm{nG}}}
=
-\frac{2t^3}{N}+O(N^{-3}), 
\qquad
\overline{M^{\mathrm{G}}}+\overline{M^{\mathrm{nG}}}=O(N^{-3}).
\label{eq:RMT_G_nG}
\end{equation}
Therefore, the cancellation is exact in the large $N$ limit
\begin{equation}
\overline{
\Tr U^{2t}
\left(\Tr U^{*t}\right)^2
}
=0.
\label{eq:RMT_three_trace_exact_zero}
\end{equation}
Thus, the Mickey Mouse contraction topology is present in zero-dimensional CUE, but it vanishes in large $N$ the 0D RMT problem. Spatial interactions assign different propagation and conversion amplitudes to the Gaussian and non-Gaussian sectors, lifting the zero-dimensional cancellation. This distinction underlies the nontrivial Gaussian--non-Gaussian dynamics, including the emergent $\mathcal{PT}$ transition, found in the TopSFF many-body problem.

\subsubsection{Finite-$N$ RMT computation}

We now compute $M(t,N)$ exactly in finite $N$, which in particular demonstrate exact cancellation. Diagonalising the unitary as $U=W\Lambda W^\dagger$, with $\Lambda=\mathrm{diag}(z_1,\ldots,z_N)$ and $z_a=e^{i\theta_a}$, the correlator becomes
\[
\overline{M(t,N)}
=
\overline{
\sum_{a,b,c=1}^{N}
z_a^{2t}z_b^{-t}z_c^{-t}
}.
\]
Separating the sum according to coincidences among $a,b,c$ gives
\begin{equation}
\overline{M(t,N)}
=
2 \overline{K(t)}+\overline{K(2t)}-2N+\overline{C_3(t)},
\label{eq:mickey_decomposition}
\end{equation}
where $K(t):={\Tr U^t\Tr U^{-t}}=\min(t,N)$ is the regular SFF for the CUE, and
$C_3(t):={\sum_{a \neq b \neq c}z_a^{2t}z_b^{-t}z_c^{-t}}$.

The CUE eigenvalues form a determinantal point process, with three-level correlator
\begin{equation}
{R_3(\theta_1,\theta_2,\theta_3)}
=
{
\sum_{a \neq b \neq c}
\delta(\theta_1-\theta_a)
\delta(\theta_2-\theta_b)
\delta(\theta_3-\theta_c)
},
\end{equation}
whose ensemble average can be evaluated in terms of the kernel $\mathscr{D}_N(\theta_i,\theta_j)$,
\begin{equation}
\overline{R_3(\theta_1,\theta_2,\theta_3)}
=
\det_{i,j=1,2,3}\mathscr{D}_N(\theta_i,\theta_j),
\qquad 
\mathscr{D}_N(\theta,\phi)
=
(2\pi)^{-1}\sum_{n=0}^{N-1}e^{in(\theta-\phi)}.
\end{equation}
We write $C_3(t)$ as 
\begin{equation}
\overline{C_3(t)}
=
\int_0^{2\pi}
d\theta_1d\theta_2d\theta_3\,
e^{i(2t\theta_1-t\theta_2-t\theta_3)}
\overline{R_3(\theta_1,\theta_2,\theta_3)}.
\end{equation}
Upon expanding the determinant over $S_3$, the identity permutation and the three transpositions vanish via $\int_0^{2\pi} d\theta \, e^{i m \theta} = 2\pi \delta_{m,0}$, because at least one non-zero Fourier charge remains unmatched. Only the two three-cycles contribute. For the cycle $(123)$, expanding the kernels introduces $n_1,n_2,n_3\in\{0,\ldots,N-1\}$, while the angular integrations over $\theta$-s impose $n_2=n_1+t$ and $n_3=n_1+2t$. 
There are therefore $\max(N-2t,0)$ allowed values of $n_1$. The second three cycle gives the same contribution, yielding
\begin{equation}
C_3(t)=2\max(N-2t,0).
\label{eq:mickey_C3}
\end{equation}
Substitution into Eq.~\eqref{eq:mickey_decomposition} gives the exact finite-$N$ result
\begin{equation}
\overline{
\Tr U^{2t}
\left(\Tr U^{*t}\right)^2
}
=
2\min(t,N)+\min(2t,N)-2N+2(N-2t)_+
=
\begin{cases}
N, & N\leq t,\\[1mm]
2t-N, & t<N<2t,\\[1mm]
0, & N\geq 2t.
\end{cases}
\label{eq:mickey_exact_piecewise}
\end{equation}
In particular, for $N\geq 2t$, $\overline{M(t,N)}$ is exactly zero, reproducing Gaussian--non-Gaussian  cancellation in large $N$ in Eq.~\eqref{eq:RMT_three_trace_exact_zero}.

\section{Review of exact solutions of SFF without topological defects}\label{app:review_sff}

In this appendix, we collect exact large-\(q\) results for the ordinary SFF, without topological defect insertions, in the three symmetry settings relevant to this work: parity inversion symmetry, time reversal symmetry, and discrete time translation symmetry.

\subsection{Parity inversion symmetry}
Here we review the analytical solutions of SFF in the parity inversion symmetric many-body quantum chaotic models derived in~\cite{upcoming}. The approach can be used to evaluate $d$-dimensional parity inversion symmetric RPM models defined in \ref{app:parity_ptRPM}. In particular, for one-dimensional parity inversion symmetric RPM with periodic boundary condition, the SFF can be evaluated exactly in large-$q$ limit as 
\begin{equation} \label{eq:regular_sff_parity}
\begin{aligned}
 \overline{K_{\text{r-p-RPM}}(t,L)}
     =
  2 (1 + e^{-2 \epsilon t} )^{\Leff}  \, .
\end{aligned}
\end{equation}
where $\Leff$ counts the number of bonds in the folded picture as described in the main text, with $\Leff= L/2-1$ for even $L$, and $\Leff= (L-1)/2-1$ for odd $L$. $\epsilon$ is the variance of the random phases in the RPM, which controls the coupling strength of the many-body system.
For Floquet parity inversion symmetric RPM with periodic boundary condition we have 
\be\label{app_eq:sff_parity}
\ba
 &  \overline{K_{\text{f-p-RPM}}(t,L)}   =  (2 t-2) (e^{-\epsilon t}-1)^2 \left(e^{-\epsilon t}-1\right)^{2 \Leff} 
 \\
 & \quad 
 +
 2 [(t-1) e^{-\epsilon}+1]^2  \left[\left(2 t^2-2 t+1\right) e^{-2\epsilon t}+2 (t-1) e^{-\epsilon t}+1\right]^{\Leff} \,. 
\ea
\ee
where $\Leff$ and $\epsilon$ are defined as above.

\subsection{TRS}
Here we review the analytical solutions of SFF in the TRS many-body quantum chaotic models derived in~\cite{wu2025}. The approach can be used to evaluate $d$-dimensional TRS RPM models defined in \ref{app_sec:trs_rpm}. In particular, for one-dimensional global TRS RPM with periodic boundary condition, the SFF can be evaluated exactly in large-$q$ limit as 
\begin{equation} \label{eq:sff_gtrs_nomag}
\begin{aligned}
    \overline{K_{\text{trs-RPM}}(t,L)}
   =
   \left[1+e^{-\epsilon \left( t-1-\delta_{0,t \,(\text{mod}\,{2})}\right)}\right]^L
   +
   \left[1-e^{-\epsilon \left( t-1-\delta_{0,t \,(\text{mod}\,{2})}\right)}\right]^L .
\end{aligned}
\end{equation}
where $\epsilon$ is the variance of the random phases in the RPM.

\subsection{Time translational symmetry}
Here we review the analytical solutions of SFF in  many-body quantum chaotic models with discrete time translational symmetry (but without other symmetries) in~\cite{chan2018spectral}. The approach can be used to evaluate $d$-dimensional RPM models with time translational symmetry defined in \ref{app:model_time_trans}. In particular, for one-dimensional RPM with periodic boundary condition, the SFF can be evaluated exactly in large-$q$ limit as 
\begin{equation} \label{eq:sff_time_Trans}
\begin{aligned}
    \overline{K_{\text{RPM}}(t,L)}
   =(t -1)(1 -e^{-\epsilon t})^L + [1 + (t - 1)e^{-\epsilon t}]^L
\end{aligned}
\end{equation}
where, again, $\epsilon$ is the variance of the random phases in the RPM.

\section{SFF with twisted boundary conditions}\label{app:sff_twisted}
For completeness, we contrast the TopSFF with the SFF under twisted boundary
conditions in space or in time. A twisted boundary condition inserts the same defect in both \(U\) and \(U^{*}\), leaving the two world-sheets topologically matched.  As a result, these observables do not isolate the domain-wall sectors
responsible for the universal scaling forms and emergent \(\mathcal{PT}\) transition discussed in the main text.

\subsection{Spatially extended twisted boundary condition}

We consider the SFF of a generic quantum many-body chaotic system with a spatially extended twisted boundary condition,
\(
K^{\mathrm{twisted}}_{\mathrm{parity}}
=
\left|\Tr[\hat S\hat U_{\mathrm{t}}(t,L)]\right|^2 .
\)
Here \(\hat S\) is the global swap operator acting in the time direction. The generic chaotic system is realised by the RPM with parity inversion and discrete time translation symmetries, defined in Eq.~\eqref{app_eq:floq_trs_rpm}, which satisfies
\([\hat S,\hat U_{\mathrm{t}}(t,L)]=0\). 
The tensor network for \(K^{\mathrm{twisted}}_{\mathrm{parity}\text{-}\mathrm{RPM}}\) consists of two space-time world-sheets with identical topology. Under the folding and unwrapping procedure of App.~\ref{app:tsff_toolbox}, the local diagrammatics contains two long loops, from \(\Tr[\hat S\hat U^t]\) and \(\Tr[\hat S\hat U^{*t}]\), but without Mickey Mouse structure or associated temporal domain walls of the type described in the main text and App.~\ref{app:leading_local_mm_diagram}. Applying the methods of App.~\ref{app:tsff_toolbox}, one computes exactly in the large-\(q\) limit
\begin{equation}
\overline{K^{\mathrm{twisted}}_{\mathrm{parity}\text{-}\mathrm{RPM}}}
=
\mathcal{Z}_{2t\text{-}\mathrm{Potts}}
=
2t\left(2t-1+e^{-2\epsilon t}\right)^{L/2-1},
\end{equation}
where
\(Z_{2t\text{-}\mathrm{Potts}}=\phi^T\boltz^{L/2-1}\phi\)
is the partition function of a \(2t\)-state Potts model with open boundary conditions. Here
\(
[\boltz]_{mn}
=
1+\delta_{mn}\left(e^{-2\epsilon t}-1\right)\) with 
\(m,n=1,\dots,2t\),
and 
\(
\phi=(1,1,\dots,1)^T
\).
This observable closely resembles the SFF of the parity-invariant RPM with open boundary conditions, cf.~Eq.~\eqref{app_eq:sff_parity}, and does not isolate temporal domain walls responsible for the universal scaling captured in TopSFF.

\subsection{Temporally extended twisted boundary condition}
We consider the SFF of a generic quantum many-body chaotic system with a temporally extended twisted boundary condition,
\(K^{\mathrm{twisted}}_{\mathrm{trs}}=|\Tr[\hat S\hat U_{\mathrm{s}}(t,L)]|^2\).
Here \(\hat S\) is the global swap operator acting in the spatial direction, and
\(\hat U_{\mathrm{s}}(t,L):=\prod_{i=1}^{L}\hat V_i(t)\) is generated by the dual transfer matrices \(\hat V_i(t)\) [Fig.~\ref{fig:tdw_mickey_mouse}(a), purple] of the RPM defined in Eq.~\eqref{app_eq:global_trs_rpm}. The RPM satisfies \([\hat V_i(t),\hat S]=0\), and therefore has global time reversal symmetry. The tensor network for \(K^{\mathrm{twisted}}_{\mathrm{trs}}\) consists of two space-time world-sheets with identical topology. Using the methods described in App.~\ref{app:tsff_toolbox} and \cite{wu2025}, one finds in the large-\(q\) limit that
\begin{equation}\label{app_eq:sff_twisted_trs}
\overline{K^{\mathrm{twisted}}_{\mathrm{trs}\text{-}\mathrm{RPM}}}
=
\overline{K_{\mathrm{trs}\text{-}\mathrm{RPM}}}\, ,
\end{equation}
where the analytical expression of \(\overline{K_{\mathrm{trs}\text{-}\mathrm{RPM}}}\) is given in \eqref{eq:sff_gtrs_nomag} and \cite{wu2025}. The exactness of the equality in \eqref{app_eq:sff_twisted_trs} is model-dependent, and generally relies on the defect operator commuting with all symmetries of the quantum many-body system. Nevertheless, we expect Eq.~\eqref{app_eq:sff_twisted_trs} to hold approximately in generic chaotic systems, and the SFF of twisted boundary condition does not isolate spatial domain walls responsible for the universal scaling captured in TopSFF.

\section{Generalized Boltzmann factor of TopSFF with spatially extended defects}\label{app:tsff_toolbox}

Here we analytically evaluate the ensemble average of TopSFF 
\begin{equation}\label{app_eq:tsff_floq_par_def}
    K_{\topo}(t,L)= \mathrm{Tr}[ \hat{\mathcal{D}} (\hat{U}\otimes \hat{U}^*) ] = \Tr [\hat{S} \hat{U}] \Tr[\hat{U}^*]  \, ,   \qquad [\hat{\mathcal{D}}, (\hat{U}\otimes \hat{U}^*)] =0\, , 
\end{equation}
where the spatially extended topological defect is $\hat{\mathcal{D}}=\hat{S}\otimes\iden$.
Here $\hat{S}$ is the global swap operator, which reverses the ordering of the $L$ qudits according to $\hat{S}\ket{a_1,a_2,\dots,a_{L-1},a_L}=\ket{a_L,a_{L-1},\dots,a_2,a_1}$, with $a_i=1,\dots,q$. We analyse $K_{\topo}$ in a generic many-body quantum chaotic system satisfying $[\hat{\mathcal{D}},\hat{U}\otimes \hat{U}^*]=0$, i.e. the system possesses parity inversion symmetry.
Such a system is minimally realised by the Random Phase Model (RPM) with parity inversion symmetry and discrete time translation symmetry, defined in App.~\ref{app:parity_ptRPM}, i.e., $\hat{U}=\hat{U}_{\mathrm{f\text{-}p\text{-}RPM}}$.
The local Hilbert-space dimension is denoted by $q$.
Our calculation uses the diagrammatic representation of random quantum circuits introduced in Ref.~\cite{chan2018solution}.
We provide exact analytic results in the large-\(q\) limit, exact analytical results for finite \(q\) and small \(t\), and numerical simulations at finite \(q\) and small \(t\).
For simplicity, the derivation below is presented for even system size $L$, although the method applies equally to both even and odd $L$.

In Appendix \ref{app:leading_local_mm_diagram}, we identify the leading local
TopSFF diagrams in the large-\(q\) expansion. These diagrams form the
topological equivalence class of the Mickey Mouse diagram and are
parametrized by \(\rho=(a,b,c,d)\), where \(a=0,1\) distinguishes
Gaussian and non-Gaussian temporal domain walls.
In \ref{app:boltz_pos_space}, we derive the generalized Boltzmann factor
\(\boltz\) underlying TopSFF in position space, and in
\ref{app:boltz_pos_prop} we describe its main properties, including
translation symmetries, exchange symmetries, and reality constraints.
In \ref{app:boltz_mom_space}, we Fourier transform the generalized
Boltzmann factor, and in \ref{app:boltz_mom_prop} we discuss its
momentum space properties, including its trace, determinant, eigenvalue
structure, and special momentum sector features.
In \ref{app:pt_sym}, we prove that, for generic interaction strength,
\(\tilde{\boltz}(k_b,k_c,k_d)\) is \(\mathcal{PT}\) symmetric precisely
when \(k_b+k_c+k_d\equiv 0 \pmod{\pi}\). In \ref{app:boundary_state}, we evaluate the boundary states in the folded representation in both position and momentum space.

\subsection{Mickey Mouse diagrams}\label{app:leading_local_mm_diagram}
\textbf{Tensor network, folding and unwrapping.} We represent the TopSFF diagrammatically using the tensor network representation, see e.g.~\cite{Or_s_2014}. 
The quantum circuit \(U_{\text{f-p-RPM}}\) in Eq.~\eqref{app_eq:par_rpm_temprand} is represented as in Fig.~\ref{fig:parity_models}(b), where time runs vertically from bottom to top and space horizontally from left to right. Each line denotes the  worldline of a single qudit, each square denotes a local unitary gate drawn from the circular unitary ensemble (CUE), and each oval denotes a two-site coupling gate. The trace operation in the TopSFF is represented by  forming a close loop from the qudit worldline at each site. 
The TopSFF is represented schematically in Fig.~\ref{fig:folded_v1}(a) using two world-sheets:\(\Tr[\hat{S}\hat{U}(t)]\), shown in white, and \(\Tr[\hat{U}^{\dagger}(t)]\), shown in grey.
The open legs at the top are identified with the corresponding legs at the bottom, thereby forming the trace.
The parity-inversion axes are shown in red and blue, with axes of the same colour identified.

We fold the world-sheet in half along the parity inversion axis as illustrated in Fig.~\ref{fig:folded_v1}(b), so that gates related by the parity inversion symmetry are stacked on top of each other after folding.
Specifically, each world-sheet, $\Tr_{\mathrm{t}}[\hat{S} \hat{U}(t,L)] $ or $\Tr_{\mathrm{t}}[\hat{U}(t,L)]$, is divided into two halves along the inversion axes, labeled 1 to 4. Folding according to Fig.~\ref{fig:folded_v1} (a,b) produces a representation in which every space-time coordinate contains two identical unitary gates together with their two complex conjugates.
In this folded picture, we construct an effective spatial chain by grouping all gates at sites $i$ and $L-i$ from both $\Tr_{\mathrm{t}}[\hat{S}\hat{U}(t,L)]$ and $\Tr_{\mathrm{t}}[\hat{U}(t,L)]$. For even $L$, the folded chain has length $L/2$.
The global swap operator now realizes a local operation exchanging the replica labelled by 1 and 3 in Fig.~\ref{fig:folded_v1}. In the folded representation, the swap operator locally commutes with the parity symmetric many-body quantum systems. 
\begin{figure}
    \centering
    \includegraphics[width=0.99\textwidth]{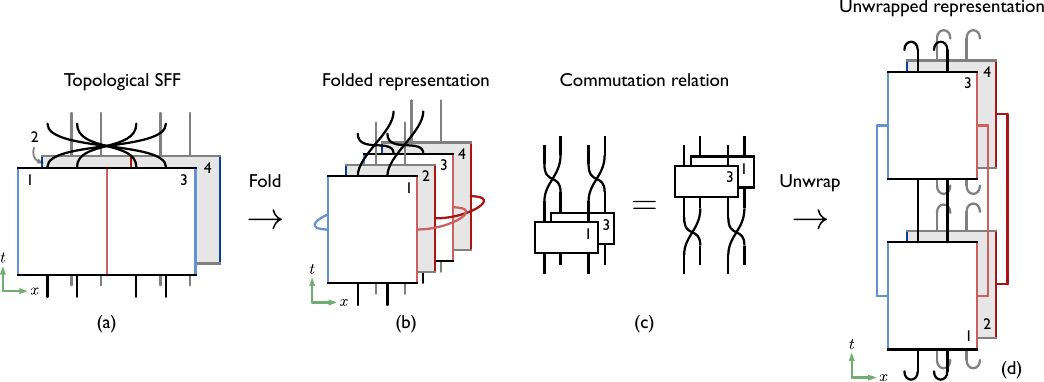}
    \caption{\textbf{Folding and unwrapping of TopSFF with spatially extended topological defects.} 
    (a) TopSFF $K_{\topo}(t,L)= \mathrm{Tr}[ \hat{\mathcal{D}} (\hat{U}\otimes \hat{U}^*) ] = \Tr [\hat{S} \hat{U}] \Tr[\hat{U}]$  is represented using two world-sheets: $\Tr[\hat{S} \hat{U}(t)]$ in white, and  $\Tr[\hat{U}(t)^\dagger]$ in grey. The open legs at the top are identified with the corresponding legs at the bottom, forming a trace. 
    The parity inversion axes are shown in reds and blues, with the axes of the same colors identified. 
    (b) For parity symmetric one-dimensional many-body quantum circuits, the folded representation stacks identical gates related by parity inversion symmetry in the world-sheets. 
    (c) In the folded representation, gates locally commute with the swap operator, which acts as the topological defect. 
    (d) We further unwrap the folded world-sheets in (b), yielding an equivalent representation in which each folded site contains three loops.
    }
    \label{fig:folded_v1}
\end{figure}

\begin{figure}
    \centering
    \includegraphics[width=0.99\textwidth]{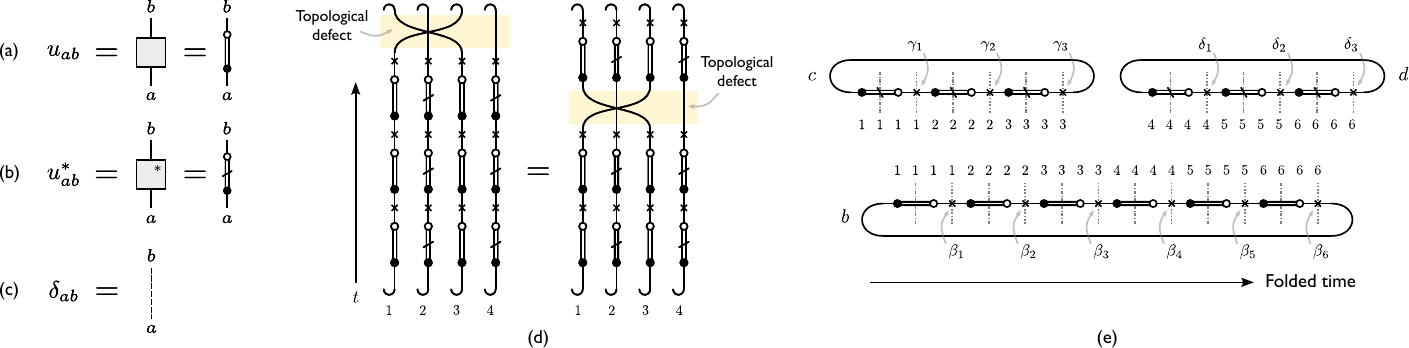}
    \caption{\textbf{Diagrammatic representation of the TopSFF before averaging} in the presence of a spatially extended topological defect.
For parity-symmetric quantum many-body systems, we use the graphical notation for
(a) unitary gates, (b) complex conjugated unitary gates, and (c) Kronecker-delta contractions.
Panel (d) shows the local tensor network diagrams contributing to the TopSFF at $t=3$ in Fig.~\ref{fig:folded_v1}(b), with the spatially extended topological defect highlighted in yellow.
Panel (e) shows the corresponding local diagrams obtained from unwrapping the folded representation in Fig.~\ref{fig:folded_v1}(d).
    }
    \label{fig:before_avg}
\end{figure}

\begin{figure}
    \centering
    \includegraphics[width=0.99\textwidth]{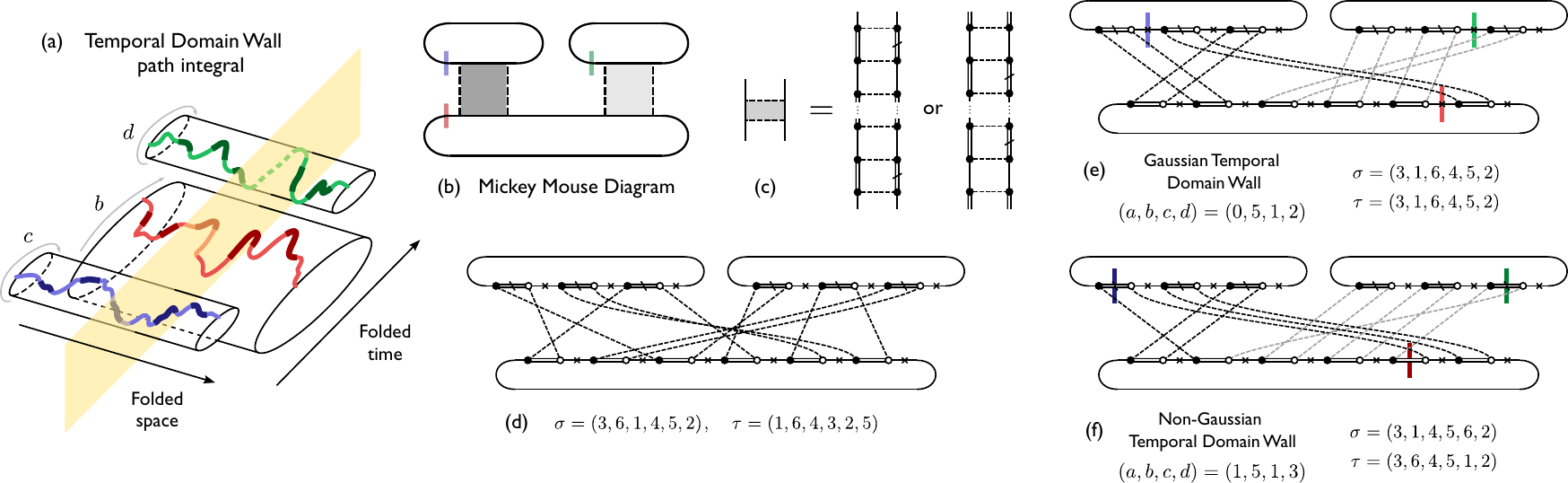}
\caption{
\textbf{Diagrammatic representation of the TopSFF after averaging} in the presence of a spatially extended topological defect for parity-symmetric quantum many-body systems.
(a) Ensemble averaging of TopSFF leads to an effective statistical mechanical problem of temporal domain walls, which lives on the loops labelled by $b$, $c$, and $d$.
(b) The leading local diagrams are the ``Mickey Mouse'' diagram, where the dashed ribbons represent a series of contractions as illustrated in (c).
(d) A representative of an ensemble averaged local diagram, labelled by the permutation data \(\sigma\) and \(\tau\).
In the limit of large local Hilbert space dimension $q$, the leading diagrams are (e) Gaussian temporal domain wall diagram ($a=0$) and 
(f) non-Gaussian temporal domain wall diagram ($a=1$).
}
    \label{fig:after_avg}
\end{figure}

\textbf{Ensemble average of Haar-random unitaries.}
The RPM ensemble average is performed in two steps: first, by integrating over the Haar-random unitary gates, or equivalently over the circular unitary ensemble (CUE), and second, by averaging over the random phases in the coupling gates.
For \(q \times q\) unitary matrices \(u\) drawn from the Haar measure on the unitary group, the ensemble average is given by
\be\label{eq:Weingarten_raw_form}
\overline{u_{a_1 b_1}\cdots u_{a_n b_n} 
        u^{\ast}_{c_1 d_1}\cdots u^{\ast}_{c_n d_n} }
    = \sum_{\sigma,\tau \in S_n} 
      \mathrm{Wg}(\sigma^{-1}\tau,q) 
      \prod_{j=1}^n 
        \delta_{a_j, c_{\sigma(j)}} 
        \delta_{b_j, d_{\sigma'(j)}} \, .
\ee
Here \(S_n\) is the symmetric group on \(n\) elements. The Weingarten function \(\mathrm{Wg}\) depends only on the cycle structure \(\{c_1,\dots,c_k\}\) of the permutation \(\sigma^{-1}\tau\), where \(c_\ell\) denotes the length of the \(\ell\)-th cycle. For example, for \(n=1\), one finds \(\mathrm{Wg}(\{1\},q)=1/q\), while for \(n=2\), the two possible cycle structures give \(\mathrm{Wg}(\{1,1\},q)=1/(q^2-1)\) and \(\mathrm{Wg}(\{2\},q)=-1/[q(q^2-1)]\).
Explicit expressions for \(\mathrm{Wg}\) can be obtained recursively~\cite{Brouwer1996}.
Importantly, the asymptotic behaviour of the Weingarten function depends only on the total number of cycles of \(\sigma^{-1}\tau\), 
\(\wg[\{c_1,c_2,\dots,c_k\};q]=O(q^{k-2\sum_{i=1}^k c_i})\).
We call a permutation pair \((\sigma,\tau)\) \textit{Gaussian} if all loop lengths of $\sigma^{-1}\tau$ are one, i.e. \(c_m=1\) for all \(m\), and \textit{non-Gaussian} otherwise.

The ensemble average over Haar-random unitary matrices admits a natural diagrammatic representation~\cite{Brouwer1996} as illustrated in Fig.~\ref{fig:before_avg}, which is used to analyse many-body quantum chaotic systems in~\cite{chan2018solution}. 
Specifically, from the tensor network diagrammatic representation of observables like TopSFF, each unitary gate drawn from the CUE is replaced by a pair of filled and empty dots, connected by double lines, as shown in Fig.~\ref{fig:before_avg}(a). Its complex conjugate is similarly represented except that a slash is added to distinguish the conjugation as in Fig.~\ref{fig:before_avg}(b). Delta functions in \eqref{eq:Weingarten_raw_form} are represented by dashed lines in Fig.~\ref{fig:before_avg}(c), which we will refer to as contractions. The random phase gates, denoted by crosses,  couple a local diagram at a given site with another local diagram at another site. Together, in the folded representation, the local diagrams in a single effective site are given by Fig.~\ref{fig:before_avg}(d) after folding, and by Fig.~\ref{fig:before_avg}(e) after unwrapping.

The permutation pair \((\sigma,\tau)\in S_{2t}\times S_{2t}\) associated with each summand in Eq.~\eqref{eq:Weingarten_raw_form} has a direct diagrammatic interpretation: \(\sigma\) specifies the contractions between the black dots of the unitary gates and those of the complex-conjugated unitary gates, while \(\tau\) specifies the corresponding contractions between the white dots. We label the black and white dots of the unitary gates along the \(b\)-loop by \(1,2,\ldots,2t\), as shown in Fig.~\ref{fig:before_avg}(e). The black and white dots of the complex-conjugated unitary gates are labelled by \(1,\ldots,t\) along the \(c\)-loop and by \(t+1,\ldots,2t\) along the \(d\)-loop. With this convention, the permutation \(\sigma\) maps the \(m\)-th black dot of the unitary gates to the \(\sigma(m)\)-th black dot of the complex-conjugated unitary gates, while \(\tau\) analogously maps the \(m\)-th white dot to the \(\tau(m)\)-th white dot. Hereafter, we write permutations in one-line notation, e.g. \(\sigma=(\sigma(1),\ldots,\sigma(2t))\in S_{2t}\), and similarly for \(\tau\).

A local TopSFF diagram is algebraically represented by a permutation pair
\((\sigma,\tau)\in S_{2t}\times S_{2t}\), where \(\sigma\) specifies the
contractions of black dots and \(\tau\) specifies the contractions of white dots. There is a natural topological equivalence relation~\cite{wu2025} among subsets of diagrams with identical scaling order in the local Hilbert-space dimension \(q\). Intuitively, this relation is obtained by first suppressing the distinction between black and white dots, and then identifying diagrams related by cyclic rotations along the \(b\)-, \(c\)-, and \(d\)-loops. 

\begin{definition}
\label{app_def:topo}
\textbf{Topological equivalence classes of TopSFF diagrams.}
Given a local TopSFF diagram represented by
\((\sigma,\tau)\in S_{2t}\times S_{2t}\), we construct a single uncoloured
contraction permutation \(\Pi_{\sigma,\tau}\in S_{4t}\) as follows. In
Fig.~\ref{fig:before_avg}(e), we order the \(4t\) endpoints on the \(b\)-loop
by interlacing black and white dots, so that the \((2m-1)\)-st and \(2m\)-th
endpoints correspond to the black and white dots at position \(m\), respectively.
We use the same convention for the \(c\)- and \(d\)-loops, with the first
\(2t\) endpoints belonging to the \(c\)-loop and the last \(2t\) endpoints
belonging to the \(d\)-loop. The uncoloured contraction permutation associated
with \((\sigma,\tau)\) is then defined by
\begin{equation}
\Pi_{\sigma,\tau}(2m-1)=2\sigma(m)-1,\qquad
\Pi_{\sigma,\tau}(2m)=2\tau(m),
\qquad m=1,\ldots,2t ,
\end{equation}
i.e., \(\Pi_{\sigma,\tau}\) records the contraction pattern after the
black/white distinction has been removed. Let \([x]_n:=1+\operatorname{Mod}(x-1,n)\). We define cyclic rotations of the
uncoloured endpoints by
\begin{equation}
R_b^{(r)}(i):=[i+r]_{4t},
\end{equation}
for the \(b\)-loop, and
\begin{equation}
R_c^{(s)}(i):=
\begin{cases}
[i+s]_{2t}, & 1\leq i\leq 2t,\\
i, & 2t+1\leq i\leq 4t,
\end{cases}
\qquad
R_d^{(u)}(i):=
\begin{cases}
i, & 1\leq i\leq 2t,\\
2t+[i-2t+u]_{2t}, & 2t+1\leq i\leq 4t,
\end{cases}
\end{equation}
for the \(c\)- and \(d\)-loops, respectively.
Two local TopSFF diagrams \((\sigma,\tau)\) and \((\sigma',\tau')\) are said to
belong to the same \textit{topological equivalence class} if their uncoloured
contraction permutations are related by cyclic rotations of the three loops,
i.e. if there exist integers \(r,s,u\) such that
\begin{equation}
\Pi_{\sigma',\tau'}
=
R_c^{(s)}
\circ
R_d^{(u)}
\circ
\Pi_{\sigma,\tau}
\circ
R_b^{(r)} .
\end{equation}
Here composition is understood from right to left. Equivalently, the
topological class of \((\sigma,\tau)\) is the orbit of
\(\Pi_{\sigma,\tau}\) under independent cyclic rotations of the \(b\)-,
\(c\)-, and \(d\)-loops.
\end{definition}

This definition preserves the cyclic ordering of endpoints on each loop, but
suppresses the black/white distinction. In particular, Gaussian and
non-Gaussian contraction patterns may belong to the same topological equivalence
class. Note that, depending on the symmetries of the ensemble, additional
operations may also be included, such as generalizations of relative twists between
\(\Tr[\hat U(t)]\) and \(\Tr[\hat U^\dagger(t)]\) in Ref.~\cite{wu2025}, in which topological equivalence classes in systems with time reversal
symmetry are defined.

\vspace{0.5cm}
\noindent\textit{Example.}
For panels (e) and (f) in Fig.~\ref{fig:after_avg}, we have \(t=3\).
The Gaussian diagram in panel (e) has
\(\sigma_{(e)}=\tau_{(e)}=(3,1,6,4,5,2)\), while the non-Gaussian diagram
in panel (f) has
\(\sigma_{(f)}=(3,1,4,5,6,2)\) and
\(\tau_{(f)}=(3,6,4,5,1,2)\).
The corresponding uncoloured contraction permutations, written in one-line
notation, are
\(\Pi_{(e)}=(5,6,1,2,11,12,7,8,9,10,3,4)\) and
\(\Pi_{(f)}=(5,6,1,12,7,8,9,10,11,2,3,4)\).
They are related by the explicit cyclic rotations
\(\Pi_{(f)}=R_c^{(5)}\circ R_d^{(1)}\circ \Pi_{(e)}\circ R_b^{(1)}\).
Thus panels (e) and (f) lie in the same topological equivalence class once the
black/white distinction is forgotten. The fact that panel (e) has
\(\sigma=\tau\), whereas panel (f) has \(\sigma\neq\tau\), reflects the distinction between Gaussian and non-Gaussian temporal domain walls, but they both belong to the Mickey Mouse topological equivalence class.

\textbf{Leading local diagrams -- Mickey Mouse diagrams.} 
Let \(a\in\{0,1\}\), \(b\in\{1,\ldots,2t\}\), and \(c,d\in\{1,\ldots,t\}\). We use the one-based modulo convention \([x]_n:=1+\operatorname{Mod}(x-1,n)\), so that \([x]_n\in\{1,\ldots,n\}\). Let \(\rho_n\) denote the right cyclic shift on an \(n\)-tuple, \(\rho_n(x_1,x_2,\ldots,x_n):=(x_n,x_1,\ldots,x_{n-1})\), and let \(\rho_n^r\) denote its \(r\)-fold composition, with powers to be modulo \(n\). For two tuples \(X=(x_1,\ldots,x_m)\) and \(Y=(y_1,\ldots,y_n)\), we define their concatenation by \(X\Vert Y:=(x_1,\ldots,x_m,y_1,\ldots,y_n)\).
Denote the $c$-fold right cyclic shift on the tuple $(1,2,\dots, t)$, and $d$-fold right cyclic shift on the tuple $(t+1,t+2,\dots, 2t)$ respectively as
\[
f_c:=\rho_t^{\,c-1}(1,2,\ldots,t),
\qquad
f_d:=\rho_t^{\,d-1}(t+1,t+2,\ldots,2t),
\]
and the $b$-fold-shifted and concatenated object as 
\[
P(b,c,d)
:=
\rho_{2t}^{\,b-1}\bigl(f_c\Vert f_d\bigr) \in S_{2t}.
\]

\begin{lemma}
\label{app_lemma:leadingtsff}
\textbf{Leading local topological spectral form factor (TopSFF) diagrams for generic  many-body quantum chaotic systems with spatially extended topological defects.} Consider the ensemble average of the TopSFF $K_{\topo}(t,L)= \mathrm{Tr}[ \hat{\mathcal{D}} (\hat{U}\otimes \hat{U}^*) ] = \Tr [\hat{S} \hat{U}] \Tr[\hat{U}]$, where the spatially extended topological defect is $\hat{\mathcal{D}}=\hat{S}\otimes\iden$ with global swap operator $\hat{S}$. For the parity-inversion-symmetric Random Phase Model with discrete time translation symmetry, defined in App.~\ref{app:parity_ptRPM}, the leading local TopSFF diagrams in the large-\(q\) expansion are the permutation pairs
\((\sigma_{abcd},\tau_{abcd})\in S_{2t}\times S_{2t}\) given by
\begin{equation}
\begin{aligned}\label{app_eq:mickey_mouse_permutations}
\sigma_{abcd}=& \, P(b,c,d),
\\
\tau_{abcd}
=& 
\begin{cases}
P(b,c,d),
& a=0,\\[1mm]
P\!\left([b-1]_{2t},[c+1]_{t},[d+1]_{t}\right),
& a=1.
\end{cases}    
\end{aligned}
\end{equation}
with \(a\in\{0,1\}\), \(b\in\{1,\ldots,2t\}\), and \(c,d\in\{1,\ldots,t\}\). There are in total $4t^3$ leading TopSFF diagrams. These leading local TopSFF diagrams are topologically equivalent to the Mickey Mouse diagram shown in Fig.~\ref{fig:after_avg}(b), represented by the equivalence class of the reference diagram
\begin{equation}
(\sigma_{0111},\tau_{0111}) .
\end{equation}
\end{lemma}

\noindent \textit{Examples.} A Gaussian leading local TopSFF diagram of Floquet TRS RPM at $t=3$ is given by Fig.~\ref{fig:after_avg} (e) with $(a,b,c,d)= (0,5,1,2)$ and $\sigma = \tau = (3,1,6,4,5,2)$. A non-Gaussian leading local TopSFF diagram of Floquet TRS RPM at $t=3$ is given by Fig.~\ref{fig:after_avg} (f) with $(a,b,c,d)= (1,5,1,3)$, $\sigma = (3,1,4,5,6,2)$, and $\tau = (3,6,4,5,1,2)$. An example of non-leading local TopSFF diagrams is given by Fig.~\ref{fig:after_avg} (d).

\vspace{0.5cm}
\begin{proof}
The proof utilizes techniques developed in \cite{chan2018solution}. The representative diagrams shown in Fig.~\ref{fig:after_avg}(e,f) are of order \(O(q^{-1})\), so the leading local diagrams are at least of order \(O(q^{-1})\). From the diagram before averaging in Fig.~\ref{fig:before_avg}, observe that any admissible local TopSFF diagram must contain at least two inter-loop contractions: one connecting the \(b\)-loop to the \(c\)-loop, and one connecting the \(b\)-loop to the \(d\)-loop. By the contraction-addition argument of Ref.~\cite{chan2018solution}, this implies that every local TopSFF diagram is of order at most \(O(q^{-1})\). It remains to identify the diagrams that saturate this bound. Starting from the minimal inter-loop contractions, every subsequent contraction addition must preserve the \(q\)-order. The only order-preserving additions are of ladder type, as illustrated in Fig.~\ref{fig:after_avg}(b): consecutive endpoints on the \(c\)-loop must be contracted to consecutive endpoints on the \(b\)-loop, and similarly consecutive endpoints on the \(d\)-loop must be contracted to consecutive endpoints on the \(b\)-loop. Hence every leading local TopSFF diagram lies in the topological equivalence class of the Mickey Mouse diagram in Fig.~\ref{fig:after_avg}(b), equivalently in the same topological class as the representative diagrams in Fig.~\ref{fig:after_avg}(e,f).
To enumerate all leading TopSFF diagrams in the topological equivalence class of the Mickey Mouse diagram, we consider the contraction $\Pi_{\sigma,\tau}$ defined in \ref{app_def:topo}. 
Given a Mickey Mouse diagram, $\Pi_{\sigma,\tau}$, we observe that there is a
unique adjacent pair of endpoints $(i,i-1)$ such that the dots $i$ and $i-1$ from the \(b\)-loop are contracted to a dot from the \(c\)-loop, and a dot from the \(d\)-loop respectively. Equivalently, we identify $i$-th dot such that $1\leq \Pi_{\sigma,\tau}(i)\leq 2t$ and $2t+1\leq \Pi_{\sigma,\tau}(i-1)\leq 4t$. $i$ can either be even (a black dot) or odd (a white dot), which we label as $a=0$ and $a=1$. See Fig.~\ref{fig:after_avg}(e,f) for examples with $a=0$ and $a=1$ respectively. For each value of \(a\), independent cyclic rotations of the three loops produce further valid leading diagrams: there are \(2t\) rotations of the
\(b\)-loop, \(t\) rotations of the \(c\)-loop, and \(t\) rotations of the
\(d\)-loop. Therefore the total number of leading local TopSFF diagrams is
\(2\times (2t)\times t\times t = 4t^3\). Using the notation \(a\in\{0,1\}\), \(b\in\{1,\ldots,2t\}\), and \(c,d\in\{1,\ldots,t\}\), these rotations give precisely the permutation pairs \((\sigma_{abcd},\tau_{abcd})\) stated in
Eq.~\eqref{app_eq:mickey_mouse_permutations}. This completes the proof.
\end{proof}

\textbf{Temporal domain walls (tDW).} 
Given a Mickey Mouse diagram represented by
the permutation \(\Pi_{\sigma,\tau}\), there is a unique
adjacent pair of endpoints \((i,[i-1]_{4t})\) along the \(b\)-loop such that the
endpoint \(i\) is contracted to the \(c\)-loop, while the endpoint
\([i-1]_{4t}\) is contracted to the \(d\)-loop. Equivalently, \(i\) is
identified by \(1\leq \Pi_{\sigma,\tau}(i)\leq 2t\) and \(2t+1\leq \Pi_{\sigma,\tau}([i-1]_{4t})\leq 4t\). With the convention in Fig.~\ref{fig:before_avg}(e), $i$ can either be even (a black dot) or odd (a white dot), which we label as $a=0$ and $a=1$. We can associate a domain wall at the bond $(i,[i-1]_{4t})$, and we call this domain wall a \textit{temporal domain wall} (tDW), since it separates two domains along folded time. Recall that we call a permutation pair \((\sigma,\sigma')\) Gaussian if all loop lengths of $\sigma^{-1}\tau$ are one, i.e. \(c_m=1\) for all \(m\), and non-Gaussian otherwise. The index \(a\) determines the Gaussian or non-Gaussian character of the tDW diagram:
\begin{itemize}
\item \textit{Gaussian tDW} \((a=0)\): the diagrams
    \(\{(\sigma_{0bcd},\tau_{0bcd})\}\), where the associated Weingarten
    factor is \(\wg(\{1^{2t}\};q)\) and positive at the leading order in the local Hilbert space dimension \(q\).
    \item \textit{Non-Gaussian tDW} \((a=1)\): the diagrams
    \(\{(\sigma_{1bcd},\tau_{1bcd})\}\), where the associated Weingarten
    factor is \(\wg(\{1^{2t-2},2\};q)\) and negative at the leading order in the local    Hilbert space dimension \(q\).
\end{itemize}
Representative examples of a Gaussian tDW diagram with \(a=0\) and a
non-Gaussian tDW diagram with \(a=1\) are shown in
Fig.~\ref{fig:after_avg}(e,f), respectively, where the tDW
on the \(b\)-loop are indicated by red lines.

\subsection{Boltzmann factor $\boltz$ in position space}\label{app:boltz_pos_space}
We now consider the average over two-site random coupling gates, acting on two neighbouring large-\(q\) TopSFF configurations
\(\rho_1=(a_1,b_1,c_1,d_1)\) and \(  \rho_2=(a_2,b_2,c_2,d_2)\). Here \(a_i\in\{0,1\}\), \(b_i=1,\ldots,2t\), and \(c_i,d_i=1,\ldots,t\).  The label \(a_i=0\) denotes a Gaussian tDW, while \(a_i=1\) denotes a non-Gaussian tDW.  We use the  modulo convention
\(
    [m]_n:=1+\operatorname{Mod}(m-1,n)
\).
Each configuration \(\rho=(a,b,c,d)\) is associated with an ordered contraction list
\(
    \mathcal C_{bcd}
    =
    \{(m,n)\}
\). This is an ordered list of pairs
\((m,n)\), where \(m\in\{1,\ldots,2t\}\) labels an endpoint on the
\(b\)-loop, indicated by the dotted lines in Fig.~\ref{fig:before_avg}(e),
while \(n\) labels the corresponding endpoint on the combined \(c\)- and
\(d\)-loops. These identifications arise from the delta functions in
\eqref{eq:Weingarten_raw_form}. Explicitly, the contraction pattern
associated with \(\rho=(a,b,c,d)\) is given by
\begin{align}
    \mathcal C_{bcd}
    :=
    &\left\{
    \left([b-1+i]_{2t},[c-1+i]_{t}\right)
    \,:\,
    i=1,\ldots,t
    \right\}
    \nonumber\\
    &\cup
    \left\{
    \left([b-1+i+t]_{2t},t+[d-1+i]_{t}\right)
    \,:\,
    i=1,\ldots,t
    \right\}.
    \label{eq:config_list_abcd}
\end{align}
Given \(\mathcal C_{bcd}\), we define the associated
permutation \(\pi_{\rho}\in S_{2t}\) by 
\(\pi_{\rho}(m)=n\) defined from 
\(    (m,n)\in \mathcal C_{bcd}\). 
The distinction between Gaussian and non-Gaussian tDW enters in that for $a=0$, colour variables where the two ladder contractions in the Mickey Mouse diagram separate, i.e., $[b-1+2t]_{2t}$ and $[b-1+t]_{2t}$, need to be identified. This identification is algebraically written according to the rule
\begin{equation}
    R_{a=0,b}^{Z}:
    \qquad
    Z_{[b-1+2t]_{2t}}
    \longmapsto
    Z_{[b-1+t]_{2t}},
    \label{eq:replacement_rule_a0}
\end{equation}
where \(Z\) denotes the formal colour variables associated with the folded site under consideration.  For \(a=1\), no such replacement is made, i.e., \(R_{a=1,b}^{Z}=\varnothing \).
Now consider two neighbouring configurations
\(\rho_1=(a_1,b_1,c_1,d_1)\) and \(\rho_2=(a_2,b_2,c_2,d_2)\).
We introduce the explicit colour labels \(X_m\) for the \(b\)-loop endpoints of the first site and \(Y_m\) for those of the second site, with \(m=1,\ldots,2t\). The phase entering the two-site average is
\begin{equation}
    \Phi_{\rho_1,\rho_2}
    =
    \left.
    \left(
           \sum_{m=1}^{2t}
    \theta_{X_m,Y_m}
    -
        \sum_{m=1}^{2t}
    \theta_{
        X_{\pi_{\rho_1}^{-1}(m)},
        Y_{\pi_{\rho_2}^{-1}(m)}}
    \right)
    \right|_{R^{X}_{a_1,b_1},\,R^{Y}_{a_2,b_2}} 
    ,
    \label{eq:phase_difference_rho}
\end{equation}
where after the contractions are imposed, the endpoint \(m\) on the combined
\(c\cup d\) loops is pulled back to the \(b\)-loop by the inverse permutations \(\pi_{\rho_1}^{-1}\) and \(\pi_{\rho_2}^{-1}\), and  we apply the rules $R^X$ and $R^Y$ after the pullback. Now expand
\(
    \Phi_{\rho_1,\rho_2}
    =
    \sum_{r,s}
    n_{rs}(\rho_1,\rho_2)\,
    \theta_{X_r,Y_s}\), 
    with 
    \(n_{rs}(\rho_1,\rho_2)\in\mathbb Z\). Since the \(\theta_{X_r,Y_s}\) are independent real Gaussian variables
with variance \(\epsilon\), we have 
\begin{equation}
\overline{  e^{i\Phi_{\rho_1,\rho_2}}}
    =
    \exp\left[
        -\frac{\epsilon}{2}
        \sum_{r,s} n_{rs}(\rho_1,\rho_2)^2
    \right].
    \label{eq:gaussian_phase_average_rho}
\end{equation}
Combining the two-site phase average with the one-site Haar averages gives the generalized Boltzmann factor
\begin{equation}
    \boltz({\rho_1, \rho_2})
    = q^{2t}
    \sqrt{ \wg(\sigma_1^{-1}\tau_1;q)}\,
    \exp\left[
        -\frac{\epsilon}{2}
        \sum_{r,s} n_{rs}(\rho_1,\rho_2)^2
    \right]
    \sqrt{\wg\!\left(\sigma_2^{-1}\tau_2;q\right)},
    \label{eq:boltz_finite_q_loop_colour_boxed}
\end{equation}
where we follow the convention in \eqref{app_eq:k_top_def_floq}, in which a trivial factor of \(q^{-1}\) is extracted outside \(\boltz\).
For example, when \(t=2\), \(\rho_1=(0,4,2,2)\) and
\(\rho_2=(1,3,2,1)\), the above construction gives
\(\Phi_{\rho_1,\rho_2}
=
\theta_{X_1,Y_3}
-
\theta_{X_1,Y_4}
-
\theta_{X_4,Y_3}
+
\theta_{X_4,Y_4}\).
Hence, the generalized Boltzmann factor gives
\(\boltz_{\rho_1,\rho_2}=i e^{-2\epsilon}\).

With this construction, the TopSFF \eqref{app_eq:tsff_floq_par_def} for  \(U_{\text{f-p-RPM}}\) in Eq.~\eqref{app_eq:par_rpm_temprand} can be evaluated exactly in the large-$q$ limit as  
\be\label{app_eq:k_top_def_floq}
\ba
\overline{K}_{\spatdef} & \, =:   q^{-\Leff -1} \overline{\Knormalized}_{\spatdef} + O(q^{-{\Leff} -2}),
\\
\overline{\mathcal{K}}_{\spatdef}  
& \, =  
\phi^{T}\boltz^{\Leff}\phi
:=
\sum_{\rho_1 ,\rho_2}\phi(\rho_1)\left[\boltz^{\Leff}\right]_{\rho_1 \rho_2}\phi(\rho_2),
\ea
\ee
where $\boltz$ and $\phi$ are respectively the generalized Boltzmann factor encoding for many-body interactions and the  boundary state under the folded representation. The effective length is $\Leff = L/2-1$ for even $L$. For odd $L$, TopSFF can be evaluated analogously.
The generalized Boltzmann factor $\boltz(\rho_1, \rho_2)$ depends only on the relative coordinates between $\rho_1$ and $\rho_2$ along the $b$-, $c$- and $d$-loops. Define the relative coordinates 
$\Delta b := (b_2 - b_1) \pmod {2t}$, 
$\Delta c := (c_2 - c_1) \pmod {t}$, and 
$\Delta d := (d_2 - d_1) \pmod {t}$. 
Translational invariance implies that the Boltzmann factor can be written as
\be\label{eq:boltz_pos_space}
\ba
\boltz
(\rho_1, \rho_2) 
 =&\; \boltz_{\text{G-G}}(\rho_1, \rho_2)
  + 
  \boltz_{\text{G-nG}}(\rho_1, \rho_2)
  + 
  \boltz_{\text{nG-nG}}(\rho_1, \rho_2) 
 \\
= &\;
\left[ \boltz_{a_1 a_2} (\Delta b, \Delta c, \Delta d) \right]_{a_1, a_2 = 0,1}
=
\begin{bmatrix}
   \boltz_{00}  & \boltz_{01} 
   \\
   \boltz_{10}  & \boltz_{11} 
\end{bmatrix}
.
\ea
\ee
where the block in the Boltzmann factor  has the following  interpretation:
\begin{itemize}
\item $\boltz_{\text{G-G}}(\rho_1, \rho_2)$: Amplitudes for the Gaussian tDW to remain as Gaussian tDW,
\item $\boltz_{\text{nG-nG}}(\rho_1, \rho_2)$: Amplitudes for the non-Gaussian tDW to remain as non-Gaussian tDW,
\item $\boltz_{\text{G-nG}}(\rho_1, \rho_2)$: Amplitudes for the DW type conversion.
\end{itemize}

In the nG-nG sector with $a_1=a_2=1$, the Boltzmann factor accounts for the many-body interactions between two non-Gaussian tDWs, and is given by
\begin{equation}\label{eq:boltz_ngng}
\begin{aligned}
\boltz_{\text{nG-nG}}(\rho_1,\rho_2)
=& \; -\delta_{a_1,1}\delta_{a_2,1}
\mu^{4t - 2 \,\xi \eta},
\\
\eta(\rho_1,\rho_2) =& \; \eta(\Delta b)
:= \left| t - \Delta b \right|,
\\
\xi(\rho_1,\rho_2) =& \; \xi(\Delta b, \Delta c, \Delta d) 
=  \delta^{(t)}_{\Delta b,\Delta c}  + \delta^{(t)}_{\Delta b,\Delta d}.
\end{aligned}
\end{equation}
where $\eta$ is the cyclic distance on the length-$2t$ cycle up to a shift, and $\delta^{(t)}_{x,y} = 1$ if $x \equiv y \ (\mathrm{mod}\ t)$ and $0$ otherwise. $\mu := e^{-\epsilon/2}$ parameterizes the interaction strength, where $\epsilon$ is the variance of the random phases in the RPM.

In the G-nG sector with $a_1 + a_2 =1$,  the Boltzmann factor for the mixed Gaussian-non-Gaussian sector is given by 
\begin{equation}\label{eq:boltz_gng}
\begin{aligned}
\boltz_{\mathrm{G\text{-}nG}}(\rho_1,\rho_2)
=& \;
\delta_{a_1+a_2,1} \,
i \mu^{4t-2\,\xi\left(\eta
+\Theta\left( \left[(a_2-a_1)\Delta b \right]\pmod{2t} - t \right) \right)}, 
\end{aligned}
\end{equation}
Here $\Theta$ is the Heaviside step function with the convention $\Theta(0)=1$.

In the G-G sector with $a_1 = a_2 =0$, the Boltzmann factor for the Gaussian-Gaussian sector is given by
\begin{equation}\label{eq:boltz_gg}
\begin{aligned}
\boltz_{\mathrm{G\text{-}G}}(\rho_1,\rho_2)
=& \;    \delta_{a_1,0}\,\delta_{a_2,0}\Big[
\delta_{\Delta b,0}\,\mu^{E_0}
+\delta_{\Delta b,t}\,\mu^{E_t}
+\big(1-\delta_{\Delta b,0} - \delta_{\Delta b,t}\big)\,
\mu^{E_\times}
\Big],
\\
E_0
=& \; 2t\big[(1-\delta_c)+(1-\delta_d)\big]
+2\alpha(1-\delta_c)(1-\delta_d),
\\
E_t
=&\; 4t-2(\delta_c+\delta_d)
+2\alpha(1-\delta_c)(1-\delta_d),
\\
E_\times
=& \;
4t + 2\alpha \, \delta_c \delta_d  -(\eta+1)2\xi
\end{aligned}
\end{equation}
where we defined $\delta_b := \delta_{b_1,b_2} = \delta_{\Delta b,0}$, and similarly for $\delta_c$ and $\delta_d$.

\subsection{Properties of $\boltz$}\label{app:boltz_pos_prop}
The generalized Boltzmann factor $\boltz$  has the following properties:
\begin{itemize}
    \item 
\textbf{$\mathcal{PT}$ symmetry.}
The Boltzmann factor $\boltz$ obeys an antiunitary $\mathcal{PT}$ symmetry,
\begin{equation}
(\mathcal{PT})\,\boltz\,(\mathcal{PT})^{-1}=\boltz
\qquad\Longleftrightarrow\qquad
\mathcal P\boltz^{*} \mathcal P^{-1}=\boltz ,
\label{eq:PT_def_T_matrix}
\end{equation}
where $\mathcal{T}$ denotes the complex conjugation in the $\{|a,b,c,d\rangle\}$ basis, and $\mathcal P$ is a unitary involution which acts as the Pauli-$z$ matrix on the $a$-space,
\begin{equation}
\mathcal P\ket{a,b,c,d} := (-1)^a\,\ket{a,b,c,d},
\qquad\Longleftrightarrow\qquad
\mathcal P = \sigma_z \otimes \iden_{2t}\otimes \iden_t\otimes \iden_t.
\label{eq:P_def_sigma_z}
\end{equation} 
This symmetry condition
$\mathcal P\boltz^*\mathcal P^{-1}=\boltz$ is equivalent to $\boltz_{00}$ and $\boltz_{11}$ being real while $\boltz_{01}$ and
$\boltz_{10}$ being purely imaginary, which is satisfied by G-G, G-nG, and nG-nG Boltzmann factors. In summary, the Boltzmann factor $\boltz$ has an antiunitary $\mathcal{PT}$-type symmetry with complex conjugation $\mathcal T$ and $\mathcal P=\sigma_z\otimes \iden$.

\item \textbf{Reality of $\tr\boltz$ and $\det\boltz$, and spectral pairing.}
From the $\mathcal{PT}$ constraint $\mathcal P\boltz^{*}\mathcal P^{-1}=\boltz$ and the facts that trace and determinant are invariant under similarity transformations, we obtain the reality condition of the trace and determinant of the Boltzmann factor,
\begin{equation}
\tr\boltz
=\tr(\boltz^{*})
=(\tr\boltz)^{*},
\qquad
\det\boltz
=\det(\boltz^{*})
=(\det\boltz)^{*}.
\end{equation}
As a consequence of the reality of $\tr\boltz$ and $\det\boltz$, the characteristic polynomial has real coefficients, so the spectrum is closed under complex conjugation, i.e. eigenvalues are real or come in conjugate pairs. Alternatively, since $\boltz$ obeys the $\mathcal{PT}$ constraint, if $\boltz v=\lambda v$, then taking complex conjugation gives $\boltz^{*}v^{*}=\lambda^{*}v^{*}$, and multiplying by $\mathcal P$ yields
\begin{equation}
\boltz\,(\mathcal P v^{*})
=\mathcal P\,\boltz^{*}v^{*}
=\lambda^{*}\,(\mathcal P v^{*}).
\label{eq:pairing_evalues}
\end{equation}
Hence $\lambda^{*}$ is also an eigenvalue of $\boltz$. Therefore the spectrum is
closed under complex conjugation, i.e., eigenvalues are either real or occur in
complex-conjugate pairs. 

\item \textbf{Exchange symmetry under $c\leftrightarrow d$.}
In the position space representation, the domain wall physics is symmetric under the exchange of $c$ and $d$. Consequently, we have
\begin{equation}
{\boltz}_{a_1 a_2}(\Delta b, \Delta c, \Delta d) = 
{\boltz}_{a_1 a_2}(\Delta b, \Delta d, \Delta c) \,.
\label{eq:boltz_cd_exchange_position}
\end{equation}

\item \textbf{Reality of $\overline{\Knormalized}_{\spatdef}$.} 
Recall that  $\overline{\mathcal{K}}_{\spatdef} = \phi^T \boltz^{\Leff} \phi$.
If the boundary states are $\mathcal{PT}$ invariant in the sense
\begin{equation}
\mathcal P \mathcal{T} \phi= \mathcal P \phi^{*} =\phi,
\qquad
\phi^{T} (\mathcal{PT})^{-1} = \phi^{\dagger}\mathcal P=\phi^T.
\label{eq:PT_invariant_boundaries}
\end{equation}
Then, using $\mathcal P\boltz^{*}\mathcal P^{-1}=\boltz$,
\begin{align}
\overline{\Knormalized}_{\spatdef}^{*}
&=\phi^\dagger \left(\boltz^{\Leff}\right)^{*}{\phi}^{*}
= \phi^\dagger  \mathcal P\boltz^{\Leff}\mathcal P {\phi}^{*}
={\phi}^T\boltz^{\Leff}{\phi}
=\overline{\Knormalized}_{\spatdef}.
\label{eq:amplitude_real}
\end{align}
Therefore, $\overline{\Knormalized}_{\spatdef}\in\mathbb{R}$ whenever the boundary states are $\mathcal{PT}$ invariant.

\end{itemize}

\subsection{Boltzmann factor $\widetilde{\boltz}$ in momentum space}\label{app:boltz_mom_space}
Utilizing the translational invariance of  ${\boltz}$, we can express the TopSFF in terms of the generalized Boltzmann factor in the momentum space as
\begin{equation} \label{eq:master-unitary} 
\overline{\Knormalized}_{\spatdef}=
\phi^T \boltz^{\Leff} \phi = \sum_{k_b,k_c,k_d} \widetilde{\phi}(-k)^T\, \widetilde{\boltz}(k)^{\Leff}\, \widetilde{\phi}(k).
\end{equation}
 The discrete Fourier transform of $\ket{\phi_{\mathrm i}}$ and ${\boltz}$ are defined as
\begin{equation}
\label{eq:four_trans_def}
\begin{aligned}
\widetilde{\phi}(a;k_b,k_c,k_d)
:=& \;
\frac{1}{\sqrt{2t^{3}}}
\sum_{b=0}^{2t-1}\sum_{c=0}^{t-1}\sum_{d=0}^{t-1}
e^{i(k_b b+k_c c+k_d d)}\,\phi(a,b,c,d), 
\\
\widetilde{\boltz}_{a_1 a_2}(\mathbf{k})
:=& \;
\sum_{\Delta b=0}^{2t-1}
\sum_{\Delta c=0}^{t-1}
\sum_{\Delta d=0}^{t-1}
\boltz_{a_1 a_2}(\Delta b,\Delta c,\Delta d)\,
e^{-i(k_b \Delta b + k_c \Delta c + k_d \Delta d)}.
\end{aligned}
\end{equation}
where \(k_b = 2\pi n_b/2t\) with \(n_b=0,1,\dots,2t-1\), 
\(k_c = 2\pi n_c/ t\) with \(n_c=0,1,\dots,t-1\), and 
\(k_d = 2\pi n_d/t \) with \(n_d=0,1,\dots,t-1\). 

Now we perform the Fourier transform of  $\boltz$ in Eq.~\eqref{eq:boltz_pos_space} to obtain $\widetilde{\boltz}(k_b,k_c,k_d)$  for general momenta.
Recall that many sectors depend on $(\Delta c,\Delta d)$ through
$\xi(\Delta b,\Delta c,\Delta d)=\delta^{(t)}_{\Delta b,\Delta c}+\delta^{(t)}_{\Delta b,\Delta d}$ with 
$\delta^{(t)}_{x,y}=1$ if  $x\equiv y\pmod t$,
and on $\Delta b$ through some integer-valued function $f(\Delta b)$. Fix $\Delta b$ and define $r:=\Delta b\pmod t\in\{0,\dots,t-1\}$, so that
$\delta^{(t)}_{\Delta b,\Delta c}=\delta_{\Delta c,r}$ and $\delta^{(t)}_{\Delta b,\Delta d}=\delta_{\Delta d,r}$.
Setting $s:=\mu^{-2 f(\Delta b)}$, we can compute the Fourier transform over $\Delta b$ and $\Delta c$ as 
\begin{equation}\label{eq:F_form}
\mathcal{F}(k_c,k_d;r,s)
:=
\sum_{\Delta c,\Delta d=0}^{t-1}
e^{-i(k_c\Delta c+k_d\Delta d)} \mu^{-2\xi f(\Delta b)}
=
\Big[S_t(k_c)+(s-1)e^{-ik_cr}\Big]\Big[S_t(k_d)+(s-1)e^{-ik_dr}\Big],
\end{equation}
where $S_t(k):=\sum_{x=0}^{t-1}e^{-ikx}=t\delta_{k,0\,(\mathrm{mod}\ 2\pi)}$. 

As a first step, we perform the Fourier transform over $c$ and $d$. For the nG-nG sector with $a_1=a_2=1$, we perform the Fourier transform on Eq.~\eqref{eq:boltz_ngng}, giving
\begin{equation}
\widetilde{\boltz}_{11}(k_b,k_c,k_d)
=
-\mu^{4t}\sum_{\Delta b=0}^{2t-1}
e^{-ik_b\Delta b}\,
\mathcal{F}\!\left(k_c,k_d;\ r(\Delta b),\ \mu^{-2\eta(\Delta b)}\right).
\label{eq:T11_general_compact}
\end{equation}
For the G-nG sector with $a_1+a_2=1$, introduce the convention $p:=a_2-a_1\in\{+1,-1\}$, so that $p=1$ corresponds to $(a_1,a_2)=(0,1)$, and $p=-1$ corresponds to $(a_1,a_2)=(1,0)$.
Recall that $\Gamma_p(\Delta b):=\Theta\big[(p\Delta b)\pmod{2t} - t\big]$ with $\Theta(0)=1$, and $f_p(\Delta b):=\eta(\Delta b)+\Gamma_p(\Delta b)$.
Then we can perform the Fourier transform on Eq.~\eqref{eq:boltz_gng}, giving
\be
\widetilde{\boltz}_{a_1 a_2}(k_b,k_c,k_d)
=
i\mu^{4t}\sum_{\Delta b=0}^{2t-1}
e^{-ik_b\Delta b}\,
\mathcal{F}\!\left(k_c,k_d;\ r(\Delta b),\ \mu^{-2 f_{p}(\Delta b)}\right),
\label{eq:T01_general_compact}
\ee
To perform the Fourier transform on Eq.~\eqref{eq:boltz_gg}, we consider the Fourier contributions for $\Delta b=0$, $\Delta b=t$, and $\Delta b\neq t$ separately. 
For $\Delta b=0$ and $\Delta b=t$, the Fourier contributions are 
\be
\begin{aligned}
\mathcal{G}_0(k_c,k_d)
&:=
\sum_{\Delta c,\Delta d=0}^{t-1}e^{-i(k_c\Delta c+k_d\Delta d)}\mu^{E_0}
\\
& =  
1+\mu^{2t}\big[(S_t(k_c)-1)+(S_t(k_d)-1)\big]
+\mu^{4t+2\alpha}(S_t(k_c)-1)(S_t(k_d)-1),
\\
\mathcal{G}_t(k_c,k_d)
&:=
\sum_{\Delta c,\Delta d=0}^{t-1}e^{-i(k_c\Delta c+k_d\Delta d)}\mu^{E_t}
\\
&=
\mu^{4t-4}
+\mu^{4t-2}\big[(S_t(k_c)-1)+(S_t(k_d)-1)\big]
+\mu^{4t+2\alpha}(S_t(k_c)-1)(S_t(k_d)-1).
\label{eq:Gt_general}
\end{aligned}
\ee
For $\Delta b\neq 0,t$, the Boltzmann factor becomes  $\mu^{E_\times} = \mu^{4t}\mu^{2\alpha \delta_c \delta_d}\mu^{-2(\eta + 1)\xi}$. We can then utilize \eqref{eq:F_form} with $f_\times(\Delta b):=\eta(\Delta b)+1$, and the Fourier transform can be evaluated as
\begin{equation}
\mathcal{G}_\times(\Delta b;\,k_c,k_d)
:=
\sum_{\Delta c,\Delta d}e^{-i(k_c\Delta c+k_d\Delta d)}\mu^{E_\times}
=
\mu^{4t}\Big[
\mathcal{F}\!\left(k_c,k_d;\ r(\Delta b),\ \mu^{-2 f_\times(\Delta b)}\right)
+(\mu^{2\alpha}-1)
\Big].
\label{eq:Gx_general}
\end{equation}
Together, we have
\begin{equation}
\widetilde{\boltz}_{00}(k_b,k_c,k_d)
=
\mathcal{G}_0(k_c,k_d)
+e^{-ik_b t}\,\mathcal{G}_t(k_c,k_d)
+\sum_{\substack{\Delta b=0\\ \Delta b\neq 0,t}}^{2t-1}
e^{-ik_b\Delta b}\,
\mathcal{G}_{\times}(\Delta b;\,k_c,k_d).
\label{eq:T00_general_compact}
\end{equation}

Next, we evaluate the remaining sum over $\Delta b$. We introduce $x:=\mu^{-2}$. It is convenient to define the following geometric sums
\begin{equation}
\begin{aligned}
\mathsf{E}(k):=& \, \sum_{r=1}^{t-1}e^{-ikr}
=
\begin{cases}
\frac{e^{-ik}\left[1-e^{-ik(t-1)}\right]}{1-e^{-ik}}, & k\not\equiv 0\ (\mathrm{mod}\ 2\pi),\\
t-1,& k\equiv 0\ (\mathrm{mod}\ 2\pi),
\end{cases}
\\
\mathsf{G}(k;y):=& \, \sum_{r=1}^{t-1}y^r e^{-ikr}
=
\frac{y e^{-ik}\big(1-(y e^{-ik})^{t-1}\big)}{1-y e^{-ik}},
\\
\mathsf{J}(k;\alpha,\beta;\,m)
:=& \, 
\sum_{r=1}^{t-1}\Big(x^{\alpha r+\beta}-1\Big)^{m} e^{-ikr},
\end{aligned}
\end{equation}
where $\alpha\in\{+1,-1\}$, $\beta\in\mathbb{Z}$ and  $m\in\{1,2\}$. We further define
\begin{align}
J_{1}(k)
&:=\mathsf{J}(k;\,+1,0;\,1)=\sum_{r=1}^{t-1}(x^{r}-1)e^{-ikr}
=\mathsf{G}(k;x)-\mathsf{E}(k),
\\
J_{2}(k)
&:=\mathsf{J}(k;\,+1,0;\,2)=\sum_{r=1}^{t-1}(x^{r}-1)^2e^{-ikr}
=\mathsf{G}(k;x^2)-2\mathsf{G}(k;x)+\mathsf{E}(k),
\\
\check{J}_{1}(k)
&:=\mathsf{J}(k;\,-1,t;\,1)=\sum_{r=1}^{t-1}(x^{t-r}-1)e^{-ikr}
=e^{-ikt}\Big(\mathsf{G}(-k;x)-\mathsf{E}(-k)\Big),
\\
\check{J}_{2}(k)
&:=\mathsf{J}(k;\,-1,t;\,2)=\sum_{r=1}^{t-1}(x^{t-r}-1)^2e^{-ikr}
=e^{-ikt}\Big(\mathsf{G}(-k;x^2)-2\mathsf{G}(-k;x)+\mathsf{E}(-k)\Big),
\\
J_{1}^{+}(k)
&:=\mathsf{J}(k;\,+1,1;\,1)=\sum_{r=1}^{t-1}(x^{r+1}-1)e^{-ikr}
=x\,\mathsf{G}(k;x)-\mathsf{E}(k),
\\
J_{2}^{+}(k)
&:=\mathsf{J}(k;\,+1,1;\,2)=\sum_{r=1}^{t-1}(x^{r+1}-1)^2e^{-ikr}
=x^2\mathsf{G}(k;x^2)-2x\,\mathsf{G}(k;x)+\mathsf{E}(k),
\\
\check{J}_{1}^{+}(k)
&:=\mathsf{J}(k;\,-1,t+1;\,1)=\sum_{r=1}^{t-1}(x^{t-r+1}-1)e^{-ikr}
=x\,e^{-ikt}\mathsf{G}(-k;x)-\mathsf{E}(k),
\\
\check{J}_{2}^{+}(k)
&:=\mathsf{J}(k;\,-1,t+1;\,2)=\sum_{r=1}^{t-1}(x^{t-r+1}-1)^2e^{-ikr}
=x^2 e^{-ikt}\mathsf{G}(-k;x^2)-2x e^{-ikt}\mathsf{G}(-k;x)+\mathsf{E}(k).
\end{align}
with all sums running from $r=1,\dots,t-1$.  For the G-nG sector, it is convenient to introduce the $p$-dependent combinations 
\begin{align}
J_m^{(p)}(k) &:= \frac{1+p}{2}\,J_m(k) + \frac{1-p}{2}\,J_m^{+}(k),
\\
\check{J}_m^{(p)}(k) &:= \frac{1+p}{2}\,\check{J}_m^{+}(k) + \frac{1-p}{2}\,\check{J}_m(k).
\end{align}
Direct evaluations of the Fourier transforms in $b$ space give the following closed expressions. 
Define the combined phases $k_{bc}:=k_b+k_c$,  $k_{bd}:=k_b+k_d$, $k_{bcd}:=k_b+k_c+k_d$. Denote the summation of phases as $S_c:=S_t(k_{c})$, $S_d:=S_t(k_{d})$ where $S_t(k):=\sum_{x=0}^{t-1}e^{-ikx}=t\delta_{k,0\ (\mathrm{mod}\ 2\pi)}$. 
For $a_1 = a_2 = 0$,
\begin{equation}
\begin{aligned}
&\widetilde{\boltz}_{00}(k_b,k_c,k_d)
=
\,\,\mathcal{G}_0(k_c,k_d)
+(-1)^{n_b}\,\mathcal{G}_t(k_c,k_d)
\\
&\qquad +\mu^{4t}\Big[
S_c S_d\,(1+(-1)^{n_b})\,\mathsf{E}(k_b)
+ S_c\big(\check{J}_{1}^{+}(k_{bd})+(-1)^{n_b} J_{1}^{+}(k_{bd})\big)
+ S_d\big(\check{J}_{1}^{+}(k_{bc})+(-1)^{n_b} J_{1}^{+}(k_{bc})\big)
\\
&\qquad 
+\big(\check{J}_{2}^{+}(k_{bcd})+(-1)^{n_b} J_{2}^{+}(k_{bcd})\big)
\Big]
+\mu^{4t}(\mu^{2\alpha}-1)\Big(2t\,\delta_{n_b,0}-(1+(-1)^{n_b})\Big).
\end{aligned}
\label{eq:T00_general_closed_J}
\end{equation}
For $a_1 = a_2 =1$,
\begin{equation}
\begin{aligned}
\widetilde{\boltz}_{11}(k_b,k_c,k_d)
= -\mu^{4t}\Big[
&S_c S_d\,(2t\,\delta_{n_b,0})
\\
&+ S_c\Big( (\mu^{-2t}-1)+\check{J}_{1}(k_{bd})+(-1)^{n_b} J_{1}(k_{bd})\Big)
\\
&+ S_d\Big( (\mu^{-2t}-1)+\check{J}_{1}(k_{bc})+(-1)^{n_b} J_{1}(k_{bc})\Big)
\\
&+ \Big( (\mu^{-2t}-1)^2+\check{J}_{2}(k_{bcd})+(-1)^{n_b} J_{2}(k_{bcd})\Big)
\Big].
\end{aligned}
\label{eq:T11_general_closed_J}
\end{equation}
Lastly, for $a_1 + a_2 = 1$,
\begin{equation}
\begin{aligned}
\widetilde{\boltz}_{a_1 a_2}(k_b,k_c,k_d)
= i\mu^{4t}\Big[
&S_c S_d\,(2t\,\delta_{n_b,0})
\\
&+ S_c\Big( (\mu^{-2t}-1)+(-1)^{n_b}(\mu^{-2}-1)
+\check{J}_{1}^{(p)}(k_{bd})+(-1)^{n_b} J_{1}^{(p)}(k_{bd})\Big)
\\
&+ S_d\Big( (\mu^{-2t}-1)+(-1)^{n_b}(\mu^{-2}-1)
+\check{J}_{1}^{(p)}(k_{bc})+(-1)^{n_b} J_{1}^{(p)}(k_{bc})\Big)
\\
&+ \Big( (\mu^{-2t}-1)^2+(-1)^{n_b}(\mu^{-2}-1)^2
+\check{J}_{2}^{(p)}(k_{bcd})+(-1)^{n_b} J_{2}^{(p)}(k_{bcd})\Big)
\Big],
\end{aligned}
\label{eq:T_mixed_general_closed_Jp}
\end{equation}

\subsection{Properties of $\widetilde{\boltz}$}\label{app:boltz_mom_prop}

$\widetilde{\boltz}(k_b,k_c,k_d)$ has the following properties, some of which are representation independent and overlap with Appendix~\ref{app:boltz_pos_prop}, but we include them here for completeness. $\widetilde{\boltz}(k_b,k_c,k_d)$ possesses $\mathcal{PT}$ symmetry at certain momentum points, which we describe in the next section.

\begin{itemize}
\item \textbf{Class AI time reversal symmetry (TRS).}
For any real $y$ and any $k$, the geometric sums satisfy $\mathsf{E}(k)^*=\mathsf{E}(-k)$, and
$\mathsf{G}(k;y)^*=\mathsf{G}(-k;y)$.
Consequently, for $m=1,2$, we have $J_m(k)^* = J_m(-k)$, and similarly for $\check{J}_m$, $J_m^{+}$, $\check{J}_m^{+}$, $J_m^{(p)}$, and $\check{J}_m^{(p)}$. The diagonal sectors of $\widetilde{\boltz}_{a_1 a_2}(k_b, k_c, k_d)$ are real linear combinations of the above geometric sums, whereas its off-diagonal sectors have an overall factor $i$. Therefore, $\widetilde{\boltz}_{a_1 a_2}(k_b, k_c, k_d)$ satisfy the following antiunitary symmetry,
\begin{equation}
(\mathcal{PT}) \widetilde{\boltz}(k_b,k_c,k_d) (\mathcal{PT})^{-1}
:=
\sigma_z\,\widetilde{\boltz}(k_b,k_c,k_d)^{*}\,\sigma_z
=
\widetilde{\boltz}(-k_b,-k_c,-k_d)\,,
\label{eq:antiu_constraint}
\end{equation}
where $\sigma_z$ is the Pauli-$z$ matrix acting on the $a$-space, and $(\cdot)^*$ denotes complex conjugation. This symmetry transformation is a class AI TRS symmetry transformation since 
\be
(\mathcal{PT})^2 = 1\, .
\ee
Relatedly, at a given momentum sector $(k_b,k_c,k_d)$, $\mathcal{PT}$ symmetry of $\widetilde{\boltz}(k_b,k_c,k_d)$ is  $\widetilde{\boltz}(k_b,k_c,k_d)=\sigma_z\, \widetilde{\boltz}^{*}(k_b,k_c,k_d)\, \sigma_z$, with $\mathcal{P}$ set to be $\sigma_z$ acting on $a$-space, and $\mathcal{T}$ set to be complex conjugation. It is clear that $\mathcal{PT}$ symmetric sectors include the sectors satisfying $(k_b,k_c,k_d)= (-k_b,-k_c,-k_d)$. We will discuss the precise criteria in the next section.

\item \textbf{Spectral pairing under momentum inversion.}
The spectral consequence of Eq.~\eqref{eq:antiu_constraint} is that eigenvalues are paired between momentum-inverted sectors. Suppose that
$\widetilde{\boltz}(k_b,k_c,k_d) v= \lambda(k_b,k_c,k_d)v$ .
Taking the complex conjugate and then multiplying by $\sigma_z$, one finds
\begin{align}
\widetilde{\boltz}(-k_b,-k_c,-k_d)\,(\sigma_z v^*)
=
\sigma_z
\widetilde{\boltz}(k_b,k_c,k_d)^*
v^*
\nonumber
=
\lambda(k_b,k_c,k_d)^*
(\sigma_z v^*) .
\end{align}
Therefore
\begin{equation}
\lambda(k_b,k_c,k_d)\in
\mathrm{Spec}\,\widetilde{\boltz}(k_b,k_c,k_d)
\quad \Longrightarrow \quad
\lambda(k_b,k_c,k_d)^*
\in
\mathrm{Spec}\,\widetilde{\boltz}(-k_b,-k_c,-k_d).
\end{equation}

\item \textbf{Exchange symmetry under $c\leftrightarrow d$.}
In the position space representation, the domain wall physics is symmetric under the exchange of $c$ and $d$. Consequently, $\widetilde{\boltz}$ is symmetric under exchanging $k_c$ and
$k_d$, i.e.,
\begin{equation}
\widetilde{\boltz}_{a_1 a_2}(k_b,k_c,k_d)=\widetilde{\boltz}_{a_1 a_2}(k_b,k_d,k_c) \, . 
\label{eq:boltz_cd_exchange}
\end{equation}

\item \textbf{Reality of $\det\widetilde{\boltz}(k)$ and $\tr\widetilde{\boltz}(k)$, and spectral pairing within each momentum sector}.
We observe that in the $a$-space, $\widetilde{\boltz}(k)$ can be brought to a generalized $\mathcal{PT}$ dimer form, 
\begin{equation}
\widetilde{\boltz}(k)=
\begin{pmatrix}
A(k) & i\,B(k)\\
i\,C(k) & D(k)
\end{pmatrix},
\qquad
A(k),B(k),C(k),D(k)\in\mathbb{R},
\qquad
\Longrightarrow 
\qquad
\det\widetilde{\boltz}(k), \tr\widetilde{\boltz}(k)
\in\mathbb{R},
\label{eq:PT_dimer_general_form}
\end{equation}
i.e. $\det \widetilde{\boltz}(k)$ is always real. The characteristic polynomial of $\widetilde{\boltz}$ is
$\chi(\lambda)=\lambda^2-(\tr\widetilde{\boltz})\,\lambda+\det\widetilde{\boltz}$, so if $\tr\widetilde{\boltz},\det\widetilde{\boltz}\in\mathbb{R}$, then $\chi(\lambda)$ has real coefficients. Hence the eigenvalues are either both real or form a complex-conjugate pair. Define the discriminant $\Delta := (\tr\widetilde{\boltz})^2-4\det\widetilde{\boltz}$. Then the eigenvalues are
\be
\lambda_{\pm}
=\frac{\tr\widetilde{\boltz}\pm\sqrt{\Delta}}{2}.
\ee
Therefore, we have
\begin{itemize}
\item If $\Delta>0$, then $\lambda_{\pm}\in\mathbb{R}$ are distinct.
\item If $\Delta=0$, then $\lambda_{+}=\lambda_{-}=\tfrac12\tr\widetilde{\boltz}$ is degenerate. This is the candidate locus for an exceptional point (EP), which occurs when the matrix is not diagonalizable at the degeneracy.
\item If $\Delta<0$, then $\lambda_{\pm}$ form a complex-conjugate pair.
\end{itemize}

\item \textbf{Bragg peaks and parity in $k_b$.}
The momentum dependence enters $\widetilde{\boltz}$ only through the elementary phases
$k_b,k_c,k_d$ and the combined momenta
$k_{bc}:=k_b+k_c$, $k_{bd}:=k_b+k_d$, and $k_{bcd}:=k_b+k_c+k_d$,
together with the Bragg factors
\be
S_c:=S_t(k_c),\qquad S_d:=S_t(k_d),\qquad
S_t(k):=\sum_{x=0}^{t-1}e^{-ikx}=t\,\delta_{k,0 \pmod{2\pi}}.
\ee
Since $S_t(k)$ vanishes for generic $k$ but jumps to $t$ at reciprocal-lattice
momenta $k\equiv 0 \pmod{2\pi}$, $\widetilde{\boltz}(k_b,k_c,k_d)$
exhibits Bragg-peak-like features in $(k_c,k_d)$.

Furthermore, recall that we defined $k_b=2\pi n_b/2t$, so that $(-1)^{n_b}=e^{-ik_b t}$. The dependence on the parity of $n_b$ is equivalently a dependence on the half translation phase $e^{-i k_b t}$ along the length-$2t$ loop in $b$. Many contributions in TopSFF appear with $1\pm(-1)^{n_b}$, thereby splitting the $k_b$-dependence into two parity sectors.

\end{itemize}

\subsection{$\mathcal{PT}$ symmetry of $\widetilde{\boltz}$}\label{app:pt_sym}

\begin{figure}
    \centering
    \includegraphics[width=0.5\linewidth]{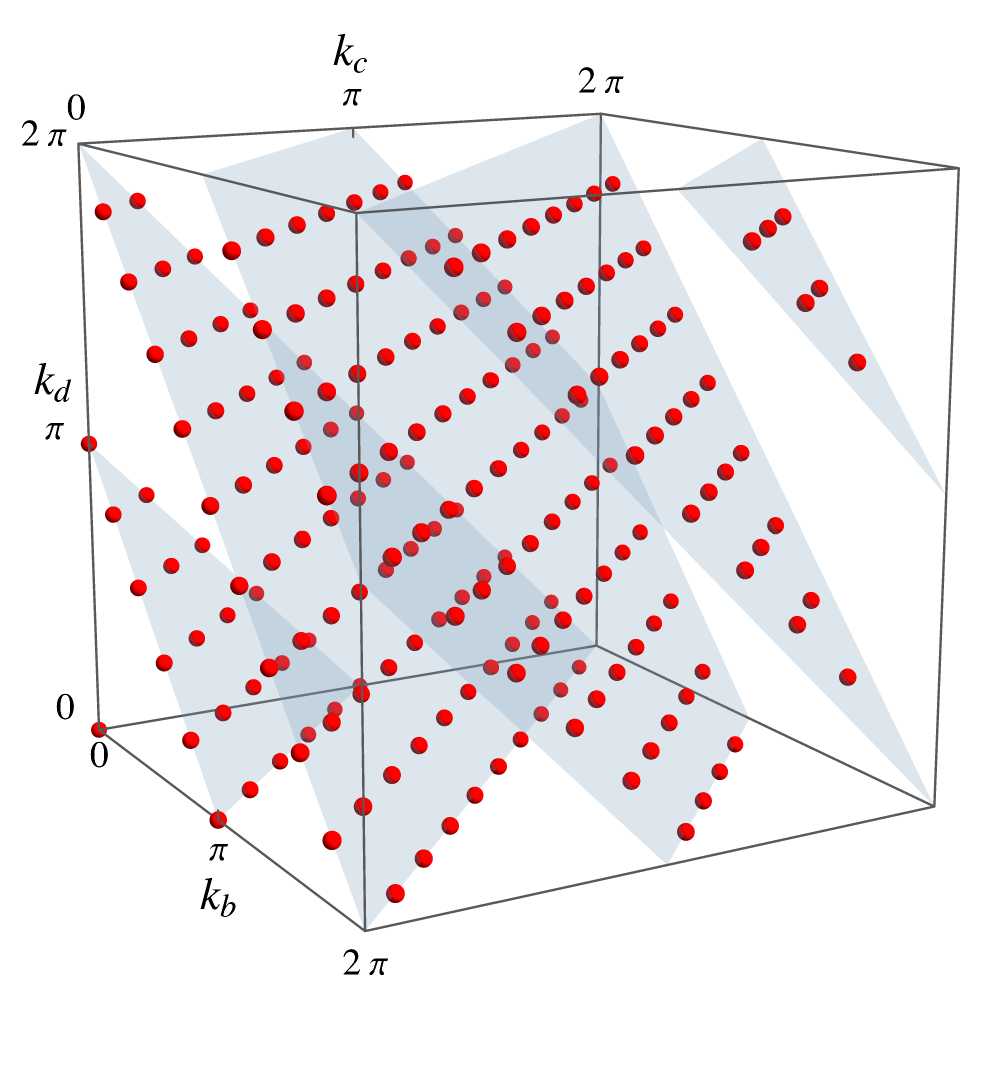}
    \caption{\textbf{EP planes.} Momentum sectors (red dots) at which the generalized Boltzmann factor $\Ttilde$ has $\mathcal{PT}$ symmetry according to Eq.~\eqref{app_eq:pt_sym}.}
    \label{app_fig:ep_mom_sectors}
\end{figure}

Set $\mathcal{P}$ to be $\sigma_z$ acting on $a$-space, and $\mathcal{T}$ to be complex conjugation. $\widetilde{\boltz}(k_b,k_c,k_d)$ is $\mathcal{PT}$ symmetric if 
\begin{equation}\label{app_eq:pt_sym_con_mom}
\widetilde{\boltz}(k_b,k_c,k_d)=\sigma_z\, \widetilde{\boltz}^{*}(k_b,k_c,k_d)\, \sigma_z.
\end{equation}
Firstly, note that the momentum points  
$(k_b,k_c,k_d)=(-k_b,-k_c,-k_d)$ 
and
$(k_b,k_c,k_d)=(-k_b,-k_d,-k_c)$ lead to $\mathcal{PT}$ symmetry due to the antiunitary symmetry \eqref{eq:antiu_constraint} and exchange symmetry \eqref{eq:boltz_cd_exchange}. We will show that this is a subset of the full $\mathcal{PT}$ symmetric momentum sectors.

\begin{proposition}
\label{prop:PT_classification_full}
Let $\mathcal P=\sigma_z$ act on the $a$-space, and let $\mathcal T$ denote complex conjugation.
For the Boltzmann factor $\widetilde{\boltz}(k_b,k_c,k_d)$ in Eqs.~\eqref{eq:T00_general_closed_J}--\eqref{eq:T_mixed_general_closed_Jp},
one has for generic $\mu$,
\begin{equation}\label{app_eq:pt_sym}
\widetilde{\boltz}(k_b,k_c,k_d)\ \text{is $\mathcal{PT}$ symmetric}
\iff
k_b+k_c+k_d\equiv 0 \pmod{\pi}.
\end{equation}
\end{proposition}

\begin{proof}
The $\mathcal{PT}$ symmetry condition \eqref{app_eq:pt_sym_con_mom} is equivalent to 
\begin{equation}
\widetilde{\boltz}_{00},\widetilde{\boltz}_{11}\in \mathbb R,
\qquad
\widetilde{\boltz}_{01},\widetilde{\boltz}_{10}\in i\mathbb R.
\label{eq:PT_entrywise_condition}
\end{equation}
First, we consider the nG-nG sector, i.e., $\widetilde{\boltz}_{11}$ in Eq.~\eqref{eq:T11_general_closed_J}. It is clear that the only possibly non-real factors are 
\be
\check{J}_m(q)+(-1)^{n_b}J_m(q).
\ee
Recall the convention $k_b=\pi n_b/t, k_c=2\pi n_c/t,
k_d=2\pi n_d/t$, and that every momentum argument appearing in the geometric sums is of the form $q \in \{\,k_b,\ k_b+k_c,\ k_b+k_d,\ k_b+k_c+k_d\,\}
= k_b+\frac{2\pi m}{t}$
for some integer $m$. Therefore, the phase factor
$e^{-iqt}
=
e^{-ik_bt}e^{-i2\pi m}
=
e^{-i\pi n_b}
=
(-1)^{n_b}$ 
becomes a sign. 
Using
\begin{equation}
\mathsf E(q)^*=\mathsf E(-q),\qquad
\mathsf G(q;y)^*=\mathsf G(-q;y).
\end{equation}
We write
\begin{equation}
\check{J}_m(q)
=
e^{-iqt}J_m(-q)
=
(-1)^{n_b}J_m(q)^*,
\qquad m=1,2.
\end{equation}
It follows that
\begin{equation}
\check{J}_m(q)+(-1)^{n_b}J_m(q)
=
2(-1)^{n_b}\left[ \check{J}_m(q) + \check{J}_m(q)^* \right]
=2(-1)^{n_b}\Re J_m(q)\in\mathbb R,
\label{eq:J_real_combo}
\end{equation}
is real. As a consequence, the nG-nG sector is always real, because Eq.~\eqref{eq:T11_general_closed_J} only contains real constants, the real factors $S_c,S_d$, and the combinations \eqref{eq:J_real_combo}. Hence
\begin{equation}
\widetilde{\boltz}_{11}(k)\in\mathbb R 
\qquad\text{for all momentum sectors.}
\label{eq:T11_always_real}
\end{equation}

Next we consider the G-G and G-nG sectors, and the $\Leftarrow$ direction of the statement \eqref{app_eq:pt_sym}. The potentially imaginary factors come from the sums $J_m^+$ in combinations like
\begin{equation}\label{app_eq:11sector_Jcomb}
\check{J}_m^+(q)+(-1)^{n_b}J_m^+(q),
\end{equation}
 which are not generically real for arbitrary $q$. Observe that if
\begin{equation}
q \equiv 0 \pmod{\pi}
\label{eq:q_real_condition}
\end{equation}
then every phase $e^{-iqr}$ is equal to $\pm 1$, and therefore
\begin{equation}
\mathsf E(q),\ \mathsf G(q;y),\ J_m(q),\ J_m^+(q),\ \check{J}_m(q),\ \check{J}_m^+(q)\in\mathbb R.
\end{equation}
In that case, $\Ttilde_{00}$ is real, and $\Ttilde_{01},\Ttilde_{10}$ are purely imaginary, so \eqref{app_eq:pt_sym_con_mom} holds. 
Next, it is enough to impose \eqref{eq:q_real_condition} on 
\begin{equation}
k_{bcd}:=k_b+k_c+k_d\equiv 0 \pmod{\pi},
\end{equation}
because every potentially non-real term of the form \eqref{app_eq:11sector_Jcomb} is either already evaluated at \(k_{bcd}\), or is multiplied by \(S_c\) or \(S_d\), which force \(k_c=0\) or \(k_d=0\), and hence reduce its argument to \(k_{bd}:=k_b+k_d\) or \(k_{bc}:=k_b+k_c\). This proves the $\Leftarrow$ direction of the statement.

Next, we prove the $\Rightarrow$ direction for generic non-fine-tuned $\mu$. 
It is enough to inspect the mixed sector with $p=+1$, namely $\widetilde{\boltz}_{01}$ in Eq.~\eqref{eq:T_mixed_general_closed_Jp}. 
The momentum $k_{bcd}:=k_b+k_c+k_d$
appears in Eq.~\eqref{eq:T_mixed_general_closed_Jp} without any prefactor $S_c$ or $S_d$, and is therefore always present through the term
\begin{equation}
\mathcal C(k_{bcd})
:=
(\mu^{-2t}-1)^2+(-1)^{n_b}(\mu^{-2}-1)^2+\check{J}_2^{+}(k_{bcd})+(-1)^{n_b}J_2(k_{bcd}).
\label{eq:Cq_def}
\end{equation}
Using the sum definitions of $J_2(k_{bcd})$ and $\check{J}_2^{+}(k_{bcd})$, we obtain
\begin{equation}
\mathcal C(k_{bcd})
=
(\mu^{-2t}-1)^2+(-1)^{n_b}(\mu^{-2}-1)^2
+\sum_{r=1}^{t-1}
\Big[(\mu^{-2(t-r+1)}-1)^2+(-1)^{n_b}(\mu^{-2r}-1)^2\Big]e^{-ik_{bcd}r}.
\label{eq:Cq_expanded}
\end{equation}
As a polynomial in $\mu^{-2}$, the coefficient of the highest power $(\mu^{-2})^{2t}$ is
\begin{equation}
1+e^{-ik_{bcd}},
\label{eq:Cq_leading}
\end{equation}
because in \eqref{eq:Cq_expanded}, the first term $(\mu^{-2t}-1)^2$ contributes $(\mu^{-2})^{2t}$, and the $r=1$ term of the third term (corresponding to $\check{J}_2^{+}(k_{bcd})$) contributes an additional $e^{-ik_{bcd}}(\mu^{-2})^{2t}$. 
Hence, if \( k_{bcd}\not\equiv 0,\pi \pmod{2\pi}\),
then $1+e^{-i k_{bcd}}\notin \mathbb R$, and therefore $\mathcal C(k_{bcd})$ is not identically real as a function of $\mu^{-2}$. 
Equivalently, for fixed such $k_{bcd}$, the condition $\mathcal C(k_{bcd})\in\mathbb R$ can hold only on a non-generic subset of parameter values $\mu$.

Now fixed-block $\mathcal{PT}$ symmetry requires, by \eqref{eq:PT_entrywise_condition}, that
\(\widetilde{\boltz}_{01}\in i\mathbb R\).
Since $\widetilde{\boltz}_{01}=i\mu^{4t}[\cdots+\mathcal C(k_{bcd})]$, this forces the bracket to be real. Therefore, for generic $\mu$,
\(k_{bcd}=k_b+k_c+k_d \equiv 0 \pmod{\pi}\).
Combined with the sufficiency condition above, this shows that the generic fixed-block $\mathcal{PT}$-symmetric sectors are those satisfying $k_{bcd}=k_b+k_c+k_d\equiv 0 \pmod{\pi}$.
\end{proof}

\subsection{Boundary states}\label{app:boundary_state}
The TopSFF with a spatially extended topological defect is given in Eq.~\eqref{app_eq:k_top_def_floq}. In this section, we describe the boundary states in the folded representation, both in position space and in momentum space. Away from the parity-inversion axes, the bulk gates are chosen to obey discrete time translation symmetry and to be compatible with the global swap defect / parity inversion symmetry. Along the parity inversion axes, the parity inversion symmetric couplings determine the form of the boundary states entering Eq.~\eqref{app_eq:k_top_def_floq}. By engineering the couplings along these axes (gates along dashed lines in Fig.~\ref{fig:parity_models}), we gain access to different momentum sectors of the generalized Boltzmann factor $\widetilde{\boltz}$, which in turn allows us to identify the universal signatures of tDW in TopSFF.

\begin{figure}
    \centering
    \includegraphics[width=0.65\linewidth]{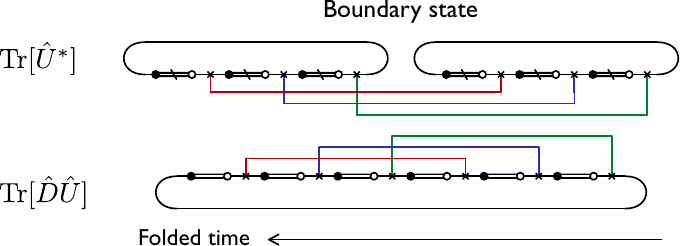}
    \caption{\textbf{Boundary states of TopSFF} are vectors defined on the edge of the emergent system in the folded representation, and encode the presence or absence of coupling across the parity-inversion axes. In this illustration, the random-phase gates are shown in colour, coupling degrees of freedom at $t=1$ (red), $t=2$ (blue),and $t=3$ (green).}
    \label{fig:boundary_state_bonds}
\end{figure}

\subsubsection{Discrete time translational symmetry and parity inversion symmetry}\label{app:bound_state_floq}
In the most natural setting, the coupling gates along the parity inversion axes are drawn analogously to the coupling gates away from the parity inversion axes. More precisely, the coupling gates of the RPM are drawn as described by \eqref{app_eq:pRPM_def3} in \ref{app:model_par_inv} subject to the requirements of (i) discrete time translational symmetry, i.e., $w(\tau,r) = w(\tau',r)$ for all $\tau$ and $\tau'$, and (ii) parity inversion symmetry $\mathrm{SWAP} \, w(\tau,r)  \, \mathrm{SWAP}  = w(\tau,r)$, where $\mathrm{SWAP}$ denotes the swap operator acting on the same Hilbert space as $w(\tau,r)$. This model is illustrated in Fig.~\ref{fig:parity_models_axes}(b). By performing the ensemble averages over the random phase gates as illustrated in Fig.~\ref{fig:boundary_state_bonds}, the boundary state $\phi$ is found to be
\begin{equation}\label{eq:bc_default_case}
\phi(a,b,c,d)  = 
i^{\delta_{a,1}} \left[ 
\mu^{2t}  + {(1-\mu^{2t})\delta_{c,d}} 
\right]\,.
\end{equation}
 The Fourier transform of $\phi$ defined in Eq.~\eqref{eq:four_trans_def} can be computed as 
\begin{equation}\label{eq:floquet_bd_state}
\widetilde{\phi}(a;k_b,k_c,k_d)
=
i^{\delta_{a,1}}
\left[
\sqrt{2t^3}\mu^{2t}\,\delta_{n_b,0}\delta_{n_c,0}\delta_{n_d,0}
+
\sqrt{2t}(1-\mu^{2t})\,\delta_{n_b,0}\delta_{n_c+n_d,\,0 \pmod{t}}
\right].
\end{equation}
where \(k_b = 2\pi n_b/ 2t\) with \(n_b=0,1,\dots,2t-1\), 
\(k_c = 2\pi n_c/ t\) with \(n_c=0,1,\dots,t-1\), and 
\(k_d = 2\pi n_d/t \) with \(n_d=0,1,\dots,t-1\).
Note that imaginary factors arising from Weingarten calculus have been absorbed into the boundary state, and that $\widetilde{\phi}$ has non-zero support only if \(k_b=0\) and \(k_d\equiv -k_c\ (\mathrm{mod}\ 2\pi)\). The tensor network representation of TopSFF with the boundary condition \eqref{eq:bc_default_case} is illustrated in Fig.~\ref{fig:dual_sim_bc}.

\subsubsection{No coupling}
Suppose the coupling gates along the parity-inversion axes of the RPM are turned off, as illustrated in Fig.~\ref{fig:parity_models_axes}(a). The parity inversion symmetric coupling along the parity inversion axes do not bias any of the TopSFF modes. Consequently, the boundary states in Eq.~\eqref{app_eq:k_top_def_floq} and their Fourier transforms are
\be
\phi(a,b,c,d) 
= i^{\delta_{a,1}}, \qquad 
\widetilde{\phi}(a;k_b,k_c,k_d)
=
i^{\delta_{a,1}}
\sqrt{2t^{3}}
\delta_{n_b,0}\delta_{n_c,0}\delta_{n_d,0},
\ee
i.e., the boundary states lie entirely in the momentum sector $(k_b,k_c,k_d)=(0,0,0)$.  This choice provides a useful reference contribution from the zero-momentum sector, which can be subtracted when one wishes to isolate non-zero momentum sectors.

\subsubsection{General periodic driving}
A simple way to gain access to different momentum sectors of the generalized Boltzmann factor is to introduce a generalized periodic driving \textit{only} along the parity inversion axes, i.e., along the dashed lines in Fig.~\ref{fig:parity_models}. The quantum circuits away from the parity inversion axes are unchanged, and hence the generalized Boltzmann factor is unchanged, but the boundary state is. To be more precise, along the parity inversion axes, the coupling gates of RPM are drawn as described by \eqref{app_eq:pRPM_def3} in \ref{app:model_par_inv}, except that the coupling gates are drawn independently  unless they are related by 
\begin{equation}
w(\tau+p,r=L/2,L)=w(\tau,r=L/2,L),
\end{equation}
where both time $\tau$ and period $p$ are discrete. 
The RPM for $p=1$ and $p=2$ are illustrated in Fig.~\ref{fig:parity_models_axes}(b) and (c) respectively.
The boundary state in position space is given by
\[
\phi(a,b,c,d)= i^{\delta_{a,1}}
\Big[\mu^{2t}+(1-\mu^{2t})\,\delta_{[c-b\ (\mathrm{mod}\ t)]\ (\mathrm{mod}\ p),\,0}\;\delta_{c,d}\Big].
\]
Thus for $t=mp$ with some integer $m$, the boundary state in momentum space is
\be\label{eq:ini_state_periodp}
\widetilde{\phi}(a;k_b,k_c,k_d)
=
i^{\delta_{a,1}}
\left[
\mu^{2t}\sqrt{2t^{3}}
\delta_{n_b,0}\delta_{n_c,0}\delta_{n_d,0}
+
(1-\mu^{2t})\frac{\sqrt{2t}}{p}
\delta_{n_b\ (\mathrm{mod}\ 2t/p),0} \,
\delta_{n_c+n_d+n_b/2\ (\mathrm{mod}\ t), 0}
\right].
\ee
A few notable cases:
\begin{itemize}
    \item \textbf{$p=1$. Floquet driving with period 1.} For $p=1$, we recover the period-1 Floquet case as described above in \ref{app:bound_state_floq}, where the period along and away from the parity inversion symmetry axes match. 
    \item \textbf{$p=2$. Floquet driving with period 2.} For $p=2$, i.e. the period-2 Floquet case, we have
\be
\widetilde{\phi}
=
i^{\delta_{a,1}}
\left[
\mu^{2t}\sqrt{2t^3}\,\delta_{n_b,0}\delta_{n_c,0}\delta_{n_d,0}
+
(1-\mu^{2t})\sqrt{\frac{t}{2}}\;
\delta_{n_b\ (\mathrm{mod}\ t),\,0}\;
\delta_{n_c+n_d+n_b/2\ (\mathrm{mod}\ t), 0}
\right].
\ee

\item \textbf{$p=t$. Temporal random driving.} For $p=t$, i.e. the temporal random case, we have
\be
\widetilde{\phi}
=
i^{\delta_{a,1}}
\left[
\mu^{2t}\sqrt{2t^3}\,\delta_{n_b,0}\delta_{n_c,0}\delta_{n_d,0}
+
(1-\mu^{2t})\sqrt{\frac{2}{t}}\;
\delta_{n_b\ (\mathrm{mod}\ 2),\,0}\;
\delta_{n_c+n_d+n_b/2\ (\mathrm{mod}\ t), 0}
\right].
\ee
\end{itemize}

\begin{figure}[ht!]
    \centering
    \includegraphics[width=0.6 \textwidth]{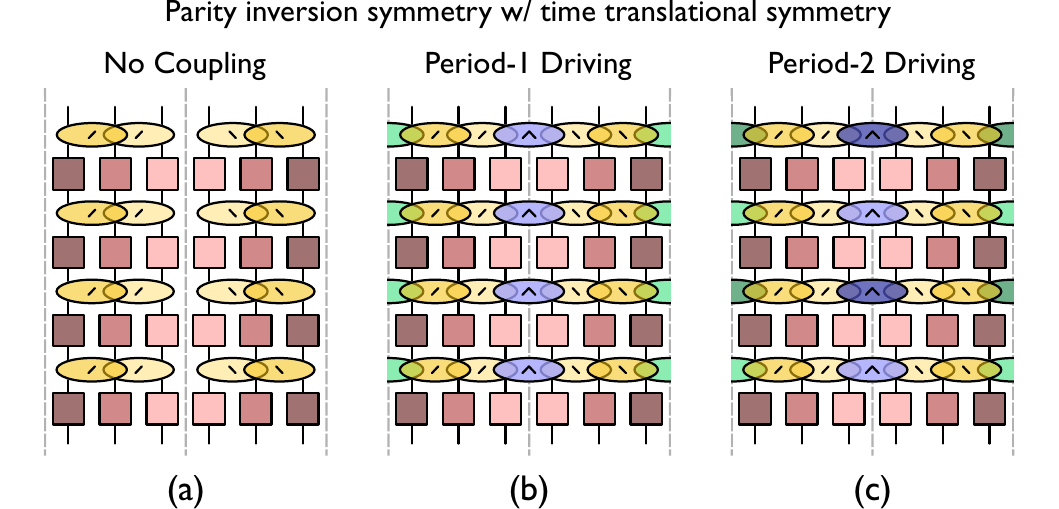}
    \caption{\textbf{Models with periodic driving along parity inversion axes.} We consider quantum chaotic many-body systems described by the parity-inversion-symmetric Random Phase Model with discrete time translational symmetry. Along the parity inversion axes, we study coupling gates corresponding to (a) no coupling, (b) period-one driving, and (c) period-two driving.
    }
    \label{fig:parity_models_axes}
\end{figure}

\begin{figure}
\centering
\includegraphics[width=0.7\textwidth]
{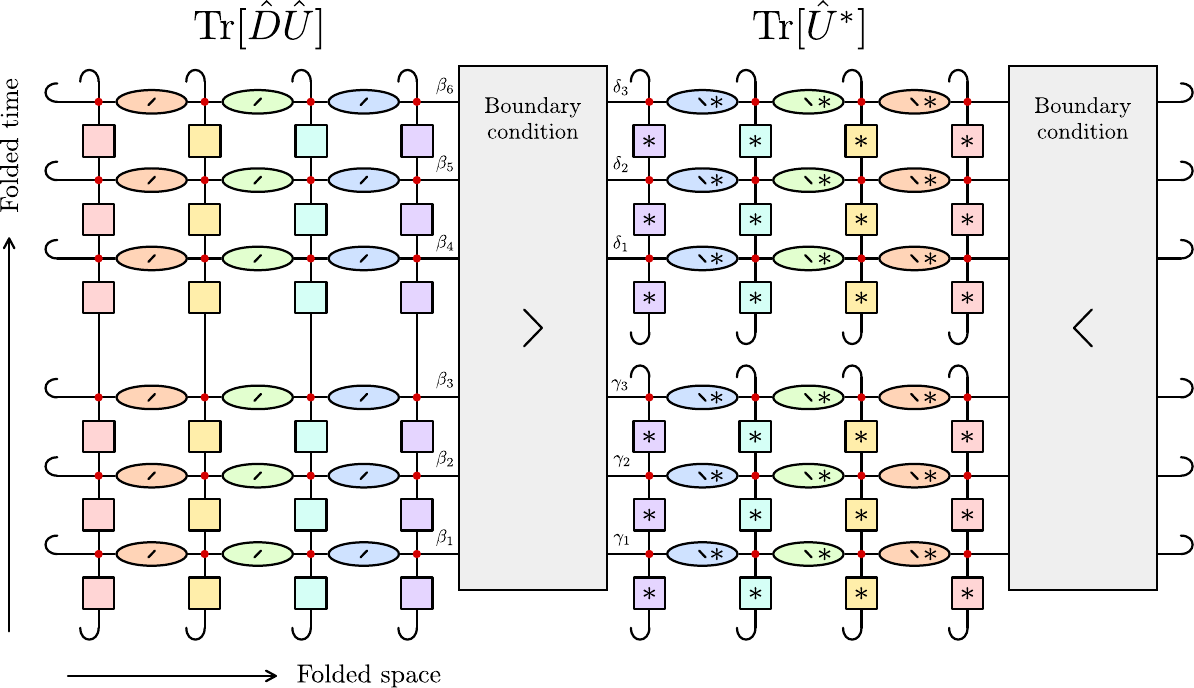}
\caption{
\textbf{Tensor network of TopSFF.}
We show the tensor network representation of TopSFF for $t=3$ and $L=8$. Gates with identical colours are identified. Red dots label where the degrees of freedom live. 
}
\label{fig:dual_sim_bc}
\end{figure}

\section{Exact TopSFF with spatially extended defects at large $q$}\label{app:tsff_exact_q_spatial}

In this appendix, we present the exact large-\(q\) evaluation of the TopSFF with a spatially extended global swap defect, using the generalized Boltzmann factor derived in \ref{app:tsff_toolbox}. 
In Sec.~\ref{app_sec:cayley_hamilton}, we express the TopSFF using the Cayley-Hamilton relation and introduce boundary-resolved TopSFF. 
In Sec.~\ref{app:tsff_fixed_mom}, we evaluate the TopSFF in selected fixed momentum sectors. 
In Sec.~\ref{app:tsff_regular_bc}, we combine these sectors to obtain the TopSFF in the time translation symmetric case. 
Finally, in Sec.~\ref{app:engineered_bc}, we show how engineered parity-inversion axes, including no-coupling and period-two driving, can be used to access and isolate selected momentum sectors such as \((\pi,\pi,0)\).

\subsection{TopSFF and Cayley-Hamilton }
\label{app_sec:cayley_hamilton}
The TopSFF with the global swap operator inserted as a spatially extended topological defect in the parity inversion symmetric and discrete time translational invariant RPM can be obtained by evaluating 
\begin{equation}
\big[\widetilde{\boltz}(k_b,k_c,k_d)^L\big]_{a_1 a_2}.
\end{equation}
for some positive integer $L$. 
The eigenvalues $\lambda_\pm$ of $\widetilde{\boltz}$ are
\be 
\lambda_\pm=\frac{\Trb\pm\sqrt{\Delta}}{2},
\qquad
\Trb:= \tr\widetilde{\boltz}, 
\qquad 
\Detb:= \det \widetilde{\boltz},
\qquad
\Delta:=\Trb^{2}-4\Detb. 
\ee 
We evaluate $\widetilde{\boltz}(k_b,k_c,k_d)^L$ by applying the Cayley-Hamilton theorem to the $2\times 2$ matrix $\widetilde{\boltz}$, so that all powers of $\widetilde{\boltz}$ are controlled by its trace and determinant as
\begin{equation}
\widetilde{\boltz}^{\,L}
=
u_L\,\widetilde{\boltz}
-
\Detb\,u_{L-1}\,\iden,
\label{eq:tildeboltz_power_closed}
\end{equation}
where $u_0=0$, $u_1=1$, and
\begin{equation}
u_{L+1}=\Trb\,u_L-\Detb\,u_{L-1}.
\end{equation}
Thus, the $L$-dependence is governed by the two scalar invariants $\Trb$ and $\Detb$ through the recurrence relation for $u_L$. 
If $\lambda_+\neq \lambda_-$, then
\begin{equation}
\widetilde{\boltz}^{\,L}
=
\frac{\lambda_+^{L}-\lambda_-^{L}}{\lambda_+-\lambda_-}\,\widetilde{\boltz}
+
\frac{\lambda_+\lambda_-^{L}-\lambda_-\lambda_+^{L}}{\lambda_+-\lambda_-}\,\iden.
\label{eq:Ttilde_power_u_v}
\end{equation}
If $\lambda_+=\lambda_-=\lambda$ (including exceptional points where $\widetilde{\boltz}$ is not diagonalizable), then the expression above has the smooth limit
\begin{equation}
\widetilde{\boltz}^{\,L}
\;\longrightarrow\;
L\lambda^{L-1}\widetilde{\boltz}
+
(1-L)\lambda^{L}\iden,
\qquad
\text{as} \qquad \Delta \rightarrow 0.
\label{eq:Ttilde_power_EP_closed}
\end{equation}
TopSFF can be written as a quadratic form of \(\widetilde{\boltz}(k)^L\) with respect to a boundary vector in momentum sector \(k\). 
Specifically, if \(\widetilde{\phi}(k)=\mathcal{N}(k)\varphi\) with \(\varphi=(\varphi_1,\varphi_2)^T\), then, using the Cayley-Hamilton identity, we obtain
\begin{equation}\label{eq:Kvarphi_general}
\ba
\Ksf_{k}
:=& \,
\varphi^T \widetilde{\boltz}(k)^L \varphi
=
g(k)\,u_L(k)
-
h \det \widetilde{\boltz}(k)\,u_{L-1}(k),
\qquad
\varphi=\binom{\varphi_1}{\varphi_2},
\\
h:=& \, \varphi^T\varphi=\varphi_1^2+\varphi_2^2,
\\
g(k):=& \, \varphi^T \widetilde{\boltz}(k)\varphi=
\Ttilde_{00}(k)\varphi_1^2
+
\bigl[\Ttilde_{01}(k)+\Ttilde_{10}(k)\bigr]\varphi_1\varphi_2
+
\Ttilde_{11}(k)\varphi_2^2.
\\
u_L(k)= & 
\begin{cases}
\dfrac{\lambda_+^L-\lambda_-^L}{\lambda_+-\lambda_-},
\qquad \qquad \qquad \qquad  & \DeltaDisc(k)\neq 0,\\
L\left(\dfrac{\Trb}{2}\right)^{L-1},
& \DeltaDisc(k)=0.
\end{cases}
\ea
\end{equation}
where eigenvalues $\lambda\pm(k)$ are distinct for $\Delta \neq 0$, and the eigenvalues coalesce $\lambda_+(k)=\lambda_-(k)=\lambda(k)$ for $\Delta=0$ at the EP.
For certain momentum sectors, the Boltzmann factor has \(\mathcal{PT}\) symmetry, and the boundary vector has the form $\varphi = (1,i)^T$, then the TopSFF is given by
\be\label{app_eq:ktop_cayley_hamilton_pt}
\ba
\Ksf_k:=& \,\varphi^T\widetilde{\boltz}(k)^L\varphi
=
g(k)\,u_L(k),
\qquad
\Ttilde(k)=
\begin{pmatrix}
\Asf(k) & i\Bsf(k)\\
i\Bsf(k) & \Csf(k)
\end{pmatrix}, \qquad \varphi=\binom{1}{i},
\\
g(k):=& \, \varphi^T\Ttilde(k)\varphi
=
\Asf(k)-\Csf(k)-2\Bsf(k),
\\
u_L(k)=&
\begin{cases}
\dfrac{\lambda_+(k)^L-\lambda_-(k)^L}{\lambda_+(k)-\lambda_-(k)},
& \DeltaDisc(k)>0
\qquad (\mathcal{PT}\text{ unbroken}),
\\
L\left(\dfrac{\Trb(k)}{2}\right)^{L-1},
& \DeltaDisc(k)=0
\qquad (\text{Exceptional point}),
\\
\bigl[\Detb(k)\bigr]^{\frac{L-1}{2}}
\dfrac{\sin\bigl(L\theta(k)\bigr)}{\sin\theta(k)},
\qquad
\theta(k):=\arccos\left(\dfrac{\Trb(k)}{2\sqrt{\Detb(k)}}\right),
& \DeltaDisc(k)<0
\qquad (\mathcal{PT}\text{ broken}).
\end{cases}
\ea
\ee
where $h=0$, and $\Asf(k),\Bsf(k),\Csf(k)\in\mathbb{R}$.
In the \(\mathcal{PT}\) broken phase, the eigenvalues form a complex-conjugate pair,
\(\lambda_\pm(k)=\rho(k)e^{\pm i\theta(k)}=x(k)\pm i y(k)\). The spectral phase \(\theta(k)\), with \(\tan\theta(k)=y(k)/x(k)\), controls the oscillation period in \(L\), while the modulus
\(\rho(k)=|\lambda_\pm(k)|=\sqrt{\Detb(k)}\)
controls the exponential envelope. For \(x(k)>0\), a smaller ratio \(|y(k)/x(k)|\) gives a smaller \(\theta(k)\) and hence a longer oscillation period \(L_{\rm osc}\sim 2\pi/\theta(k)\). In the same phase, TopSFF grows exponentially if \(\rho(k)>1\), decays exponentially if \(\rho(k)<1\), and is purely oscillatory  if \(\rho(k)=1\).

\subsubsection{Boundary-resolved TopSFF}\label{app_sec:boundary_resolved_topsff}
At the exceptional point, we would like to observe the linear-in-$L$ behaviour in Eq.~\eqref{app_eq:ktop_cayley_hamilton_pt}, which arises from Jordan non-diagonality. However, this contribution can be suppressed if $g(k)=0$ at the EP, as will occur
in the case considered below. To avoid this cancellation, we instead consider boundary-resolved TopSFFs defined with the boundary vectors
\begin{equation}
    v_0=
    \begin{pmatrix}
        1\\0
    \end{pmatrix},
    \qquad
    v_1=
    \begin{pmatrix}
        0\\i
    \end{pmatrix}.
\end{equation}
which have supports only in the Gaussian tDW basis state and non-Gaussian tDW basis state respectively. We define the three boundary-resolved amplitudes
\begin{align}
    \Ksf_{00}
    &:=
    v_0^T \widetilde{\mathcal B}(k)^L v_0,
    \\
    \Ksf_{11}
    &:=
    v_1^T \widetilde{\mathcal B}(k)^L v_1,
    \\
    \Ksf_{01}
    &:=
    v_0^T \widetilde{\mathcal B}(k)^L v_1 .
\end{align}
such that, for $\varphi=v_0+v_1= (1,i)^T$, we have
\begin{equation}
    \Ksf_k
    :=
    \varphi^T \widetilde{\mathcal B}(k)^L \varphi
    =
    \Ksf_{00}+2\Ksf_{01}+\Ksf_{11}.
\end{equation}
Using the Cayley-Hamilton form as before, the boundary-resolved TopSFF are
\begin{align}
    \Ksf_{00}
    &=
    \Asf(k)u_L(k)-\Detb(k)u_{L-1}(k),
    \\
    \Ksf_{11}
    &=
    -\Csf(k)u_L(k)+\Detb(k)u_{L-1}(k),
    \\
    \Ksf_{01}
    &=
    -\Bsf(k)u_L(k).
\end{align}
Equivalently, for $\Delta(k)\neq 0$,
\begin{align}
    \Ksf_{00}
    &=
    \frac{
    \bigl(\Asf(k)-\lambda_-(k)\bigr)\lambda_+(k)^L
    -
    \bigl(\Asf(k)-\lambda_+(k)\bigr)\lambda_-(k)^L
    }
    {\lambda_+(k)-\lambda_-(k)},
    \\
    \Ksf_{11}
    &=
    -\frac{
    \bigl(\Csf(k)-\lambda_-(k)\bigr)\lambda_+(k)^L
    -
    \bigl(\Csf(k)-\lambda_+(k)\bigr)\lambda_-(k)^L
    }
    {\lambda_+(k)-\lambda_-(k)},
    \\
    \Ksf_{01}
    &=
    -\Bsf(k)
    \frac{\lambda_+(k)^L-\lambda_-(k)^L}
    {\lambda_+(k)-\lambda_-(k)} .
\end{align}
At an exceptional point, $\Delta(k)=0$, the two eigenvalues coalesce,
\begin{equation}
\lambda_{\rm EP}(k) :=  \lambda_+(k)=\lambda_-(k)
    =
    \frac{\tr \widetilde{\boltz}(k)}{2}.
\end{equation}
The transfer matrix can be written in terms of the identity matrix and a nilpotent matrix,
\begin{equation}
    \widetilde{\mathcal B}(k)
    =
    \lambda_{\rm EP}(k)\iden + N(k),
    \qquad
    N(k)^2=0 .
\end{equation}
In the original Gaussian/non-Gaussian basis, the nilpotent matrix is
\begin{equation}
    N(k)
    =
    \begin{pmatrix}
        \delta(k) & i\Bsf(k)\\
        i\Bsf(k) & -\delta(k)
    \end{pmatrix},
    \qquad
    \delta(k):=\frac{\Asf(k)-\Csf(k)}{2}.
\end{equation}
At the exceptional point, we have
\(    \delta(k)^2=\Bsf(k)^2\), with  \(\delta(k)=s \Bsf(k)\)
and     \(s=\pm 1\). Therefore
\begin{equation}
    \widetilde{\mathcal B}(k)^L
    =
    \lambda_{\rm EP}(k)^L \iden
    +
    L\lambda_{\rm EP}(k)^{L-1}N(k),
\end{equation}
and hence
\begin{align}
    \Ksf_{00}\big|_{\rm EP}
    &=
    \lambda_{\rm EP}(k)^L
    +
    L\,\delta(k)\lambda_{\rm EP}(k)^{L-1},
    \\
    \Ksf_{11}\big|_{\rm EP}
    &=
    -\lambda_{\rm EP}(k)^L
    +
    L\,\delta(k)\lambda_{\rm EP}(k)^{L-1},
    \\
    \Ksf_{01}\big|_{\rm EP}
    &=
    -\Bsf(k)\,
    L\lambda_{\rm EP}(k)^{L-1}.
\end{align}
To observe TopSFF's polynomial scaling in $L$ due to the exceptional point and Jordan non-diagonality, we compute the quantity
\begin{equation}\label{app_eq:tsff_jordan}
  \left.  \frac{\Ksf_{00}+ \Ksf_{11}}
    {[\tr \widetilde{\boltz}(k)/2]^{L-1}} \right|_{\rm EP}= 2  \delta(k) L \, \propto \,  L \, ,
\end{equation}
which is linear in $L$ at the exceptional point. We compute
\(
 (\Ksf_{00}+ \Ksf_{11})/[\tr \widetilde{\boltz}(k)/2]^{L-1}
\)
across the exceptional point at finite $q$ in
Fig.~\ref{fig:jordan_diff_t} and at large $q$ in Fig.~\ref{fig:pipi0_jordan_nondiag_largeq}. In both cases, the results show excellent
agreement with the linear scaling in $L$ in  Eq.~\eqref{app_eq:tsff_jordan}.

\subsection{TopSFF for fixed momentum sectors}\label{app:tsff_fixed_mom}
Next, we evaluate the TopSFF for four $\mathcal{PT}$ symmetric momentum sectors that arises from evaluating TopSFF with period-1 and period-2 driving along the parity inversion axes. 
It is convenient to define
\begin{equation}
S_2(n):=\frac{1-\mu^{-2n}}{1-\mu^{-2}},
\qquad
S_4(n):=\frac{1-\mu^{-4n}}{1-\mu^{-4}},
\qquad
\widehat S_2:=S_2(t-1)=\frac{1-\mu^{-2(t-1)}}{1-\mu^{-2}}.
\end{equation}
and
\begin{equation}\label{app_eq:geom_sum_abcd}
\begin{aligned}
A_n
&:=\sum_{r=1}^{n}(\mu^{-2r}-1)^2
=\mu^{-4}S_4(n)-2\mu^{-2}S_2(n)+n,
\\
B_n
&:=\sum_{r=1}^{n}(t-1+\mu^{-2r})^2
=n(t-1)^2+2(t-1)\mu^{-2}S_2(n)+\mu^{-4}S_4(n),
\\
C_t&:=A_t-(\mu^{-2}-1)^2,
\\
D_t&:=B_t-(t-1+\mu^{-2})^2.
\end{aligned}
\end{equation}

\subsubsection{TopSFF for \((k_b,k_c,k_d)=(0,0,0)\)}
We consider the TopSFF for \((k_b,k_c,k_d)=(0,0,0)\), which is $\mathcal{PT}$ symmetric as derived in \ref{app:pt_sym}.
We use the shorthand \(\zeroedge\equiv (k_b,k_c,k_d)=(0,0,0)\) for both the momentum sector and the corresponding subscripts. The matrix elements of the generalized Boltzmann factor are
\begin{equation}\label{app_eq:ABC_000}
\begin{aligned}
\Asf_{\zeroedge}
&=
\Big(1+2(t-1)\mu^{2t}+(t-1)^2\mu^{4t+2}\Big)
+\Big(\mu^{4t-4}+2(t-1)\mu^{4t-2}+(t-1)^2\mu^{4t+2}\Big)
\nonumber\\
&\qquad
+2\mu^{4t}\Big(D_t+(t-1)(\mu^2-1)\Big),
\\
\Bsf_{\zeroedge}
&=2\mu^{4t}B_t,
\\
\Csf_{\zeroedge}
&=
-\mu^{4t}\Big[t^2+\bigl(t-1+\mu^{-2t}\bigr)^2+2B_{t-1}\Big].
\end{aligned}
\end{equation}
Using Eq.~\eqref{app_eq:ktop_cayley_hamilton_pt}, we obtain TopSFF at \((k_b,k_c,k_d)=(0,0,0)\) as
\begin{equation}\label{app_eq:tsff_zero_edge}
\Ksf_{\zeroedge}:=\varphi^T\widetilde{\boltz}(0,0,0)^L\varphi
=
\bigl(\Asf_{\zeroedge}-\Csf_{\zeroedge}-2\Bsf_{\zeroedge}\bigr)\,
u_L(0,0,0).
\end{equation}
A representative example of the TopSFF behaviour in the \((k_b,k_c,k_d)=(0,0,0)\) sector is shown in Fig.~\ref{fig:tsff_three_sectors_data}.

\subsubsection{TopSFF for \((k_b,k_c,k_d)=(0,k,-k)\)}
We consider the TopSFF for \((k_b,k_c,k_d)=(0,k,-k)\) with \(k\neq0\), which is $\mathcal{PT}$ symmetric as derived in \ref{app:pt_sym}.
We use the shorthand \(\zerobulk\equiv (k_b,k_c,k_d)=(0,k,-k)\) for both the momentum sector and the corresponding subscripts. The matrix elements of the generalized Boltzmann factor are
\begin{equation}\label{app_eq:ABC_0k}
\begin{aligned}
\Asf_{\zerobulk}
&=
\Big(1-2\mu^{2t}+\mu^{4t+2}\Big)
+\Big(\mu^{4t-4}-2\mu^{4t-2}+\mu^{4t+2}\Big)
+2\mu^{4t}\Big(C_t+(t-1)(\mu^2-1)\Big),
\\
\Bsf_{\zerobulk}
&=2\mu^{4t}A_t,
\\
\Csf_{\zerobulk}
&=
-\mu^{4t}\Big[\bigl(\mu^{-2t}-1\bigr)^2+2A_{t-1}\Big].
\end{aligned}
\end{equation}
Using Eq.~\eqref{app_eq:ktop_cayley_hamilton_pt}, we obtain TopSFF at \((k_b,k_c,k_d)=(0,k,-k)\) with \(k\neq0\) as
\begin{equation}\label{app_eq:tsff_zero_bulk}
\Ksf_{\zerobulk}:=\varphi^T\widetilde{\boltz}(0,k,-k)^L\varphi
=
\bigl(\Asf_{\zerobulk}-\Csf_{\zerobulk}-2\Bsf_{\zerobulk}\bigr)\,
u_L(0,k,-k), 
\end{equation}
for \( k\neq 0\).
A representative example of the TopSFF behaviour in the \((k_b,k_c,k_d)=(0,k,-k)\) sector is shown in Fig.~\ref{fig:tsff_three_sectors_data}.

\subsubsection{TopSFF for \((k_b,k_c,k_d)=(\pi,\pi,0)\)}
We consider the TopSFF for \((k_b,k_c,k_d)=(\pi,\pi,0)\), which is $\mathcal{PT}$ symmetric as derived in \ref{app:pt_sym}.
We use the shorthand \(\piedge\equiv (k_b,k_c,k_d)=(\pi,\pi,0)\) for both the momentum sector and the corresponding subscripts. The matrix elements of the generalized Boltzmann factor are
\begin{equation}\label{app_eq:ABC_pipi0}
\begin{aligned}
\Asf_{\piedge}
&=
1+\mu^{4t-4}
+(t-2)\bigl(\mu^{2t}+\mu^{4t-2}\bigr)
-2(t-1)\mu^{4t+2}
+2t\,\mu^{4t-4}\widehat S_2
-2t(t-1)\mu^{4t}
+2\mu^{4t}C_t
+2\mu^{4t}(1-\mu^2),
\\
\Bsf_{\piedge}
&=
\mu^{4t}\Big[
t\Big((\mu^{-2t}-1)+(\mu^{-2}-1)+(\mu^{-2}+\mu^{-4})\widehat S_2-2(t-1)\Big)
+(\mu^{-2t}-1)^2+(\mu^{-2}-1)^2+A_{t-1}+C_t
\Big],
\\
\Csf_{\piedge}
&=
-\mu^{4t}\Big[
t\Big((\mu^{-2t}-1)+2\bigl(\mu^{-2}\widehat S_2-(t-1)\bigr)\Big)
+(\mu^{-2t}-1)^2+2A_{t-1}
\Big].
\end{aligned}
\end{equation}
Using Eq.~\eqref{app_eq:ktop_cayley_hamilton_pt}, we can simplify $g_{\piedge}$ and obtain TopSFF at \((k_b,k_c,k_d)=(\pi,\pi,0)\) as
\begin{equation}\label{app_eq:tsff_pi_edge}
\Ksf_{\piedge}:=\varphi^T\widetilde{\boltz}(\pi,\pi,0)^L\varphi
=
\mu^{4t-4}(1-\mu^2)\Bigl(2t\mu^4-t\mu^2+\mu^2-1\Bigr)\,
u_L(\pi,\pi,0).
\end{equation} 
A representative example of the TopSFF behaviour in the \((k_b,k_c,k_d)=(\pi,\pi,0)\) sector with \(\mathcal{PT}\) symmetry breaking transition is shown in Fig.~\ref{app_fig:large_q_PT_trans}. The signatures of  \(\mathcal{PT}\) symmetry breaking transition and the exceptional point are discussed in Appendix \ref{app:tsff_and_pt} and in Appendix \ref{sec:EP_coalescence_pipi0}.

\subsubsection{TopSFF for \((k_b,k_c,k_d)=(\pi,k,\pi-k)\)}
We consider the TopSFF for \((k_b,k_c,k_d)=(\pi,k,\pi-k)\) with \(k\neq0,\pi\), which is $\mathcal{PT}$ symmetric as derived in \ref{app:pt_sym}. We use the shorthand \(\pibulk\equiv (k_b,k_c,k_d)=(\pi,k,\pi-k)\) for both the momentum sector and the corresponding subscripts. The matrix elements of the generalized Boltzmann factor are
\begin{equation}\label{app_eq:ABC_pibulk}
\begin{aligned}
\Asf_{\pibulk}
&=
1+\mu^{4t-4}-2\mu^{2t}-2\mu^{4t-2}+2\mu^{4t}+2\mu^{4t}C_t,
\\
\Bsf_{\pibulk}
&=
\mu^{4t}\Big[(\mu^{-2t}-1)^2+(\mu^{-2}-1)^2+A_{t-1}+C_t\Big],
\\
\Csf_{\pibulk}
&=
-\mu^{4t}\Big[(\mu^{-2t}-1)^2+2A_{t-1}\Big].
\end{aligned}
\end{equation}
Using Eq.~\eqref{app_eq:ktop_cayley_hamilton_pt}, we can simplify $g_{\pibulk}$ and obtain TopSFF at \(\pibulk\equiv (k_b,k_c,k_d)=(\pi,k,\pi-k)\) with \(k\neq0,\pi\) as
\begin{equation}\label{app_eq:tsff_pi_bulk}
\Ksf_{\pibulk}:=\varphi^T\widetilde{\boltz}(\pi,k,\pi-k)^L\varphi
=
-\mu^{4t-4}(1-\mu^2)^2\,u_L(\pi,k,\pi-k),
\end{equation}
for \( k\neq 0, \pi\). A representative example of the TopSFF behaviour in the \((k_b,k_c,k_d)=(\pi,k,\pi-k)\) sector is shown in Fig.~\ref{fig:tsff_three_sectors_data}.

\begin{figure}[ht]
\centering
\includegraphics[width=0.32\textwidth]{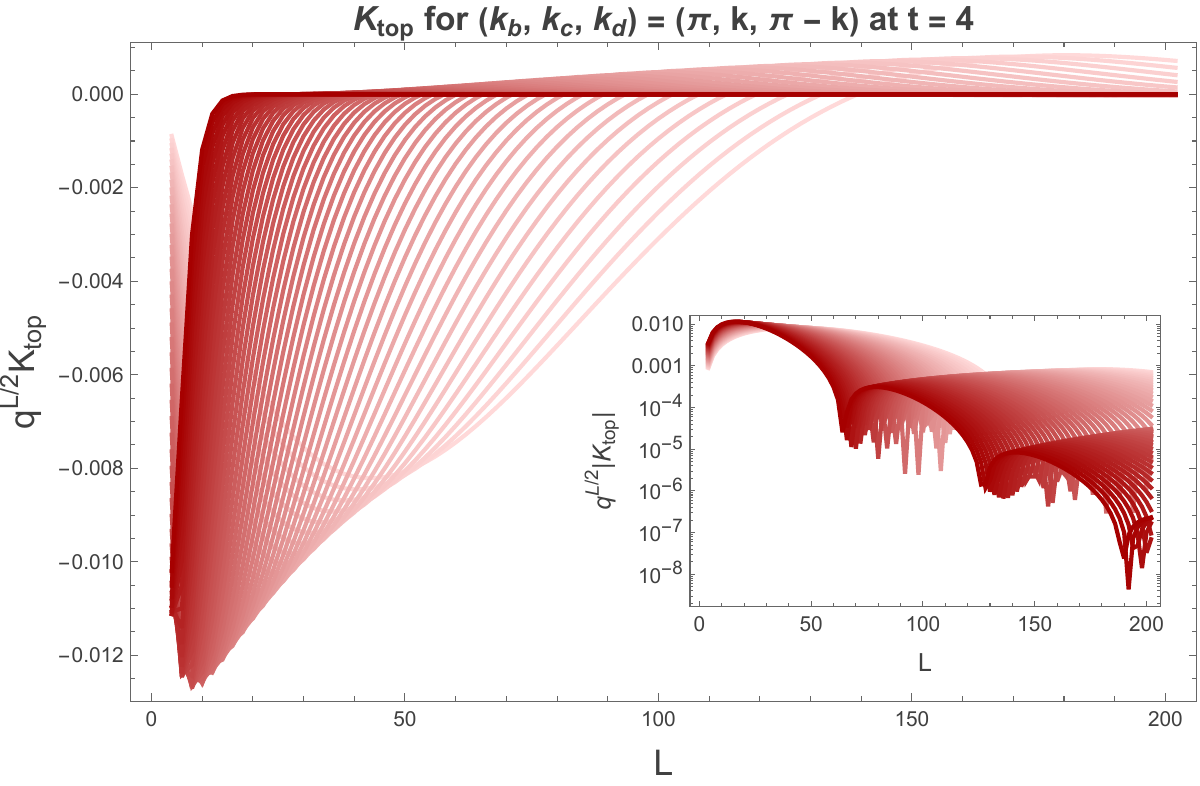}
\hfill
\includegraphics[width=0.32\textwidth]{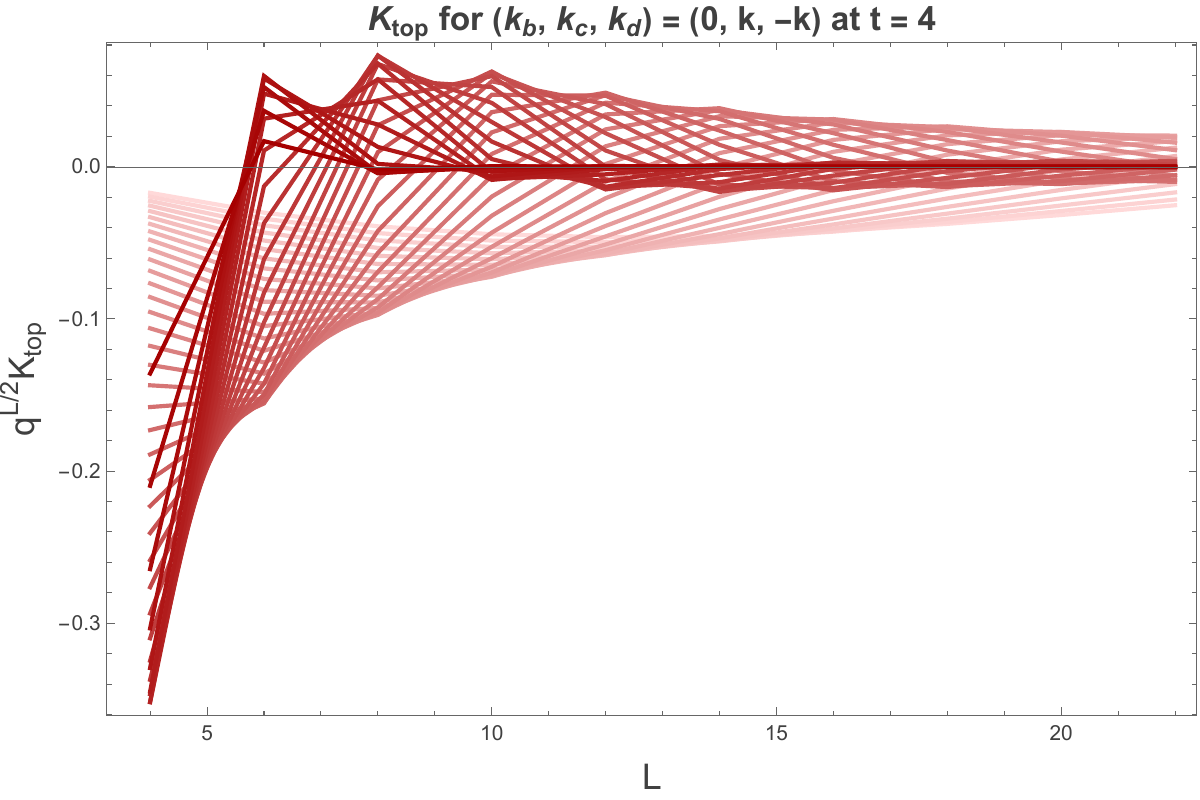}
\hfill
\includegraphics[width=0.32\textwidth]{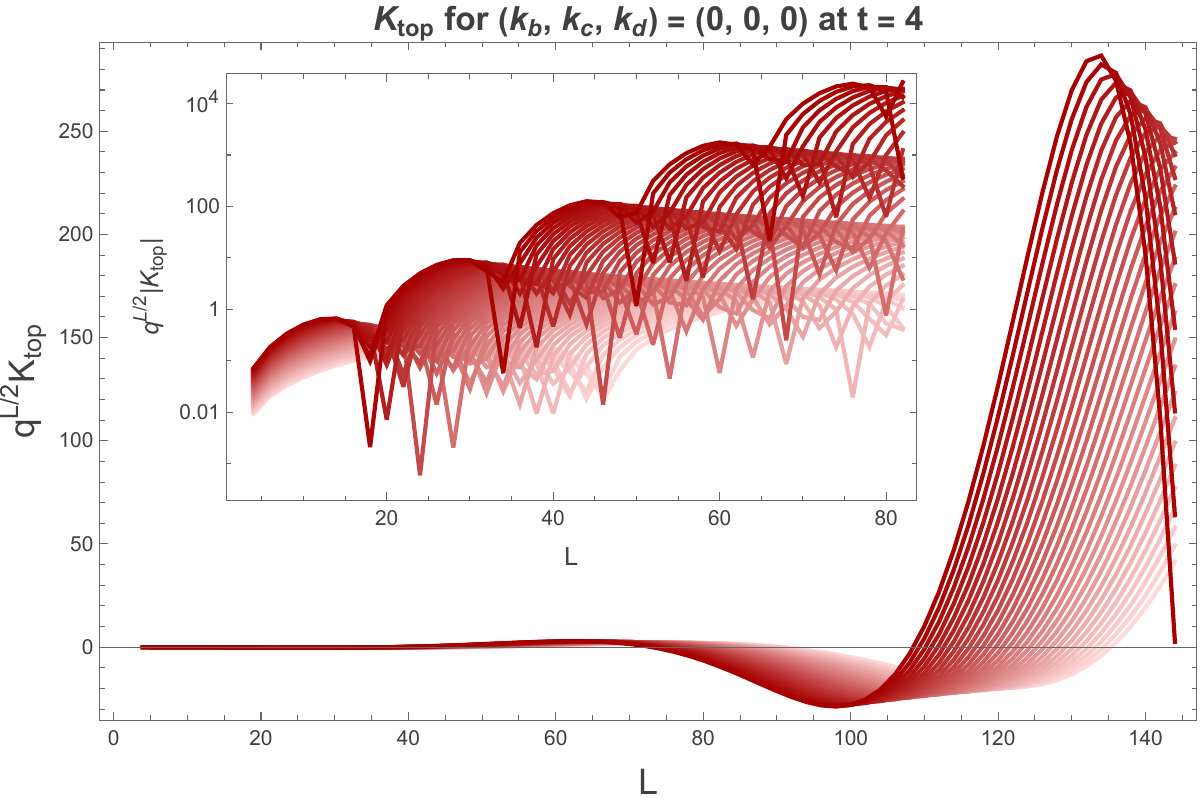}
\caption{
\textbf{Exact large-$q$ TopSFF in different momentum sectors at $t=4$.}
We show the exact large-$q$ TopSFF as a function of $L$ in the momentum sectors $(k_b,k_c,k_d)=(\pi,k,\pi-k)$, $(0,k,-k)$, and $(0,0,0)$, from left to right, respectively. The corresponding $\mu$ ranges are $\mu\in[0.6,0.9]$, $\mu\in[0.7,0.99]$, and $\mu\in[0.6,0.62]$, sampled in 50, 40, and 30 slices, respectively, and shown from light to dark red with increasing $\mu$. 
In the left and right panels, the TopSFF exhibits exponentially decaying and exponentially growing oscillatory behaviour, respectively. These oscillations arise from complex-conjugate eigenvalues of the Boltzmann factor, with modulus smaller than $1$ in the left panel and larger than $1$ in the right panel, see Fig.~\ref{fig:grid_4_sectors}. In the middle panel, the TopSFF mostly displays an exponentially decaying envelope, although an exceptional point exists in this sector, as described by \eqref{app_eq:ep_0k_sector}, the associated oscillatory behaviour is largely obscured by this decay.
Insets in the left and right panels show the corresponding semi-log plots for $\mu\in[0.6,0.7]$, where peaks and troughs correspond to the oscillations in the main panels.
}
\label{fig:tsff_three_sectors_data}
\end{figure}

\begin{figure}
    \centering

    \includegraphics[width=0.49\textwidth]{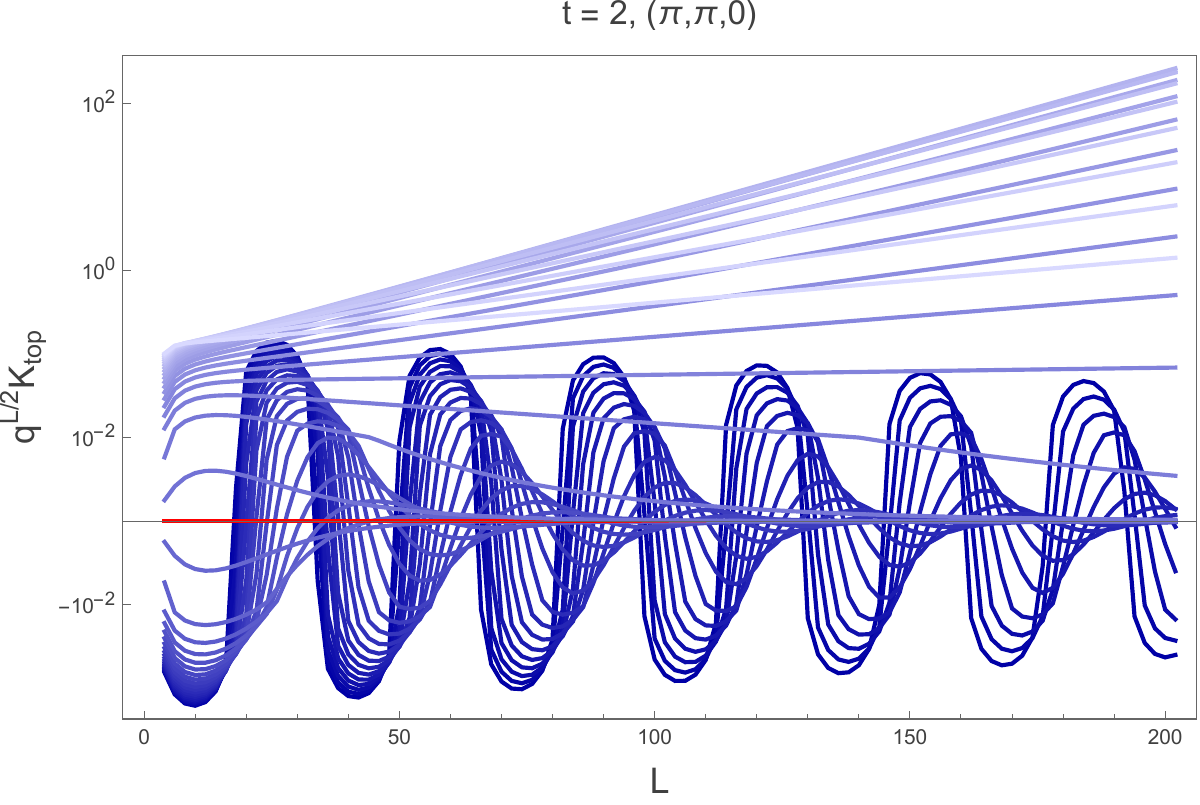}
    \hfill
    \includegraphics[width=0.49\textwidth]{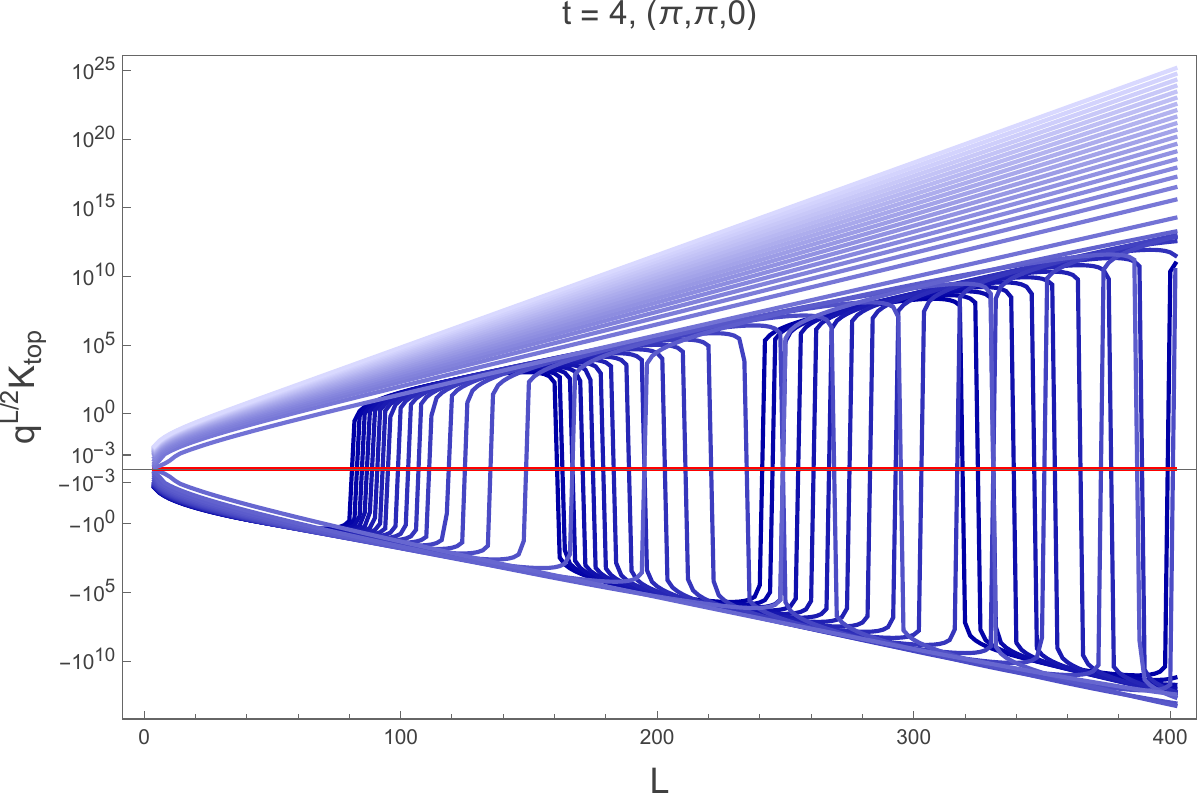}

    \vspace{0.5em}

    \includegraphics[width=0.49\textwidth]{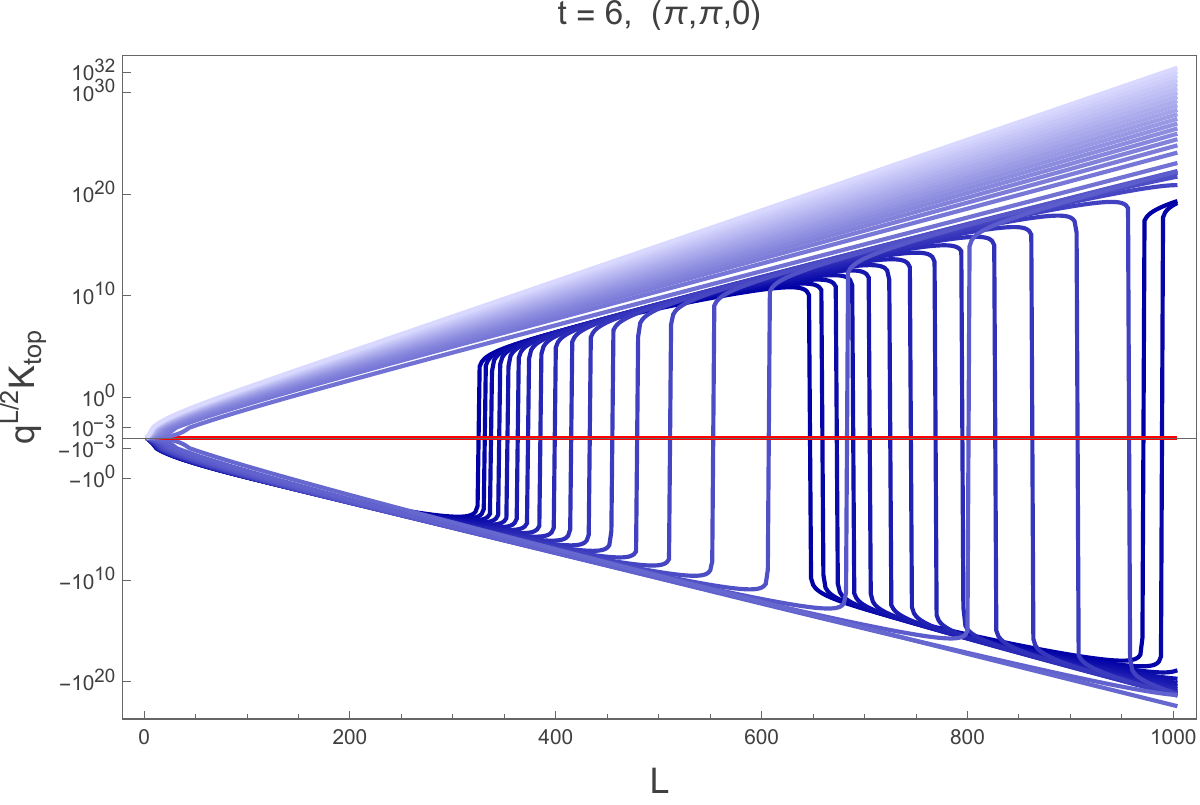}
    \hfill
    \includegraphics[width=0.49\textwidth]{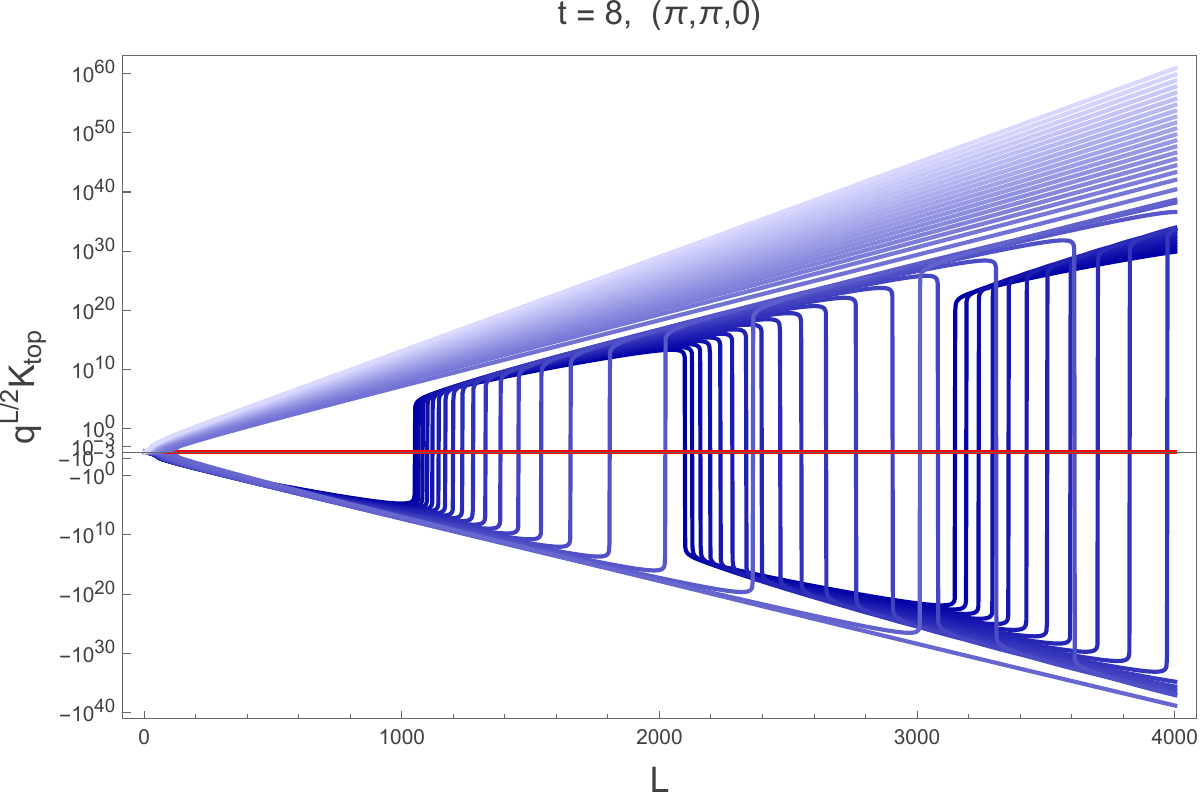}

    \caption{
\textbf{$\mathcal{PT}$ symmetry breaking transitions in the exact large-$q$ TopSFF in the $(\pi,\pi,0)$ sector.}
We plot the exact large-$q$ TopSFF as a function of $L$ for $t=2,4,6,8$, focusing on the neighbourhood
$\mu \in [\mu_{\EP}(t)-\delta\mu,\mu_{\EP}(t)+\delta\mu]$
of the exceptional point $\mu_{\EP}(t)$ given in Eq.~\eqref{eq:muEP_closed}.
For each value of $t$, this interval is sampled using $40$ slices in $\mu$, with
$\delta\mu=0.1,0.02,0.01,0.125$ for $t=2,4,6,8$, respectively.
Curves at smaller $\mu$, corresponding to stronger interaction strength $\epsilon$, are shown in darker blue.
The TopSFF evaluated exactly at the exceptional point $\mu=\mu_{\EP}$ is shown in red.
Note that for this particular boundary condition, the prefactor multiplying the exceptional point contribution vanishes, and hence the TopSFF is identically zero. The Jordan non-diagonality can be extracted by probing boundary-resolved TopSFF, as shown in Fig.~\ref{fig:pipi0_jordan_nondiag_largeq} and in Appendix \ref{sec:EP_coalescence_pipi0}.
}
    \label{app_fig:large_q_PT_trans}
\end{figure}

\subsection{TopSFF with time translational symmetry}\label{app:tsff_regular_bc}
Here we provide exact leading-$q$ solution of TopSFF with the global swap operator inserted as spatially extended topological defects, in the presence of parity inversion symmetry and discrete time translational symmetry. For completeness, we restate the TopSFF in \eqref{eq:master-unitary} here as
\be\label{app_eq:tsff_largeq_def}
\ba
\overline{K}_{\spatdef} & \, =:   q^{-\Leff -1} \overline{\Knormalized}_{\spatdef} + O(q^{-{\Leff} -2}),
\\
\overline{\Knormalized}_{\spatdef} & \, =
\phi^T \boltz^{\Leff} \phi = \sum_{k_b,k_c,k_d} \widetilde{\phi}(-k)^T\, \widetilde{\boltz}(k)^{\Leff}\, \widetilde{\phi}(k).
\ea
\ee
where $\boltz$ and $\widetilde{\boltz}$ are the generalized Boltzmann factor in position and momentum space respectively. The components of the boundary state in momentum space are given by 
\begin{equation}
\widetilde{\phi}(a;k_b,k_c,k_d)
=
i^{\delta_{a,1}}
\left[
\sqrt{2t^3}\mu^{2t}\,\delta_{n_b,0}\delta_{n_c,0}\delta_{n_d,0}
+
\sqrt{2t}(1-\mu^{2t})\,\delta_{n_b,0}\delta_{n_c+n_d,\,0 \pmod{t}}
\right],
\end{equation}
which is nonzero only if \(k_b=0\) and \(k_d\equiv -k_c \pmod{2\pi}\). Restricted to these momenta, the generalized Boltzmann factor depends only on whether \((k_b,k_c,k_d)=(0,0,0)\) or \((k_b,k_c,k_d)=(0,k,-k)\). Note that the generalized Boltzmann factors at these momentum sectors have $\mathcal{PT}$ symmetry. We use the subscript \(\zeroedge\) to denote the momentum sector \((k_b,k_c,k_d)=(0,0,0)\), and the subscript \(\zerobulk\) to denote the momentum sector \((k_b,k_c,k_d)=(0,k,-k)\) with \(k\neq 0\). Thus, TopSFF can be written as
\begin{equation}
\label{eq:two-block-decomp}
\overline{\Knormalized}_\spatdef =
2t\Big[
\mathcal{N}_{\zeroedge}^2\,\varphi^T \widetilde{\boltz}^{\Leff}(0,0,0) \, \varphi 
+(t-1)\mathcal{N}_{\zerobulk}^2\,
\varphi^T \widetilde{\boltz}^{\Leff}(0,k,-k) \, \varphi
\Big],
\end{equation}
with
\be
\mathcal{N}_{\zeroedge} := \left[1+(t-1) \mu^{2t}\right],\qquad 
 \mathcal{N}_{\zerobulk} := \left(1-\mu^{2t}\right), \qquad
 \varphi = \binom{1}{i}.
\ee
Using \eqref{app_eq:ktop_cayley_hamilton_pt}, the TopSFF with global swap operator as a spatially extended topological defect can be evaluated as
\begin{equation}
\label{eq:Kspatdef-final}
\begin{aligned}
\overline{\Knormalized}_{\spatdef}
&=
2t
\left[1+(t-1) \mu^{2t}\right]^2 \Ksf_{\zeroedge}
+ 2t (t-1) \left(1-\mu^{2t}\right)^2
\Ksf_{\zerobulk}\,,
\end{aligned}
\end{equation}
where $\Ksf_{\zeroedge}$ and $\Ksf_{\zerobulk}$ are given in \eqref{app_eq:tsff_zero_edge} and \eqref{app_eq:tsff_zero_bulk} respectively.
The implications of the presence of $\mathcal{PT}$ symmetry and the existence of exceptional points are discussed below. As a consistency check, for $t=1$ and $\alpha=1$, \eqref{eq:Kspatdef-final} reduces to
\be\label{eq:floqparity_largeq_t1_check}
\ba
\overline{\Knormalized}_{\spatdef}
=& \,
2(\mu^4-1)\,
\frac{\lambda_{\zeroedge,+}^{\Leff}-\lambda_{\zeroedge,-}^{\Leff}}{\lambda_{\zeroedge,+}-\lambda_{\zeroedge,-}}
=
-\frac{4\sqrt{1-\mu^4}}{\sqrt{\mu^4+7}}
\left(\sqrt{2(1-\mu^4)}\right)^{\Leff}
\sin(\Leff\theta),
\\
\lambda_{\zeroedge,\pm}
=& \, 
\frac{1-\mu^4}{2}\pm \frac{1}{2}\sqrt{(\mu^4-1)(\mu^4+7)}=:r e^{\pm i\theta}.
\ea
\ee
where $r=\sqrt{2(1-\mu^4)}$ and $\cos\theta=\sqrt{(1-\mu^4)/8}$ or $\sin\theta=\sqrt{(\mu^4+7)/8}$. Eq.~\eqref{eq:floqparity_largeq_t1_check} coincides with (i) the large-$q$ expansion of the exact finite-$q$ result of the TopSFF in Eq.~\eqref{eq:t1_finite_q_v2} at $t=1$, and (ii) the exact large-$q$ calculation of the TopSFF for the RPM at $t=1$ in the absence of discrete time translational symmetry in Eq.~\eqref{eq:K_spatdef_ana_temp_rand}.

\subsection{TopSFF with engineered parity inversion axes}\label{app:engineered_bc}

 Different momentum sectors of TopSFF, in particular the $(\pi,\pi,0)$ sector containing the $\mathcal{PT}$ symmetry breaking transition, can be accessed by engineering the coupling gates along the parity inversion axes. We focus on two choices: no coupling and period-2 driving along these axes. Comparing the corresponding TopSFF allows us to isolate the $(k_b,k_c,k_d)=(\pi,\pi,0)$ contribution.

\subsubsection{No coupling}
Here we compute the TopSFF in Eq.~\eqref{app_eq:tsff_largeq_def} exactly for RPM with discrete time translational symmetry, but without coupling along the parity inversion axes. The generalized Boltzmann factor $\widetilde{\boltz}$ remains the same, while the boundary states in Eq.~\eqref{app_eq:k_top_def_floq} in momentum space are given by
\be 
\widetilde{\phi}(a;k_b,k_c,k_d)
=
i^{\delta_{a,1}}
\sqrt{2t^{3}}
\delta_{n_b,0}\delta_{n_c,0}\delta_{n_d,0},
\ee
Since the boundary states lie entirely in the momentum sector $(k_b,k_c,k_d)=(0,0,0)$, the TopSFF without coupling along parity inversion axes can be evaluated as
\be
\overline{\Knormalized}_\spatdef = 
8t^3  \Ksf_{\zeroedge}\, , 
\ee
where $\Ksf_{\zeroedge}$ is given in \eqref{app_eq:tsff_zero_edge}.
This case is useful for removing contributions in TopSFF associated with the zero momentum sector.

\subsubsection{Period-2 driving}
Here we compute the TopSFF in Eq.~\eqref{app_eq:tsff_largeq_def} exactly for RPM with discrete time translational symmetry, but with period-2 driving along the parity inversion axes. The generalized Boltzmann factor $\widetilde{\boltz}$ remains the same, while the boundary states in Eq.~\eqref{app_eq:k_top_def_floq} in momentum space are given by
\begin{equation} \label{app_eq:ini_vector_p2}
\begin{aligned}
\widetilde{\phi}(a;k_b,k_c,k_d)
=& \,
i^{\delta_{a,1}}
\left[
\mu^{2t}\sqrt{2t^3}\,\delta_{n_b,0}\delta_{n_c,0}\delta_{n_d,0}
+
(1-\mu^{2t})\sqrt{\frac{t}{2}}\;
\delta_{n_b\ (\mathrm{mod}\ t),\,0}\;
\delta_{n_c+n_d+n_b/2\ (\mathrm{mod}\ t), 0}
\right],
\end{aligned}
\end{equation}
which has support in the $\mathcal{PT}$ symmetric sectors \( (k_b,k_c,k_d) \in \{(0,k,-k), (\pi,k, \pi-k) \} \). Defining
\begin{equation}
    \begin{aligned}
\Ncal_0:=
\sqrt{\frac t2}\Big[1+(2t-1)\mu^{2t}\Big],
\qquad
\Ncal_\star:= \sqrt{\frac t2}\,(1-\mu^{2t}) \, ,
    \end{aligned}
\end{equation}
the TopSFF with period-2 driving along parity inversion axes can be evaluated as 
\be
\overline{\Knormalized}_\spatdef = 
\mathcal{N}^2_0\Ksf_{\zeroedge}
+(t-1)  \mathcal{N}^2_\star \Ksf_{\zerobulk}
+ \mathcal{N}^2_\star\Big[
2 \Ksf_{\piedge}
+(t-2)\Ksf_{\pibulk}
\Big].
\ee
where $\Ksf_{\zeroedge}$, $\Ksf_{\zerobulk}$, $\Ksf_{\piedge}$ and $\Ksf_{\pibulk}$ are respectively given in \eqref{app_eq:tsff_zero_edge}, \eqref{app_eq:tsff_zero_bulk}, \eqref{app_eq:tsff_pi_edge} and \eqref{app_eq:tsff_pi_bulk}.

\subsubsection{Isolating $(k_b, k_c,k_d)=(\pi,\pi,0)$ sector by engineering parity inversion axes}
We extract the TopSFF directly in the $(\pi,\pi,0)$ sector by computing TopSFF in finite-\(q\) via  dual spatial propagation of the quantum circuit, as discussed in Secs.~\ref{subsec:finite_q_symbolic_tsff} and \ref{app:num_sim}.
However, in principle, TopSFF of $(k_b, k_c,k_d)=(\pi,\pi,0)$ and its \(\mathcal{PT}\) symmetry breaking transition can be extracted by engineering the parity inversion axes. 
More specifically, comparing the no-coupling boundary condition with period-2 driving allows one to separate the dominant contributions between the $k_b=0$ sector and the $k_b=\pi$ sector. As shown in Fig.~\ref{fig:grid_4_sectors}, among the momentum sectors accessed by period-2 driving, namely $(0,0,0)$, $(0,k,-k)$, $(\pi,\pi,0)$, and $(\pi,k,\pi-k)$, only the $(0,0,0)$ and $(\pi,\pi,0)$ sectors contain eigenvalues with modulus larger than unity. The contributions from the remaining sectors therefore decay exponentially with system size $L$. Consequently, at large $L$, the contribution from the $(\pi,\pi,0)$ sector can in principle be isolated from the period-2 TopSFF by subtracting the corresponding zero-momentum contribution,
\be
\mathsf{K}_{\piedge}
=
\overline{\Knormalized}_{p=2}
-
\beta\,
\overline{\Knormalized}_{\mathrm{no\text{-}coupling}} ,
\ee
where $\overline{\Knormalized}_{p=2}$ denotes TopSFF with period-2 driving along the parity-inversion axes, $\overline{\Knormalized}_{\mathrm{no\text{-}coupling}}$ denotes TopSFF without coupling along these axes, and $\beta$ is chosen to fit the zero momentum sector contribution. This subtraction isolates the $(\pi,\pi,0)$ contribution up to terms that are exponentially small in $L$. In practice, however, we extract the TopSFF directly in the $(\pi,\pi,0)$ sector by computing TopSFF in finite-\(q\) via  dual spatial propagation of the quantum circuit, as discussed in Secs.~\ref{subsec:finite_q_symbolic_tsff} and \ref{app:num_sim}.

\section{Emergent \(\mathcal{PT}\) symmetric dimers}
\label{app:PT_mapping}

In this appendix, we first review the two-mode \(\mathcal{PT}\) dimer, and then show how the TopSFF generalized Boltzmann factor reduces to a canonical \(\mathcal{PT}\) symmetric dimer. This gives an interpretation of the \(\mathcal{PT}\) unbroken regime, \(\mathcal{PT}\) broken regime, and the exceptional point in terms of Gaussian/non-Gaussian tDW detuning and conversion.

\subsection{Review of $\mathcal{PT}$ symmetric dimers}
Consider a single complex amplitude $\psi(t)$ evolving under a generator $H$ as
\be
i\frac{d}{dt}\psi = H\psi.
\ee
If $H=i\gamma$ with $\gamma\in\mathbb{R}$, then
\be
\frac{d}{dt}\psi=\gamma\psi
\quad\Rightarrow\quad
\psi(t)=e^{\gamma t}\psi(0).
\ee
Thus $\gamma>0$ gives exponential growth of amplitude, which is often referred to as ``gain''. If instead $H=-i\gamma$, then
\be
\frac{d}{dt}\psi=-\gamma\psi
\quad\Rightarrow\quad
\psi(t)=e^{-\gamma t}\psi(0),
\ee
which leads to  exponential decay of amplitude, which is referred to as ``loss''. 

Consider the canonical $\mathcal{PT}$ dimer
\begin{equation}
H=
\begin{pmatrix}
+i\gamma & J\\
J & -i\gamma
\end{pmatrix},
\label{eq:PT_dimer}
\end{equation}
where $J,\gamma\in\mathbb{R}$. The diagonal entries $\pm i\gamma$ represent balanced gain/loss, while $J$ coherently couples the two modes. In particular, we have $\tr H = i\gamma-i\gamma =0$,
so the gain and loss are exactly balanced at the level of the trace. The eigenvalues of the $\mathcal{PT}$ dimer are
\begin{equation}
E_\pm=\pm\sqrt{J^2-\gamma^2},
\label{eq:Epm_PT_dimer}
\end{equation}
and a right eigenvector $\psi=(u,v)^T$ of the eigenvalue $E$ satisfies
\begin{equation}
(i\gamma-E)u+Jv=0
\qquad\Longrightarrow\qquad
\frac{v}{u}=\frac{E-i\gamma}{J}.
\label{eq:ratio_v_over_u}
\end{equation}
The Hamiltonian \eqref{eq:PT_dimer} satisfies the $\mathcal{PT}$ symmetry
\begin{equation}
(\mathcal{PT})H(\mathcal{PT})^{-1}=H,
\qquad
\mathcal{P}=\sigma_x,
\label{eq:PT_sym_H}
\end{equation}
and with $\mathcal{T}$ being complex conjugation. However, individual eigenvectors need not be invariant under $\mathcal{PT}$.
If $H|\psi\rangle=E|\psi\rangle$, then
\begin{equation}
H(\mathcal{PT}|\psi\rangle)=E^*(\mathcal{PT}|\psi\rangle),
\label{eq:PT_maps_E_to_Estar}
\end{equation}
so $\mathcal{PT}$ maps an eigenvector at $E$ to an eigenvector at $E^*$.

The phases of the model are given as follow.
\begin{itemize}
\item \textbf{Unbroken $\mathcal{PT}$ phase ($|\gamma|<|J|$): real spectrum and balanced modes.}
In this strong coupling regime with $E_\pm\in\mathbb{R}$, the time evolution is governed by pure phases, $e^{-iE_\pm t}$, and there is no exponential growth/decay in time: the eigenmodes are dynamically balanced.
The right eigenvectors can be chosen as $\mathcal{PT}$ eigenstates
\begin{equation}
    \mathcal{PT}|\psi_\pm\rangle=e^{i\theta_\pm}|\psi_\pm\rangle.
\end{equation}
Using $E^2=J^2-\gamma^2$ in \eqref{eq:ratio_v_over_u} gives
\begin{equation}
\left|\frac{v}{u}\right|
=\frac{|E-i\gamma|}{|J|}
=\frac{\sqrt{E^2+\gamma^2}}{|J|}
=\frac{\sqrt{J^2}}{|J|}
=1.
\label{eq:balanced_weights}
\end{equation}
Thus each eigenvector has equal weight on the gain and loss sites (up to a relative phase): the eigenvectors samples both sites so that net gain/loss cancels.

\item \textbf{Broken $\mathcal{PT}$ phase ($|\gamma|>|J|$): complex-conjugate spectrum and polarized modes.}
Write
\begin{equation}
E_\pm=\pm i\kappa,
\qquad
\kappa:=\sqrt{\gamma^2-J^2}>0.
\end{equation}
Then time evolution contains exponentials, $e^{-iE_\pm t}=e^{\pm \kappa t}$, so one eigenmode amplifies while the other eigenmode decays. $\mathcal{PT}$ cannot leave an eigenvector invariant, since it maps $E_\pm\mapsto E_\pm^*=E_\mp$. In other words,
\begin{equation}
\mathcal{PT}|\psi_+\rangle \propto |\psi_-\rangle,\qquad
\mathcal{PT}|\psi_-\rangle \propto |\psi_+\rangle.
\end{equation}
The ratio of the components of eigenvector becomes
\begin{equation}
\frac{v}{u}=\frac{i\kappa-i\gamma}{J}
=i\,\frac{\kappa-\gamma}{J},
\end{equation}
so generically $\bigl|v/u\bigr|\neq 1$, i.e., the right eigenvectors become polarized toward one site/mode. This eigenvector polarization is the signature accompanying $\mathcal{PT}$ symmetry breaking.

\item \textbf{Exceptional point ($|\gamma|=|J|$): coalescence of eigenvalues and eigenvectors.}
At threshold, $E_+=E_-=0$ and the two eigenvalues merge. Generically the corresponding eigenvectors coalesce, so $H$ becomes non-diagonalizable and is represented by a Jordan block. This coalescence of eigenvalues and eigenvectors defines the exceptional point.
\end{itemize}

\subsection{TopSFF and $\mathcal{PT}$ dimers}\label{app:tsff_and_pt}
According to Proposition \ref{prop:PT_classification_full}, in the momentum space, the generalized Boltzmann factor $\widetilde{\boltz}(k_b, k_c,k_d)$ has $\mathcal{PT}$ symmetry  with complex conjugation $\mathcal{T}$ and $\mathcal{P}:= \sigma_z$, if $k_b+k_c+k_d\equiv 0 \pmod{\pi}$. At such symmetric points, the generalized Boltzmann factor $\widetilde{\boltz}(k_b, k_c,k_d)$ is a 2-by-2 matrix in $a$-space with the form
\begin{equation}
\widetilde{\boltz}=
\begin{pmatrix}
\widetilde{\boltz}_{00} & \widetilde{\boltz}_{01}\\
\widetilde{\boltz}_{10} & \widetilde{\boltz}_{11}
\end{pmatrix}
=
\begin{pmatrix}
\Asf & i\Bsf\\
i\Bsf & \Csf
\end{pmatrix}
=d_{0}\,\iden+d_{z}\,\sigma_z+i\Bsf\,\sigma_x,
\qquad
d_{0}:=\frac{\Asf+\Csf}{2},\qquad
d_{z}:=\frac{\Asf-\Csf}{2}.
\label{eq:Mp_def_concise_nop}
\end{equation}
where $\Asf,\Bsf,\Csf\in\mathbb{R}$, and we have left labels of momentum sector implicit. In terms of the temporal domain wall basis,
\begin{equation}
\{ \ket{\gauss}:=\ket{a=0; k_b,k_c,k_d}\quad\text{(Gaussian tDW)},\qquad
\ket{\ngauss}:=\ket{a=1; k_b,k_c,k_d}\quad\text{(non-Gaussian tDW)}\,\},
\end{equation}
 the entry of the Boltzmann factor  has the following  interpretation:
\begin{itemize}
\item $\widetilde{\boltz}_{00} \equiv \Asf$: Amplitude for the Gaussian tDW to remain as Gaussian tDW,
\item $\widetilde{\boltz}_{11} \equiv  \Csf$: Amplitude for the non-Gaussian tDW to remain as non-Gaussian tDW,
\item $\widetilde{\boltz}_{01} \equiv  i\Bsf$: Amplitude for the DW type conversion with a $\pi/2$ phase.
\end{itemize}
Now introduce a unitary transformation
\begin{equation}
U:=\frac{1}{\sqrt{2}}
\begin{pmatrix}
1 & 1\\
1 & -1
\end{pmatrix},
\qquad
U\sigma_zU^\dagger=\sigma_x,
\qquad
U\sigma_xU^\dagger=\sigma_z.
\label{eq:U_def_Ttilde}
\end{equation}
Then the rotated TopSFF Boltzmann factor is the canonical $\mathcal{PT}$ dimer up to a trace shift,
\begin{equation}
U\widetilde{\boltz}U^\dagger
=
d_0\,\iden
+d_z\,\sigma_x
+i\Bsf\,\sigma_z
=
d_0\,\iden
+
\begin{pmatrix}
i\Bsf& \frac{\Asf-\Csf}{2}\\
\frac{\Asf-\Csf}{2} & -i\Bsf
\end{pmatrix}
\equiv
d_0\,\iden
+
\begin{pmatrix}
+i\gamma & J\\
J & -i\gamma
\end{pmatrix},
\label{eq:Ttilde_rotated_ordered_pm}
\end{equation}
where the coupling strength and gain/loss rates are determined respectively by
\begin{equation}
    J:=d_z=\frac{\Asf-\Csf}{2} = \frac{\widetilde{\boltz}_{00} - \widetilde{\boltz}_{11}}{2},
\qquad
\gamma:=\Bsf = -i\widetilde{\boltz}_{01} .
\end{equation}
Note that after the mapping, $U\widetilde{\boltz}U^\dagger$ is  invariant under the $\mathcal{PT}$ symmetry with complex conjugation $\mathcal{T}$ and the standard parity symmetry operator $U\mathcal{P}U^{\dagger}=\sigma_x$. The basis is now
\begin{equation}
\left\{ 
\ket{+}:=\tfrac{1}{\sqrt{2}}\left(\ket{\gauss}+ \ket{\ngauss}\right) ,
\qquad
\ket{-}:= \tfrac{1}{\sqrt{2}}\left(\ket{\gauss}- \ket{\ngauss} \right)  \right\},
\end{equation}
The eigenvalues of $\widetilde{\boltz}$ are
\( \lambda_\pm=
\frac{\Asf+\Csf}{2}
\pm \sqrt{\Delta}\) with the transition-controlling discriminant
\begin{equation}
\Delta:=J^2-\gamma^2
=\left(\frac{\Asf-\Csf}{2}\right)^2-\Bsf^2 
=  
\left(\frac{\widetilde{\boltz}_{00} - \widetilde{\boltz}_{11}}{2}\right)^2
-
\left(-i\widetilde{\boltz}_{01}\right)^2
,
\label{eq:lambda_discriminant_gauss}
\end{equation}
i.e., the \(\mathcal{PT}\) symmetry-breaking transition is controlled by whether the Gaussian/non-Gaussian tDW conversion amplitude exceeds the detuning between the Gaussian and non-Gaussian tDW propagation amplitudes.

TopSFF can be obtained by evaluating $\varphi^{T}\widetilde{\boltz}^{\,\Leff}\varphi$ with the vector $\varphi=(\varphi_1,\varphi_2)^T$. As discussed in \ref{app_sec:cayley_hamilton}, using the
Cayley-Hamilton theorem on $\widetilde{\boltz}$, we can write
\begin{equation}
\label{eq:phi_T_TL_phi_general}
\varphi^{T}\widetilde{\boltz}^{\,\Leff}\varphi
=
A(\Leff)\, \varphi^{T}\widetilde{\boltz}\varphi
+
B(\Leff)\varphi^{T}\varphi \, ,
\end{equation}
where 
\begin{equation}
\label{eq:CH_decomposition}
A(\Leff):=\frac{\lambda_{+}^{\Leff}-\lambda_{-}^{\Leff}}{\lambda_{+}-\lambda_{-}},
\qquad
B(\Leff):=\frac{\lambda_{+}\lambda_{-}^{\Leff}-\lambda_{-}\lambda_{+}^{\Leff}}{\lambda_{+}-\lambda_{-}}.
\end{equation}
and, in terms of components of $\varphi$,
\begin{equation}
\begin{aligned}
\label{eq:phiT_T_phi_components}
\varphi^{T}\widetilde{\boltz}\,\varphi
=
\Asf\,\varphi_1^{2}
+\Csf\,\varphi_2^{2}
+2i\Bsf\,\varphi_1\varphi_2,
\qquad
\varphi^{T}\varphi=\varphi_1^{2}+\varphi_2^{2}\, .
\end{aligned}
\end{equation}
For a given momentum sector $(k_b,k_c,k_d)$ where $\mathcal{PT}$ symmetry is present, the $\mathcal{PT}$ symmetry breaking transition has the following interpretation:
\begin{itemize}
\item \textbf{Unbroken $\mathcal{PT}$ phase ($\Delta>0$ and $\left| \widetilde{\boltz}_{01} \right| < \left| \frac{\widetilde{\boltz}_{00} - \widetilde{\boltz}_{11}}{2} \right| $): Polarization of Gaussian/non-Gaussian tDWs.}
In the rotated basis the right eigenvectors can be chosen as $\mathcal{PT}$ eigenstates, i.e. they are balanced across the effective gain/loss modes.
In the unrotated tDW basis, the eigenvectors are dominated by \(\ket{\gauss}\) or \(\ket{\ngauss}\), so the propagation eigenmodes do not exhibit substantial coherent conversion between Gaussian and non-Gaussian tDWs.

The eigenvalues $\lambda_{\pm}\in\mathbb{R}$ are distinct. Take the convention where $|\lambda_{+}|>|\lambda_{-}|$. From
\eqref{eq:CH_decomposition},
\begin{equation}
A(\Leff)=c\,\lambda_{+}^{\Leff-1},
\qquad
c:=\frac{1-(\lambda_{-}/\lambda_{+})^{\Leff}}{1-(\lambda_{-}/\lambda_{+})}
\approx \frac{1}{1-(\lambda_{-}/\lambda_{+})},
\end{equation}
and similarly $B(\Leff)=\mathcal{O}\left(\lambda_{+}^{\Leff-1} \right)$. Hence
$\varphi^{T}\widetilde{\boltz}^{\,\Leff}\varphi$ grows or decays exponentially without oscillations with the dominant real eigenvalue.

\item \textbf{Broken $\mathcal{PT}$ phase ($\Delta<0$ and  $\left| \widetilde{\boltz}_{01} \right| > \left| \frac{\widetilde{\boltz}_{00} - \widetilde{\boltz}_{11}}{2} \right| $): Hybridization of Gaussian/non-Gaussian tDWs.}
In the rotated basis the eigenvalues form a complex-conjugate pair and the eigenvectors are not $\mathcal{PT}$ eigenstates ($\mathcal{PT}$ exchanges them).
In the unrotated tDW basis, this polarization manifests as hybridization between
\(\ket{\gauss}\) and \(\ket{\ngauss}\): each propagation eigenmode involves nontrivial coherent hybridization of the Gaussian and non-Gaussian tDW sectors, with one mode effectively amplified and the other mode attenuated under repeated multiplication by \(\widetilde{\boltz}\).

The eigenvalues  $\lambda_{\pm}$ form a complex-conjugate pair. Write
\begin{equation}
\lambda_{\pm}=d_0\pm i\sqrt{|\Delta|}=r\,e^{\pm i\theta},
\qquad
r:=|\lambda_{\pm}|=\sqrt{d_0^2+|\Delta|},
\qquad
\theta:=\arctan\left(\frac{\sqrt{|\Delta|}}{d_0}\right).
\end{equation}
In particular,
\begin{equation}
A(\Leff)
=\frac{\lambda_{+}^{\Leff}-\lambda_{-}^{\Leff}}{\lambda_{+}-\lambda_{-}}
=
r^{\Leff-1}\,\frac{\sin(\Leff\theta)}{\sin\theta},
\end{equation}
and $B(\Leff)=\mathcal{O}(r^{\Leff-1})$. Therefore,
$\varphi^{T}\widetilde{\boltz}^{\,\Leff}\varphi$ has an envelope $r^{\Leff-1}$ with oscillations
set by $\sin(\Leff\theta)/\sin\theta$.

\item \textbf{Exceptional point ($\Delta=0$ and and $\left| \widetilde{\boltz}_{01} \right| = \left| \frac{\widetilde{\boltz}_{00} - \widetilde{\boltz}_{11}}{2} \right| $): onset of hybridization and mode coalescence.}

At $\Delta=0$ the two eigenvalues coalesce, and generically the two eigenvectors merge into a
single eigenvector in the Jordan form. This is the $\mathcal{PT}$ threshold where the tDW sector hybridization switches on non-analytically.

The eigenvalue $\lambda_{+}=\lambda_{-}=d_0$ and $\widetilde{\boltz}$ is generically non-diagonalizable. Taking the coalescing limit in \eqref{eq:CH_decomposition} yields
\begin{equation}
A(\Leff)=\Leff\, d_{0}^{\,\Leff-1},
\qquad
B(\Leff)=(1-\Leff)\,d_{0}^{\,\Leff}.
\end{equation}
Thus $\varphi^{T}\widetilde{\boltz}^{\,\Leff}\varphi$ acquires the characteristic Jordan non-diagonality, i.e., an extra factor of $\Leff$ relative to the generic cases. See \ref{app_sec:boundary_resolved_topsff} for the extraction of the linear scalling in $\Leff$ from TopSFF.
\end{itemize}

\section{TopSFF and exceptional points (EP)}\label{app:tsff_ep}

In this appendix, we analyse exceptional points (EP) of the generalized Boltzmann factor in the \(\mathcal{PT}\) symmetric momentum sectors. Section \ref{app:selected_ep} solves the exceptional-point condition in the four selected momentum sectors \((0,k,-k)\), \((0,0,0)\), \((\pi,\pi,0)\), and \((\pi,k,\pi-k)\). Section \ref{app:ep_data} presents the resulting EP data. Section \ref{sec:EP_coalescence_pipi0} presents the signatures of eigenvalue and eigenvector coalescence, including the Jordan-block contribution exposed by boundary-resolved TopSFFs.

\subsection{EP in selected momentum sectors}\label{app:selected_ep}
We study exceptional points (EPs) of the generalized Boltzmann factor in the $\mathcal{PT}$-symmetric sectors. In these sectors, the Boltzmann factor takes the form
\begin{equation}
\widetilde{\boltz}=
\begin{pmatrix}
\Asf & i\Bsf\\
i\Bsf & \Csf
\end{pmatrix},
\end{equation}
with eigenvalues
\begin{equation}
\lambda_\pm=
\frac{\Asf+\Csf}{2}
\pm \sqrt{\Delta},
\qquad
\Delta:=
\left(\frac{\Asf-\Csf}{2}\right)^2-\Bsf^2 .
\label{eq:lambda_discriminant_gauss_v2}
\end{equation}
Exceptional points occur when the discriminant vanishes, $\Delta=0$, equivalently
\begin{equation}
\Asf-\Csf=\pm 2\Bsf .
\label{app_eq:ep_condition}
\end{equation}
For each $\mathcal{PT}$-symmetric sector, we define the branch polynomials
\begin{equation}
P^\pm(\mu):=\Asf-\Csf\mp 2\Bsf,
\label{eq:Ppm_def}
\end{equation}
so that the EP condition is $P^\pm(\mu)=0$. Below we identify the corresponding values of $\mu$ for which $\Delta=0$. See Sec.~\ref{app:tsff_and_pt} for further discussion of the associated $\mathcal{PT}$ symmetry breaking transition.

\subsubsection{EP for $(k_b,k_c,k_d)=(0,k,-k)$}
For the sector $(k_b,k_c,k_d)=(0,k,-k)$ with $k\neq0$, the matrix elements are given in Eq.~\eqref{app_eq:ABC_0k}.

\textbf{$(+)$ branch: $\Asf-\Csf=+2\Bsf$.}
On this branch, the geometric sums cancel, and
\begin{equation}
P^+_{\zerobulk}(\mu)
=
\mu^{t-1}(\mu^2-1)\bigl(2t\mu^4-\mu^2+1\bigr).
\end{equation}
The factor $\mu^2-1$ gives the boundary solution $\mu=1$, while for $t\ge2$ the factor $\mu^{t-1}$ gives the boundary solution $\mu=0$. The remaining quadratic has discriminant $1-8t<0$ for all $t\ge1$, so it has no real root. Hence there is no EP with $\mu\in(0,1)$ on the $(+)$ branch.

\textbf{$(-)$ branch: $\Asf-\Csf=-2\Bsf$.}
On this branch,
\begin{equation}
\label{eq:Pstar}
\begin{aligned}
P^-_{\zerobulk}(\mu)
=&\,
8\sum_{m=0}^{t-1}\mu^{4m}
-16\sum_{m=t}^{2t-1}\mu^{2m}
-\mu^{4t-4}+2\mu^{4t-2}+(6t-1)\mu^{4t}
+2t\mu^{4t+2}.
\end{aligned}
\end{equation}
One checks that
\begin{equation}
P^-_{\zerobulk}(1)=0,
\qquad
P^-_{\zerobulk}(0)=8 \quad (t\ge2),
\qquad
\frac{\mathrm{d}}{\mathrm{d}\mu}P^-_{\zerobulk}(1)=4t>0.
\end{equation}
For $t\ge2$, $P^-_{\zerobulk}(\mu)$ is negative just to the left of $\mu=1$ and positive at $\mu=0$. Therefore, by continuity there is at least one non-trivial root in $(0,1)$. Thus the $(-)$ branch admits an interior EP. At large $t$, the interior root approaches $\mu=1$ as
\begin{equation} \label{app_eq:ep_0k_sector}
\mu_{\EP}(t)
=
1-\frac{1}{t^2}+O(t^{-3}),
\qquad t\gg1.
\end{equation}
In Fig.~\ref{fig:EP_polynomials_and_scaling}, we plot $P^-_{\zerobulk}$ whose roots determine $\mu_{\EP}(t)$.

\begin{figure}[ht]
\centering
\includegraphics[width=0.32\textwidth]{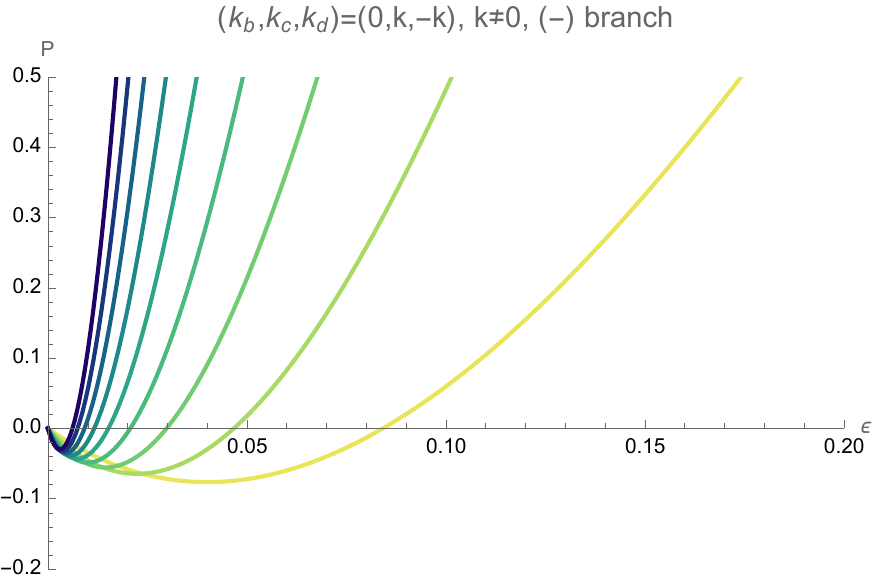}
\hfill
\includegraphics[width=0.32\textwidth]{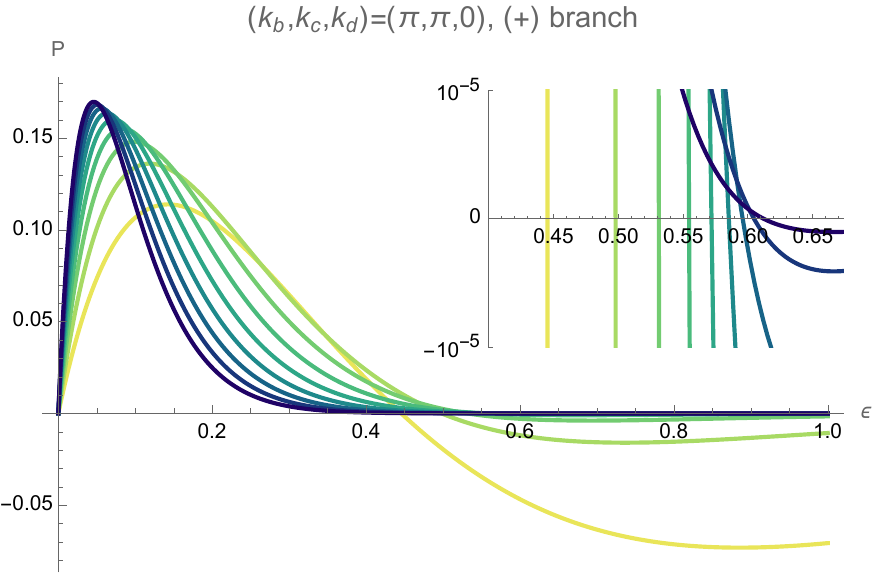}
\hfill
\includegraphics[width=0.32\textwidth]{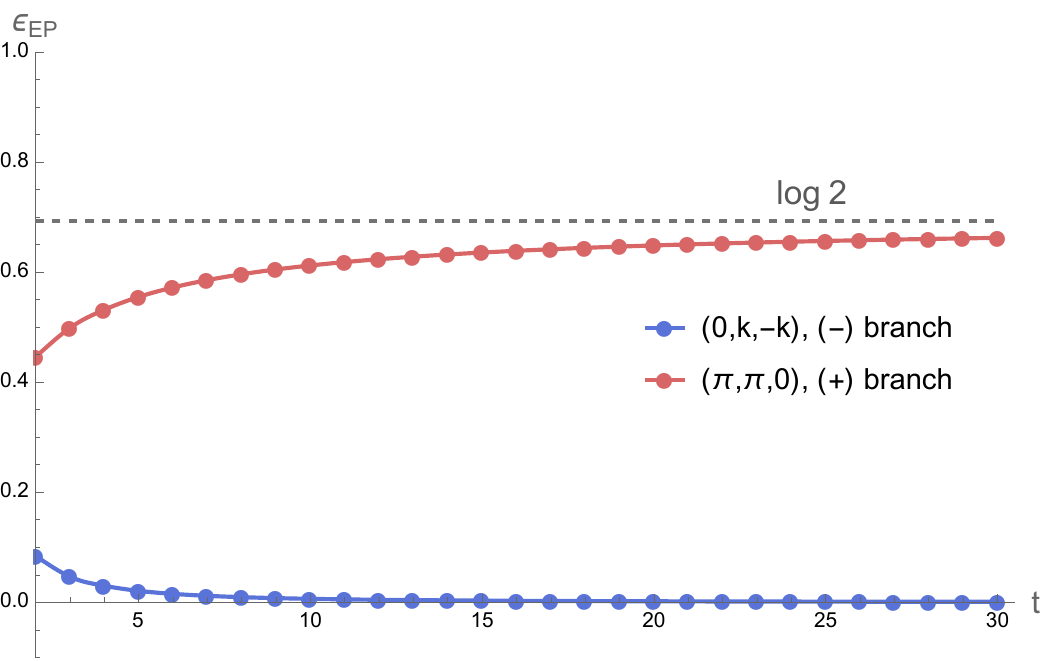}
\caption{
\textbf{Exceptional points in the $(0,k,-k)$ and $(\pi,\pi,0)$ sectors at large $q$.}
Left and middle panels show $P^-_{\zerobulk}(\epsilon)$ and $P^+_{\piedge}(\epsilon)$, respectively, as functions of $\epsilon\equiv -2\ln\mu$, with $t=2,\ldots,10$ shown from yellow to blue. For each $t\ge2$, a non-trivial root with $\mu\in(0,1)$ signals an exceptional point. The inset in the middle panel enlarges the root region for the $(\pi,\pi,0)$ sector. Right panel: exceptional-point locations versus $t$ for the $(\pi,\pi,0)$ sector and the $(0,k,-k)$ sector. As $t$ increases, the exceptional point in the $(\pi,\pi,0)$ sector approaches $\epsilon_{\EP}=\log 2$, see \eqref{eq:muEP_closed}.
}
\label{fig:EP_polynomials_and_scaling}
\end{figure}

\subsubsection{EP for $(k_b,k_c,k_d)=(0,0,0)$}
For the sector $(k_b,k_c,k_d)=(0,0,0)$, the matrix elements are given in Eq.~\eqref{app_eq:ABC_000}.

\textbf{$(+)$ branch: $\Asf_0-\Csf_0=+2\Bsf_0$.}
On this branch, the geometric sums cancel, and
\begin{equation}
P^+_{\zeroedge}(\mu)
=
\mu^{t-1}(\mu^2-1)\Bigl[2t(t-1)\mu^4+(2t-1)\mu^2+1\Bigr].
\end{equation}
The factor $\mu^2-1$ gives the boundary solution $\mu=1$, while for $t\ge2$ the factor $\mu^{t-1}$ gives the boundary solution $\mu=0$. The quadratic in $\mu^2$ has discriminant $(2t-1)^2-8t(t-1)=1+4t(1-t)\le 1$, which is negative for all $t\ge2$, so there is no real root. Hence there is no EP with $\mu\in(0,1)$ on the $(+)$ branch.

\textbf{$(-)$ branch: $\Asf_0-\Csf_0=-2\Bsf_0$.}
The EP condition is 
$P^-_{\zeroedge}(\mu)=0$,
with
\begin{equation}
\label{eq:P0}
\begin{aligned}
P^-_{\zeroedge}(\mu)
={}&
8\sum_{m=0}^{t-1}\mu^{2m}
+16(t-1)\sum_{m=t}^{2t-1}\mu^{m}
+\bigl(8t(t-1)^2-2t^2+4t-1\bigr)\mu^{2t}
-\mu^{2t-2}-2(t-1)\mu^{2t-1}
+2t(t-1)\mu^{2t+3}.
\end{aligned}
\end{equation}
This can be rearranged as
\begin{equation}
\begin{aligned}
P^-_{\zeroedge}(\mu)
={}&
8\sum_{m=0}^{t-2}\mu^{2m}
+7\mu^{2t-2}
+16(t-1)\sum_{m=t}^{2t-2}\mu^{m}
+14(t-1)\mu^{2t-1}
\\
&\qquad
+\bigl(8t(t-1)^2-2t^2+4t-1\bigr)\mu^{2t}
+2t(t-1)\mu^{2t+3}.
\end{aligned}
\end{equation}
All coefficients are non-negative, and in fact strictly positive for $t\ge2$. Therefore $P^-_{\zeroedge}(\mu)>0$ for $\mu\in(0,1)$, and the $(-)$ branch has no interior solution.

\subsubsection{EP for $(k_b,k_c,k_d)=(\pi,\pi,0)$}
For the sector $(k_b,k_c,k_d)=(\pi,\pi,0)$, the matrix elements are given in Eq.~\eqref{app_eq:ABC_pipi0}.

\textbf{$(+)$ branch: $\Asf-\Csf=+2\Bsf$.}
The $(+)$ branch simplifies to
\begin{equation}
P^+_{\piedge}(\mu)
=
\mu^{4t-4}(1-\mu^2)\Bigl(2t\mu^4-(t-1)\mu^2-1\Bigr).
\end{equation}
Besides the boundary solutions $\mu=0$ (for $t\ge2$) and $\mu=1$, the only interior EP comes from the quadratic equation
\begin{equation}\label{app_eq:pipi0_mu_pt}
2t\mu^4-(t-1)\mu^2-1=0.
\end{equation}
Solving for $\mu^2$ gives
\begin{equation}
\mu_{\EP}(t)
=
\sqrt{\frac{(t-1)+\sqrt{(t-1)^2+8t}}{4t}}
\;\xrightarrow{t \gg 1}\;
\frac{1}{\sqrt{2}}
\, .
\label{eq:muEP_closed}
\end{equation}
In terms of $\epsilon$ with $\mu = e^{-\epsilon /2}$, we have
\begin{equation}
\epsilon_{\EP}(t)
=
\log\!\left[
\frac{\sqrt{(t-1)^2+8t}-(t-1)}{2}
\right]
\;\xrightarrow{t\gg 1}\;
\log 2
\, .
\label{eq:epsilonEP_closed}
\end{equation}
See Fig.~\ref{fig:EP_polynomials_and_scaling} for the locations of the EP for $t \leq 30$.
We are particularly interested in the case of $t=2$ at large $q$, where $\mu_{\EP}(t=2) = 0.8002$ and $\epsilon_{\EP}(t=2) = -2 \log \mu_{\EP}(t=2) = 0.4457$, since we can compare these predictions with our exact computation at finite $q$.

\textbf{$(-)$ branch: $\Asf-\Csf=-2\Bsf$.}
Using Eqs.~\eqref{app_eq:ABC_pipi0} and \eqref{app_eq:geom_sum_abcd}, the branch polynomial
\begin{equation}
\begin{aligned}
P^-_{\piedge}(\mu)
={}&
8\sum_{m=0}^{t-1}\mu^{4m}
+8(t-2)\sum_{m=t}^{2t-3}\mu^{2m}
+(8t-17)\mu^{4t-4}
+(7t-14)\mu^{4t-2}
\\
&\qquad
-\bigl(8t^2-11t+1\bigr)\mu^{4t}
-2t\,\mu^{4t+2}.
\end{aligned}
\label{eq:Pminus_piedge_poly}
\end{equation}
It can be shown that $P^-_{\piedge}(\mu)$ is positive for $\mu\in(0,1)$, and hence this branch gives no additional root in $\mu\in(0,1)$.

\subsubsection{EP for $(k_b,k_c,k_d)=(\pi,k,\pi-k)$}
For the sector $(k_b,k_c,k_d)=(\pi,k,\pi-k)$ with $k\neq0,\pi$, the matrix elements are given in Eq.~\eqref{app_eq:ABC_pibulk}.

\textbf{$(+)$ branch: $\Asf-\Csf=+2\Bsf$.}
Using Eq.~\eqref{app_eq:ABC_pibulk}, one finds
\begin{equation}
P^+_{\pibulk}(\mu)
=
-\mu^{4t-4}(1-\mu^2)^2.
\end{equation}
Hence the $(+)$ branch has only the boundary solution $\mu=1$, and no EP in the interior $\mu\in(0,1)$.

\textbf{$(-)$ branch: $\Asf-\Csf=-2\Bsf$.}
On this branch, using the definitions of $A_{t-1}$ and $C_t$ in \eqref{app_eq:geom_sum_abcd}, we obtain
\begin{equation}
P^-_{\pibulk}(\mu)
=
7\mu^{4t-4}(1-\mu^2)^2
+8\sum_{r=2}^{t-1}\mu^{4t-4r}(1-\mu^{2r})^2
+8(1-\mu^{2t})^2,
\label{eq:Ppibulk_positive}
\end{equation}
which has positive terms, leading to $P^-_{\pibulk}(\mu)>0$ for $\mu\in(0,1)$. Therefore, the $(-)$ branch has no interior solution. Hence the $(\pi,k,\pi-k)$ sector has no exceptional point for $\mu\in(0,1)$.

\subsubsection{EP for $(k_b,k_c,k_d)=(2\pi/t, - 2\pi/t, 0)$}
For odd \(t\geq 3\), we focus on the momentum sector
\(
k
=
(k_b,k_c,k_d)=(2\pi/t, - 2\pi/t, 0)
\),
which has a leading eigenvalue with moduli exceeding unity, giving the exponential growth with system size discussed below, and an EP at finite interaction strength in large \(t\).  

In this momentum sector, \(n_b=2\), so that \((-1)^{n_b}=1\), while
\(S_t(k_c)=0\),  \(S_t(k_d)=t\), and 
\(k_b+k_c=k_b+k_c+k_d=0 \pmod{2\pi}\).
Moreover, all terms depending on \(k_b+k_d\) in the Boltzmann factor are multiplied by \(S_t(k_c)\) and therefore
vanish.  These are precisely the same reduced momentum-space
conditions as in the even \(t\) sector
\(\pi,\pi,0)\).  Consequently, the
generalized Boltzmann factor in the odd-\(t\) sector is given by
the same analytic expressions as Eq.~\eqref{app_eq:ABC_pipi0}. The EP condition therefore factorizes as
\begin{equation}
P_{\mathrm{o}}^{+}(\mu)
=
\mu^{4t-4}(1-\mu^2)
\left[
2t\mu^4-(t-1)\mu^2-1
\right].
\end{equation}
Besides the boundary solutions \(\mu=0,1\), the interior
exceptional point obeys
\begin{equation}
2t\mu_{\mathrm{EP}}^4
-(t-1)\mu_{\mathrm{EP}}^2-1=0,
\end{equation}
and hence
\begin{equation}
\mu_{\EP}(t)
=
\sqrt{\frac{(t-1)+\sqrt{(t-1)^2+8t}}{4t}}
\;\xrightarrow{t \gg 1}\;
\frac{1}{\sqrt{2}}
\, .
\end{equation}
Using \(\mu=e^{-\epsilon/2}\), this gives
\begin{equation}
\epsilon_{\mathrm{EP}}(t)
=
\log\!\left[
\frac{
\sqrt{(t-1)^2+8t}-(t-1)
}{2}
\right]
\;\xrightarrow{t\gg 1}\; \log 2.
\end{equation}
Thus, the exceptional-point location has the same functional
form for even and odd \(t\), but occurs respectively in the
momentum sectors \((\pi,\pi,0)\) for even \(t\) and \(\left(2\pi/t,-2\pi/t,0\right)\) for odd \(t\).

\subsection{EP data}\label{app:ep_data}
Here we provide plots of the complex eigenvalues of the generalized Boltzmann factor against the interaction strength in Fig.~\ref{fig:grid_4_sectors} and Fig.~\ref{fig:grid_all_sectors}, demonstrating $\mathcal{PT}$ symmetry breaking transition in certain momentum sectors. 
In Fig.~\ref{fig:grid_4_sectors}, we plot the complex eigenvalues of the generalized Boltzmann factor of $(k_b,k_c,k_d) \in \{(0,0,0),(0,2\pi m/t,-2\pi m/t ), (\pi,\pi,0), (\pi, 2\pi m/t,\pi -2\pi m/t ) \}$ with $m=1$. Among these, $(\pi,\pi,0)$ sector exhibits a particularly prominent $\mathcal{PT}$ transition, for which $|\lambda_\pm|$ exceeds $1$ over a finite range of $\mu$. Importantly, the EP of $(\pi,\pi,0)$ persists at finite $\mu$, see \eqref{eq:muEP_closed}.
In Fig.~\ref{fig:grid_all_sectors}, we plot complex eigenvalues of Boltzmann factor in all momentum sectors $(k_b,k_c,k_d) =( 2\pi n_b/2t ,2\pi n_c/t,2\pi n_d/t )$ with $n_b=[1,2t]$ and $n_c, n_d \in [1,t]$. $\Re \lambda_\pm$. For $t=2$, $\mathcal{PT}$ symmetry breaking transitions occur in $(k_b,k_c,k_d)\in \{(0,\pi,\pi), (\pi/2,0,0), (\pi/2,0,\pi), (\pi/2,\pi,0), (\pi,0,\pi), (\pi,\pi,0), (3\pi/2, 0, 0),  (3\pi/2, 0, \pi),  (3\pi/2, \pi, 0) \}$.
    Again, the transitions in the $(\pi,0,\pi)$ and $(\pi,\pi,0)$ sectors (related by the exchange symmetry between coordinates $c$ and $d$) are the most prominent, since $|\lambda_\pm|$ can exceed $1$, which is also generally the case for $t>2$.

\begin{figure}[ht]
    \centering
    \includegraphics[width=1\textwidth]{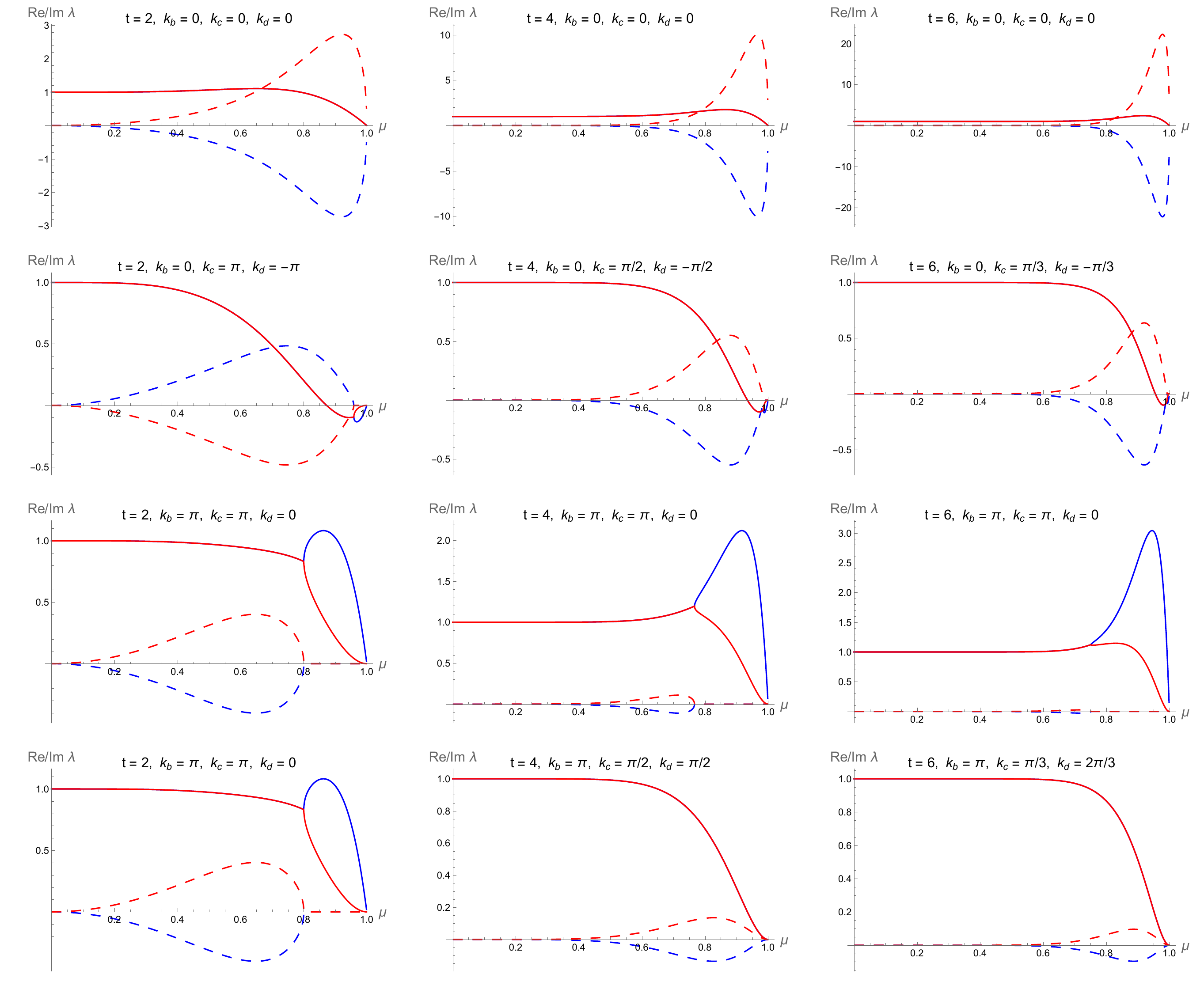}
    \caption{\textbf{Eigenvalues of Boltzmann factor in 4 momentum sectors for $t=2,4,6$ at large $q$.}
    Complex eigenvalues of the generalized Boltzmann factor of $(k_b,k_c,k_d) \in \{(0,0,0),(0,2\pi m/t,-2\pi m/t ), (\pi,\pi,0), (\pi, 2\pi m/t,\pi -2\pi m/t ) \}$ with $m=1$. $\Re \lambda_\pm$ are plotted in red and blue solid lines. $\Im \lambda_{\pm}$ are plotted in red and blue dashed lines. Exceptional points are shown in the sectors $(\pi,\pi,0)$ and $(0,2\pi m/t,-2\pi m/t )$.
    Among these, $(\pi,\pi,0)$ sector exhibits a particularly prominent $\mathcal{PT}$ transitions, for which $|\lambda_\pm|$ exceeds $1$ over a finite range of $\mu$. Note that the EP of $(\pi,\pi,0)$ persists at finite $\mu$ (equivalently $\epsilon = -2\ln \mu$), see \eqref{eq:muEP_closed}. 
    }
    \label{fig:grid_4_sectors}
\end{figure}

\begin{figure}[ht]
    \centering
    \includegraphics[width=1\textwidth]{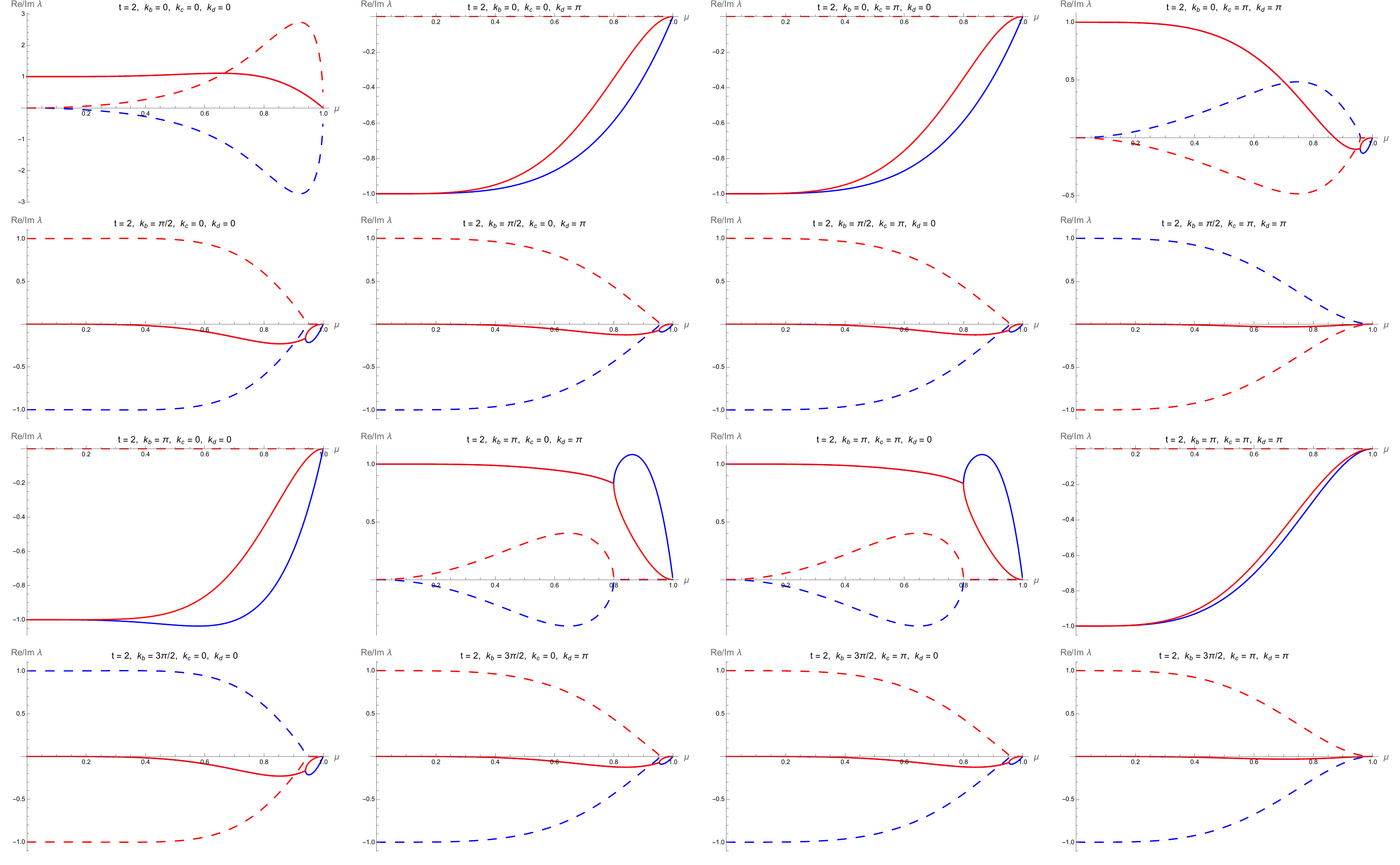}
    \caption{\textbf{Eigenvalues of Boltzmann factor in all momentum sectors for $t=2$  at large $q$.}
    Complex eigenvalues of Boltzmann factor of $(k_b,k_c,k_d) =( 2\pi n_b/2t ,2\pi n_c/t,2\pi n_d/t )$ for $n_b=[1,2t]$ and $n_c, n_d \in [1,t]$. $\Re \lambda_\pm$ are plotted in red and blue solid lines. 
    Note that the spectra are invariant under both the exchange of the two momentum labels
\(
(k_b,k_c,k_d)\sim (k_b,k_d,k_c)
\) (equivalence between second and third columns)
and the simultaneous momentum inversion
\(
(k_b,k_c,k_d)\sim (-k_b,-k_c,-k_d)
\) (equivalence between second and forth rows).
    $\Im \lambda_{\pm}$ are plotted in red and blue dashed lines. For $t=2$, real-to-complex spectral bifurcations  occur in $(k_b,k_c,k_d)\in \{(0,\pi,\pi), (\pi/2,0,0), (\pi/2,0,\pi), (\pi/2,\pi,0), (\pi,0,\pi), (\pi,\pi,0), (3\pi/2, 0, 0),  (3\pi/2, 0, \pi),  (3\pi/2, \pi, 0) \}$. Note that only sectors satisfying $k_x+k_y + k_z \equiv 0 \pmod{\pi}$ have $\mathcal{PT}$ symmetry. Among these, the bifurcation in the $(\pi,0,\pi)$ and $(\pi,\pi,0)$ sectors are the most prominent, since $|\lambda_\pm|$ can exceed $1$, which is also generally the case for $t>2$, see Fig.~\ref{fig:grid_4_sectors}.
    }
    \label{fig:grid_all_sectors}
\end{figure}

\begin{figure}[ht]
    \centering
    \includegraphics[width=1\textwidth]{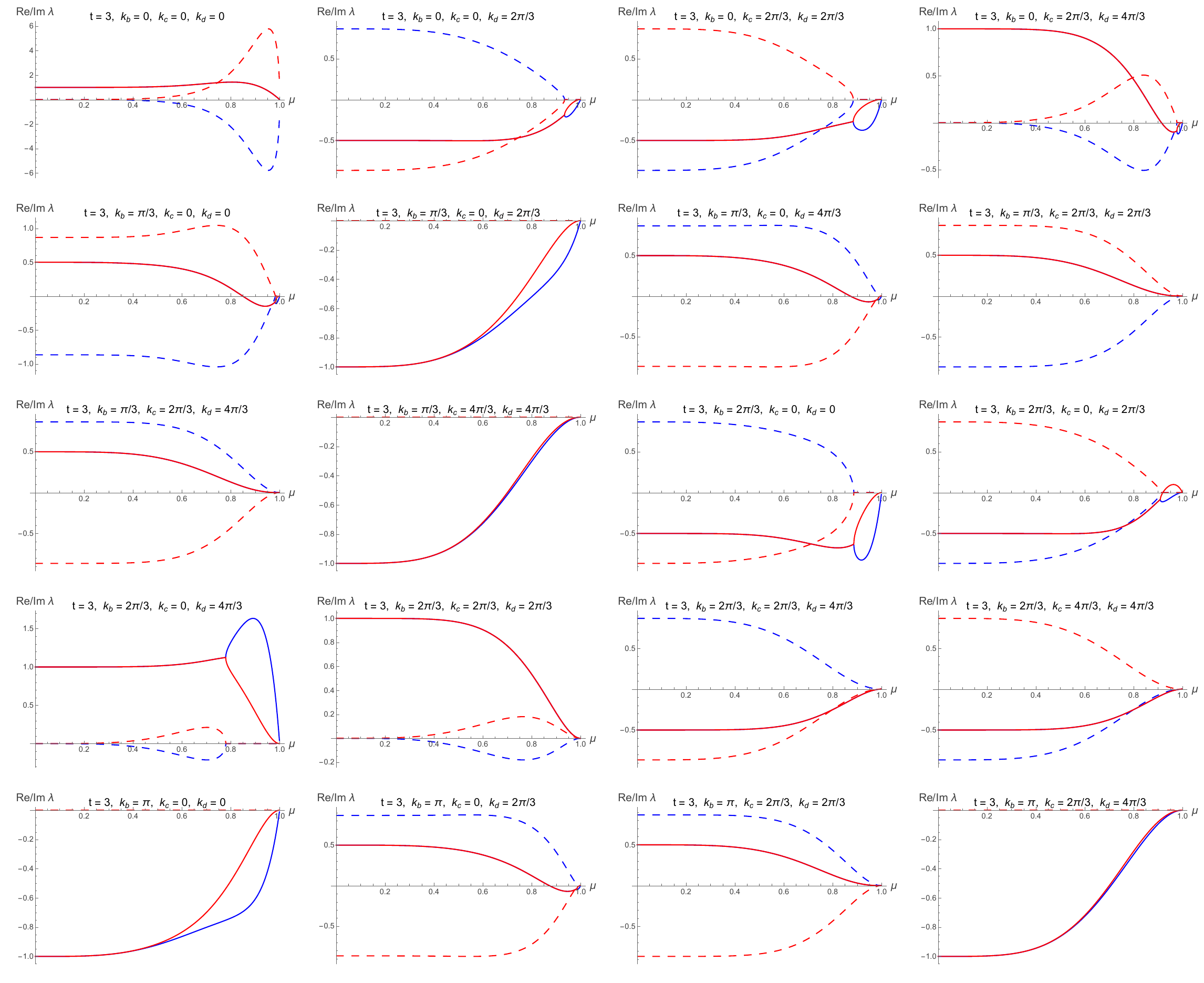}
    \caption{\textbf{Eigenvalues of Boltzmann factor in all momentum sectors for $t=3$  at large $q$.}
    Complex eigenvalues of Boltzmann factor of $(k_b,k_c,k_d) =( 2\pi n_b/2t ,2\pi n_c/t,2\pi n_d/t )$ for $n_b=[1,2t]$ and $n_c, n_d \in [n_c,t]$. $\Re \lambda_\pm$ are plotted in red and blue solid lines. $\Im \lambda_{\pm}$ are plotted in red and blue dashed lines. 
    Note that the spectra are invariant under both the exchange of the two momentum labels
\(
(k_b,k_c,k_d)\sim (k_b,k_d,k_c),
\)
and the simultaneous momentum inversion
\(
(k_b,k_c,k_d)\sim (-k_b,-k_c,-k_d)
\).
    For $t=3$, there are many $\mathcal{PT}$ symmetry breaking transitions, notably in the $(2\pi/3,0,4\pi/3)$, $(2\pi/3,4\pi/3,0)$, $(4\pi/3,0,2\pi/3)$, and $(4\pi/3,2\pi/3,0)$, with $|\lambda_\pm|$  exceeding $1$.
    }
    \label{fig:grid_all_sectors_t3}
\end{figure}

\subsection{Eigenvalue and eigenvector coalescence at EP} \label{sec:EP_coalescence_pipi0}

\begin{figure}
    \centering
    \includegraphics[width=0.33\textwidth]{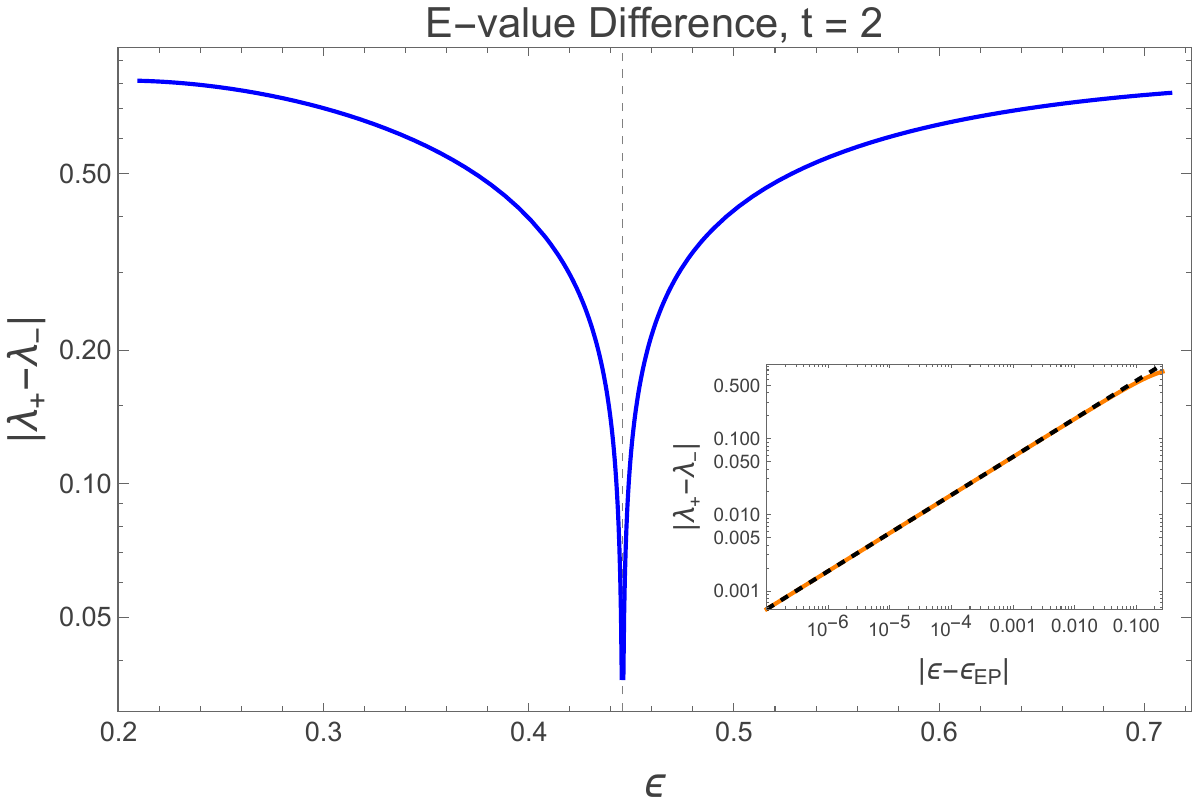}
    \includegraphics[width=0.34\textwidth]{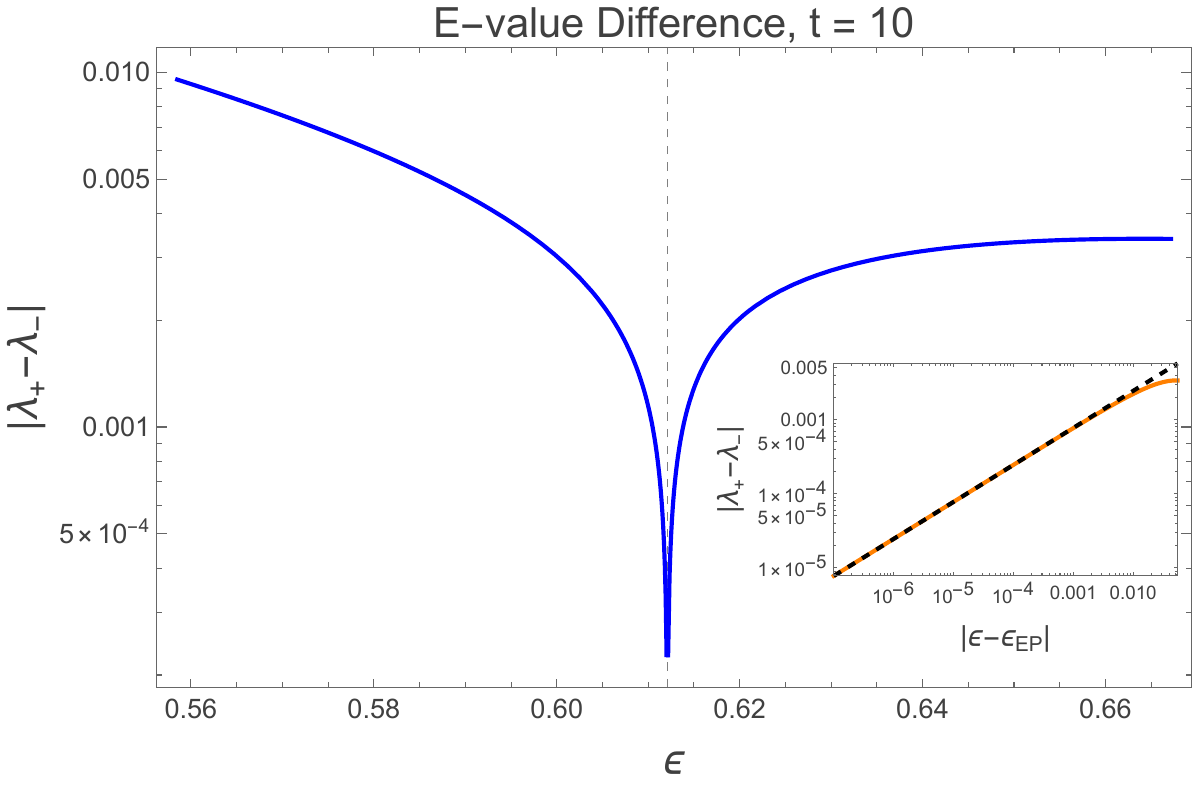}
    \caption{
    \textbf{Square-root eigenvalue splitting at the EP at large $q$.}
    For the exact large-$q$ Boltzmann factor of TopSFF in the $(\pi,\pi,0)$ sector, we plot the eigenvalue splitting
    $|\lambda_{+}-\lambda_{-}|$ as a function of the interaction parameter $\epsilon$ for $t=2$ (left) and $t=10$ (right).
    The vertical dashed line marks the exceptional point
    $\epsilon_{\EP}(t)=-2\log \mu_{\EP}(t)$, with $\mu_{\EP}(t)$ given in Eq.~\eqref{eq:muEP_closed}.
    At this point the two eigenvalues coalesce, so that $|\lambda_{+}-\lambda_{-}|$ vanishes.
    The insets show the same data as a function of $|\epsilon-\epsilon_{\EP}|$ on log-log axes.
    The dashed black guide highlights the characteristic EP scaling
    $|\lambda_{+}-\lambda_{-}| \propto |\epsilon-\epsilon_{\EP}|^{1/2}$ as given in \eqref{app_eq:eval_diff}.
    }
    \label{fig:pipi0_evalue_diff_largeq}
\end{figure}

\begin{figure}
    \centering
    \includegraphics[width=0.33\textwidth]{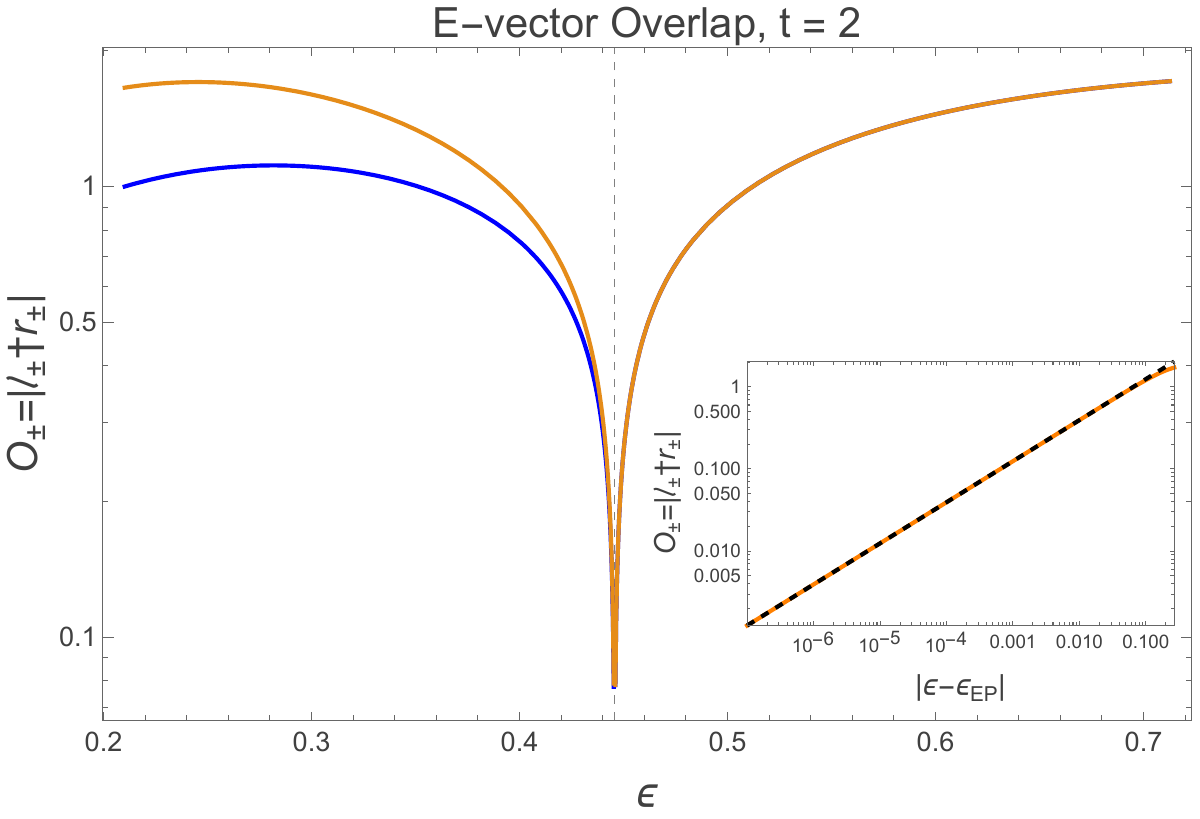}
    \includegraphics[width=0.342\textwidth]{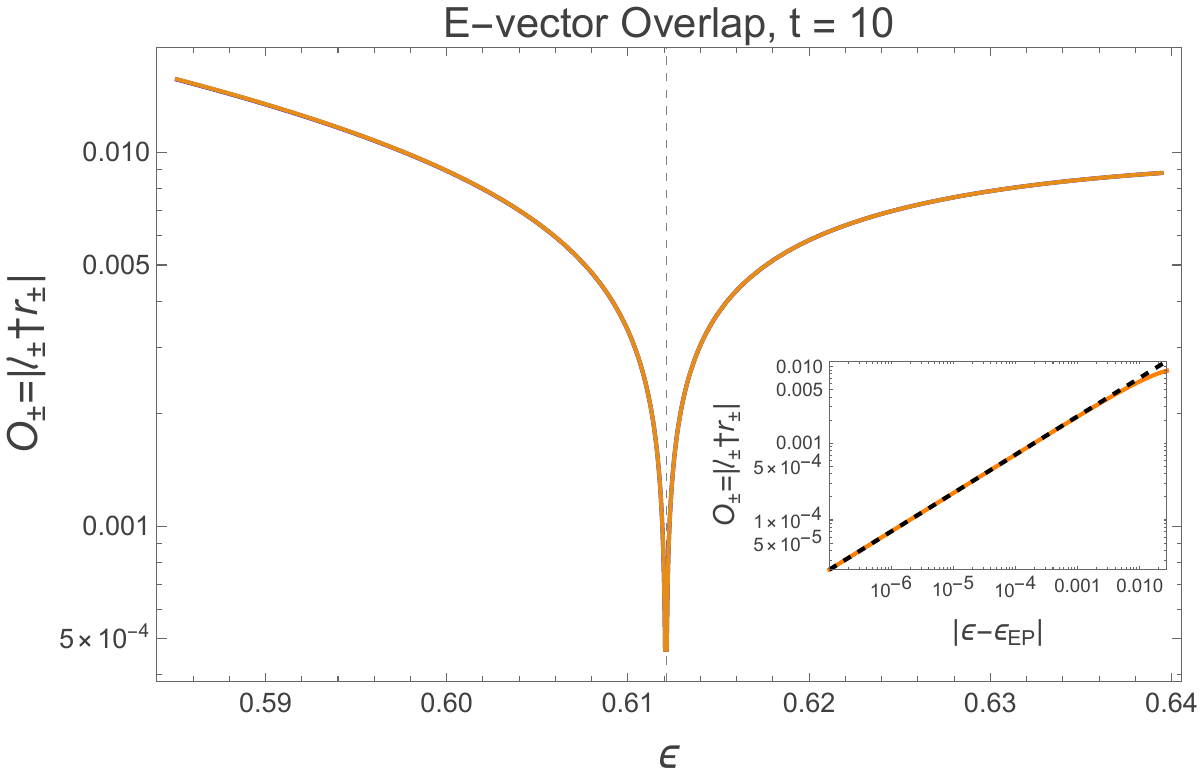}
    \caption{
    \textbf{Eigenvector self-orthogonality at the EP at large $q$.}
    We plot the left-right eigenvector overlap
    $O_{\pm}=|\ell_{\pm}^{\dagger} r_{\pm}|$
    for the two eigenmodes of the exact large-$q$ Boltzmann factor of TopSFF in the $(\pi,\pi,0)$ sector, for $t=2$ (left) and $t=10$ (right).
    The vertical dashed line marks
    $\epsilon_{\EP}(t)=-2\log \mu_{\EP}(t)$.
    The overlap collapses to zero at the exceptional point, showing that the coalescing eigenmode becomes self-orthogonal in the biorthogonal sense.
    The insets plot the approach to the exceptional point on log-log axes and show the expected scaling
    $O_{\pm}\propto |\epsilon-\epsilon_{\EP}|^{1/2}$ as given in \eqref{eq:overlap_sqrt_scaling}.
    }
    \label{fig:pipi0_evector_overlap_largeq}
\end{figure}

\begin{figure}
    \centering
    \includegraphics[width=0.33\textwidth]{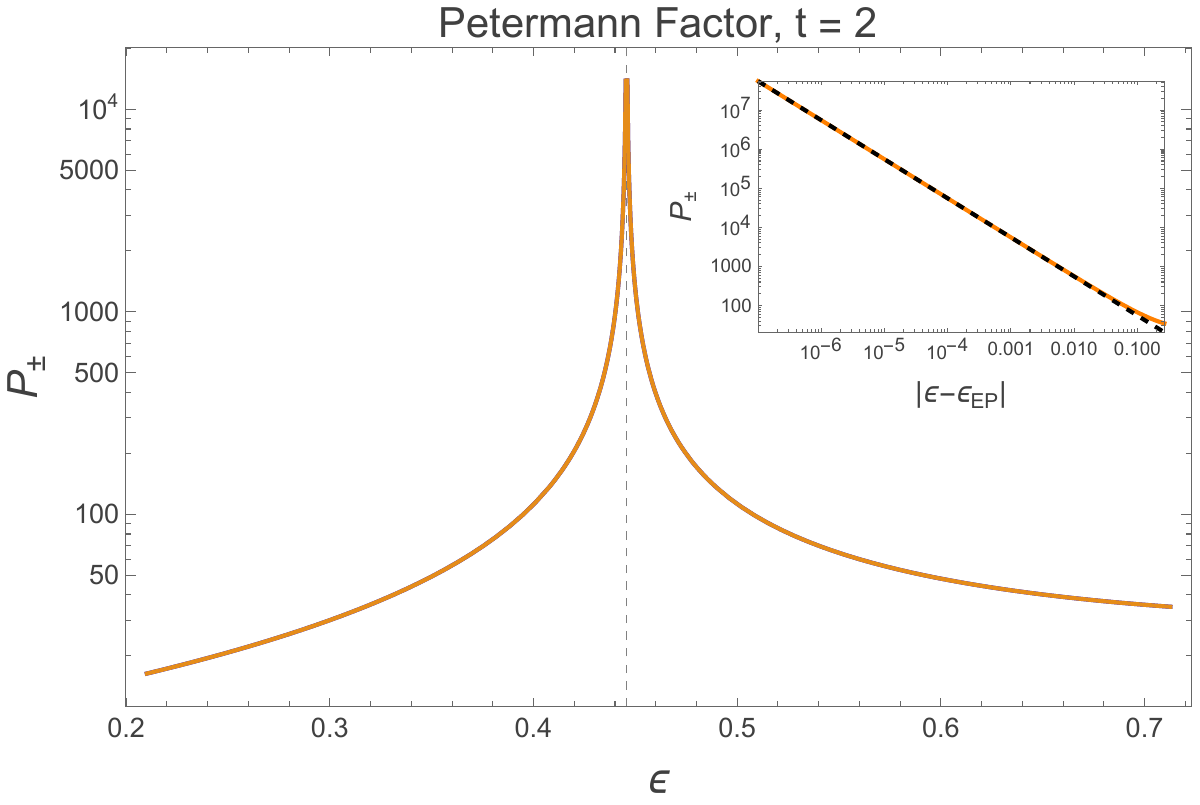}
    \includegraphics[width=0.323\textwidth]{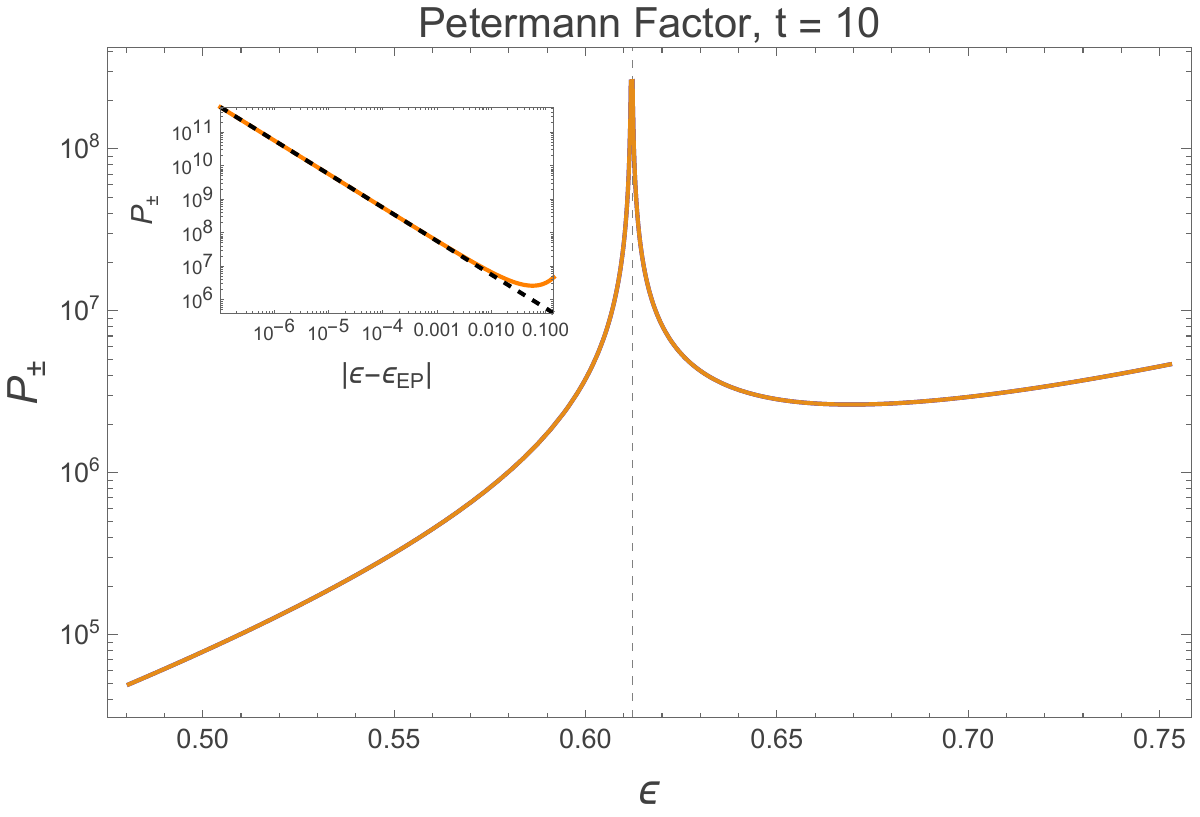}
    \caption{
    \textbf{Divergent Petermann factor at the EP at large $q$.}
    We plot the Petermann factor $P_{\pm}$ of the two eigenmodes of the exact large-$q$ Boltzmann factor of TopSFF in the $(\pi,\pi,0)$ sector, for $t=2$ (left) and $t=10$ (right).
    The vertical dashed line marks
    $\epsilon_{\EP}(t)=-2\log \mu_{\EP}(t)$.
    As the exceptional point is approached, the biorthogonal overlap vanishes and the Petermann factor diverges.
    The insets show the same data against $|\epsilon-\epsilon_{\EP}|$ on log-log axes, with the dashed black guide indicating the EP divergence as described by
    $P_{\pm}\propto |\epsilon-\epsilon_{\EP}|^{-1}$ in \eqref{eq:Petermann_EP_scaling}.
    }
    \label{fig:pipi0_petermann_largeq}
\end{figure}

\begin{figure}
    \centering
    \includegraphics[width=0.33\textwidth]{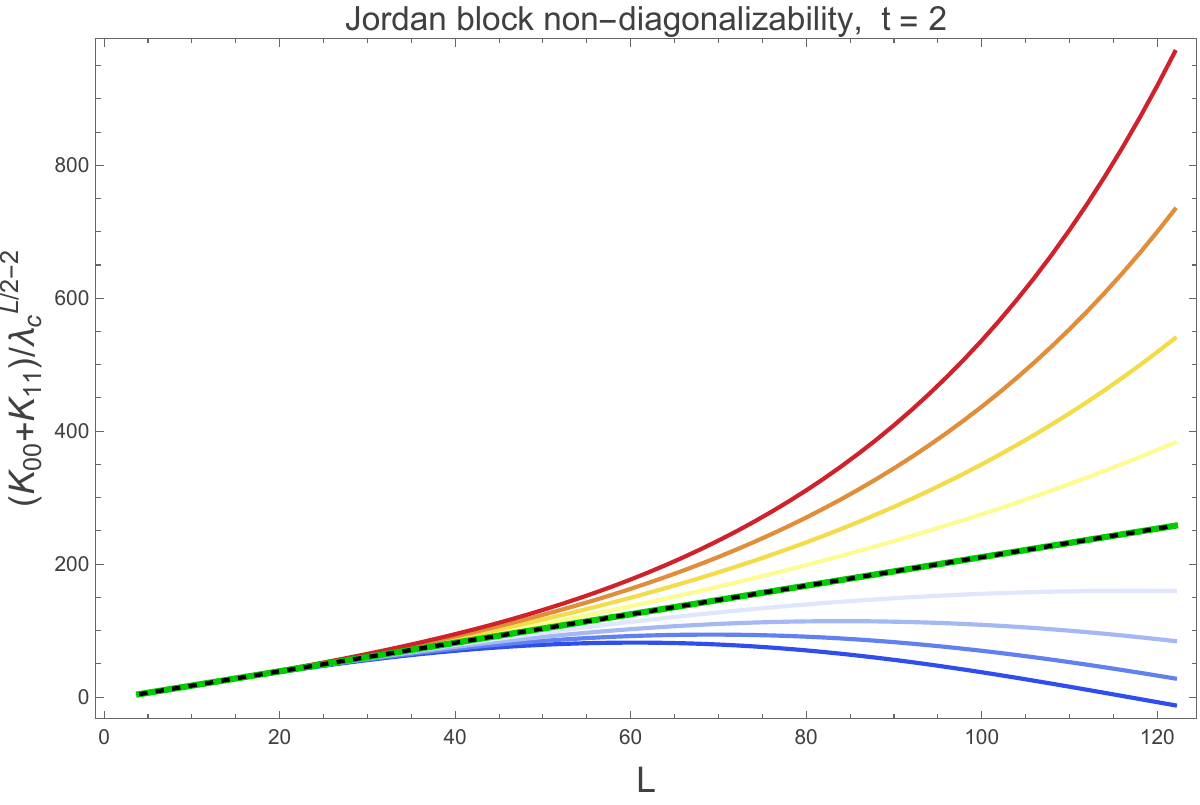}
    \includegraphics[width=0.334\textwidth]{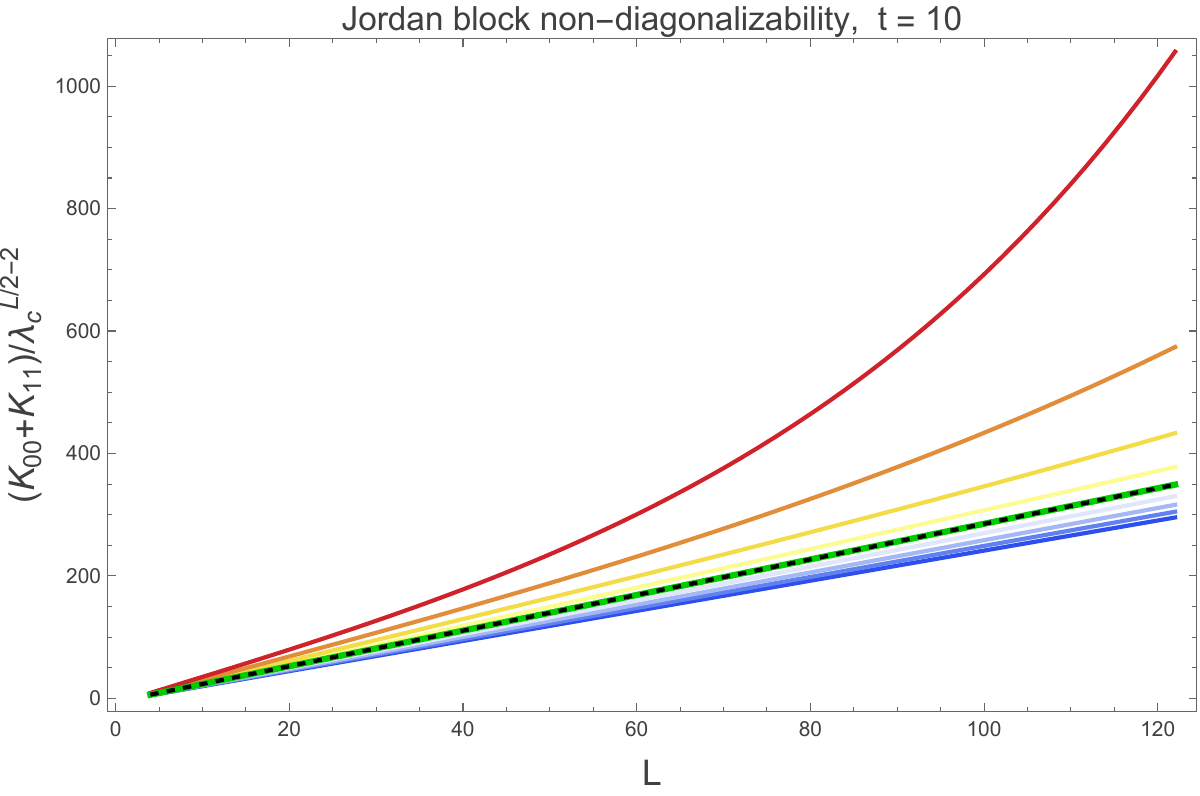}
    \caption{
    \textbf{Jordan block non-diagonality of TopSFF at the EP at large $q$.}
    We plot the boundary-resolved TopSFF
       $(\Ksf_{00}+\Ksf_{11})/\lambda_{\mathrm{c}}^{L/2-2}$
    in the \((k_b,k_c,k_d)=(\pi,\pi,0)\) sector as a function of the physical system size \(L\) for \(t=2\) (left) and \(t=10\) (right), where \(\lambda_c(\mu):=\Tr\widetilde{\mathcal B}/2 = (\lambda_+ + \lambda_-)/2\) reduces to the coalesced eigenvalue \(\lambda_{\EP}\) at \(\epsilon=\epsilon_{\EP}\). The curves correspond to different values of \(\epsilon \in [\epsilon_{\EP}-\delta\epsilon,\epsilon_{\EP}+\delta\epsilon]\), with \(\delta\epsilon=10^{-3}\) for \(t=2\) and \(\delta\epsilon=0.1\) for \(t=10\). 
    The TopSFF at EP (green line) agree with a  guideline linear in \(\Leff\) and \(L\) (dashed black line), as predicted in Eq.~\eqref{app_eq:tsff_jordan}.  
    At the EP, the two eigenvalues of the generalized Boltzmann factor coalesce and the matrix becomes non-diagonalizable, producing the Jordan-block scaling
        \(\Ksf_{00}+\Ksf_{11}
        \propto
        L_{\rm eff}\lambda^{L_{\rm eff}-1}\), resulting in a linear growth with \(L\) after dividing the TopSFF by the exponential envelope.   Away from \(\mu_{\EP}\), the eigenvalues split, the matrix is diagonalizable, and the curves depart from the linear EP form. 
    }
    \label{fig:pipi0_jordan_nondiag_largeq}
\end{figure}

In this section, we discuss the signatures of exceptional points (EPs) of the $\mathcal{PT}$-symmetric generalized Boltzmann factor, including eigenvalue coalescence and eigenvector coalescence, as diagnosed by eigenvector self-orthogonality and the divergence of the Petermann factor. We focus on the sector $(k_b,k_c,k_d)=(\pi,\pi,0)$, although the discussion applies more generally to any $\mathcal{PT}$-symmetric momentum sector.
In the momentum sector $(k_b,k_c,k_d)=(\pi,\pi,0)$, the generalized Boltzmann factor is $\mathcal{PT}$ symmetric and takes the form
\begin{equation}
\widetilde{\mathcal B}_{\pi_{\mathrm e}}
=
\begin{pmatrix}
\mathsf A_{\pi_{\mathrm e}} & i\,\mathsf B_{\pi_{\mathrm e}}\\
i\,\mathsf B_{\pi_{\mathrm e}} & \mathsf C_{\pi_{\mathrm e}}
\end{pmatrix},
\end{equation}
where the matrix entries are defined in Eq.~\eqref{app_eq:ABC_pipi0}, and we use the shorthand \(\piedge \equiv (k_b,k_c,k_d)=(\pi,\pi,0)\) for this momentum sector.
At an exceptional point, the two eigenvalues and their corresponding eigenvectors coalesce. This is a second order exceptional point (EP2), since the algebraic multiplicity is 2, while the geometric multiplicity is 1. In what follows, we examine both eigenvalue coalescence and eigenvector coalescence in the vicinity of the EP.

\begin{itemize}
    \item \textbf{Square root scaling of eigenvalue coalescence.} 
    The eigenvalues are 
\begin{equation}
\lambda_\pm
=
\frac{\Tr \Ttilde_{\piedge} \pm \sqrt{\Delta_{\pi_{\mathrm e}}}}{2}, 
\qquad 
\Delta_{\pi_{\mathrm e}}
:=
\Tr \Ttilde_{\piedge}^2 - 4\det \Ttilde_{\piedge}
=
(\mathsf A_{\pi_{\mathrm e}}-\mathsf C_{\pi_{\mathrm e}})^2-4\mathsf B_{\pi_{\mathrm e}}^2.
\end{equation}
At the exceptional point, the discriminant vanishes so that the two eigenvalues coalesce,
\begin{equation}
\Delta_{\pi_{\mathrm e}}(\mu_{\EP})=0 \quad 
\Longrightarrow  \quad 
\lambda_+(\mu_{\EP})=\lambda_-(\mu_{\EP})=\frac{1}{2}\Tr \Ttilde_{\piedge}(\mu_{\EP}).
\end{equation}
Expanding near the exceptional point, we have generically
\begin{equation}\label{app_eq:lambda_expand}
\lambda_\pm(\mu)
=
\lambda_{\EP}
\pm
\frac{1}{2}\sqrt{\Delta'_{\pi_{\mathrm e}}(\mu_{\EP})}\,\sqrt{\mu-\mu_{\EP}}
+
O(\mu-\mu_{\EP}), \qquad 
\Delta'_{\pi_{\mathrm e}}(\mu_{\EP}):= \frac{\mathrm d \Delta_{\pi_{\mathrm e}}}{\mathrm d \mu}\neq 0,
\end{equation}
Thus, the eigenvalue splitting, i.e., the difference between the two eigenvalues, exhibits the characteristic square-root scaling near the EP,
\begin{equation}\label{app_eq:eval_diff}
\lambda_+ - \lambda_-  \propto \sqrt{\Delta(\mu)} \propto \sqrt{|\mu-\mu_{\EP}|}  \propto
    \sqrt{|\epsilon-\epsilon_{\EP}|}.
\end{equation}
Note that this square-root eigenvalue splitting occurs on both sides of the EP. In the unbroken $\mathcal{PT}$ phase, the two eigenvalues split along the real axis, whereas in the broken $\mathcal{PT}$ phase they split along the imaginary axis. See Fig.~\ref{fig:grid_4_sectors} and Fig.~\ref{fig:pipi0_evalue_diff_largeq} for square root eigenvalue at the EP ar large $q$ for various $t$.

\item \textbf{Square root scaling of eigenvector coalescence.}
A convenient choice of right and left eigenvectors, satisfying
\(\widetilde{\boltz}(\pi,\pi,0)\,r_\pm=\lambda_\pm r_\pm\) and
\(\widetilde{\boltz}_{\pi,\pi,0}^{\dagger}\,\ell_\pm=\lambda_\pm^*\,\ell_\pm\), is
\begin{equation}
r_\pm:=
\begin{pmatrix}
i\Bsf\\
\lambda_\pm-\Asf
\end{pmatrix},
\qquad
\ell_\pm:=
\begin{pmatrix}
-i\Bsf\\
\lambda_\pm^*-\Asf
\end{pmatrix}.
\label{eq:right_evec_choice}
\end{equation}

\begin{itemize}
    \item \textbf{Eigenvector self-orthogonality.}
Define the biorthogonal overlap
\begin{equation}
\Overlap_\pm:=|\ell_\pm^\dagger r_\pm|.
\end{equation}
Away from the exceptional point, the left and right eigenvectors may be chosen biorthogonal and normalized so that
$\ell_n^\dagger r_m=\delta_{nm}$. For the explicit choice in Eq.~\eqref{eq:right_evec_choice}, one finds
$\ell_\pm^\dagger r_\pm
=
(\lambda_\pm-\Asf)^2-\Bsf^2$.
Using
$\lambda_\pm-\Asf
=
\frac{\Csf-\Asf\pm\sqrt{\Delta}}{2}$, and 
$\Delta=(\Asf-\Csf)^2-4\Bsf^2$,
we obtain
\begin{equation}
\ell_\pm^\dagger r_\pm
=
\frac{1}{2}\left[\Delta\mp(\Asf-\Csf)\sqrt{\Delta}\right].
\label{eq:overlap_closed_general}
\end{equation}
At an EP2, \(\Delta=0\), we have 
\begin{equation}
\Overlap_\pm(\mu_{\EP})=0 \qquad \text{at EP},
\label{eq:self_orthogonality}
\end{equation}
i.e., the coalesced eigenvector is self-orthogonal at the EP.
Expanding around the EP, we write 
$\Delta(\mu)=\Delta'(\mu_{\rm EP}) (\mu - \mu_{\EP})+O((\mu - \mu_{\EP})^2)$
with \(\Delta'(\mu_{\rm EP})\neq 0\). 
Therefore, 
around the EP, we have
\begin{equation}
\ell_\pm^\dagger r_\pm
=
\mp \frac{\Asf-\Csf}{2}\sqrt{\Delta}\,\Big[1+O(\sqrt{\Delta})\Big].
\end{equation}
Provided \(\Asf-\Csf\neq 0\) at \(\mu=\mu_{\rm EP}\), which is generically the case, we obtain
\begin{equation}
\Overlap_\pm(\mu)=\left|\ell_\pm^\dagger r_\pm \right|
\propto \left|\sqrt{\Delta(\mu)}\right|
\propto \sqrt{\left|\mu-\mu_{\rm EP}\right|}
\propto \sqrt{\left|\epsilon-\epsilon_{\rm EP}\right|}.
\label{eq:overlap_sqrt_scaling}
\end{equation}
i.e., the biorthogonal overlap vanishes with a square root exponent at a generic EP2. See Fig.~\ref{fig:pipi0_evector_overlap_largeq} for the vanishing biorthogonal overlap at the EP in the $(\pi,\pi,0)$ sector at large-$q$.

\item \textbf{Divergence of the Petermann factor.} A standard measure of eigenvector coalescence in non-Hermitian systems is the
Petermann factor
\begin{equation}
\Peter_\pm(\mu):=
\frac{\|\ell_\pm\|^2\,\|r_\pm\|^2}
{|\ell_\pm^\dagger r_\pm|^2},
\label{eq:Petermann_def}
\end{equation}
where \(r_\pm\) and \(\ell_\pm\) are respectively the right and left
eigenvectors, and \(\|\cdot\|\) denotes the Euclidean norm. This quantity is invariant under independent rescalings of \(r_\pm\) and \(\ell_\pm\). At a generic second-order exceptional point, we have \(\Delta(\mu_{\rm EP})=0\)
and \(A-C\neq0\), and the norms \(\|r_\pm\|\) and \(\|\ell_\pm\|\) remain finite and non-zero as
\(\mu\to\mu_{\rm EP}\). Using \eqref{eq:overlap_sqrt_scaling}, we obtain
\begin{equation}
\Peter_\pm(\mu)
\propto
\frac{1}{|\ell_\pm^\dagger r_\pm|^2}
\propto
\frac{1}{|\Delta(\mu)|}
\propto
\frac{1}{|\mu-\mu_{\rm EP}|}
\propto
\frac{1}{|\epsilon-\epsilon_{\rm EP}|},
\label{eq:Petermann_EP_scaling}
\end{equation}
i.e. the Petermann factor diverges with exponent \(-1\) at a generic EP2. See Fig.~\ref{fig:pipi0_petermann_largeq} for the divergent Petermann factor
at the EP in the \((\pi,\pi,0)\) sector at large \(q\).

\end{itemize}

\end{itemize}

\section{Exact TopSFF with spatially extended defects at finite $q$}\label{app:finite_q_exact}
In this appendix, we present an exact finite-\(q\) evaluation of the TopSFF with a spatially extended swap defect. We first formulate the computation for finite \(q\) and small \(t\) using an ensemble-averaged transfer matrix, and then give an analytic evaluation at \(t=1\) and finite \(q\).

\subsection{Exact finite-$q$ symbolic evaluation of TopSFF}
\label{subsec:finite_q_symbolic_tsff}
Consider the exact finite-\(q\) symbolic computation of the TopSFF, \(   K_{\topo}(t,L)= \mathrm{Tr}[ \hat{\mathcal{D}} (\hat{U}\otimes \hat{U}^*) ] = \Tr [\hat{S} \hat{U}] \Tr[\hat{U}^\dagger]\), with  \([\hat{\mathcal{D}}, (\hat{U}\otimes \hat{U}^*)] =0\), 
where the spatially extended topological defect is $\hat{\mathcal{D}}=\hat{S}\otimes\iden$.
Here $\hat{S}$ is the global swap operator, which reverses the ordering of the $L$ qudits according to $\hat{S}\ket{a_1,a_2,\dots,a_{L-1},a_L}=\ket{a_L,a_{L-1},\dots,a_2,a_1}$, with $a_i=1,\dots,q$. We analyse $K_{\topo}$ in a generic many-body quantum chaotic system satisfying $[\hat{\mathcal{D}},\hat{U}\otimes \hat{U}^*]=0$, i.e. the system possesses parity inversion symmetry.
Such a system is minimally realised by the Random Phase Model (RPM) with parity inversion symmetry and discrete time translation symmetry, defined in App.~\ref{app:parity_ptRPM}, i.e., $\hat{U}=\hat{U}_{\mathrm{f\text{-}p\text{-}RPM}}$.
The local Hilbert-space dimension is denoted by $q$.
Instead of generating an ensemble of quantum circuits and computing the averaged
TopSFF, this approach computes the generalized Boltzmann factor after the ensemble average over the quantum circuits.
In particular, this computation keeps \textit{all} permutations in Eq.~\eqref{eq:Weingarten_raw_form}
at finite \(q\), as opposed to retaining only the large-\(q\) Mickey Mouse diagrams.
For a single folded site, the two permutations appearing in the Weingarten
contraction are labelled by the pair \((\sigma,\tau)\in S_{2t}\times S_{2t}\).
The Kronecker deltas generated by Eq.~\eqref{eq:Weingarten_raw_form} identify
the local time indices into closed loops. We denote the resulting set of loops by
\(\mathcal L(\sigma,\tau)\), and its cardinality by
\(\ell(\sigma,\tau):=|\mathcal L(\sigma,\tau)|\). 

\subsubsection{Colour basis}
At finite \(q\), each loop carries an explicit colour index \(a \in \{1,\dots, q\}\). The basis state for a
single folded site is labelled by \(\ket{\rho}\equiv\ket{\sigma,\tau;\mathbf a}\), with \((\sigma,\tau)\in S_{2t}\times S_{2t}\), and \(\mathbf a=(a_1,\ldots,a_{\ell(\sigma,\tau)})\in \{1,\ldots,q\}^{\ell(\sigma,\tau)}\). In order to perform the average over the 2-site coupling gate, consider the colour index \(\Theta_\rho(x)\) associated to each endpoint,
see dotted lines along \(b\)-, \(c\)- and \(d\)-loops in
Fig.~\ref{fig:before_avg}(e), for a given configuration \(\rho\). The endpoint
label \(x\) runs over \(1,\ldots,4t\). The first \(2t\) labels are the endpoints
associated to \(\Tr[\hat S \hat U]\), and the labels \(2t+1,\ldots,4t\) are the
endpoints associated to \(\Tr[\hat U^\dagger]\).
For a given pair \((\sigma,\tau)\), the loop partition
\(\mathcal L(\sigma,\tau)\) is obtained by following the contractions around the
folded site. Let \(R_U\) be the cyclic rotation acting on the \(2t\) endpoints
associated to \(\Tr[\hat S\hat U]\), and let \(R_{U^\dagger}\) be the cyclic
rotation acting separately on the two \(t\)-blocks associated to
\(\Tr[\hat U^\dagger]\). Explicitly, for \(m\in\{1,\ldots,2t\}\),
\[
    R_U(m)
    =
    \begin{cases}
    m+1, & 1\leq m<2t,\\
    1, & m=2t,
    \end{cases}
\]
while
\[
    R_{U^\dagger}(m)
    =
    \begin{cases}
    m+1, & 1\leq m<t,\\
    1, & m=t,\\
    m+1, & t+1\leq m<2t,\\
    t+1, & m=2t.
    \end{cases}
\]
Starting from an endpoint \(y\in\{1,\ldots,2t\}\) associated to
\(\Tr[\hat S\hat U]\), one first reaches the endpoint \(2t+\tau(y)\) associated
to \(\Tr[\hat U^\dagger]\), and then returns to an endpoint associated to
\(\Tr[\hat S\hat U]\). The resulting map on the
\(\Tr[\hat S\hat U]\) endpoints is
\[
    \Phi_{\sigma,\tau}(y)
    :=
    R_U^{-1}\sigma^{-1}R_{U^\dagger}\tau(y).
\]
The loops are the orbits of \(\Phi_{\sigma,\tau}\), together with their
associated endpoints on \(\Tr[\hat U^\dagger]\). Explicitly, we write
\[
    L_{\alpha}
    =
    \mathcal O_{\alpha}
    \cup
    \{\,2t+\tau(y):y\in\mathcal O_{\alpha}\,\},
    \qquad
    \mathcal O_{\alpha}\in{\rm Cyc}(\Phi_{\sigma,\tau}) ,
\]
and the set of all loops are given by \(\mathcal L(\sigma,\tau)=\{L_1,\ldots,L_{\ell(\sigma,\tau)}\}\). The total number of coloured basis states is
\[
N(t,q)
=
\sum_{\sigma,\tau\in S_{2t}} q^{\ell(\sigma,\tau)}
=
(2t)!
\sum_{\pi\in S_{2t}} q^{C(\pi)}
=
(2t)!\prod_{r=0}^{2t-1}(q+r),
\]
where \(C(\pi)\) denotes the number of cycles of \(\pi\).  For example, one finds, for \(t=2\),
\(
    \{N(2,q)\}_{q=2}^{5}
    =
    \{2880,\;8640,\;20160,\;40320\},
\)
and for \(t=3\), 
\(
    \{N(3,q)\}_{q=2}^{3}
    =
    \{3628800,\;14515200\}.
\)

The two-site coupling comes from averaging the random diagonal phases in the
coupling gates. 
Define the colour function as
\(
    \Theta_{\rho}(x)=a_{\alpha} 
\) 
    if and only if
\( 
    x\in L_{\alpha}.
\)
Let
\(\rho_1=(\sigma_1,\tau_1;\mathbf a_1)\) and
\(\rho_2=(\sigma_2,\tau_2;\mathbf a_2)\)
be two adjacent one-site basis states. We write
\(\Theta_{\rho_1}(x),\Theta_{\rho_2}(x)\in\{1,\ldots,q\}\) for the colours
carried by endpoint \(x\) in the loop configurations \(\rho_1\) and
\(\rho_2\), respectively. The endpoints are split into those associated to
\(\Tr[\hat S \hat U]\), \(x=1,\ldots,2t\), and those associated to
\(\Tr[\hat U^\dagger]\), \(x=2t+1,\ldots,4t\). The phase average depends only on
the imbalance of colour pairs between the endpoints associated to
\(\Tr[\hat S \hat U]\) and the endpoints associated to
\(\Tr[\hat U^\dagger]\). Then the diagonal phase average gives the factor
\begin{align}
    D_{\rho_1\rho_2}(\mu)
    =
    \mu^{Q(\rho_1,\rho_2)},
    \qquad
    Q(\rho_1,\rho_2)
    :=
    \sum_{r,s=1}^{q}
    \left[
    \sum_{x=1}^{2t}
    \delta_{\Theta_{\rho_1}(x),r}
    \delta_{\Theta_{\rho_2}(x),s}
    -
    \sum_{x=2t+1}^{4t}
    \delta_{\Theta_{\rho_1}(x),r}
    \delta_{\Theta_{\rho_2}(x),s}
    \right]^2 .
    \label{eq:finite_q_phase_exponent}
\end{align}
From this expression, we see that the interaction does not depend on the
absolute values of the loop colours, but only on which loop colours coincide with which other loop colours. We will use this fact to compress the size of $\boltz(\rho_1, \rho_2)$ below. Combining the one-site Weingarten function \(\wg(\sigma^{-1}\tau, q)\) with the two-site phase average gives the exact finite-\(q\) generalized Boltzmann factor as an \(N(t,q)\)-dimensional transfer matrix
\begin{align}
    \boltz(\rho_1, \rho_2)
    =
    \sqrt{\wg(\sigma_1^{-1}\tau_1,q)}
    \,
    D_{\rho_1\rho_2}(\mu)
    \,
    \sqrt{\wg(\sigma_2^{-1}\tau_2,q)} .
    \label{eq:finite_q_generalized_boltzmann}
\end{align}
Note that, at finite \(q\), the generalized Boltzmann factor \(\boltz\) acts on the full permutation-and-colour basis, and therefore has a much larger matrix dimension than the large-\(q\) generalized Boltzmann factor, which is restricted to the leading Mickey Mouse diagrams.
The generalized Boltzmann factor can have signs and phases, and as a result, the emergent physics is non-Hermitian.
Although each Weingarten weight is increasingly suppressed in the large-\(q\) expansion as \(t\) increases, the matrix dimension of the generalized Boltzmann factor \(\boltz(\rho_1,\rho_2)\) grows rapidly with \(t\). After summing over the finite-\(q\) basis states in the transfer-matrix contraction, the unnormalized
TopSFF scales as \(\overline{K_{\topo}}=O(q^{-\Leff})\).

The TopSFF in a momentum sector is obtained by contracting powers of the
position-space transfer matrix \(\boltz\) between momentum-resolved boundary
states via \(\overline{K_{\topo}}=\phi_{-k}^T\boltz^{\Leff}\phi_k\). We consider
a single momentum sector \(k=(k_b,k_c,k_d)\), where
\(k_b=2\pi n_b/(2t)\), \(k_c=2\pi n_c/t\), and \(k_d=2\pi n_d/t\), with
\(n_b=0,1,\ldots,2t-1\) and \(n_c,n_d=0,1,\ldots,t-1\).
Let \((\sigma_{abcd},\tau_{abcd})\), defined in
Eq.~\eqref{app_eq:mickey_mouse_permutations}, denote the Mickey Mouse diagrams
labelled by \(a=0,1\), \(b=1,\ldots,2t\), and \(c,d=1,\ldots,t\). We define the
phase function
\begin{align}
    \Psi_{k}(\sigma,\tau)
    =
    \begin{cases}
    \exp\!\left[
        i\left(
        b k_b + c k_c + d k_d
        \right)
    \right],
    &
    \text{if }
    (\sigma,\tau)
    =
    (\sigma_{abcd},\tau_{abcd}),
    \\
    0,
    &
    \text{otherwise}.
    \end{cases}
    \label{eq:finite_q_momentum_phase_function}
\end{align}
The corresponding boundary state has components
\(\braket{\sigma,\tau;\mathbf a|\phi_k}
=
\sqrt{\wg(\sigma^{-1}\tau,q)}\,\Psi_k(\sigma,\tau)\). This boundary state is independent of the explicit loop-colour assignment \(\mathbf a\), and the colour degrees of freedom enter only through the bulk transfer matrix \(\boltz\). 

\subsubsection{Orbit basis}
Further, the colour structure in Eq.~\eqref{eq:finite_q_phase_exponent} allows an exact compression. Since the bulk interaction and the momentum boundary states are invariant under a global relabelling of the \(q\) colours,
one may replace explicit colour assignments by their equality patterns. Let \(m=\ell(\sigma,\tau)\), and write \(\mathcal L(\sigma,\tau)=\{L_1,\ldots,L_m\}\). Instead of assigning an explicit colour \(a_\alpha\in\{1,\ldots,q\}\) to each loop \(L_\alpha\), we assign an equality pattern, namely a partition \(\lambda=\{P_1,\ldots,P_r\}\) of the loop labels \(\{1,\ldots,m\}\), with
\(1\leq r\leq \min(m,q)\). Two loops have the same colour if and only if their labels belong to the same block of \(\lambda\), i.e.
\(a_\alpha=a_\beta \) if \(\alpha,\beta\in P_j\) for some $j$. 
For example, for \(m=3\), the partition \(\lambda=\{\{1,2\},\{3\}\}\) represents all explicit colour assignments with
\(a_1=a_2\neq a_3\), such as \((a_1,a_2,a_3)=(1,1,2)\), \((3,3,1)\), or
\((q,q,1)\). The compressed one-site basis state is labelled by
\(
    \ket{\rho}_{\rm orb}
    =
    \ket{\sigma,\tau;\lambda}_{\rm orb}.
\)
For a partition \(\lambda=\{P_1,\ldots,P_r\}\),  define the 
block-colour function \(\widehat{\Theta}_{\sigma,\tau;\lambda}(x)=j\) if \(x\in L_\alpha\) with \(\alpha\in P_j\). Choosing a physical colours amounts to choosing \(r\) distinct colours using the map \(\iota:\{1,\ldots,r\}\to\{1,\ldots,q\}\) with distinct values \(\iota(j)=c_j\), and set
\(
    \Theta_{\sigma,\tau;\lambda,\iota}(x)
    =
    \iota\!\left(
    \widehat{\Theta}_{\sigma,\tau;\lambda}(x)
    \right)
\).
For a fixed equality pattern with \(r\) colour blocks, the number of explicit colour assignments represented by this orbit is
\(
    w(\lambda)=(q)_r
    =
    q(q-1)\cdots(q-r+1)
\).
The number of possible equality patterns for \(m\) loops is the restricted Bell number
\(
    B_m^{(q)}
    =
    \sum_{r=1}^{\min(m,q)} S(m,r),
\)
where \(S(m,r)\) is the Stirling number of the second kind. Hence an \(m\)-loop
sector contributes \(B_m^{(q)}\) orbit states rather than \(q^m\) explicit-colour
states. The corresponding basis size can be evaluated as 
\[
    N_{\rm orb}(t,q)
    =
   \sum_{\sigma,\tau\in S_{2t}}
    B_{\ell(\sigma,\tau)}^{(q)} 
    =
    (2t)!
    \sum_{m=1}^{2t}
    c(2t,m)
    \sum_{r=1}^{\min(m,q)}
    S(m,r),
    \label{eq:finite_q_orbit_basis_count}
\]
where \(c(n,m)\) is the unsigned Stirling number of the first kind. 
For fixed \(t\), the orbit-compressed basis size \(N_{\rm orb}(t,q)\) stops growing once \(q\geq 2t\) and becomes independent of $q$, since an \(m\)-loop sector can use at most \(m\leq 2t\) distinct colours. This allows us to compute TopSFF for arbitrary large $q$ for the case of $t=2$. For example, one finds for \(t=2\), 
\(
    \{N_{\rm orb}(2,q)\}_{q=2}^{5}
    =
    \{1440,1728,1752,1752\},
\)
and for \(t=3\), 
\(
    \{N_{\rm orb}(3,q)\}_{q=2}^{5}
    =
    \{1814400,2678400,2894400,2916000\}
\). In a normalized orbit basis,
\(
    \ket{\sigma,\tau;\lambda}_{\rm orb}
    =
    w(\lambda)^{-1/2}
    \sum_{\mathbf a\in{\rm Orb}(\lambda)}
    \ket{\sigma,\tau;\mathbf a},
\)
the matrix element of the compressed transfer matrix is
\begin{equation}
    \boltz_{\rm orb}(\rho_1,\rho_2)
    =
    \frac{1}{\sqrt{w(\lambda_1)w(\lambda_2)}}
    \sqrt{\wg(\sigma_1^{-1}\tau_1,q)}
    \sqrt{\wg(\sigma_2^{-1}\tau_2,q)}
    \sum_{\iota_1}
    \sum_{\iota_2}
    \mu^{
    Q\left(
    \rho_1,\iota_1;
    \rho_2,\iota_2
    \right)
    },
    \label{eq:finite_q_orbit_boltzmann}
\end{equation}
where \(\rho_j=(\sigma_j,\tau_j;\lambda_j)\), and the sums are over all choices
of distinct physical colours for the blocks of \(\lambda_j\), equivalently over
maps \(\iota_j:\{1,\ldots,r_j\}\to\{1,\ldots,q\}\) with distinct values. The
exponent \(Q\) is the same phase exponent as in
Eq.~\eqref{eq:finite_q_phase_exponent}, but with \(\Theta_{\rho_j}(x)\) replaced
by \(\Theta_{\sigma_j,\tau_j;\lambda_j,\iota_j}(x)\). In the same normalized orbit basis, the momentum-resolved boundary state becomes
\begin{equation}
    \braket{\sigma,\tau;\lambda|\phi_k}_{\rm orb}
    =
    \sqrt{w(\lambda)}
    \sqrt{\wg(\sigma^{-1}\tau,q)}
    \Psi_k(\sigma,\tau),
\end{equation}
such that \(\overline{K_{\topo}}=\phi_{\mathrm{orb},-k}^T\boltz^{\Leff}_{\rm orb}\phi_{\mathrm{orb},k}\).

\subsubsection{Fourier transformation}
The ensemble-averaged transfer matrix is invariant under independent cyclic
translations along the three folded-time loops. These translations form the
group
$\Gamma_t=\mathbb Z_{2t}^{(b)}\times\mathbb Z_t^{(c)}
\times\mathbb Z_t^{(d)}$.
For a translation $ r=(r_b,r_c,r_d)\in\Gamma_t$, let $T_{r}$ denote the
corresponding matrix representation, the projector onto a momentum sector
$k=(k_b,k_c,k_d)$ is
\begin{align}
    P_k
    =
    \frac{1}{2t^3}
    \sum_{r_b=0}^{2t-1}
    \sum_{r_c,r_d=0}^{t-1}
    e^{-i(k_b r_b+k_c r_c+k_d r_d)}
    T_{r}.
    \label{eq:finite_q_momentum_projector}
\end{align}
The momenta are labelled by integer indices $k_b=2\pi n_b/2t$ and
$k_c=2\pi n_c/t$, $k_d=2\pi n_d/t$, with
$n_b=0,\ldots,2t-1$ and $n_c,n_d=0,\ldots,t-1$. Using that $[\boltz_{\rm orb},T_{ r}]=0$, the colour-orbit transfer matrix block diagonalises as
$\boltz_{\rm orb}=\bigoplus_k\boltz_{{\rm orb},k}$; hence, the momentum-resolved TopSFF is
$\overline{K_{\topo,k}}
=\phi_{{\rm orb},-k}^{T}
\boltz_{{\rm orb},k}^{\Leff}
\phi_{{\rm orb},k}$.

\subsubsection{Quotient reduction}
\label{subsubsec:finite_q_direct_quotient}
Let $\alpha=(\sigma_\alpha,\tau_\alpha;\lambda_\alpha)$ denote a state in the orbit basis, $\alpha$'s action on the $4t$ endpoints set generates a colour-equality pattern. We denote this endpoint pattern by $z_\alpha$. Explicitly, two endpoints $x,y\in\{1,\ldots,4t\}$ belong to the
same block of $z_\alpha$ if $\widehat{\Theta}_{\sigma_\alpha,\tau_\alpha;\lambda_\alpha}(x) = \widehat{\Theta}_{\sigma_\alpha,\tau_\alpha;\lambda_\alpha}(y)$.
In other words, $z_\alpha$ records which endpoint colours coincide, but not which permutation pair and contraction loop partition generated these coincidences. This parametrisation is useful since the two-site phase average depends only on the endpoint patterns $z_\alpha$ and $z_\beta$. We write the colour-averaged interaction
kernel $\mathcal D_{z_\alpha z_\beta}(\mu,q)$ between two endpoint patterns as
\begin{align}
    \mathcal D_{zz'}(\mu,q)
    =
    \frac{1}{\sqrt{w(z)w(z')}}
    \sum_{\iota}
    \sum_{\iota'}
    \mu^{Q(z,\iota;z',\iota')},
    \label{eq:finite_q_endpoint_kernel}
\end{align}
where $w(z) = q!/[q- r(z)]!$, with $r(z)$ the number of colour blocks in $z$.
The exponent $Q$ is the same phase-imbalance exponent as in
Eq.~\eqref{eq:finite_q_phase_exponent}, and the maps $\iota$ is defined below \eqref{eq:finite_q_orbit_boltzmann}. Note that the endpoint pattern does not determine the Weingarten function.
Different states $\alpha$ and $\alpha'$ can have $z_\alpha=z_{\alpha'}$ but
different cycle structures of
$\sigma_\alpha^{-1}\tau_\alpha$ and
$\sigma_{\alpha'}^{-1}\tau_{\alpha'}$. The permutation labels must therefore be retained when evaluating the Weingarten factors, and may be summed only afterwards. To express this reduction, define the rectangular matrix
$F_{\alpha z}
=
\sqrt{\wg(\sigma_\alpha^{-1}\tau_\alpha,q)}
\,\delta_{z,z_\alpha}$.
The colour-orbit transfer matrix then factorizes as
\begin{align}
    \boltz_{\rm orb}
    =
    F \mathcal DF^{T}.
    \label{eq:finite_q_direct_factorization}
\end{align}
The matrix $G:=F^{T}F$ is diagonal in the position-space endpoint-pattern
basis, with
\begin{align}
    G_{zz'}
    =
    \delta_{zz'}\,g_z(q),
    \qquad
    g_z(q)
    =
    \sum_{\alpha:\,z_\alpha=z}
    \wg(\sigma_\alpha^{-1}\tau_\alpha,q).
    \label{eq:finite_q_weingarten_aggregate}
\end{align}
$g_z(q)$ is the  sum of the Weingarten functions of all permutation-loop-colour states which produce that pattern $z$. 
We refer to this exact compression as a ``quotient reduction''. 
In a fixed momentum sector $k$, the map $F_k^{T}$ identifies orbit-basis
vectors that differ by an element of $\ker F_k^{T}$. Hence the dynamically
relevant space is the quotient
\(
\mathcal{H}_k/\ker F_k^{T}
\simeq
\operatorname{im}F_k^{T},
\)
which is represented by the endpoint-pattern states. In particular, any
$v\in\ker F_k^{T}$ is annihilated by the bulk transfer matrix,
$\boltz_{{\rm orb},k}v
=F_k\mathcal D_kF_k^{T}v=0$.

The boundary states are generated in the same way. Define
$\eta_k:=F_k^{T}\phi_{{\rm orb},k}$, where $F_k$ denotes the restriction of $F$ to momentum sector $k$. Its components are
\begin{equation}
    [\eta_k]_z
    =
    \sqrt{w(z)}
    \sum_{\alpha:\,z_\alpha=z}
    \wg(\sigma_\alpha^{-1}\tau_\alpha,q)
    \Psi_k(\sigma_\alpha,\tau_\alpha).
\end{equation}
Hence both the bulk transfer matrix and the boundary vectors retain the full
permutation-dependent Weingarten information.

For $\Leff\geq1$, Eq.~\eqref{eq:finite_q_direct_factorization} gives
\begin{align}
    \boltz_{{\rm orb},k}^{\Leff}
    =
    F_k\mathcal D_k
    (G_k\mathcal D_k)^{\Leff-1}
    F_k^{T}.
    \label{eq:finite_q_factorized_power}
\end{align}
We therefore define the exact quotient transfer matrix as
$\boltz_{{\rm quot},k}:=G_k\mathcal D_k$, together with the quotient boundary
vectors
$r_k:=\eta_k$ and
$\ell_{-k}^{T}:=\eta_{-k}^{T}\mathcal D_k$.
The momentum-resolved TopSFF becomes
\begin{align}
    \overline{K_{\topo,k}}(\Leff)
    =
    \ell_{-k}^{T}
    \boltz_{{\rm quot},k}^{\Leff-1}
    r_k.
    \label{eq:finite_q_quotient_topSFF}
\end{align}
Equation~\eqref{eq:finite_q_quotient_topSFF} evaluates exactly the same
TopSFF as the full colour-orbit transfer matrix, but requires repeated
multiplication only in the much smaller endpoint-pattern quotient space. We emphasise that the simplification via quotient is exact rather than a numerical approximation.

\subsection{Exact finite-\(q\) analytical evaluation of TopSFF at \(t=1\)}

For \(t=1\) at finite-\(q\), we evaluate exactly the TopSFF, \(   K_{\topo}(t,L)= \mathrm{Tr}[ \hat{\mathcal{D}} (\hat{U}\otimes \hat{U}^*) ] = \Tr [\hat{S} \hat{U}] \Tr[\hat{U}^\dagger]\), for the Random Phase Model (RPM) with parity inversion symmetry, defined in App.~\ref{app:parity_ptRPM}. At \(t=1\), imposing time translation symmetry makes no difference, since there is only a single time step.  The computation follows the colour-orbit construction above. For $t=1$, the explicit colour basis has size \(N(1,q)=2q(q+1)\), while the compressed basis has size \(N_{\rm orb}(1,q)=6\) for
\(q\geq2\). The boundary state further restricts the problem to a
three-dimensional invariant subspace. Diagonalizing the corresponding
finite-\(q\) generalized Boltzmann factor gives
\begin{align}
\overline{K_{\spatdef}}(t=1,\Leff,\mu,q)
&=
\frac{2(1-\mu^4+\rho)}{(q+1)(\mu^4+7-\rho)}
\lambda_{+}^{\Leff}
+
\frac{2(1-\mu^4)}{(q+1)\rho}
\lambda_{-}^{\Leff},
\label{app_eq:finite_q_t1_tsff}
\end{align}
where \(\lambda_{\pm}(\mu,q)=(1-\mu^4\pm\rho)/[2(q+1)]\),
\(\rho(\mu)=\sqrt{(\mu^4-1)(\mu^4+7)}\), \(\mu=e^{-\epsilon/2}\in[0,1]\), and
\(\Leff=L/2-1\) for even \(L\) (the case for odd \(L\) can be analysed easily). For \(0\leq\mu<1\), \(\rho=i\sqrt{(1-\mu^4)(\mu^4+7)}\), so \(\lambda_{\pm}\) form a complex-conjugate pair, and the \(t=1\) finite-\(q\)
Boltzmann factor is in the broken \(\mathcal{PT}\) symmetric regime. Writing
\(\lambda_{\pm}=\sqrt{2(1-\mu^4)}\,e^{\pm i\theta}/(q+1)\), with
\(\cos\theta=\sqrt{(1-\mu^4)/8}\) and
\(\sin\theta=\sqrt{(\mu^4+7)/8}\), the result becomes
\begin{equation}
\overline{K_{\spatdef}}(t=1,\Leff,\mu,q)
=
-\frac{4\sqrt{1-\mu^{4}}}{\sqrt{\mu^{4}+7}}\,
\Big(\sqrt{2(1-\mu^{4})}\Big)^{\Leff}
\sin(\Leff\theta)
\left(\frac{1}{q+1}\right)^{\Leff+1},
\label{eq:t1_finite_q_v2}
\end{equation}
which makes the oscillatory dependence on \(\Leff\) explicit. Note that \(\overline{K_{\spatdef}}(t=1,\Leff,\mu,q)=O(q^{-\Leff-1})=O(q^{-L/2})\). The finite-\(q\) answer has the same \(\mu\)- and \(\Leff\)-dependence as the
large-\(q\) result. The only difference is the replacement
\(q^{-(\Leff+1)}\) to \((q+1)^{-(\Leff+1)}\). Therefore, at \(t=1\), finite \(q\) preserves the
\(\mathcal{PT}\) symmetric structure and the oscillatory behaviour in \(L\), but renormalizes the overall large-\(q\) factor from \(q^{-\Leff-1}\) to \((q+1)^{-\Leff-1}\).  The finite-\(q\) TopSFF at \(t=1\) is plotted as a function of the system size
\(L\) in Fig.~\ref{fig:tsff_t1}.

\begin{figure}[ht]
    \centering
    \includegraphics[width=0.475\textwidth]{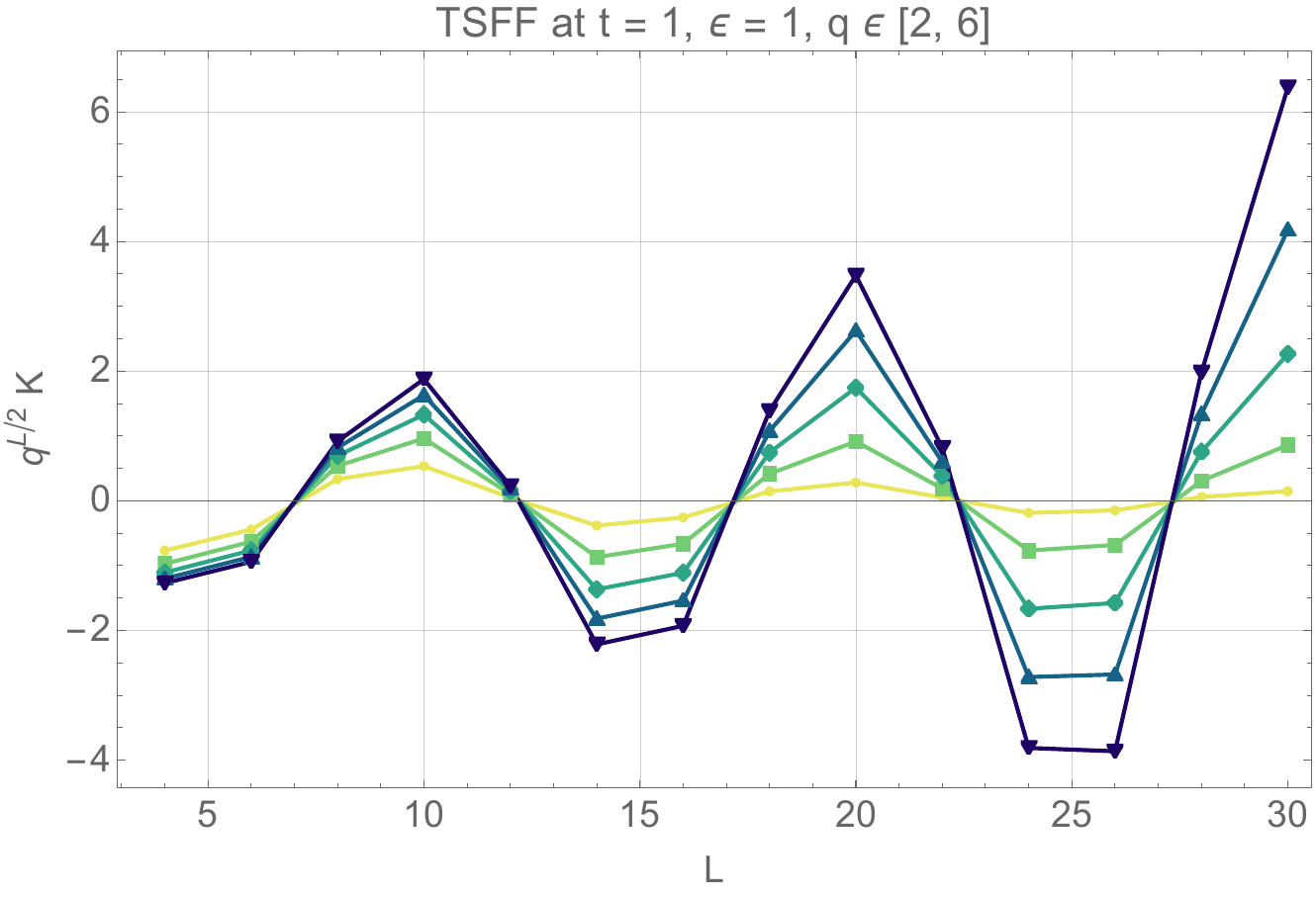}
    \caption{\textbf{Exact TopSFF at $t=1$ at finite $q$.}
    TopSFF \eqref{eq:t1_finite_q_v2} of the parity symmetric RPM at $\epsilon=1$ (with or without discrete time translational symmetry) against system size $L$ at $t=1$ from $q=2$ (yellow) to $q=6$ (blue).}
    \label{fig:tsff_t1}
\end{figure}

\section{TopSFF with spatially extended topological defects without time translational symmetry}\label{app:tsff_tmat}
In this appendix, we exactly compute the TopSFF with a spatially extended defect for the RPM in the absence of discrete time translation symmetry. We show that the TopSFF still gives rise to an emergent non-Hermitian \(\PT\)-symmetric single-particle problem for a temporal domain wall. However, any \(\PT\) symmetry breaking transitions are pushed towards the boundaries of parameter space, corresponding to zero or infinite interaction strength, as \(t\) increases in the pre-Heisenberg regime. This contrasts with the time-translation symmetric case, where a \(\PT\) transition can occur at finite interaction strength for arbitrary \(t\).

\subsection{Boltzmann factor $\boltz$ in position space}
\begin{figure}[ht]
    \centering
    \includegraphics[width=0.48\textwidth]{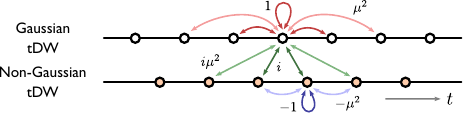}
    \caption{\textbf{Emergent non-Hermitian single-particle model}, as defined in \eqref{eq:ladder_model}. The single particle residing in the top (bottom) chain corresponds to the Gaussian (non-Gaussian) temporal domain wall. The intra-chain and inter-chain hopping terms are Hermitian and non-Hermitian respectively. We illustrate hopping amplitudes up to $O(\mu^2)$.
    }
    \label{fig:tight_binding}
\end{figure}

Here we exactly compute TopSFF for a spatially extended defect for the RPM in the absence of discrete time translational symmetry. As in the case with time translational symmetry, we compute the TopSFF
\(
K_{\topo}(t,L)
=
\mathrm{Tr}\!\left[\hat{\mathcal D}(\hat U\otimes \hat U^*)\right]
=
\Tr[\hat S\hat U]\Tr[\hat U^\dagger]\), where 
\(\hat{\mathcal D}=\hat S\otimes \iden\)
with  global swap operator $\hat S$.  $\hat{U}$ is defined in \eqref{app_eq:par_rpm_temprand}, which commutes with $\hat{D}$, but does not have discrete time translational symmetry, i.e., the RPM is temporally random and globally parity-inversion symmetric.
Upon ensemble averaging over \eqref{app_eq:par_rpm_temprand}, the TopSFF \(\overline{\Knormalized}_{\spatdef}\) can be expressed in terms of a transfer matrix, which describes a single-particle non-Hermitian hopping problem:
\be\label{app_eq:k_top_def_floq_temp_rand}
\ba
\overline{K}_{\spatdef} & \, =:   q^{-\Leff -1} \overline{\Knormalized}_{\spatdef} + O(q^{-{\Leff} -2}),
\\
\overline{\mathcal{K}}_{\spatdef}  
& \, =  
\phi^{T}\boltz^{\Leff}\phi
:=
\sum_{\sigma_1 ,\sigma_2}\phi(\sigma_1)\left[\boltz^{\Leff}\right]_{\sigma_1 \sigma_2}\phi(\sigma_2),
\ea
\ee
where $\boltz$ and $\phi$ are respectively the generalized Boltzmann factor encoding the many-body interactions and the boundary state in the folded representation. The effective length is $\Leff = L/2-1$ for even $L$. For odd $L$, TopSFF can be evaluated analogously.
The local TopSFF configurations are labelled by \(\sigma=(a,b)\) in the folded representation. 
Here \(a\in \{0,1\}\) labels the type of temporal domain wall (tDW): Gaussian tDW with $a=0$ consists only of Gaussian contractions, while non-Gaussian tDW with $a=1$ involves both Gaussian and non-Gaussian contractions. Further, \(b\in\{1,2,\dots,2t\}\) labels the position along the local length-$2t$ loop in $\tr[\hat{D} \hat{U}]$ in the folded representation. There are in total $4t$ local TopSFF configurations. Writing \(\sigma_i=(a_i,b_i)\) and \(\Delta b := (b_2-b_1)\pmod{2t}\), the generalized Boltzmann factor can be written as
\begin{equation}\label{eq:boltz_temporal_random_pos}
\boltz(\sigma_1,\sigma_2)
=
\big[\boltz_{a_1 a_2}(\Delta b)\big]_{a_1,a_2=0,1}
=
\begin{bmatrix}
\boltz_{00} & \boltz_{01}
\\
\boltz_{10} & \boltz_{11}
\end{bmatrix},
\end{equation}
where the matrix elements are
\begin{equation}\label{eq:boltz_temporal_random_entries}
\begin{aligned}
\boltz_{00}(\Delta b)
=&\;
\delta^{(2t)}_{\Delta b,0}
+\mu^{4t-4}\,\delta^{(2t)}_{\Delta b,t}
+\sum_{\ell=1}^{t-1}\mu^{4\ell-4}
\Big(
\delta^{(2t)}_{\Delta b,\ell}
+
\delta^{(2t)}_{\Delta b,-\ell}
\Big),
\\
\boltz_{11}(\Delta b)
=&\;
-\delta^{(2t)}_{\Delta b,0}
-\mu^{4t}\,\delta^{(2t)}_{\Delta b,t}
-\sum_{\ell=1}^{t-1}\mu^{4\ell}
\Big(
\delta^{(2t)}_{\Delta b,\ell}
+
\delta^{(2t)}_{\Delta b,-\ell}
\Big),
\\
\boltz_{01}(\Delta b)
=&\;
i\sum_{\ell=1}^{t}\mu^{4\ell-4}
\Big(
\delta^{(2t)}_{\Delta b,-\ell}
+
\delta^{(2t)}_{\Delta b,\ell-1}
\Big),
\\
\boltz_{10}(\Delta b)
=&\;
i\sum_{\ell=1}^{t}\mu^{4\ell-4}
\Big(
\delta^{(2t)}_{\Delta b,\ell}
+
\delta^{(2t)}_{\Delta b,1-\ell}
\Big).
\end{aligned}
\end{equation}
where \(\delta^{(2t)}_{x,y}=1\) if \(x\equiv y\ (\mathrm{mod}\ 2t)\) and \(0\) otherwise.

It is natural to represent $\boltz$ as a single-particle hopping problem of the temporal domain wall in second-quantized form by introducing creation and annihilation operators \(\hat{c}_{A,j}^{(\dagger)}\) and \(\hat{c}_{B,j}^{(\dagger)}\), 
\begin{equation}
|a=0,b\rangle := \hat{c}_{A,b}^\dagger |0\rangle,
\qquad
|a=1,b\rangle := \hat{c}_{B,b-\frac12}^\dagger |0\rangle,
\end{equation}
with periodic boundary condition \(b\equiv b+2t\). Here \(|0\rangle\) denotes the vacuum state annihilated by all annihilation operators, and \(A\) and \(B\) label the two sublattices. In this basis,
\begin{equation}\label{eq:H_boltz_pos_space}
H
=
\sum_{\sigma_1,\sigma_2}
\boltz(\sigma_1,\sigma_2)\,
|\sigma_1\rangle\langle \sigma_2|,
\qquad
\boltz(\sigma_1,\sigma_2):=\langle \sigma_1|H|\sigma_2\rangle .
\end{equation}
Note that such a Hamiltonian representation can also be made for Eq.~\eqref{eq:boltz_pos_space} in the presence of discrete time translational symmetry. In terms of creation and annihilation operators, in the absence of time translational symmetry, the Hamiltonian can then be written as
\begin{IEEEeqnarray}{rl}\label{eq:ladder_model}
H =& \, \, H_A + H_B + H_{AB} \, ,
\\ \nonumber
H_A =& \, \,\sum_{j=1}^{2t} \Big[ \hat{c}_{A,j}^\dagger \hat{c}_{A,j}
+ \mu^{4t-4} \hat{c}_{A,j}^\dagger \hat{c}_{A,j+t}
+ \sum_{\ell=1}^{t-1} \mu^{4\ell-4}\big( \hat{c}_{A,j}^\dagger \hat{c}_{A,j+\ell} + \hat{c}_{A,j}^\dagger \hat{c}_{A,j-\ell} \big) \Big] \, ,
\\ \nonumber
H_B =& \, \, - \sum_{j=1/2}^{2t-1/2} \Big[ \hat{c}_{B,j}^\dagger \hat{c}_{B,j}
+ \mu^{4t} \hat{c}_{B,j}^\dagger \hat{c}_{B,j+t}
+ \sum_{\ell=1}^{t-1} \mu^{4\ell}\big( \hat{c}_{B,j}^\dagger \hat{c}_{B,j+\ell} + \hat{c}_{B,j}^\dagger \hat{c}_{B,j-\ell} \big) \Big] \, ,
\\ \nonumber
H_{AB} =& \, \, i \sum_{j=1}^{2t}\sum_{\ell=1}^{t} \mu^{4\ell-4}
\Big(
\hat{c}_{A,j}^\dagger \hat{c}_{B,j-\ell-\frac12}
+\hat{c}_{A,j}^\dagger \hat{c}_{B,j+\ell-\frac32}
+\hat{c}_{B,j-\ell-\frac12}^\dagger \hat{c}_{A,j}
+\hat{c}_{B,j+\ell-\frac32}^\dagger \hat{c}_{A,j}
\Big) .
\end{IEEEeqnarray}
Here \(\hat{c}_{A,j}^{(\dagger)}\) with \(j=1,2,\dots,2t\) and \(\hat{c}_{B,j}^{(\dagger)}\) with \(j=1/2,3/2,\dots,2t-1/2\) act on the two legs of the ladder, with periodic identification \(\hat{c}^{(\dagger)}_{X,2tm+j}=\hat{c}^{(\dagger)}_{X,j}\) for \(X=A,B\).

\subsection{Boltzmann factor $\widetilde{\boltz}$ in momentum space}
Utilizing the translational invariance of  ${\boltz}$, we can express the TopSFF in terms of the generalized Boltzmann factor in momentum space as
\begin{equation} \label{eq:master-unitary_temp_rand} 
\overline{\Knormalized}_{\spatdef}=
\phi^T \boltz^{\Leff} \phi
=
\sum_{k} \widetilde{\phi}(-k)^T\, \widetilde{\boltz}(k)^{\Leff}\, \widetilde{\phi}(k),
\end{equation}
where the discrete Fourier transforms of $\phi$ and ${\boltz}$ are defined as
\begin{equation}
\label{eq:four_trans_def_temp_rand}
\begin{aligned}
\widetilde{\phi}(a;k)
:=& \;
\frac{1}{\sqrt{2t}}
\sum_{b=0}^{2t-1}
e^{ik b}\,\phi(a,b), 
\\
\widetilde{\boltz}_{a_1 a_2}(k)
:=& \;
\sum_{\Delta b=0}^{2t-1}
\boltz_{a_1 a_2}(\Delta b)\,
e^{-ik \Delta b }.
\end{aligned}
\end{equation}
where \(k = 2\pi n/(2t)\) with \(n=0,1,\dots,2t-1\).
A compatible momentum-space convention for the annihilation operators is
\begin{equation}
\hat{c}_{A,b}
=
\frac{1}{\sqrt{2t}}\sum_k e^{-ikb}\hat{c}_{A,k},
\qquad
\hat{c}_{A,k}
=
\frac{1}{\sqrt{2t}}\sum_{b=1}^{2t} e^{ikb}\hat{c}_{A,b},
\label{eq:FT_uniform_b}
\end{equation}
and
\begin{equation}
\hat{c}_{B,b-\frac12}
=
\frac{1}{\sqrt{2t}}\sum_k e^{-ikb}\hat{c}_{B,k},
\qquad
\hat{c}_{B,k}
=
\frac{1}{\sqrt{2t}}\sum_{b=1}^{2t} e^{ikb}\hat{c}_{B,b-\frac12}.
\label{eq:IFT_uniform_b}
\end{equation}
where \(k=\pi n/t\), \(n=0,1,\dots,2t-1\). 
With these conventions, the hopping Hamiltonian becomes
\begin{equation}
H=\sum_k
\begin{pmatrix}
\hat c_{A,k}^\dagger & \hat c_{B,k}^\dagger
\end{pmatrix}
\widetilde{\boltz}(k)
\begin{pmatrix}
\hat c_{A,k}\\
\hat c_{B,k}
\end{pmatrix},
\qquad
\widetilde{\boltz}(k)=
\begin{pmatrix}
\widetilde{\boltz}_{00}(k) & \widetilde{\boltz}_{01}(k)\\
\widetilde{\boltz}_{10}(k) & \widetilde{\boltz}_{11}(k)
\end{pmatrix}.
\label{eq:H_bloch_uniform}
\end{equation}
Evaluating the Fourier sums gives
\begin{align}
\widetilde{\boltz}_{00}(k)
&=
1+\mu^{4t-4}e^{-itk}+2\sum_{\ell=1}^{t-1}\mu^{4\ell-4}\cos(\ell k)
=: \Asf(k),
\\
\widetilde{\boltz}_{11}(k)
&=
-\Big[1+\mu^{4t}e^{-itk}+2\sum_{\ell=1}^{t-1}\mu^{4\ell}\cos(\ell k)\Big]
=: \Csf(k),
\\
\widetilde{\boltz}_{01}(k)
&=
i\sum_{\ell=1}^{t}\mu^{4\ell-4}\Big(e^{i\ell k}+e^{-i(\ell-1)k}\Big)
=
i\,e^{ik/2}\Bsf(k),
\\
\widetilde{\boltz}_{10}(k)
&=
i\sum_{\ell=1}^{t}\mu^{4\ell-4}\Big(e^{-i\ell k}+e^{i(\ell-1)k}\Big)
=
i\,e^{-ik/2}\Bsf(k),
\end{align}
where
\begin{equation}
\widetilde{\boltz}(k)=
\begin{pmatrix}
\Asf(k) & i\,e^{ik/2}\Bsf(k)\\
i\,e^{-ik/2}\Bsf(k) & \Csf(k)
\end{pmatrix},
\label{eq:boltzk_uniform}
\end{equation}
and after evaluating the geometric sums,
\begin{align}
\Asf(k)
&=
\frac{(1-\mu^4)(1-\mu^4+2\cos k)-(1-\mu^8)\mu^{4t-4}e^{-itk}}
{1-2\mu^4\cos k+\mu^8},
\\
\Csf(k)
&=
-\frac{(1-\mu^8)\bigl(1-\mu^{4t}e^{-itk}\bigr)}
{1-2\mu^4\cos k+\mu^8},
\\
\Bsf(k)
&=
\frac{2(1-\mu^4)\bigl(1-\mu^{4t}e^{-itk}\bigr)\cos(k/2)}
{1-2\mu^4\cos k+\mu^8}.
\end{align}
In this basis, the off-diagonal entries differ by the phase factors \(e^{\pm ik/2}\). We can remove the off-diagonal phases by a momentum-dependent unitary transformation,
\begin{equation}
\hat{\mathbf c}_k :=
\begin{pmatrix}
\hat c_{A,k}\\
\hat c_{B,k}
\end{pmatrix},
\qquad
\hat{\mathbf d}_k :=
\begin{pmatrix}
\hat d_{A,k}\\
\hat d_{B,k}
\end{pmatrix},
\qquad
\hat{\mathbf d}_k = U(k)\,\hat{\mathbf c}_k,
\qquad
U(k):=
\begin{pmatrix}
1 & 0\\
0 & e^{ik/2}
\end{pmatrix}.
\end{equation}
In terms of \(\hat d\), the \(B\)-sublattice Fourier transform becomes
\begin{equation}
\hat c_{B,b-\frac12}
=
\frac{1}{\sqrt{2t}}\sum_k e^{-ik(b+\frac12)}\hat d_{B,k}
=
\frac{1}{\sqrt{2t}}\sum_k e^{-ikj}\hat d_{B,k},
\qquad
j:=b+\frac12 .
\end{equation}
In terms of the \(\hat d\),
\begin{equation}
H=\sum_k
\begin{pmatrix}
\hat d_{A,k}^\dagger & \hat d_{B,k}^\dagger
\end{pmatrix}
\widetilde{\boltz}_{\mathcal{PT}}(k)
\begin{pmatrix}
\hat d_{A,k}\\
\hat d_{B,k}
\end{pmatrix},
\qquad
\widetilde{\boltz}_{\mathcal{PT}}(k)
=
U(k)\,\widetilde{\boltz}(k)\,U(k)^\dagger
=
\begin{pmatrix}
\Asf(k) & i\Bsf(k)\\
i\Bsf(k) & \Csf(k)
\end{pmatrix}.
\label{eq:boltzk_PT_gauge}
\end{equation}
Since \(k=\pi n/t\) implies \(e^{-itk}=(-1)^n\), the coefficients \(\Asf(k)\), \(\Bsf(k)\), and \(\Csf(k)\) are real. Therefore \(\widetilde{\boltz}_{\mathcal{PT}}(k)\) is in the \(\mathcal{PT}\) symmetric dimer form; see Sec.~\ref{app:tsff_and_pt}.

\begin{figure}[ht]
    \centering
    \includegraphics[width=1\textwidth]{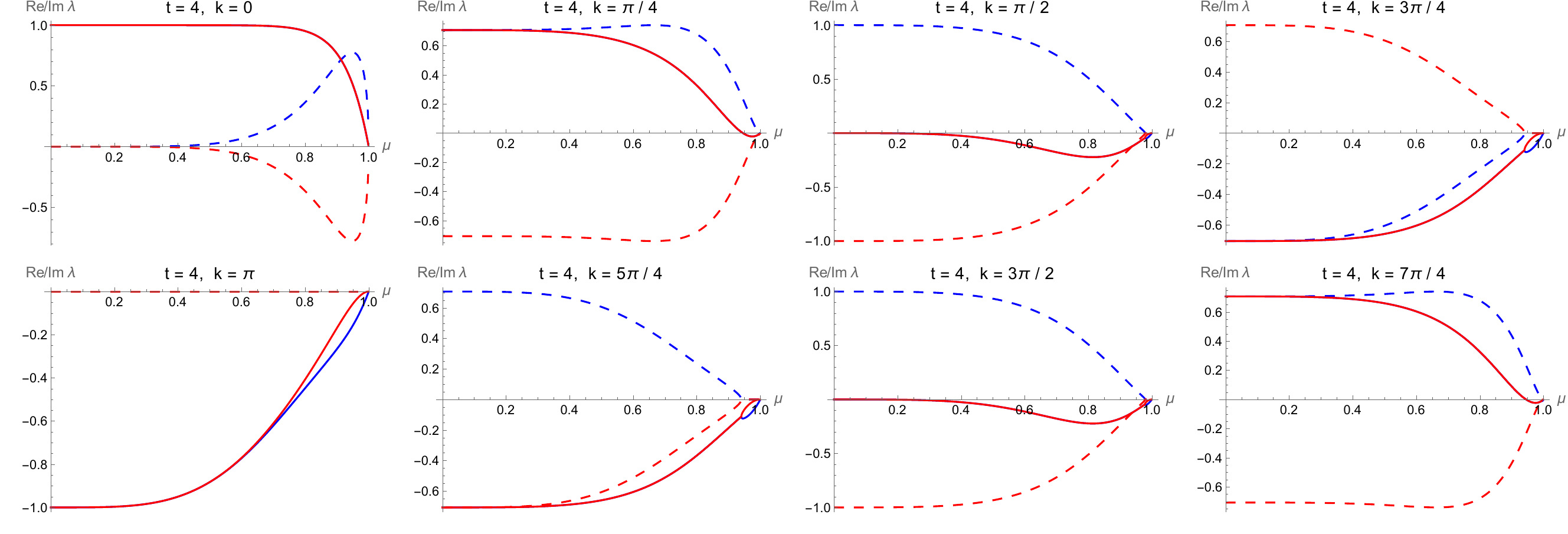}
    \caption{\textbf{Eigenvalues of Boltzmann factor for $t=4$ in the absence of time translational symmetry.}
    Complex eigenvalues of the generalized Boltzmann factor of $k \in \{0,1,2,\dots, 2t-1 \}$. $\Re \lambda_\pm$ are plotted in red and blue solid lines. $\Im \lambda_{\pm}$ are plotted in red and blue dashed lines. All sectors are $\mathcal{PT}$  symmetric. Exceptional points are shown in the sectors $k = \pi/2, 3\pi /4, 5\pi/4, 3\pi/2$.
    Unlike the case with time translational symmetry,  $|\lambda_\pm|$ do not exceed $1$, see \ref{fig:grid_4_sectors}. 
    }
    \label{fig:temp_rand_eval_t4}
\end{figure}

\subsection{Properties of $\widetilde{\boltz}(k)$ and $\mathcal{PT}$ symmetry}\label{app:boltz_mom_prop_temp_rand}
Here we describe the properties of the momentum space generalized Boltzmann factor in Eq.~\eqref{eq:boltzk_uniform} and Eq.~\eqref{eq:boltzk_PT_gauge}:
\begin{itemize}
    \item \textbf{Reality of $\det \widetilde{\boltz}(k)$ and $\tr \widetilde{\boltz}(k)$, and spectral pairing.} Since the allowed momenta are
\(
k=\pi n/t
\)
with \(n=0,1,\dots,2t-1\), we have
\(
e^{-itk}=(-1)^n\in\mathbb R
\).
It follows directly from Eqs.~\eqref{eq:boltzk_uniform}--\eqref{eq:boltzk_PT_gauge} that
\begin{equation}
\Asf(k),\Bsf(k),\Csf(k)\in\mathbb R,
\qquad
\Asf(-k)=\Asf(k),\quad
\Bsf(-k)=\Bsf(k),\quad
\Csf(-k)=\Csf(k).
\label{eq:ABC_even_real_temp_rand}
\end{equation}
As a consequence,
\begin{equation}
\tr \widetilde{\boltz}(k)=\Asf(k)+\Csf(k)\in\mathbb R,
\qquad
\det \widetilde{\boltz}(k)=\Asf(k)\Csf(k)+\Bsf(k)^2\in\mathbb R.
\end{equation}
Hence the characteristic polynomial has real coefficients, and the eigenvalues are
\begin{equation}
\lambda_\pm(k)
=
\frac{\tr\widetilde{\boltz}(k)\pm \sqrt{\Delta(k)}}{2},
\qquad
\Delta(k):=\big[\Asf(k)-\Csf(k)\big]^2-4\Bsf(k)^2\in\mathbb R.
\label{eq:disc_temp_rand}
\end{equation}

    \item \textbf{$\mathcal{PT}$ symmetry and antiunitary symmetry with respect to $\sigma_z$.} Take $\mathcal{P}:= \sigma_z$ and $\mathcal{T}$ to be complex conjugation, then the generalized Boltzmann factor is $\mathcal{PT}$ symmetric in the sense that
    \be
\sigma_z\,\widetilde{\boltz}_{\mathcal{PT}}(k)^*\,\sigma_z
=\widetilde{\boltz}_{\mathcal{PT}}(k)
,
\qquad 
\sigma_z\,\widetilde{\boltz}(k)^*\,\sigma_z=\widetilde{\boltz}(-k),
    \ee
i.e., every momentum sector is $\mathcal{PT}$ symmetric. As described in more details in Sec.~\ref{app:PT_mapping}, the standard $\mathcal{PT}$ symmetry breaking transition applies here and is controlled by \(\Delta(k)\) as
\begin{itemize}
    \item if \(\Delta(k)>0\), the two eigenvalues are real and distinct.
    \item if \(\Delta(k)=0\), the eigenvalues and eigenvectors coalesce, giving the candidate locus for an exceptional point.
    \item if \(\Delta(k)<0\), the eigenvalues form a complex-conjugate pair.
\end{itemize}

    \item \textbf{Complex symmetricity.} Let the superscript $T$ denote transposition, then the generalized Boltzmann factor satisfies
    \be
    \widetilde{\boltz}_{\mathcal{PT}}(k)^T
=
\widetilde{\boltz}_{\mathcal{PT}}(k),
\qquad 
\widetilde{\boltz}(k)^T=\widetilde{\boltz}(-k).
    \ee
    
    \item \textbf{Pseudo-Hermiticity with respect to $\sigma_z$.} Combining the $\mathcal{PT}$ symmetry (or antiunitary symmetry) and complex symmetricity above, the  generalized Boltzmann factor has pseudo-Hermiticity condition,
    \be
\sigma_z\,\widetilde{\boltz}_{\mathcal{PT}}(k)^\dagger\,\sigma_z
=
\widetilde{\boltz}_{\mathcal{PT}}(k),
\qquad 
\sigma_z\,\widetilde{\boltz}(k)^\dagger\,\sigma_z=\widetilde{\boltz}(k).
    \ee
Pseudo-Hermiticity is not an independent property once complex symmetricity and \(\mathcal{PT}\) symmetry are established.
\end{itemize}

\subsection{Boundary states}\label{app:bound_state_temp_rand}
The TopSFF with a spatially extended topological defect is given in Eq.~\eqref{app_eq:k_top_def_floq_temp_rand}. Here we describe the boundary states in the folded representation, both in position space and in momentum space. Away from the parity-inversion axes, the bulk of the quantum many-body system is parity-inversion symmetric, since it commutes with the global swap-defect operator.
Along the parity-inversion axes, the parity-inversion-symmetric coupling gates of the RPM are chosen according to Eqs.~\eqref{app_eq:pRPM_def3} and \eqref{app_eq:rpm_parity_temp_rand} in Sec.~\ref{app:model_par_inv}, subject to the constraint
\(\mathrm{SWAP}\, w(t,r)\, \mathrm{SWAP} = w(t,r)\),
where \(\mathrm{SWAP}\) denotes the swap operator acting on the same Hilbert space as \(w(t,r)\).
Coupling gates at different times are drawn independently, so discrete time translational symmetry is absent, i.e., \(w(t,r) \neq w(t',r)\) for  \(t \neq t'\).
After performing the ensemble average over the random phase gates, the boundary state \(\phi\) and its Fourier transform, defined in Eq.~\eqref{eq:four_trans_def_temp_rand}, are given by
\be \label{eq:temp_Rand_bd_state}
\phi(a,b) 
= i^{\delta_{a,1}}, 
\qquad 
\widetilde{\phi}(a;k)
=
i^{\delta_{a,1}}
\sqrt{2t}\,
\delta_{n,0},
\ee
where \(k = 2\pi n/(2t)\) with \(n=0,1,\dots,2t-1\).
Note that the imaginary factors arising from the Weingarten calculus have been absorbed into the boundary state. The boundary state in Eq.~\eqref{eq:temp_Rand_bd_state} lies entirely in the momentum sector \(k=0\). Hence
\be
\overline{\Knormalized}_{\spatdef} = 
\widetilde{\phi}(0)^T\, \widetilde{\boltz}(0)^{\Leff}\, \widetilde{\phi}(0).
\ee
In fact, in the absence of time translational symmetry, turning off the coupling gates along the parity-inversion axes of the RPM gives rise to the same boundary state as in Eq.~\eqref{eq:temp_Rand_bd_state}. The presence of these coupling gates can only be detected by probing higher moments of the TopSFF and likewise of the SFF~\cite{chan2020lyap}, or in the presence of time translational symmetry, see Sec.~\ref{app:tsff_toolbox}.

\subsection{Exact expression of TopSFF without time translational symmetry}

Since the boundary state has support only at \(k=0\), using the Cayley-Hamilton form in Eq.~\eqref{app_eq:ktop_cayley_hamilton_pt}, we obtain the TopSFF as
\begin{equation}
\overline{\Knormalized}_{\spatdef}
=
\widetilde{\phi}(0)^T\,\widetilde{\boltz}(0)^{\Leff}\,\widetilde{\phi}(0)
=
2t\,\varphi^T\widetilde{\boltz}(0)^{\Leff}\varphi
=2t\,g(0)\,u_{\Leff}(0),
\end{equation}
where $\varphi = (1,i)^T$, and $u_{\Leff}(k)$ and $g(k)$ are defined in Eq.~\eqref{app_eq:ktop_cayley_hamilton_pt}. 
For \(k=0\), one finds
\begin{equation}
\begin{aligned}
g(0)=& \, \Asf(0)-\Csf(0)-2\Bsf(0)
=
\,-\mu^{4t-4}(1-\mu^4),
\\
\tr\widetilde{\boltz}(0)
=& \,
2-(1+\mu^4)\mu^{4t-4},
\\
\det\widetilde{\boltz}(0)
=& \,
(1-\mu^{4t})(1+\mu^{4t-4}),
\\
\Delta(0)
:=& \,
\tr \, \widetilde{\boltz}(0)^2-4\det \widetilde{\boltz}(0)
=
\mu^{4t-8}\Big[\mu^{4t}(1+6\mu^4+\mu^8)-8\mu^4\Big] \leq 0\,,
\\
\lambda_\pm(0)
=& \, 
\frac{\tr\widetilde{\boltz}(0)\pm \sqrt{\Delta(0)}}{2}
=
1-\frac{(1+\mu^4)\mu^{4t-4}}{2}
\pm \frac{1}{2}\sqrt{\Delta(0)}
\end{aligned}
\end{equation}
Observe that for $k=0$ sector, $\Delta(0)<0$ for $\mu \in (0,1)$, i.e., there is not an exceptional point for $\mu \in (0,1)$, and the generalized Boltzmann factor is in the \(\mathcal{PT}\) broken phase with $\lambda_{\pm}$ forming complex conjugate pairs. Therefore, the TopSFF with spatially extended topological defects in the absence of time translational symmetry can be evaluated as 
\be \label{eq:K_spatdef_ana_temp_rand}
\ba
\overline{\Knormalized}_{\spatdef}
=&\;
-2t\mu^{4t-4}\left(1-\mu^4\right)\,
\frac{\lambda_+^{\Leff}-\lambda_-^{\Leff}}{\sqrt{\Delta}}
=
-\frac{4t\mu^{4t-4}\left(1-\mu^{4}\right)}{\sqrt{-\Delta}}
|\lambda|^{\Leff}\,\sin\!\big(\Leff \arg\lambda\big),
\\
\lambda
\equiv &\; \lambda_+ = \lambda^*_-=
1-\frac{1}{2}(1+\mu^4)\mu^{4t-4}
+\frac{1}{2}\sqrt{\Delta},
\\
\Delta
=&\;
\mu^{4t-8}\Big[\mu^{4t}\big(1+6\mu^{4}+\mu^{8}\big)-8\mu^{4}\Big]
\;\le 0 \, .
\ea
\ee
where $\Leff = L/2-1 $ for  even $L$, and $\Leff= (L-1)/2-1$ for odd $L$. $\mu = e^{-\epsilon/2}$ with $\epsilon \in [0, \infty)$ is the RPM parameter that sets the coupling strength, with larger $\epsilon$ corresponding to stronger coupling. Typical TopSFF behaviours are plotted Fig.~\ref{fig:tsff_no_tts_two_panels}.

\begin{figure}[t]
    \centering
      \centering
    \includegraphics[width=0.4\linewidth]{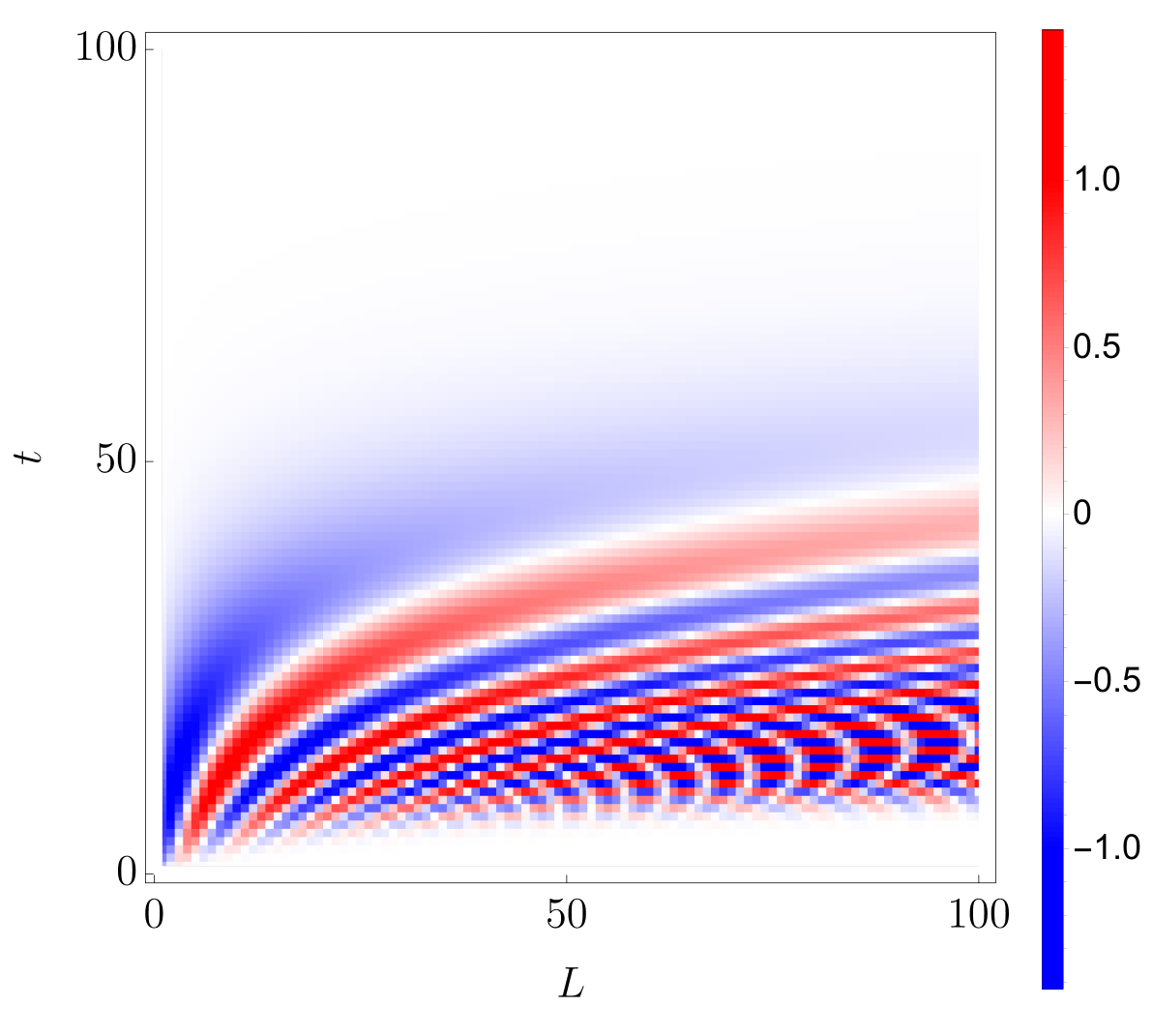}
    \includegraphics[width=0.4\linewidth]{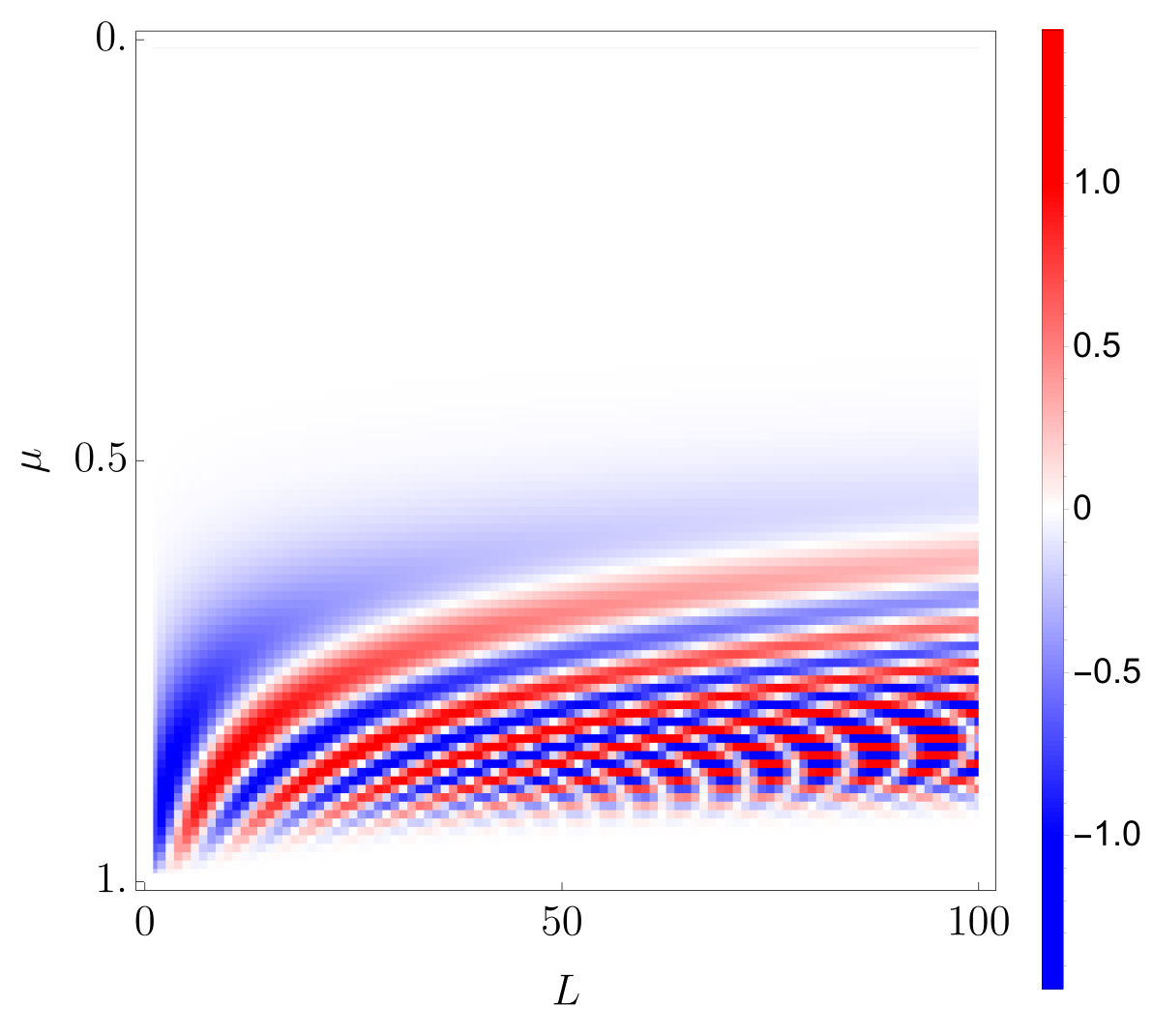}
    \caption{\textbf{TopSFF without time translational symmetry at large \(q\).}
    $\overline{\Knormalized_{\spatdef}}$ in the $t$-$L$ plane at $\mu=0.85$ (left) and in the $\mu$-$L$ plane at $t=15$ (right).}
    \label{fig:tsff_no_tts_two_panels}
\end{figure}

\subsection{Consistency check with finite-$q$ result at $t=1$} 

As a consistency check, at $t=1$, $\overline{\Knormalized_{\spatdef}}$ with spatially extended topological defects reduces to
\be\label{eq:tsff_temprand_t1}
\overline{\Knormalized_{\spatdef}}
=
-\frac{4\sqrt{1-\mu^{4}}}{\sqrt{\mu^{4}+7}}
\Big(\sqrt{2(1-\mu^{4})}\Big)^{\Leff}
\sin\big(\Leff\,\theta\big)\,,
\ee
where the angle $\theta$ is defined such that $\cos\theta=\sqrt{(1-\mu^4)/8}$ or 
$\sin\theta=\sqrt{(\mu^4+7)/8}$. 
Eq.~\eqref{eq:tsff_temprand_t1} coincides with (i) the large-$q$ expansion of the exact finite-$q$ result of the TopSFF in Eq.~\eqref{app_eq:finite_q_t1_tsff} at $t=1$, and with (ii) the exact large-$q$ calculation of the TopSFF in Eq.~\eqref{eq:floqparity_largeq_t1_check} for the  Floquet parity symmetric RPM at $t=1$.

\section{Temporally extended topological defects}\label{app:temp_ext}

In this appendix, we analyse TopSFF with temporally extended topological defects inserted in the spatial evolution. We consider two representative defects: a global swap defect in time reversal symmetric dynamics, and a time translation defect in Floquet dynamics. In both cases, we derive exact large-\(q\) expressions and the associated Thouless scaling forms for the topological free energy cost.

\subsection{Temporally extended global swap defects}

Here we consider the TopSFF of a temporally extended global swap defect \(\hat{\mathcal D}=\hat S\otimes \hat \iden\) acting on \( U_{\mathrm{s}} \otimes  U_{\mathrm{s}}\) in the spatial direction.  The one-dimensional generic quantum many-body chaotic circuit \(\hat U_{\mathrm{s}}(t,L):= \prod_{i=1}^L  \hat{V}(t,L)\) is  generated by the space-time dual transfer matrices \(\hat{V}(t,L)\) [Fig.~\ref{fig:tdw_mickey_mouse}(a) purple] of RPM defined in \eqref{app_eq:global_trs_rpm}, satisfying the \([\hat{V}(t,L), \hat S]=0\), and therefore has time reversal symmetry (TRS). We will also use the RHM with TRS, defined in Eq.~\eqref{app_eq:rhm_gtrs}, to verify the universal scaling form derived below from the RPM.

Using the methods described in App.~\ref{app:tsff_toolbox} and Ref.~\cite{wu2025}, the TopSFF with temporally extended defects in the TRS RPM can be evaluated exactly in the large-\(q\) limit as
\(
\overline{K_{\topo}}
=
\Tr\!\left[\boltz^{L}\sigma_x\right],
\)
where \(
\alpha(t)=t-1-\delta_{0,t\,\mathrm{mod}\,2}
\) and
\(
[\boltz]_{mn}
=
1+(1-\delta_{mn})\bigl(e^{-\epsilon\alpha(t)}-1\bigr)\)
is the transfer matrix, or equivalently the Boltzmann factor, of an effective Ising model.
\(\sigma_x\) is the Pauli-\(x\) matrix, which emerges because the global swap operator exchanges the role of the ladder mode and the twisted mode described in~\cite{wu2025}. 
 Thus, upon ensemble averaging, the temporally extended global swap topological defect becomes a spin-flip defect in the effective Ising description. Diagonalizing \(\boltz\) gives an exact expression for TopSFF with temporally extended defects in the large-\(q\) limit as
\begin{equation}
\overline{K_{\topo}}
=
\bigl[1+e^{-\epsilon\alpha(t)}\bigr]^L
-
\bigl[1-e^{-\epsilon\alpha(t)}\bigr]^L .
\end{equation}
In the Thouless double scaling limit~\cite{chan2021trans}, where
\(\ell\equiv L/\Lth\) is held fixed with \(\Lth=e^{\epsilon\alpha(t)}\), the TopSFF obeys
\begin{equation}
\kappa_{\topo}
:=
\lim_{\substack{L,t\to\infty\\ \ell=L/\Lth}}
\overline{K_{\topo}}
=
2\sinh\ell .
\label{eq:tsff_gtrs_scaling}
\end{equation}
In the absence of topological defects, the large-\(q\) SFF of the TRS RPM is
\(\lim_{q\to\infty}\overline{K}=[1+e^{-\epsilon \alpha(t)}]^L+[1-e^{-\epsilon \alpha(t)}]^L\), see App.~\ref{app:review_sff} and \cite{wu2025}. In the same Thouless scaling limit, the corresponding SFF scaling form is \(\kappa:=\lim_{t,L\to\infty}\overline{K}=2\cosh\ell\).
Interpreting the TopSFF and SFF as partition functions, with \(t\) playing the role of an inverse temperature, we define the topological free energy difference as
\begin{equation}
\begin{aligned}
\Delta F_{\topo}
:=
-\frac{1}{t}
\ln\frac{\overline{K_{\topo}}}{\overline{K}}
=
-\frac{1}{t}
\ln\!\left[
\tanh\!\left(L\,\operatorname{atanh} e^{-\epsilon \alpha(t)}\right)
\right],
\end{aligned}
\end{equation}
where the second equality is evaluated at the large-\(q\) limit. 

$\Delta F_{\topo}$ has the bounds $0 \leq \Delta F_{\topo} \leq \epsilon$ 
where $\epsilon$ is the spatial domain wall (sDW) effective tension and also controls the coupling strength of the quantum many-body system.   $\Delta F_{\topo}$ is monotonically increasing in $t$ and decreasing in $L$, i.e., $\partial_t \Delta F_{\topo} \geq 0$, and $\partial_L \Delta F_{\topo} \leq 0$.
At late time and fixed system size, we have $\Delta F_{\topo} = \epsilon - \tfrac{\ln L}{t}$. On the other hand, in the thermodynamic limit with large system size $L$ at fixed time $t$, we have $\Delta F_{\topo} = \tfrac{2}{t}e^{-2L \atanh e^{-\epsilon \alpha(t)}}$, and the spatial correlation length can be extracted as $\xi = \tfrac{1}{2 \atanh e^{-\epsilon \alpha(t)} } \to \tfrac{1}{2 e^{-\epsilon \alpha(t)} }  \propto \Lth$. 
Together, we have $\lim_{L\to \infty} \lim_{t \to \infty} \Delta F_{\topo}  = \epsilon$ and $\lim_{t \to \infty}  \lim_{L\to \infty} \Delta F_{\topo}  = 0$, i.e., the late-time and thermodynamic limits do not commute.

We begin with the first Thouless scaling limit, where \(t,L\to\infty\) with the Thouless ratio \(\ell\equiv L/\Lth=L e^{-\epsilon\alpha(t)}\) held fixed. In this limit, we have the free energy scaling form
\be
\mathscr{F}_{\topo, 1}\equiv \lim_{\substack{L,t \to \infty\\ \ell \equiv  L/\Lth}}  t \Delta F_{\topo} = - \ln \tanh \ell \,.
\label{eq: ratio gtrs-scaling}
\ee
In Fig.~\ref{fig:F_scaling_limits} left panel, we provide a finite-\(t\) check of the corresponding Thouless scaling limit, where the finite-\(t\) curves at fixed \(\mu=0.67\) collapse onto this universal scaling function as \(t\) is increased. The inset shows the unscaled \(\Delta F_{\topo}\) as a function of \(L\). This demonstrates that the temporally extended global swap defect measures the free-energy cost of a spatial domain wall, which becomes universal in the Thouless scaling limit.

We take the second Thouless scaling limit to be where $t$ and $L$ are sent to infinity, while $\tau \equiv t/\tth = \epsilon t/\ln L$ is fixed, we can write $ \Delta F_{\topo} = - \tfrac{1}{t} \ln L^{1- \tau }$ and
\be\label{app_eq:trs_2thoulessform}
\mathscr{F}_{\topo, 2} \equiv \lim_{\substack{L,t \to \infty\\ \tau  \equiv  t/\tth}}  \Delta F_{\topo} = 
\begin{cases}
0 \quad  & \tau  \leq 1 \,,
\\
\epsilon \left(1- \tfrac{1}{\tau } \right) & \tau  > 1\,,
\end{cases}
\ee
i.e., in the Thouless scaling limit, the universal scaling form has  a sharp crossover at $\tau =1$, i.e. at $t = \tth$. In Fig.~\ref{fig:F_scaling_limits} right panel, we test this second Thouless scaling limit at fixed \(\mu=0.67\), using finite-\(L\) data while varying \(t\). As \(L\) is increased, the curves collapse onto the universal scaling form \(\mathscr{F}_{\topo,2}\), with a rounded finite-size crossover sharpening near \(\tau=1\). The inset shows the corresponding unscaled \(\Delta F_{\topo}\) as a function of \(t\), illustrating the crossover from the thermodynamic-limit regime, where \(\Delta F_{\topo}\simeq 0\), to the late-time regime, where the topological free energy cost becomes finite.

As apparent in Fig.~\ref{fig:F_scaling_limits}, the convergence to the second Thouless scaling form is slower than that of the first because the two limits have parametrically different finite-size corrections. 
In the first Thouless limit, fixing \(\ell=Le^{-\epsilon t}\) gives \(L\operatorname{atanh}e^{-\epsilon t}=L(e^{-\epsilon t}+e^{-3\epsilon t}/3+\cdots)\), so the free energy difference differs from its scaling form only by power-law corrections,
\be
t\Delta F_{\topo}
=
\mathscr F_{\topo,1}(\ell)+O(L^{-2}).
\ee
at large finite system size \(L\).

In contrast, in the second Thouless scaling limit, fixing \(\tau=t/\tth=\epsilon t/\ln L\) gives \(e^{-\epsilon t}=L^{-\tau}\), and hence \(Le^{-\epsilon t}=L^{1-\tau}\). The sharp crossover at \(\tau=1\) is therefore rounded over the regime \(|\tau-1|\ln L=O(1)\). Correspondingly, the approach to the limiting scaling form is only logarithmic,
\be
\Delta F_{\topo}
=
\mathscr{F}_{\topo,2}(\tau)
+
O\!\left(\frac{1}{\ln L}\right),
\ee
at fixed \(\tau=t/\tth=\epsilon t/\ln L\). Equivalently, the finite-size rounding of the crossover has width \(\delta\tau\sim1/\ln L\). Thus the second scaling form suffers from logarithmically slow convergence with \(L\).

\begin{figure}
    \centering
    \includegraphics[width=0.49\linewidth]{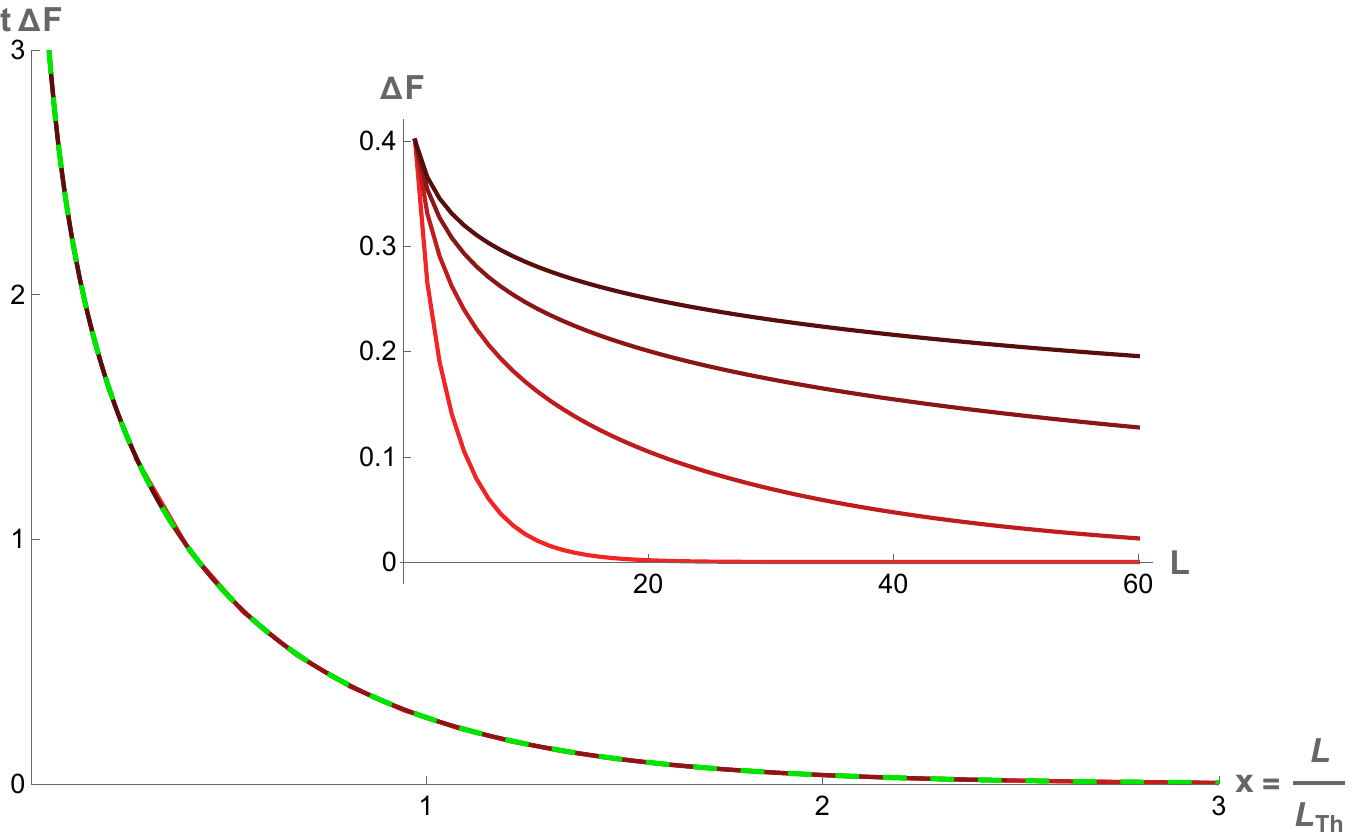}
    \hfill
    \includegraphics[width=0.49\linewidth]{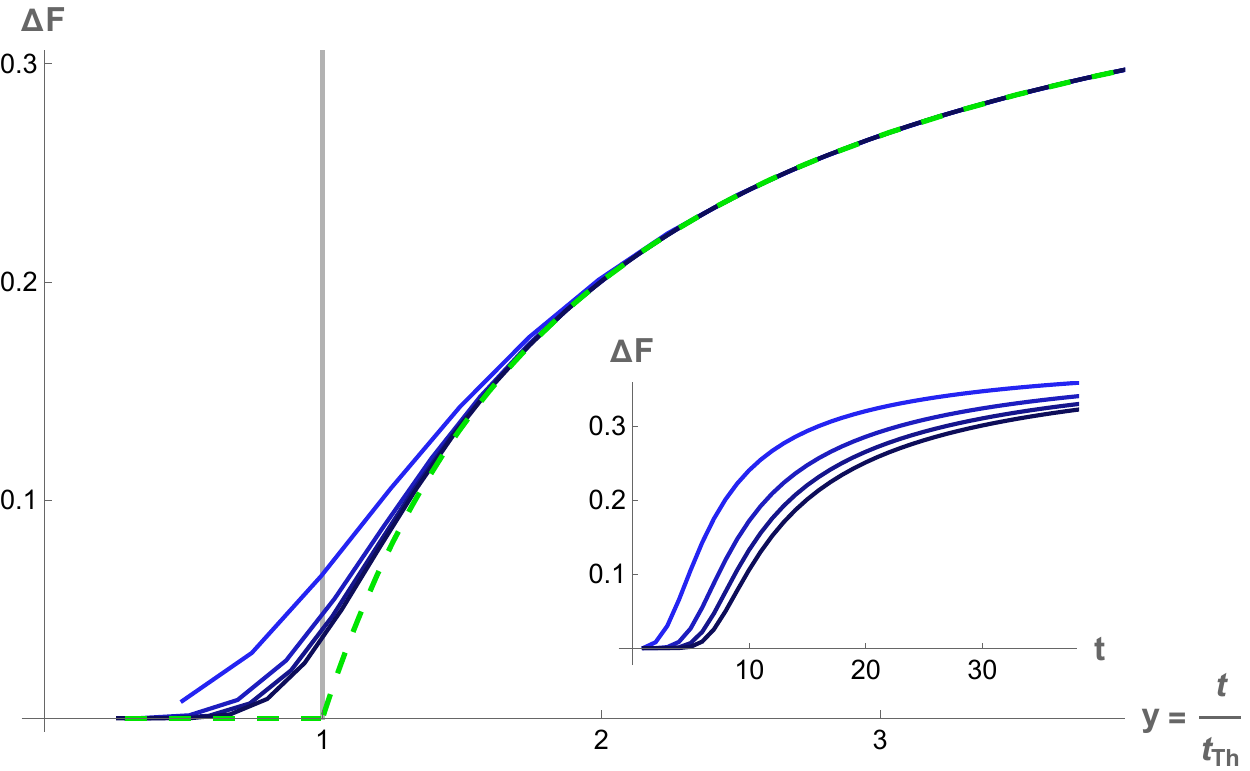}
    \caption{
    \textbf{Scaling collapse of free energy differences for temporally extended global swap defects.}
    Left: \(\mathscr{F}_{\topo, 1}\) plotted as a function of \(\ell\equiv L/\Lth\) at fixed \(\mu=0.67\), using finite-\(t\) data with \(t\in\{5,10,15,20\}\) and varying \(L\). Darker red curves denote larger \(t\). The inset shows the corresponding unscaled free-energy difference \(\Delta F_{\topo}\) as a function of \(L\).
    Right: \(\mathscr{F}_{\topo, 2}\) plotted as a function of \(y=t/\tth\) at fixed \(\mu=0.67\), using finite-\(L\) data with \(L\in\{5,10,15,20\}\) and varying \(t\). Darker blue curves denote larger \(L\). The inset shows \(\Delta F_{\topo}\) as a function of \(t\).
    }
    \label{fig:F_scaling_limits}
\end{figure}

\subsubsection{TRS and time translation symmetry}
The TopSFF with temporally extended global swap defects for TRS RPM with additional discrete time translational symmetry can be exactly evaluated in the large-\(q\) limit as 
\be
\ba
\, &\overline{K_{\topo}}
 =
\begin{cases}
\left[ 1 +t e^{-\epsilon(t-1)} + (t-1) e^{-\epsilon t}\right]^L
\\
\quad  -\left[1-t e^{-\epsilon (t-1)}+(t-1) e^{-\epsilon t}\right]^L , \quad & t \text{ odd}
\\
\left[1+ \frac{1}{2} t e^{-\epsilon (t-2)}+\left(\frac{3 t}{2}-1\right) e^{-\epsilon t}\right]^L
\\
\quad  -\left[1+ \frac{1}{2} t e^{-\epsilon (t-2)}-\left(\frac{t}{2}+1\right) e^{-\epsilon t}\right]^L , 
& t \text{ even} \, .
\end{cases}
\ea
\ee
From this expression, one can straightforwardly generalize the above procedure to derive similar universal scaling form for \(\overline{K_{\topo}}\) as in the previous section and \cite{wu2025}.

\subsection{Temporally extended time translation defects.} 
Now we consider the TopSFF of a temporally extended time translation defect \(\hat{\mathcal D}=\hat T\otimes \hat \iden\) acting on \( U_{\mathrm{s}} \otimes  U_{\mathrm{s}}\) in the spatial direction. Here \(T\) is the time translational operator defined by \(T \ket{a_1, a_2, \dots a_t} = \ket{a_t, a_1, \dots a_{t-1}}\), which cyclically shifts the time indices.
The one-dimensional generic quantum many-body chaotic circuit \(\hat U_{\mathrm{s}}(t,L):= \prod_{i=1}^L  \hat{V}(t,L)\) is realized generated by the space-time dual transfer matrices \(\hat{V}(t,L)\) [Fig.~\ref{fig:tdw_mickey_mouse}(a) purple] of RPM defined in \eqref{app_eq:floqonly_rpm_geom}, satisfying the \([\hat{V}(t,L), \hat T]=0\), and therefore has discrete time translational symmetry. 
We will also use the RHM with time translational symmetry, defined in Eq.~\eqref{app_eq:floqonly_rhm_geom}, to verify the universal scaling form derived below from the RPM.

The TopSFF with the discrete time translational operator as topological defects for the Floquet RPM can be evaluated exactly in the large-\(q\) limit as 
\be
\overline{K_{\topo}} = 
\left[1+ (t-1) e^{-\epsilon t}\right]^L
-(1- e^{-\epsilon t})^L \, ,
\ee
Taking the Thouless scaling limit~\cite{chan2021trans} where $\ell \equiv  L/\Lth$ is fixed with  $\Lth = e^{\epsilon t} /t$,  the TopSFF scaling form gives 
\be
\kappa_{\topo} :=  \lim_{\substack{L,t \to \infty\\ \ell \equiv  L/\Lth}} \left[ \, \overline{K_\topo}  +1 \right] = e^\ell  \, .
\ee
In the absence of topological defect, the SFF can be evaluated in the large-\(q\) limit as $\overline{K} = [1+ (t-1) e^{-\epsilon t}]^L+ (t-1) (1-e^{-\epsilon t})^L$~\cite{chan2018spectral}. In the Thouless scaling limit, the corresponding SFF scaling form is $\kappa := \lim_{t,L\to \infty} \overline{K} -t + 1 = e^\ell - \ell$.

Interpreting the TopSFF and SFF as partition functions, with \(t\) playing the role of an inverse temperature, we define the topological free energy difference as
\begin{equation}
\begin{aligned}
\Delta F_{\topo}
:=
-\frac{1}{t}
\ln\frac{\overline{K_{\topo}}}{\overline{K}}
 =  -\frac{1}{t} \ln \, \frac{
1-\omega^L 
}{
1 +  (t-1) \omega^L} \,,
\end{aligned}
\end{equation}
where $\omega =  (1- e^{-\epsilon t})/[1+ (t-1) e^{-\epsilon t}]$, and the second equality is evaluated at the large-\(q\) limit.

Similar to the global TRS case, $\Delta F_{\topo}$ has the bounds $0 \leq \Delta F_{\topo} \leq \epsilon$ where $\epsilon$ is the spatial domain wall (sDW) effective tension and also controls the coupling strength of the quantum many-body system. $\Delta F_{\topo}$ is monotonically increasing in $t$ and decreasing in $L$, i.e., $\partial_t \Delta F_{\topo} \geq 0$, and $\partial_L \Delta F_{\topo} \leq 0$.
 The upper bound is obtained by expanding at small $e^{-\epsilon t}$. 
At late time and at fixed system size $L$, we have $\Delta F_{\topo} = \epsilon - \tfrac{\ln L}{t}$. On the other hand, in the thermodynamic limit with large system size $L$ at fixed time $t$, we have $\Delta F_{\topo} = 
e^{-L/\xi}$, where the spatial correlation length is $\xi= \left[ \ln \tfrac{1+(t-1) e^{-\epsilon t}}{ 1- e^{-\epsilon t}} \right]^{-1} \to 1/t e^{-\epsilon t}  \propto \Lth$.
Together, we have $\lim_{L\to \infty} \lim_{t \to \infty} \Delta F_{\topo} = \epsilon$ and $\lim_{t \to \infty}  \lim_{L\to \infty} \Delta F_{\topo} = 0$, i.e., the late-time and thermodynamic limits do not commute. 

We now take the corresponding first Thouless scaling limit for the
time translation topological defect.  
In the scaling limit \(L,t\to\infty\) with
\(\ell:=L/\Lth=Lte^{-\epsilon t}\) fixed, we have \(e^{-\epsilon t}\to0\). Thus
\(
\ln\omega
=
\ln(1-e^{-\epsilon t})
-
\ln\left[1+(t-1)e^{-\epsilon t}\right]
=
-t e^{-\epsilon t}
+
O(t^2e^{-2\epsilon t}),
\)
so \(L\ln\omega=-\ell+O(\ell^2/L)\), and hence \(\omega^L\to e^{-\ell}\).
The free energy therefore becomes
\(
t\Delta F_{\topo}
=
-\ln\!\left[
\frac{1-e^{-\ell}}{1+(t-1)e^{-\ell}}
\right]
=
\ln t-\ln(e^\ell-1)+o(1)
\).  
The free energy therefore becomes, keeping the leading finite-\(t\) dependence,
\begin{equation}
t\Delta F_{\topo}-\ln t
=
-\ln\!\left[
\frac{t(1-e^{-\ell})}{1+(t-1)e^{-\ell}}
\right].
\end{equation}
Therefore, in the large-\(q\) limit, the universal scaling
form of TopSFF with a temporally extended time translation defect in the Thouless scaling limit is
\begin{equation}
\mathscr{F}_{\topo, 1}^{\mathrm{floq}}
\equiv 
\lim_{\substack{L,t\to\infty\\ \ell \equiv Lte^{-\epsilon t}}}
\left[
t\Delta F_{\topo}-\ln t
\right]
=
-\ln(e^\ell-1).
\end{equation}
This convergence uses \(1+(t-1)e^{-\ell}\simeq t e^{-\ell}\), and therefore requires
\((t-1)e^{-\ell}\gg1\) or \(\ell\ll \ln t\). Thus, at finite \(t\), the asymptotic scaling form is valid only before the crossover at \(\ell\sim\ln t\).

We consider next the second Thouless scaling limit in which
\(\tau := t/\tth = \epsilon t / \ln L\) is kept fixed as \(L,t\to\infty\). Since \(e^{-\epsilon t}=L^{-\tau}\), we have 
\(\ln \omega
    =
    \ln(1-e^{-\epsilon t})-\ln\left[1+(t-1)e^{-\epsilon t}\right]
    \simeq
    -te^{-\epsilon t}\), and therefore
\(    \omega^L \simeq
    \exp\left[-tL^{1-\tau}\right]\).
For \(\tau \leq 1\), we have \(\omega^L\to0\) and hence
\(
    \Delta F_{\topo}
    \longrightarrow 0 
\).
For \(\tau>1\), we have \(1-\omega^L\simeq tL^{1-\tau}\), while
\(1+(t-1)\omega^L\simeq t\). Therefore,
\begin{equation}
    \Delta F_{\topo}
    \simeq
    -\frac1t \ln L^{1-\tau}
    =
    \frac{(\tau-1)\ln L}{t}
    =
    \epsilon\left(1-\frac1{\tau}\right).
\end{equation}
Hence, in the large-\(q\) limit, the second Thouless scaling form of TopSFF with temporally extended time translation defect is
\begin{equation}
\mathscr{F}_{\topo, 2} \equiv \lim_{\substack{L,t \to \infty\\ \tau  \equiv  t/\tth}}  \Delta F_{\topo} = 
    \begin{cases}
        0, & \tau\leq 1, \\[4pt]
        \epsilon\left(1-\dfrac1{\tau}\right), & \tau>1 ,
    \end{cases}
\end{equation}
which has the same functional form as in Eq.~\eqref{app_eq:trs_2thoulessform}.

\section{Finite-\(q\) numerical simulations}\label{app:num_sim}

\subsection{Spatially extended topological defects}

We numerically simulate the finite-$q$ TopSFF with a spatially extended topological defect,
\(
K_{\topo}(t,L)
=
\mathrm{Tr}\!\left[\hat{\mathcal D}(\hat U\otimes \hat U^*)\right]
=
\Tr[\hat S\hat U]\Tr[\hat U^\dagger]\), where 
\(\hat{\mathcal D}=\hat S\otimes \iden\)
with  global swap operator $\hat S$ . We take $\hat U$ to be the parity inversion symmetric Random Phase Model (RPM) with discrete time translation symmetry, defined in App.~\ref{app:parity_ptRPM}. The local Hilbert-space dimension is denoted by $q$. 

We compute the TopSFF using two complementary finite-$q$ methods. 
Firstly, we estimate the TopSFF by sampling over random circuit realisations. Secondly, by constructing a transfer matrix after the ensemble average of the RPM tensor network, we exactly compute TopSFF for small \(t\) and arbitrary \(L\), as described in App.~\ref{subsec:finite_q_symbolic_tsff}. 
The two methods agree, as demonstrated in Fig.~\ref{fig:parity_tsff_temporal_random}. 
We primarily use the second method, despite its restriction to small time up to $t=4$, because it is fluctuation free and exact up to machine precision. This enables access to large \(L\), allowing us to verify at finite \(q\) the \(\mathcal{PT}\) symmetry breaking signatures predicted in the large-\(q\) limit.

 \begin{figure}
    \centering
    \includegraphics[width=0.32\textwidth]{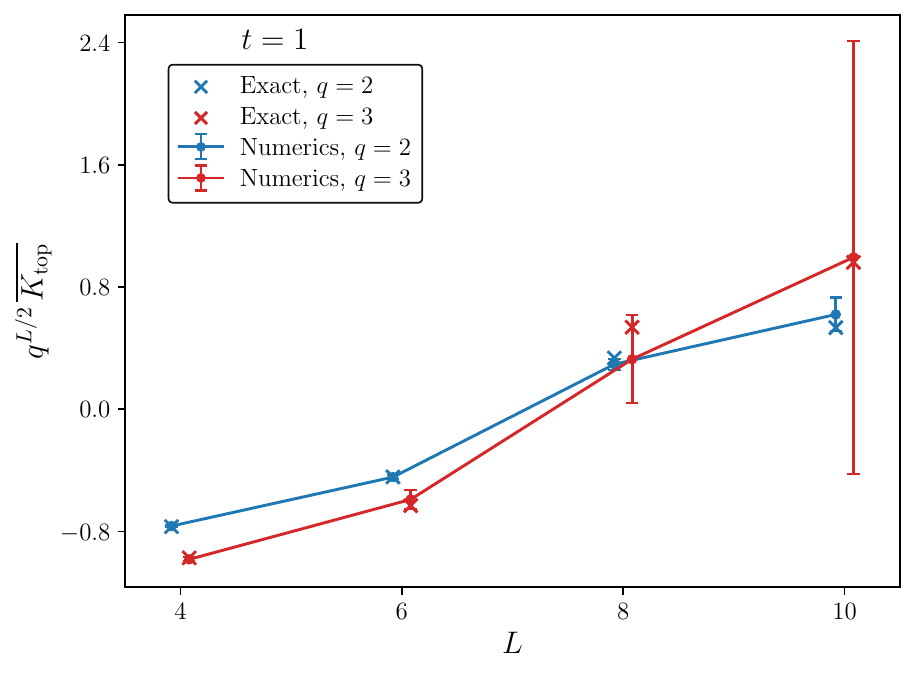}
    \hfill
    \includegraphics[width=0.32\textwidth]{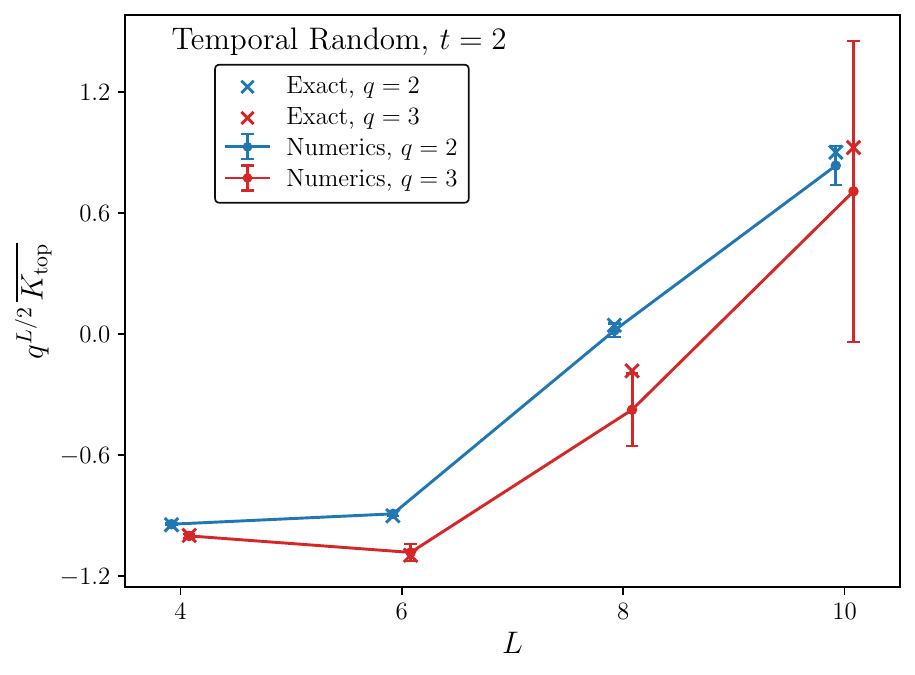}
\hfill
     \includegraphics[width=0.32\linewidth]{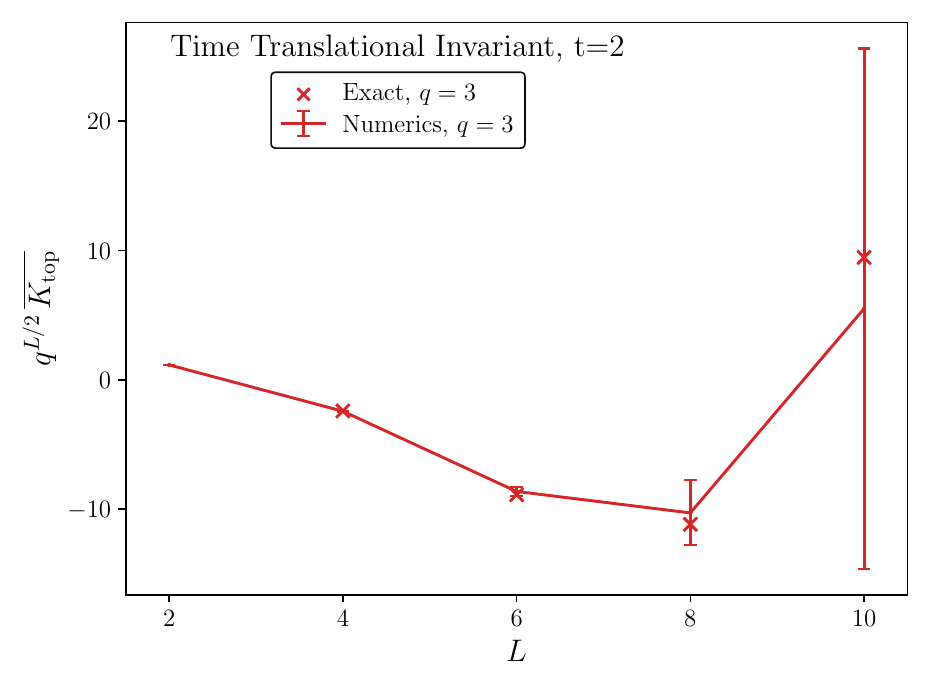}
    \caption{
    \textbf{Two independent methods of obtaining finite-$q$ small-$t$ TopSFF $\overline{K_{\mathrm{top}}}$} for parity symmetric RPM~\ref{app_eq:par_rpm_temprand} at variance $\epsilon=1$ via sampling over $5\times 10^{6}$ circuit realizations (dots) and exact symbolic computation (crosses) as described in \ref{subsec:finite_q_symbolic_tsff}. 
    In both panels, blue data denotes $q=2$ and red data denotes $q=3$. The exact finite-$q$ results are overlaid for comparison, showing good agreement with the numerical results. 
    }
    \label{fig:parity_tsff_temporal_random}
\end{figure}

 \begin{figure}
    \centering
  \includegraphics[width=0.32\textwidth]{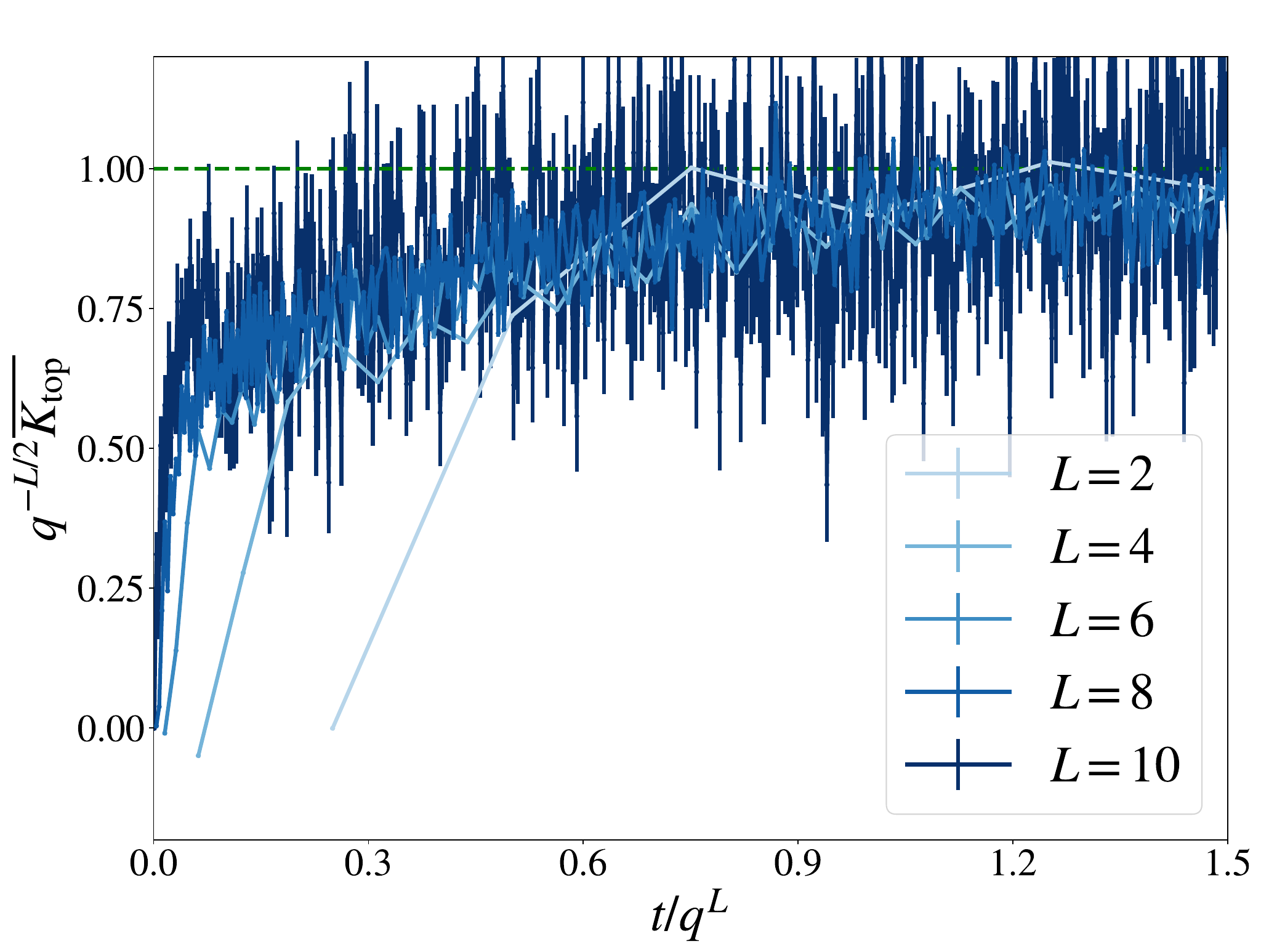}  \includegraphics[width=0.32\textwidth]{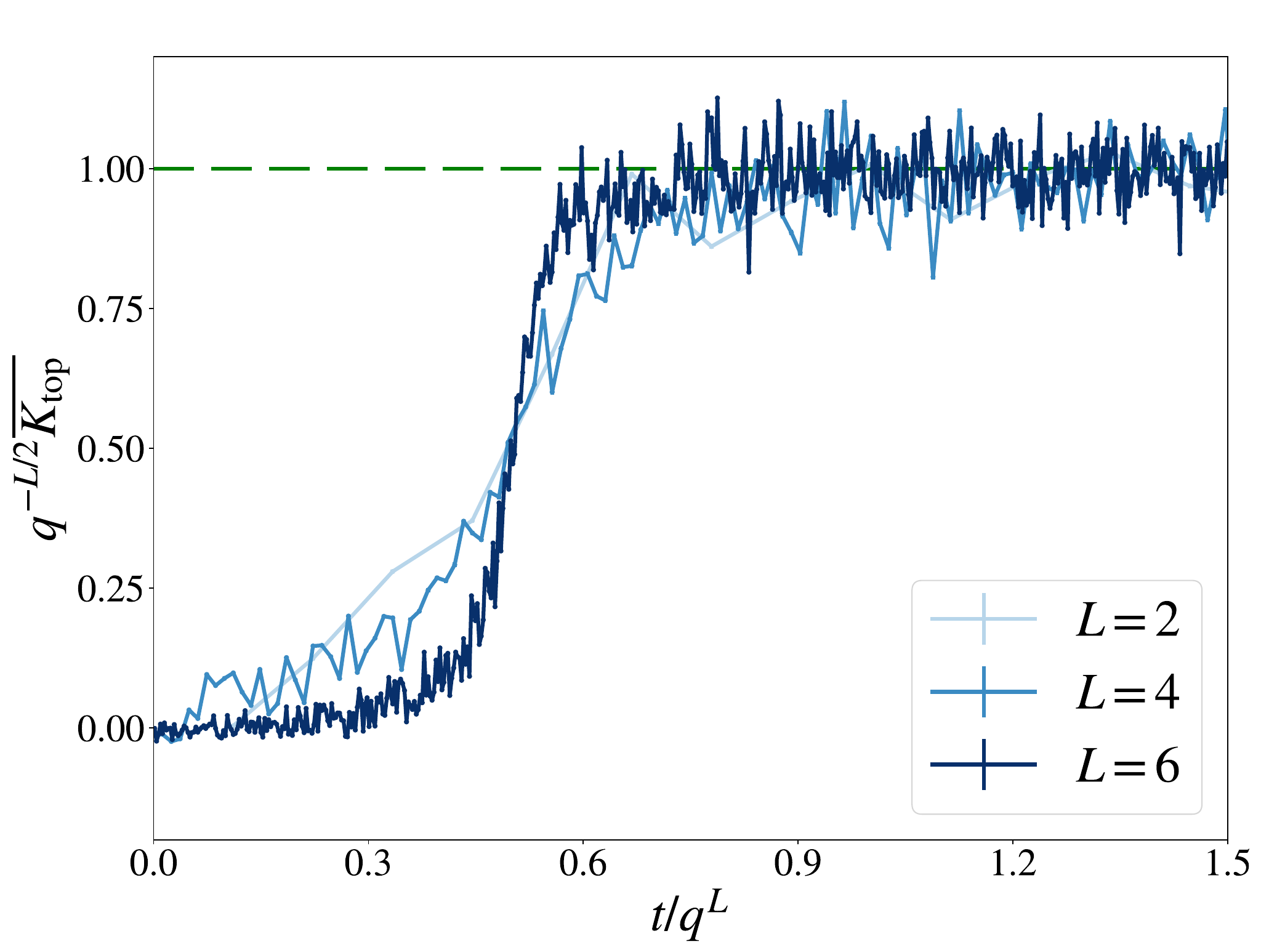}
     \includegraphics[width=0.32\linewidth]{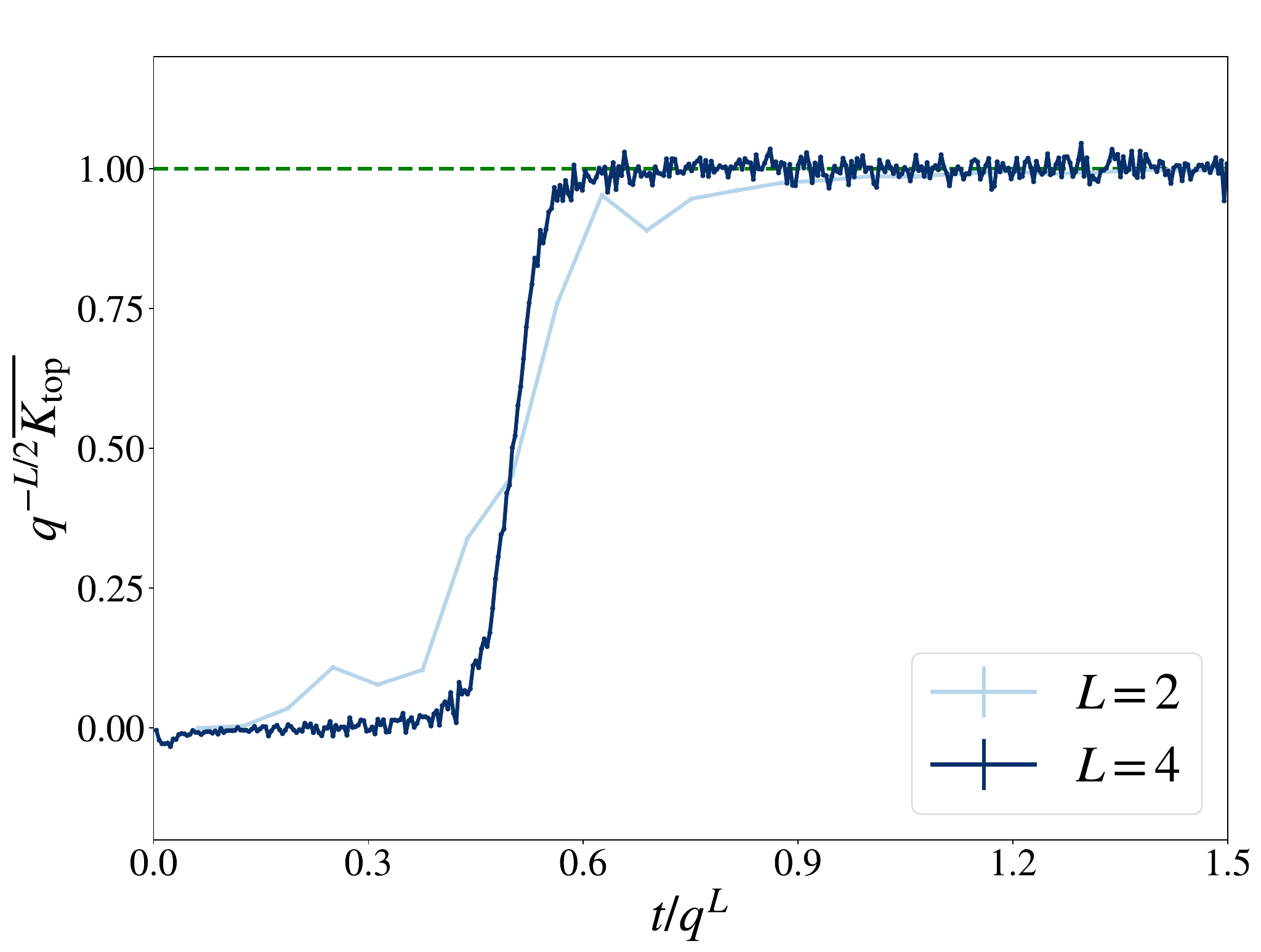}
     \caption{
\textbf{TopSFF $\overline{K_{\mathrm{top}}}$ against time at finite $q$} for the parity symmetric RPM~\eqref{app_eq:par_rpm_floq} at variance $\epsilon=1$, obtained by sampling for $q=2$ (left), $q=3$ (middle), and $q=4$ (right).
We plot the rescaled quantity $q^{-L/2}\overline{K_{\mathrm{top}}}$ against $t/q^L$ for several system sizes.
For $q=3,4$, the data show good agreement with the predicted long-time plateau and Heisenberg time scaling from Eq.~\eqref{app_eq:general_tsff_h_rmt}, namely $\overline{K_{\mathrm{top}}}\propto q^{L/2}$ for $t\gtrsim t_{\mathrm{Hei}}$, with $t_{\mathrm{Hei}}\propto q^L$.
However, the pre-Heisenberg regime for $q=3,4$ is not well described by the conventional SFF ramp-plateau behaviour. This regime is investigated in the main text and below.
For $q=2$, the many-body dynamics is non-chaotic, and the TopSFF approaches the plateau on a timescale shorter than the Heisenberg time.
}
    \label{fig:k_vs_t}
\end{figure}

 \begin{figure}
    \centering
  \includegraphics[width=0.33\textwidth]{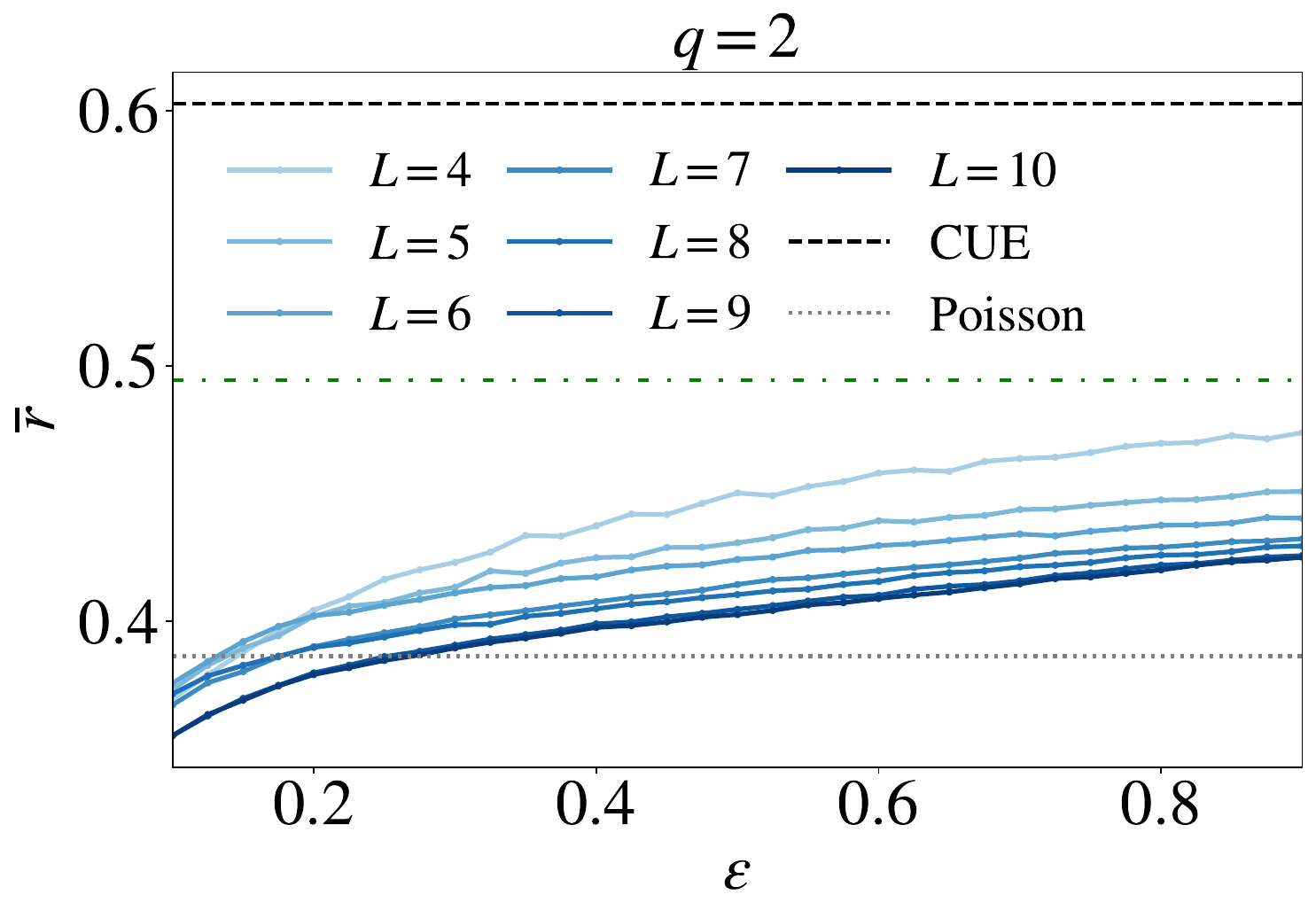}
  \vspace{0.5em}
  \includegraphics[width=0.33\textwidth]{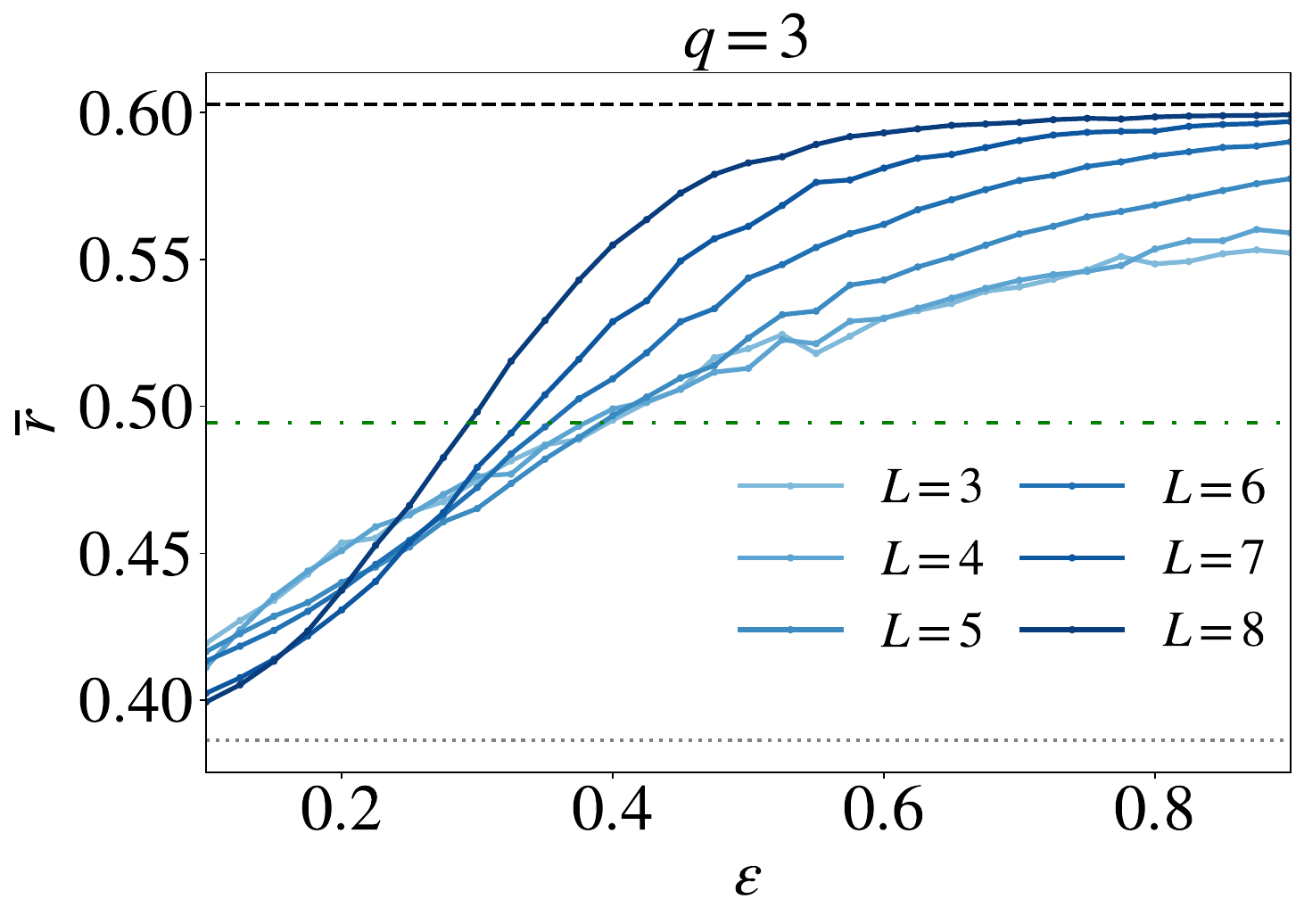}
  \vspace{0.5em}
  \includegraphics[width=0.33\linewidth]{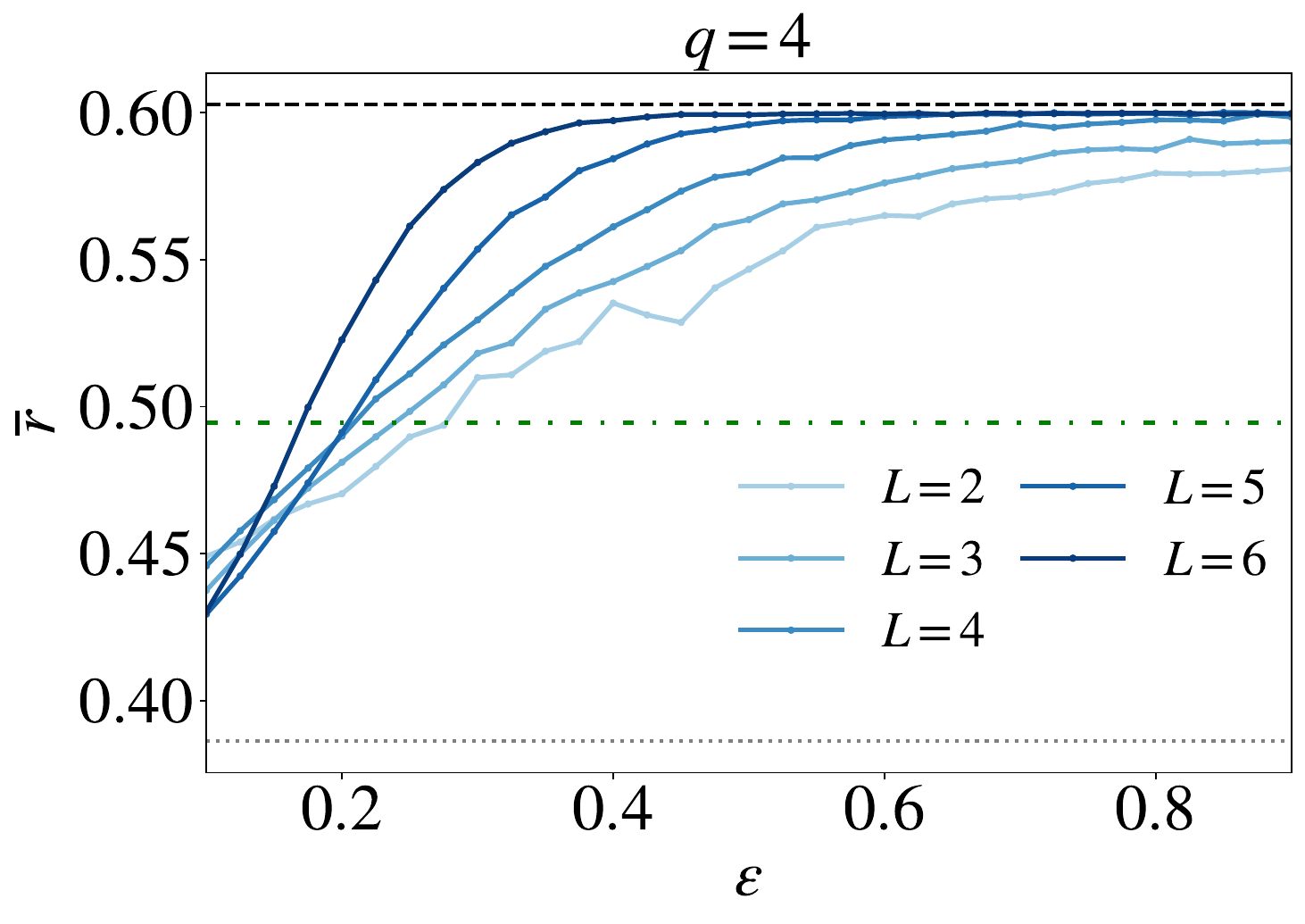}
  \includegraphics[width=0.33\linewidth]{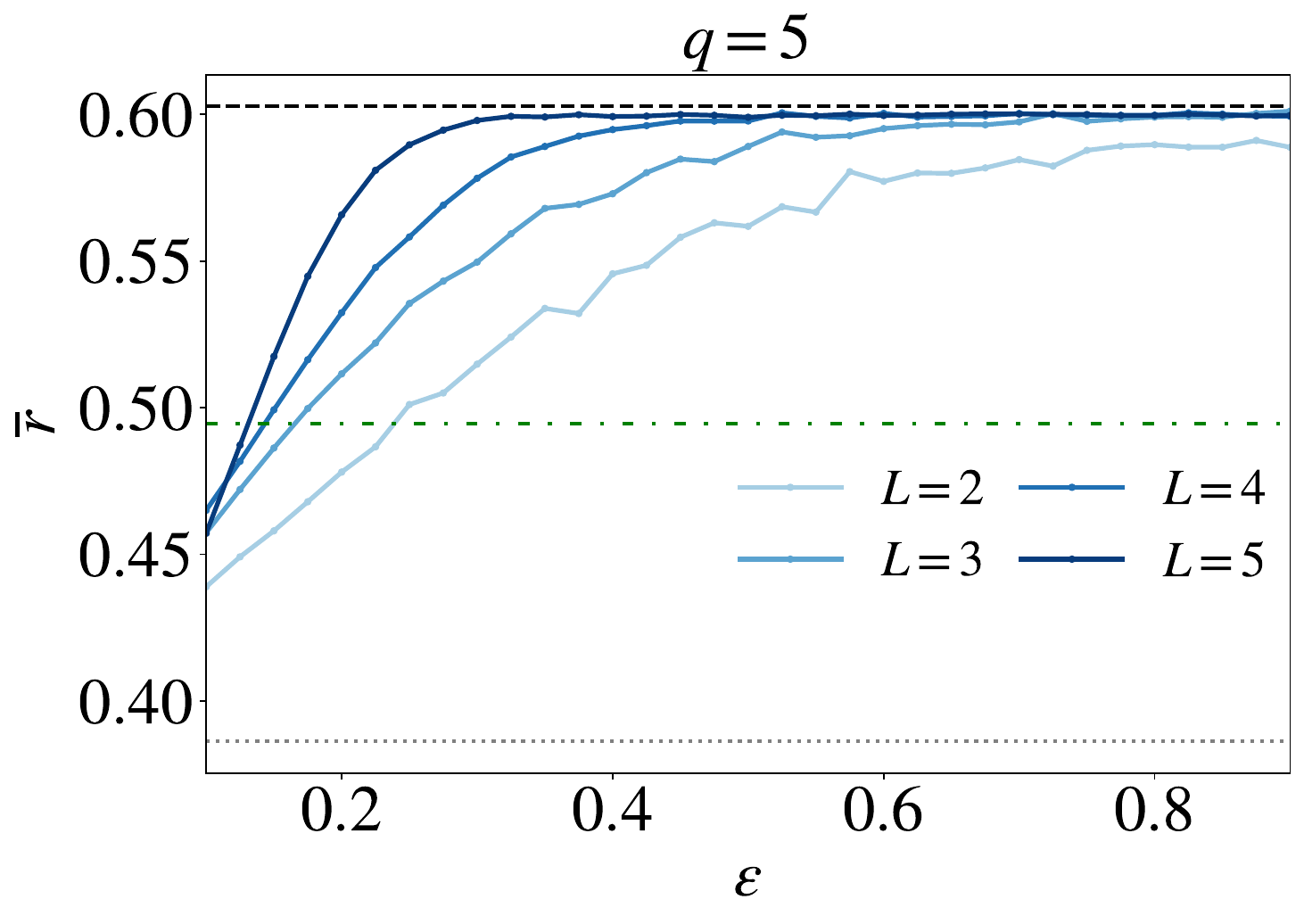}
     \caption{\textbf{Spacing ratio $\overline{r}$  against interaction strength $\epsilon$} for $q=2$ (top left), $q=3$ (top right), $q=4$ (bottom left), and $q=5$ (bottom right) for various $L$. The grey dotted, black dashed, and green dash-dotted lines indicate the $\overline{r}$ value for Poisson statistics ($2\ln2-1$), CUE statistics (0.5996) and the midpoint $\overline{r}^*$ between them (0.4929), respectively. 
     Data for increasing system sizes \(L\) are plotted with darker colours.
    We show only data from the \(+1\) parity sector, and the \(-1\) parity sector exhibits similar behaviour.
     The data show that, for fixed \(q\), the crossover shifts to smaller \(\epsilon\) as \(L\) increases. 
     Defining the crossover interaction strength \(\epsilon^*\) via 
     \(\overline r(\epsilon^*)=\overline r^*\),
     we observe that the extracted \(\epsilon^*\) decreases with increasing \(q\).
     } 
    \label{fig:r_ratio_finite_q_L}
\end{figure}

\begin{figure}
    \centering
    \includegraphics[width=0.33\linewidth]{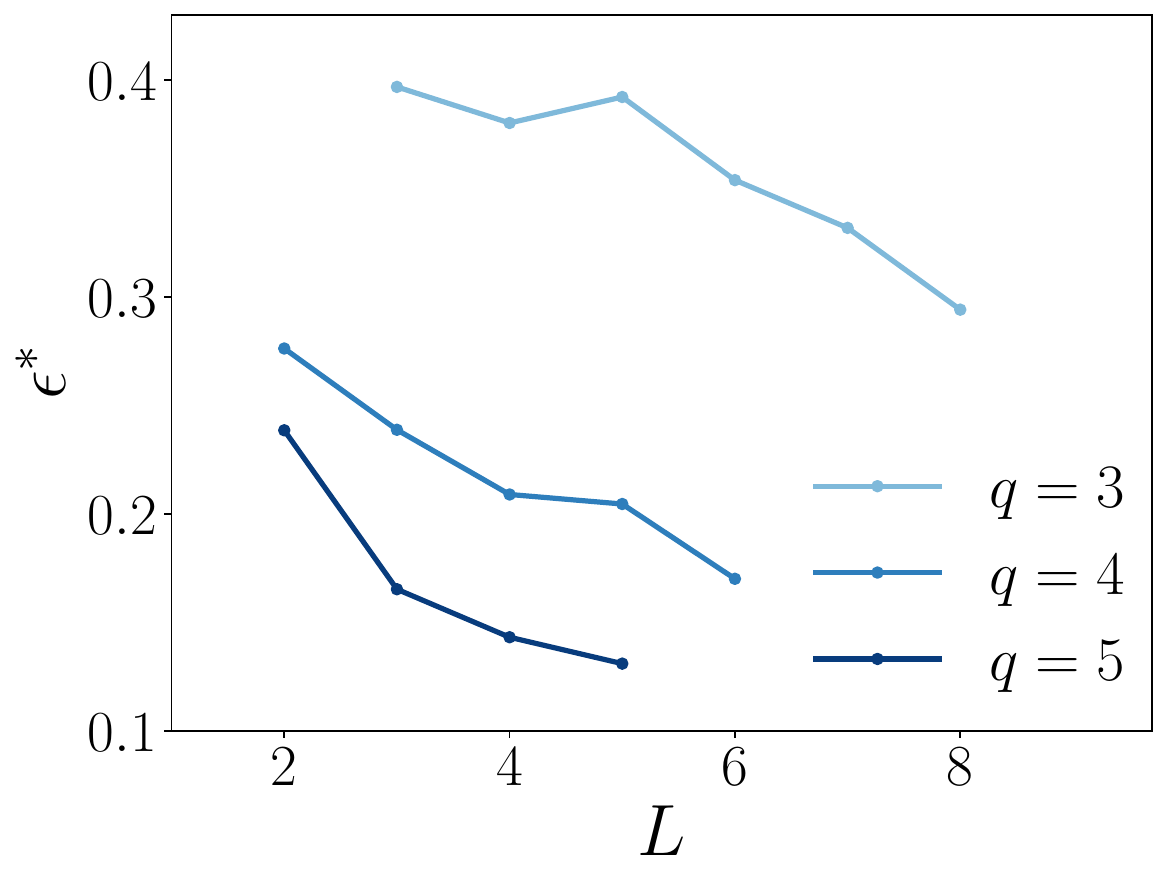}
    \caption{ 
    \textbf{Crossover interaction strength $\epsilon^*$ against system size $L$}. \(\epsilon^*\) is defined by 
     \(\overline r(\epsilon^*)=\overline r^*\) (see Fig.~\ref{fig:r_ratio_finite_q_L}). Data for increasing \(q\) are plotted with darker colours.
    The decrease of \(\epsilon^\ast\) for both increasing \(L\) and \(q\) indicates that the non-chaotic regime shrinks as either the system size or the local Hilbert-space dimension is increased.} 
    \label{fig:eps_star_parity}
\end{figure}

Figure~\ref{fig:k_vs_t} shows the time dependence of the TopSFF for the parity-symmetric RPM. For \(q=3,4\), at late times, the data approach the plateau predicted \(\overline{K_{\mathrm{top}}}\propto q^{L/2}\) by Eq.~\eqref{app_eq:general_tsff_h_rmt}, consistent with the many-body quantum system being described by a zero-dimensional RMT after the Heisenberg time \(t_{\mathrm{Hei}}\propto q^L\). However, for \(q=3,4\), unlike the conventional SFF, the TopSFF does not show a clear ramp-plateau behaviour in the pre-Heisenberg regime. The TopSFF behaviour in the pre-Heisenberg regime is investigated in the main text and below. For $q=2$, the many-body dynamics is non-chaotic, and the TopSFF approaches the plateau on a timescale shorter than the Heisenberg time.

As a complementary spectral diagnostic to the SFF, we also compute the spacing ratio $\overline{r}$~\cite{OganesyanHuse}
\begin{align}
    r_n:=\frac{\min(\delta_n,\delta_{n+1})}{\max(\delta_n,\delta_{n+1})},
\qquad
\overline r = \overline{\langle r_n\rangle_n},
\end{align}
where \(\langle\cdots\rangle_n\) denotes an average over quasienergy levels labelled by \(n\), and
the overline denotes an ensemble average. The gaps are
\(\delta_n=\theta_{n+1}-\theta_n\), where the quasienergies
\(\{\theta_n\}\) of the Floquet operator are ordered increasingly on the unit circle.
The results for the parity symmetric RPM model are displayed in Fig.~\ref{fig:r_ratio_finite_q_L}. For all \(q\)'s, the $\overline{r}$ increases from the Poisson value as $\epsilon$, $q$, or $L$ increases, indicating a crossover towards more chaotic dynamics. For \(q=2\), the parity symmetric RPM is non-chaotic, consistent with previous studies of the RPM without symmetries~\cite{chan2018spectral}. 
For \(q\geq 3\), \(\overline r\) crosses over from the Poisson value towards the Gaussian Unitary Ensemble (GUE) or Circular Unitary Ensemble (CUE) values as \(\epsilon\) increases. 
We quantify this crossover by \(\epsilon^*(q,L)\), defined as the value of \(\epsilon\) at which \( \overline{r}(\epsilon^*) \) reaches the midpoint between the Poisson and CUE values. 
The \(L\)-dependence of \(\epsilon^*(q,L)\), plotted in Fig.~\ref{fig:eps_star_parity}, shows that the crossover shifts to smaller \(\epsilon\) as \(L\) increases.
This spectral diagnostic should be contrasted with the CH stability map in Fig.~\ref{fig:CH_validity}, where the effective two-mode TopSFF description is robust only in the strongly interacting chaotic regime for \(q\geq 4\).

As an illustrative example, in Fig.~\ref{fig:floq_par_tsff_t2}, we compute the length dependence of momentum-unresolved TopSFF for period-one driving along the parity-inversion axes and the momentum-resolved TopSFF in the \(k=(0,0,0)\) sector for $t=2$, \(q=2,3\). For period-one driving, the dominant contribution comes from the \(k=(0,0,0)\) sector, consistent with the large-\(q\) analysis. 
However, the quantitative test described below [Fig.~\ref{fig:CH_validity}] shows that neither \(q=2\) nor \(q=3\) is described by the large-\(q\) $2\times2$ transfer matrix TopSFF theory. This description applies only for \(q\geq 4\) for accessable even \(t\), and we focus on regimes where this large-$q$ theory applies. While the deviation in the \(q=2\) RPM may be related to the absence of chaos, as previously noted in Ref.~\cite{chan2018spectral}, the corresponding \(q=3\) behaviour instead reflects a limitation of the two-mode description. We find (not shown) that, for \(q=3\), several eigenvalues of the exact transfer matrix have comparable moduli, so that a higher-mode description is required. In contrast, in the regime where the large-\(q\) description applies, the two leading eigenvalues are well separated from the rest and dominate the TopSFF.

\begin{figure}
    \centering 
    \includegraphics[width=0.24\linewidth]{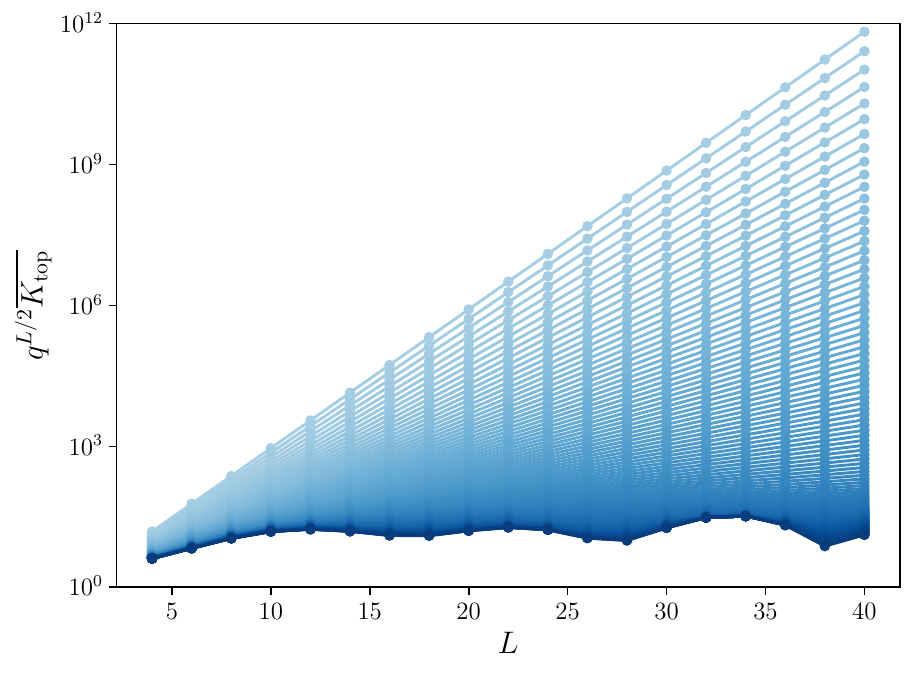}
    \hfill
    \includegraphics[width=0.24\linewidth]{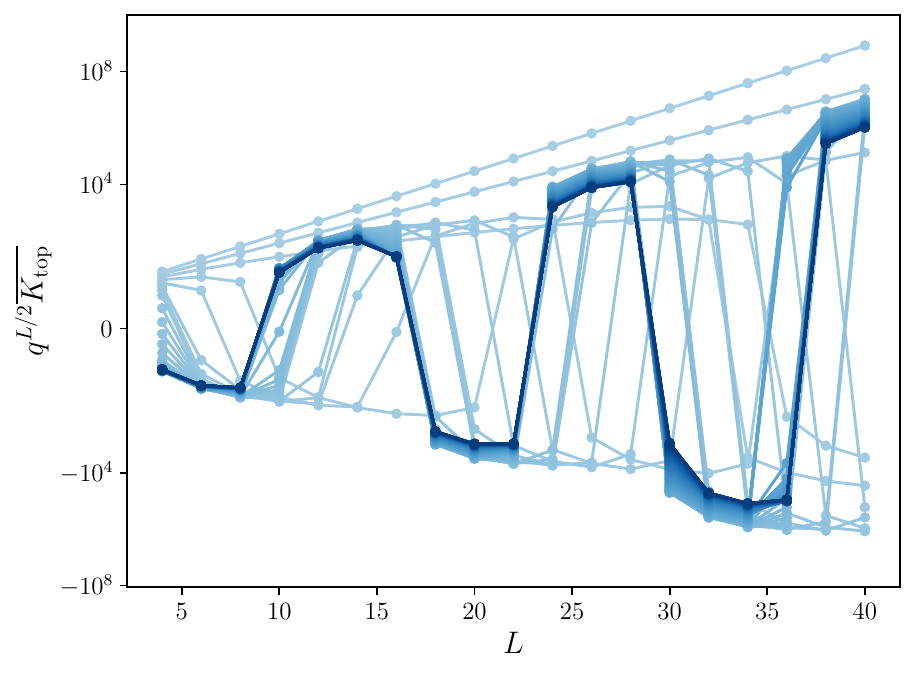}
    \hfill
    \includegraphics[width=0.24\linewidth]{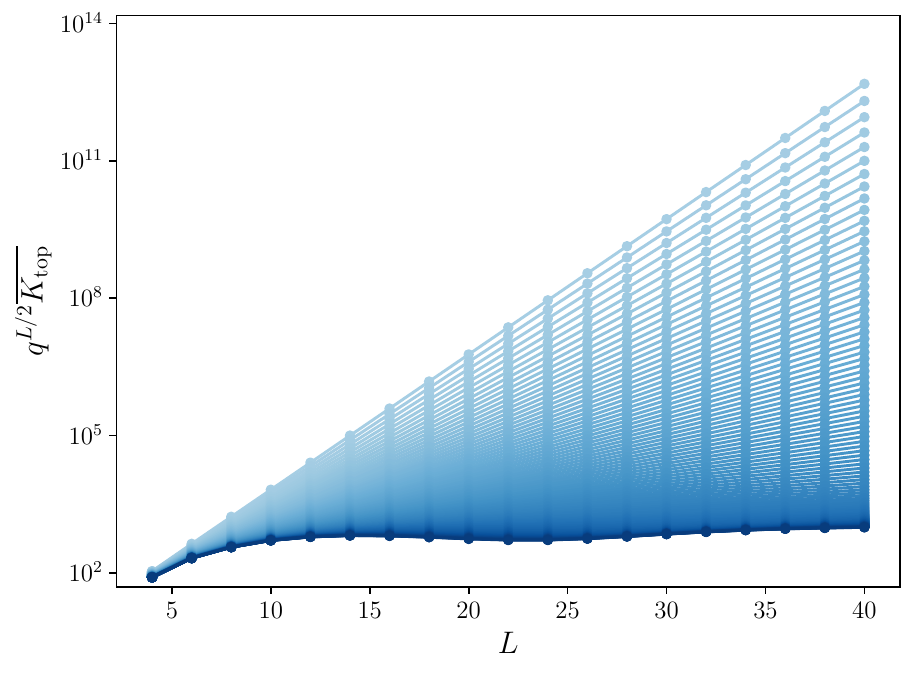}
    \hfill
    \includegraphics[width=0.24\linewidth]{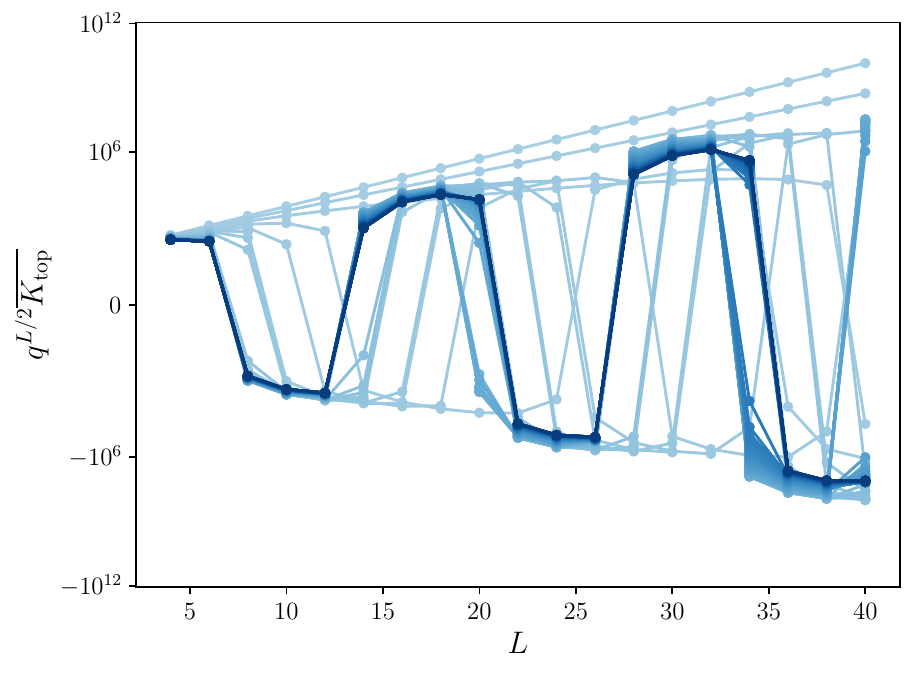}
    \caption{\textbf{Exact TopSFF at \(t=2\) for finite \(q\).}
From left to right: momentum-unresolved TopSFF with period-1 driving along parity inversion axes for \(q=2\) and \(q=3\), and momentum-resolved TopSFF in the \(k=(0,0,0)\) sector for \(q=2\) and \(q=3\). We plot \(q^{L/2}\overline{K_{\mathrm{top}}}\) for 100 values of \(\epsilon\) ranging from \(0.01\) (lighter curves) to \(2\) (darker curves). Note that a test described below [Fig.~\ref{fig:CH_validity}] shows that the large-\(q\) description of the TopSFF applies only for \(q\geq 4\).
}
    \label{fig:floq_par_tsff_t2}
\end{figure}

\begin{figure}
    \centering
    \includegraphics[width=\linewidth]{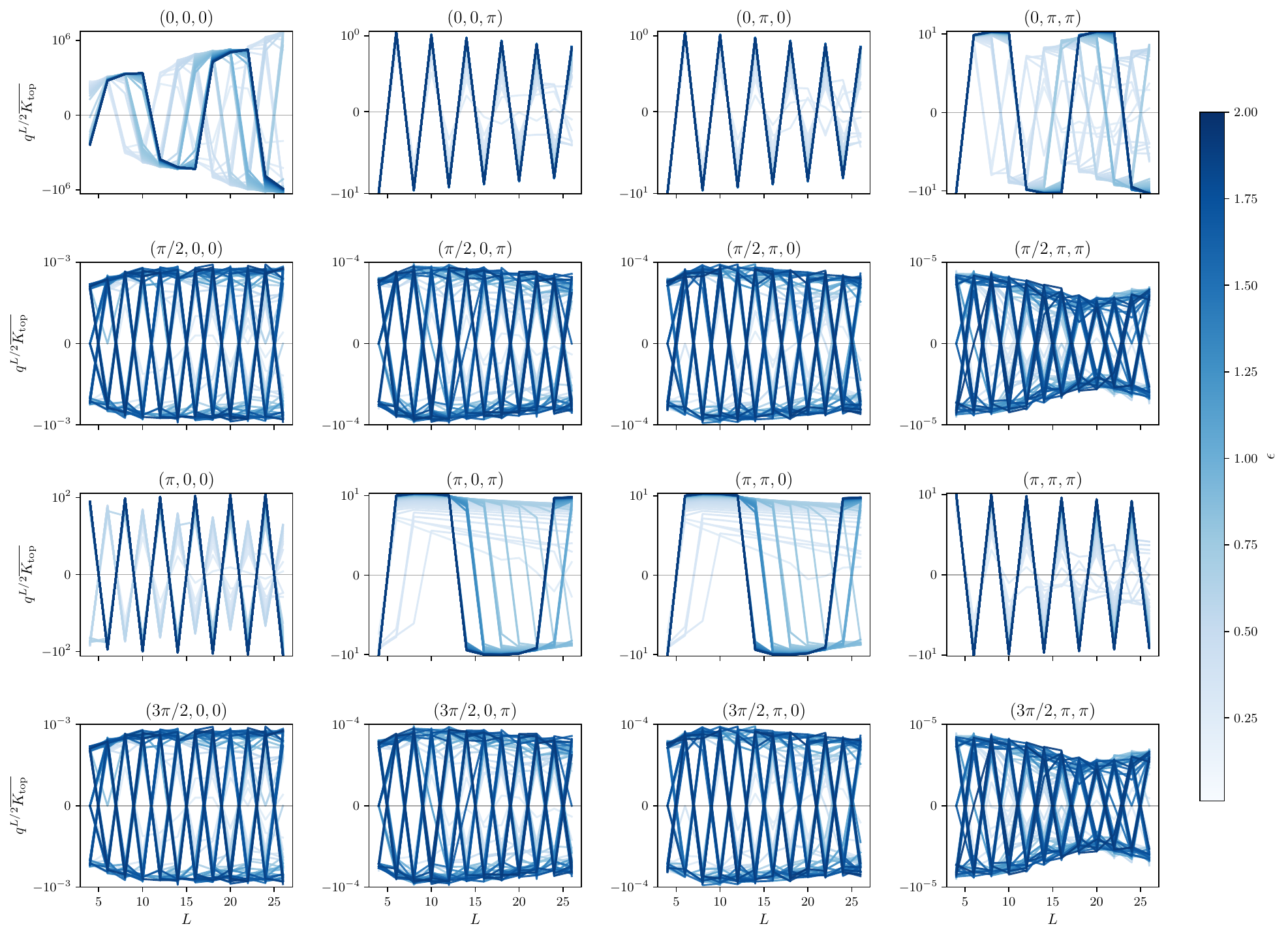}
    \caption{\textbf{Exact TopSFF at $t=2$ and $q=4$ for all momentum sectors.} Consistent with the  large-$q$ analysis in Fig.~\ref{fig:grid_all_sectors}, the data exhibit the expected symmetries under the exchange $k_c \leftrightarrow k_d$ and under momentum inversion $k \leftrightarrow -k$: in the figure, column 2 is identical to column 3, and row 2 is identical to row 4. In particular, as predicted in Fig.~\ref{fig:grid_all_sectors}, the sectors $(\pi,0,0)$ and $(\pi,\pi,\pi)$ are governed by a negative real eigenvalue, leading to sign-flipping behaviour as a function of $L$. The sector $(0,0,0)$ displays oscillatory behaviour, consistent with the dominant eigenvalues forming a complex conjugate pair. The sectors exhibiting a $\mathcal{PT}$ transition, such as $(\pi,\pi,0)$, are analysed in the main text.
  }
    \label{fig:t2_q4_allmom}
\end{figure}

To draw further parallels with the large-$q$ theory, Fig.~\ref{fig:t2_q4_allmom} provides an extensive sector-by-sector check of the finite-\(q\) TopSFF at \(t=2\) and \(q=4\) where the 2-mode description applies. We sweep all momentum sectors \(k=(k_b,k_c,k_d)\), and find that the exact finite-\(q\) data obey the symmetry constraints predicted by the large-\(q\) generalized Boltzmann factor, including invariance under \(k_c\leftrightarrow k_d\) and momentum inversion \(k\leftrightarrow -k\).
The qualitative sector dependence in finite-\(q\) in Fig.~\ref{fig:t2_q4_allmom} is also consistent with the large-\(q\) prediction in Fig.~\ref{fig:grid_all_sectors}: symmetry-related sectors coincide, while sectors governed by real negative or complex-conjugate leading eigenvalues show the corresponding monotonic/sign-changing or oscillatory dependence on \(L\). 
In the following, we focus on the \(k=(\pi,\pi,0)\) sector. As discussed in the main text, this sector is singled out by the large-\(q\) analysis because its leading eigenvalues can have moduli exceeding unity for sufficiently large \(q\), giving exponential growth with system size, while also exhibiting a \(\mathcal{PT}\) transition at finite interaction strength. It therefore provides the cleanest setting to diagnose the emergent \(\PT\) transition of TopSFF.

In Fig.~\ref{fig:parity_tsff_classification_t234}, we examine the finite-\(q\) dynamics in the momentum sector \(k=(\pi,\pi,0)\), showing the exact TopSFF at \(t=2,3,4\) for two different values of $q$, while sweeping \(\epsilon\).
For \(\epsilon<\epsilon_{\mathrm{EP}}\), the data are non-oscillatory at sufficiently large \(L\), with either exponential decay or growth, consistent with the large-\(q\) prediction. This can be seen in the first row for $t=2$, in which $q=4$ (left column) are exponentially decaying compared to the exponential growth for $q=30$.
This is consistent with the \(\mathcal{PT}\) unbroken phase. For \(\epsilon>\epsilon_{\mathrm{EP}}\), the TopSFF exhibits oscillations in \(L\) with an exponential envelope, consistent with the \(\mathcal{PT}\) broken phase. 
Near \(\epsilon=\epsilon_{\mathrm{EP}}\), as discussed in the main text, the unresolved TopSFF does not directly expose the Jordan-block contribution, because the corresponding prefactor vanishes for the physical boundary vector. Below, we first make this qualitative finite-\(q\) signature quantitative by extracting \(\epsilon_{\mathrm{EP}}\), together with effective traces, determinants and eigenvalues, using the Cayley-Hamilton relation. We then resolve the TopSFF into Gaussian and non-Gaussian boundary sectors to diagnose the polynomial-in-\(L\) scaling induced by Jordan non-diagonality.

\begin{figure}
    \centering

    \includegraphics[width=0.24\textwidth]{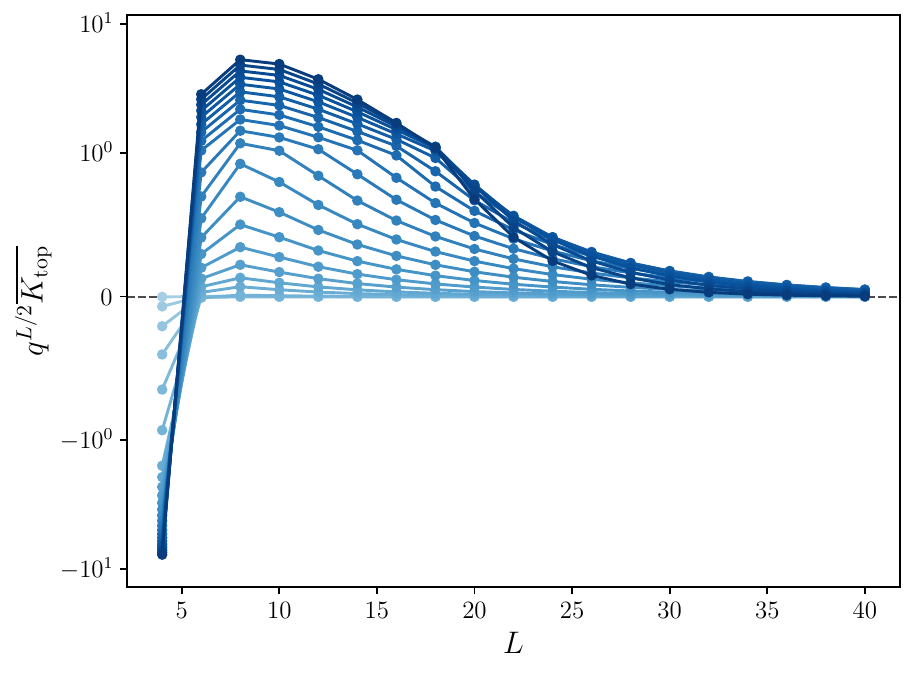}
    \hfill
    \includegraphics[width=0.24\textwidth]{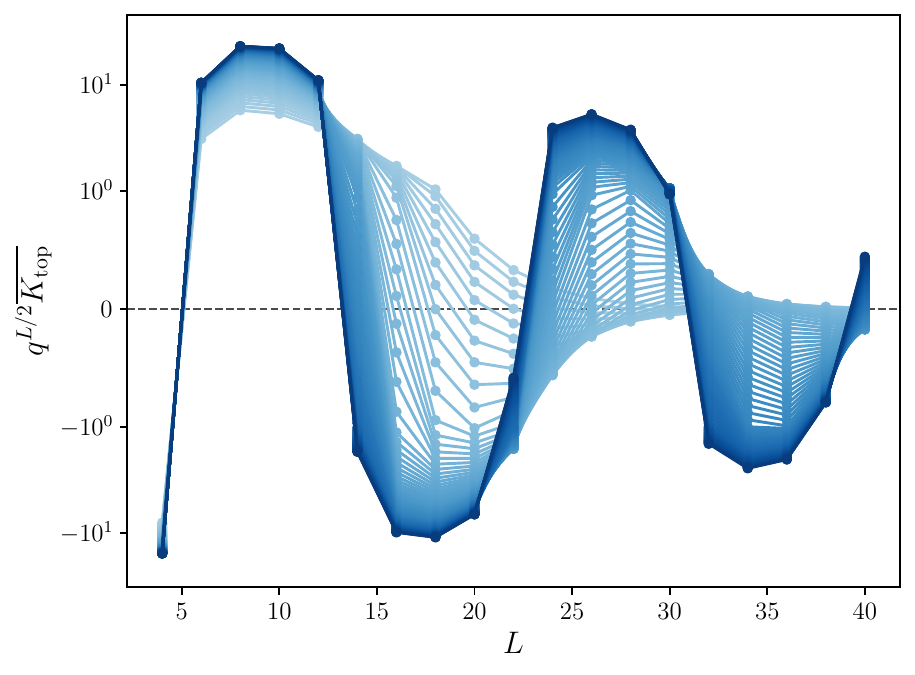}
    \hfill
    \includegraphics[width=0.24\textwidth]{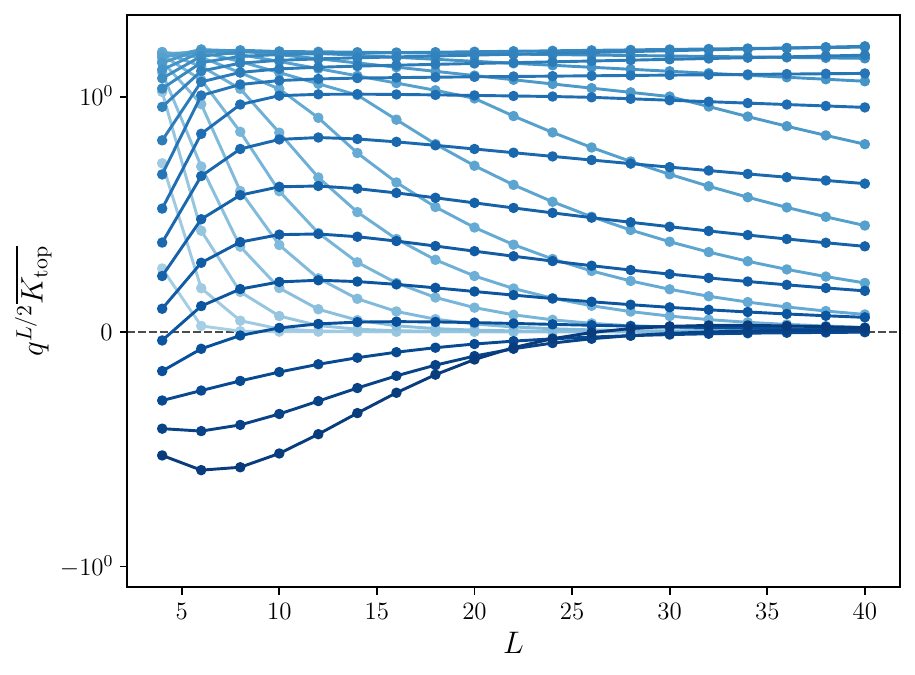}
    \hfill
    \includegraphics[width=0.24\textwidth]{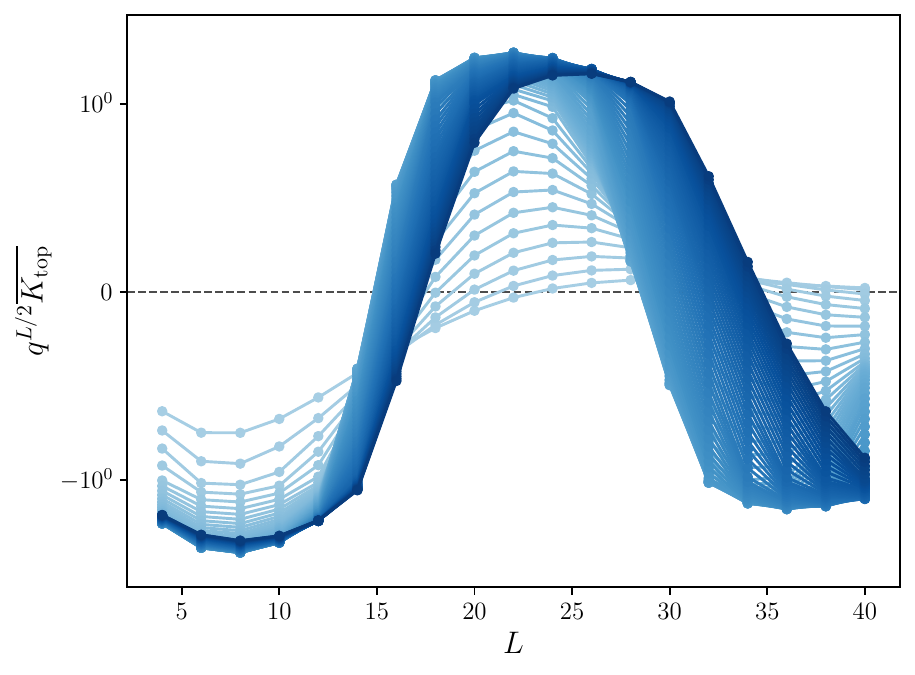}

    \vspace{0.25em}
    
    \centering
    \includegraphics[width=0.24\linewidth]{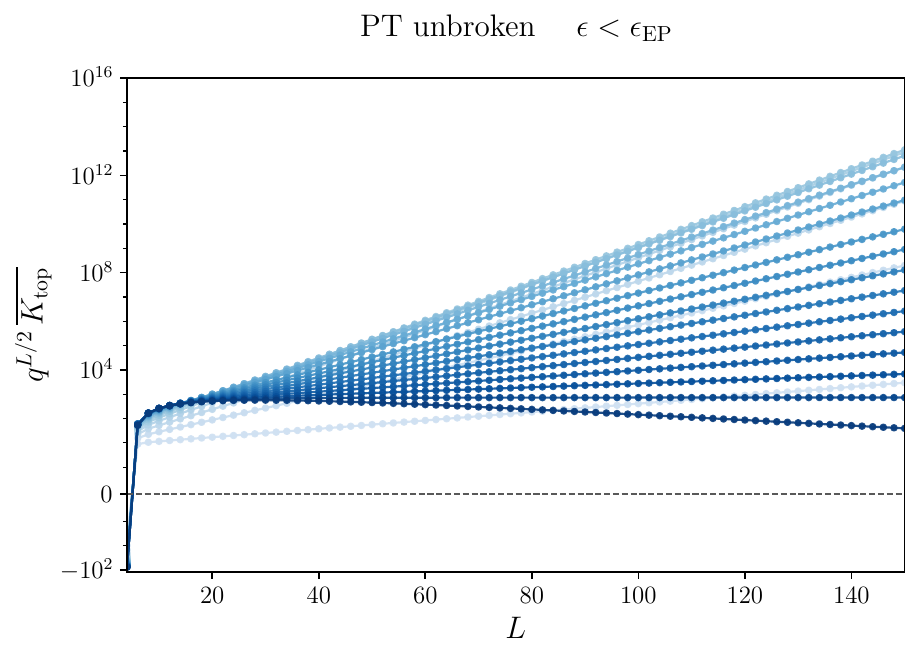}
    \hfill
    \includegraphics[width=0.24\linewidth]{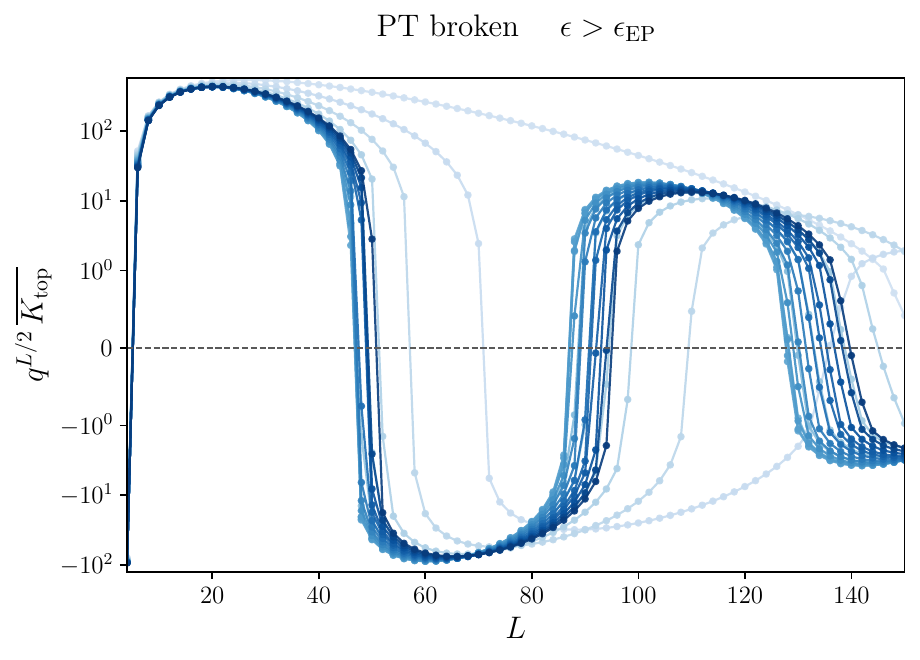}
    \hfill
    \includegraphics[width=0.24\linewidth]{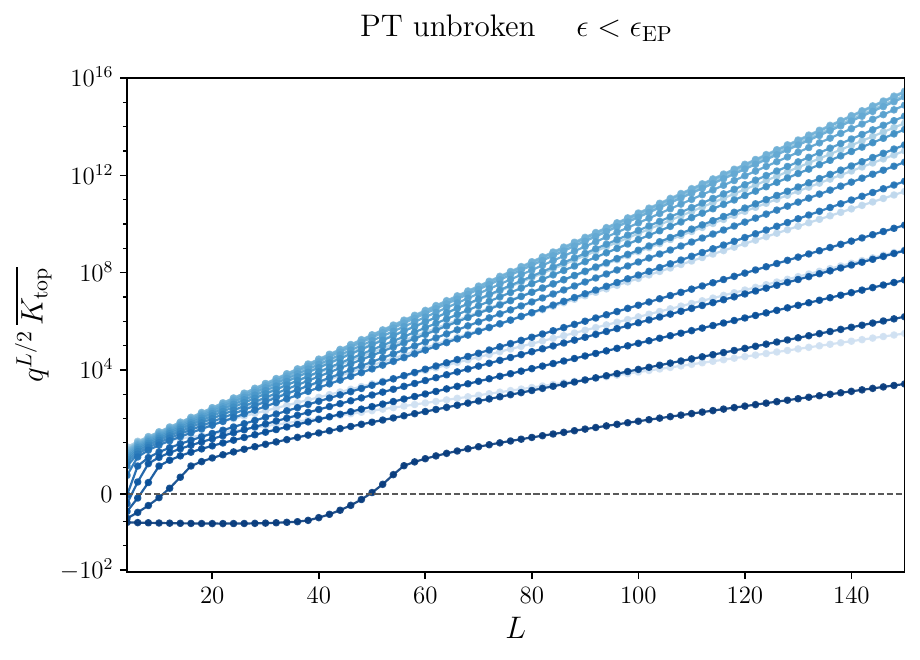}
    \hfill
    \includegraphics[width=0.24\linewidth]{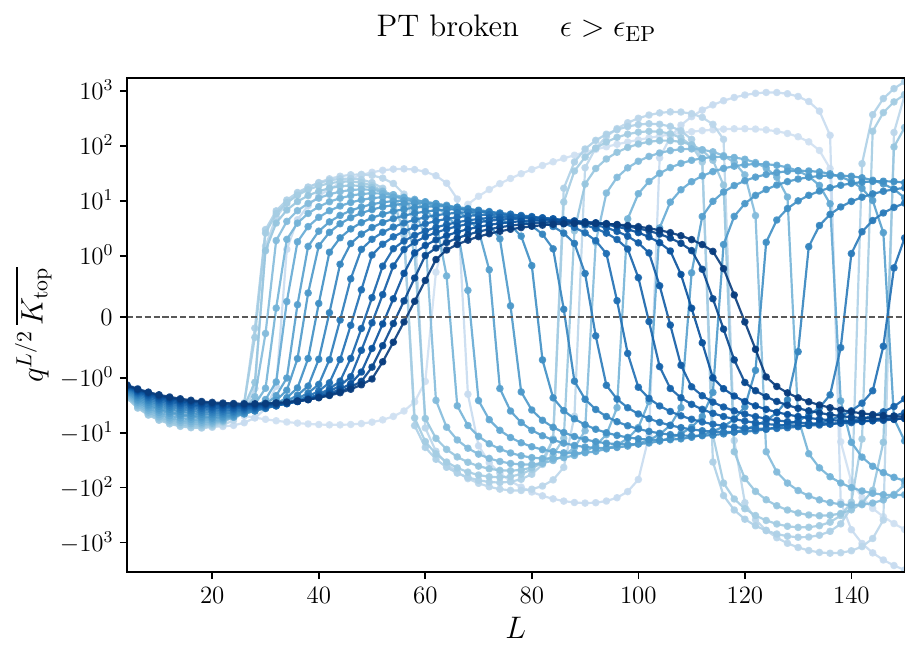}

    \vspace{0.25em}
    
    \centering
    \includegraphics[width=0.24\linewidth]{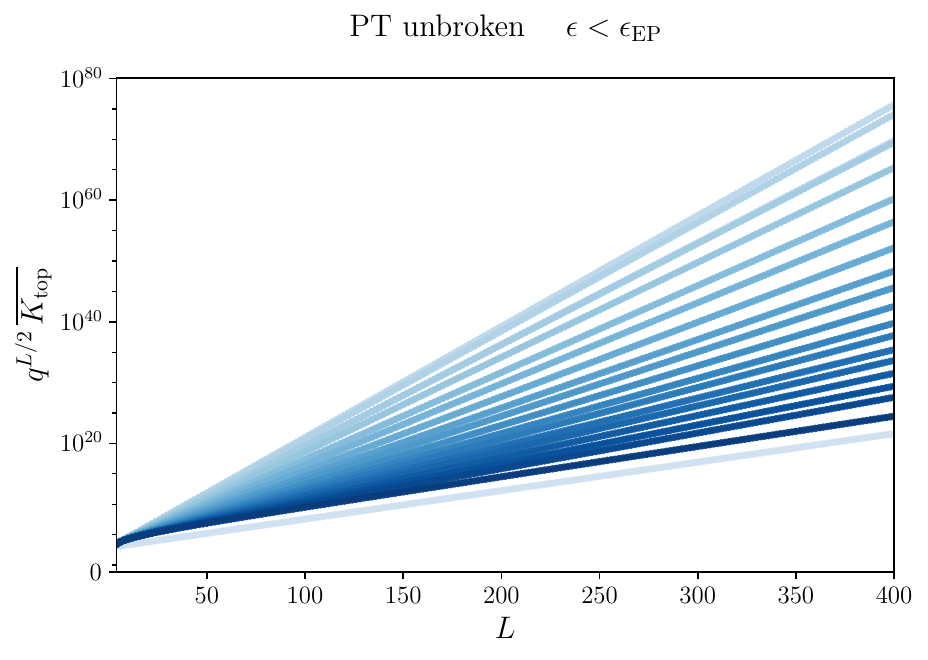}
    \hfill
    \includegraphics[width=0.24\linewidth]{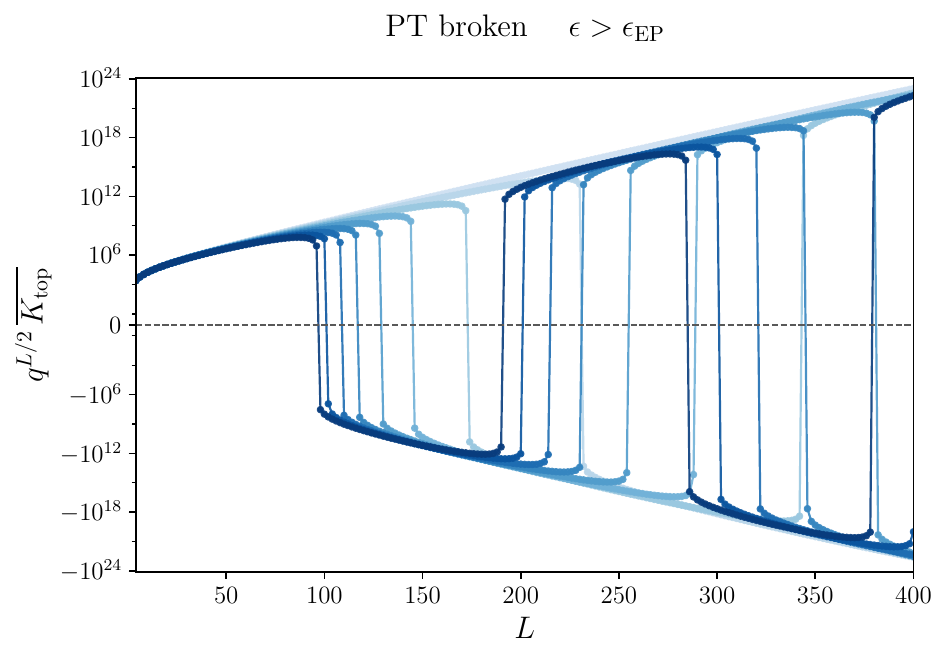}
    \hfill
    \includegraphics[width=0.24\linewidth]{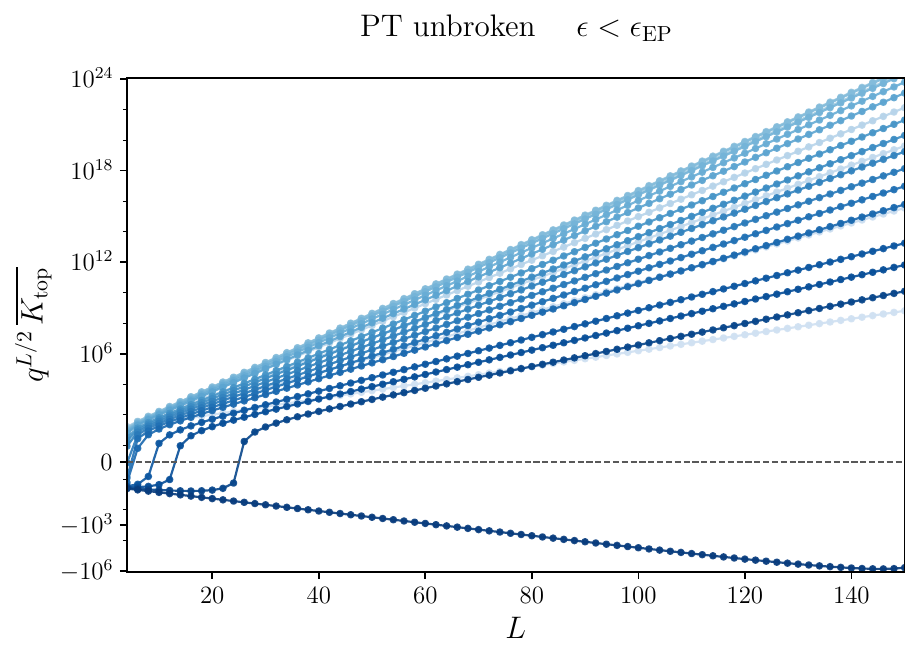}
    \hfill
    \includegraphics[width=0.24\linewidth]{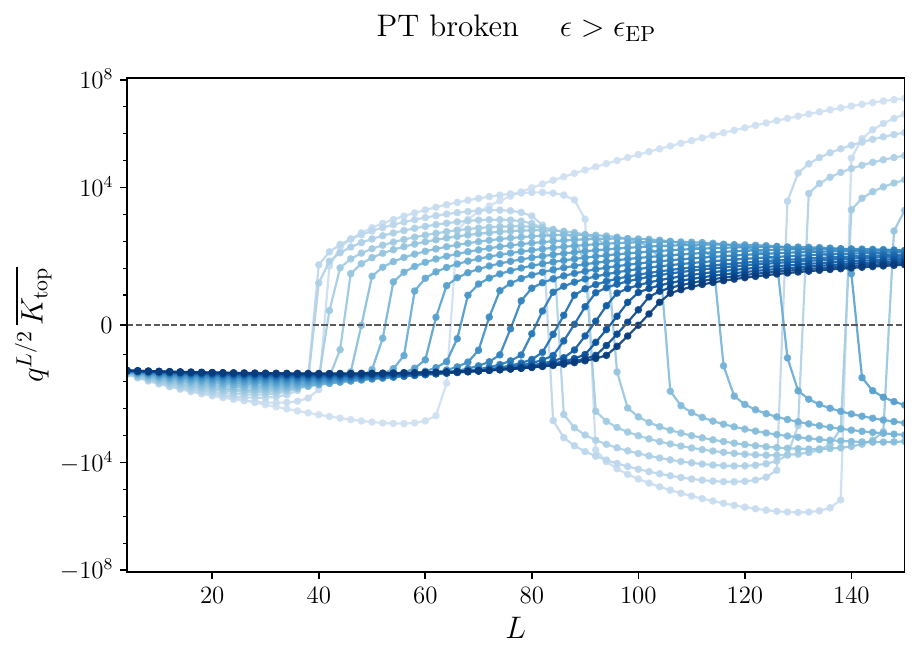}
    \caption{
\textbf{Exact finite-$q$ TopSFF and $\mathcal{PT}$ transitions for $t=2,3,4$.}
We plot $q^{L/2}\overline{\mathcal{K}_{\mathrm{top}}}$ as a function of $L$ for 100 values of $\epsilon$ in the range $0.01\leq\epsilon\leq2$.
The rows, from top to bottom, correspond to
$t=2$ in the $(\pi,\pi,0)$ sector,
$t=3$ in the $(2\pi/t,-2\pi/t,0)$ sector,
and $t=4$ in the $(\pi,\pi,0)$ sector.
For each $t$, two values of $q$ are shown:
$q=4,20$ for $t=2$,
$q=6,30$ for $t=3$,
and $q=4,30$ for $t=4$.
Within each row, the first two panels correspond to the smaller $q$ and the last two to the larger $q$.
For each $q$, the left and right panels show non-oscillatory TopSFF behaviour in the $\mathcal{PT}$ unbroken phase with $\epsilon<\epsilon_{\mathrm{EP}}$, and oscillatory TopSFF behaviour in the $\mathcal{PT}$ broken phase with $\epsilon>\epsilon_{\mathrm{EP}}$ respectively.
Within each panel, lighter to darker blue denotes increasing $\epsilon$.
For finite $q$, the TopSFF can also exhibit exponential decay with $L$, as visible, for example, among the light blue curves for $q=4$ and $t=2$, consistent with the finite $q$ effective-eigenvalue analysis.
}
    \label{fig:parity_tsff_classification_t234}
\end{figure}

\subsubsection{Cayley-Hamilton extraction and stability}\label{App:CH_verify}

In the large-\(q\) expansion, TopSFF in each momentum sector is described by a
$2\times2$ transfer matrix. As discussed in App.~\ref{app_sec:cayley_hamilton},
the Cayley-Hamilton (CH) relation then allows one to extract the trace and
determinant of the transfer matrix, and hence its eigenvalues and signatures of
the \(\mathcal{PT}\) transition.
At finite \(q\), the TopSFF can be evaluated as
\(\overline{\mathcal{K}_{\mathrm{top}}}(\ell)
=
{\phi}^{\,T}\widetilde{\boltz}^{\ell}{\phi}\),
with the finite-\(q\) transfer matrix defined in
App.~\ref{subsec:finite_q_symbolic_tsff}, and with \(\ell \equiv \Leff\). 
After transforming to momentum space, \(\widetilde{\boltz}\) is block diagonal, and the exact quotient reduction removes redundant permutation degrees of freedom within each block. Nevertheless, the resulting momentum-sector transfer matrices generally remain higher than two dimensional. 
Using hints from the large-$q$ results, we conjecture, and verify numerically, that at finite \(q\) each decoupled momentum sector may be described by an effective $2\times2$ transfer matrix, characterised by the two-dimensional Cayley-Hamilton relation,
\(\widetilde{\boltz}^{2}
=
\widehat{\Trb}\,\widetilde{\boltz}
-
\widehat{\Detb} \,\iden\),
where \(\Trb:=\tr\widetilde{\boltz}\) and
\(\Detb:=\det\widetilde{\boltz}\).
Multiplying this identity by \({\phi}^{T}\widetilde{\boltz}^{\ell}\)
from the left and by \({\phi}\) from the right gives the 
recurrence relation
\begin{equation}
    \overline{\mathcal{K}_{\mathrm{top}}}(\ell+2)
    =
    \widehat{\Trb}\, \overline{\mathcal{K}_{\mathrm{top}}}(\ell+1)
    -
    \widehat{\Detb}\, \overline{\mathcal{K}_{\mathrm{top}}}(\ell).
    \label{eq:CH_recurrence_TopSFF}
\end{equation}
Thus, from four consecutive TopSFF values, one can locally estimate the effective trace and
determinant as
\begin{equation}
    \widehat{\Trb}_{\ell}
    :=
    \frac{
    \overline{\mathcal{K}_{\mathrm{top}}}(\ell+1)\overline{\mathcal{K}_{\mathrm{top}}}(\ell+2)
    -
    \overline{\mathcal{K}_{\mathrm{top}}}(\ell)\overline{\mathcal{K}_{\mathrm{top}}}(\ell+3)
    }{
    \overline{\mathcal{K}_{\mathrm{top}}}(\ell+1)^{2}
    -
    \overline{\mathcal{K}_{\mathrm{top}}}(\ell)\overline{\mathcal{K}_{\mathrm{top}}}(\ell+2)
    },
\qquad 
    \widehat{\Detb}_{\ell}
    :=
    \frac{
    \overline{\mathcal{K}_{\mathrm{top}}}(\ell+2)^{2}
    -
    \overline{\mathcal{K}_{\mathrm{top}}}(\ell+1)\overline{\mathcal{K}_{\mathrm{top}}}(\ell+3)
    }{
    \overline{\mathcal{K}_{\mathrm{top}}}(\ell+1)^{2}
    -
    \overline{\mathcal{K}_{\mathrm{top}}}(\ell)\overline{\mathcal{K}_{\mathrm{top}}}(\ell+2)
    },
    \label{eq:CH_det_extraction}
\end{equation}
and we define the effective discriminant as
\begin{equation}
    \widehat{\Delta}_\ell(\epsilon)
    :=
    \widehat{\Trb}_\ell(\epsilon)^2
    - 4 \widehat{\Detb}_\ell(\epsilon) .
\end{equation}
At finite \(q\), the CH relation is overdetermined when applied over a range of \(\ell\). The extracted quantities \(\widehat{\Trb}_{\ell}\) and \(\widehat{\Detb}_{\ell}\) therefore acquire an \(\ell\)-dependence, and can be interpreted as effective trace and determinant only in regimes where this dependence is weak. 
\begin{figure}[ht]
    \centering
    \includegraphics[width=0.49\linewidth]{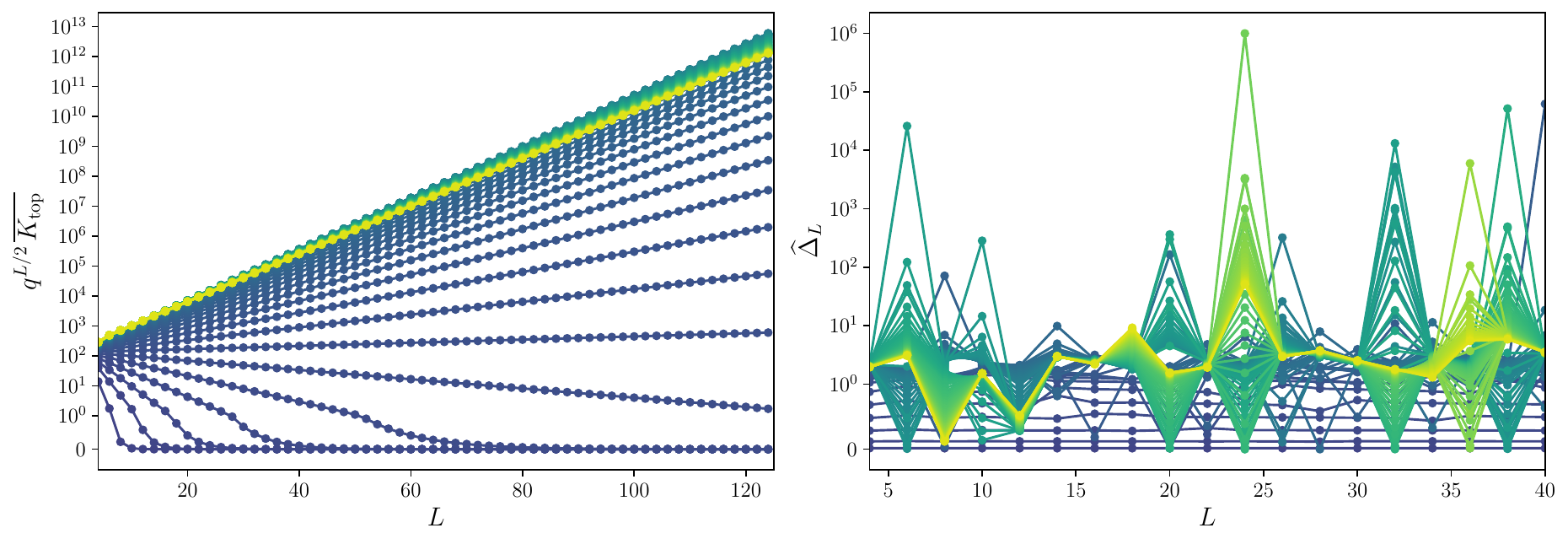}
    \includegraphics[width=0.49\linewidth]{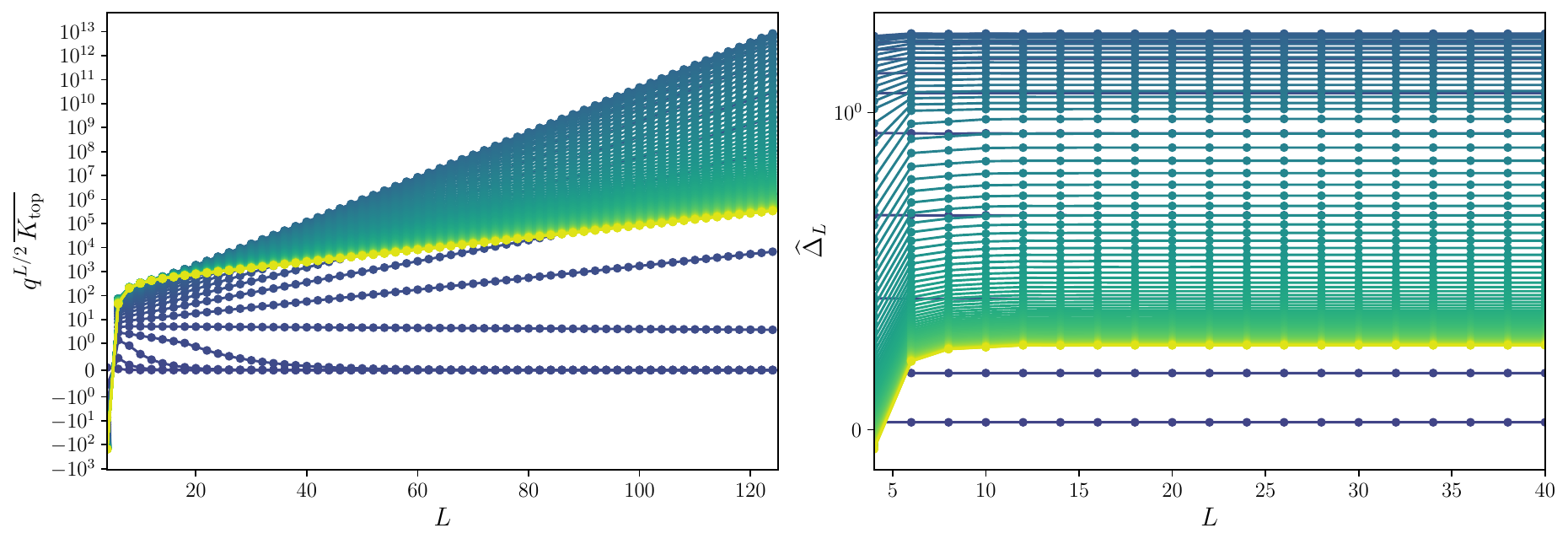}
    \vspace{0.5em}
     \includegraphics[width=0.49\linewidth]{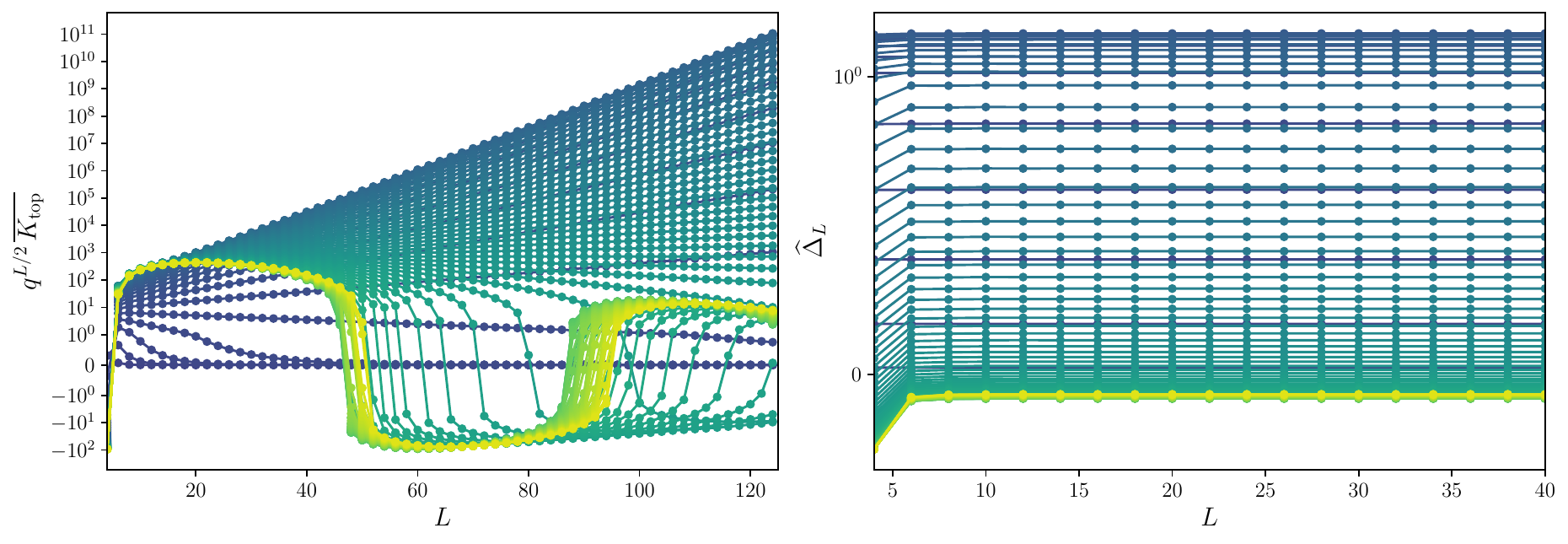}
    \vspace{0.5em}
     \includegraphics[width=0.49\linewidth]{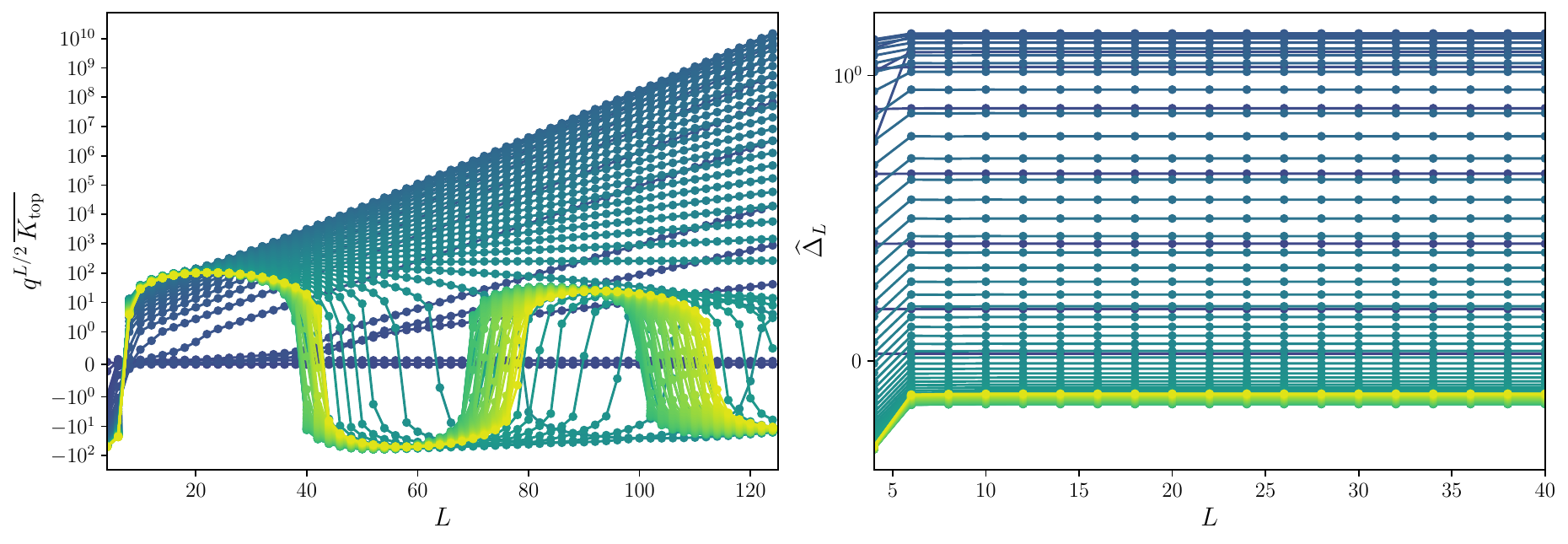}
    \vspace{0.5em}
    \includegraphics[width=0.49\linewidth]{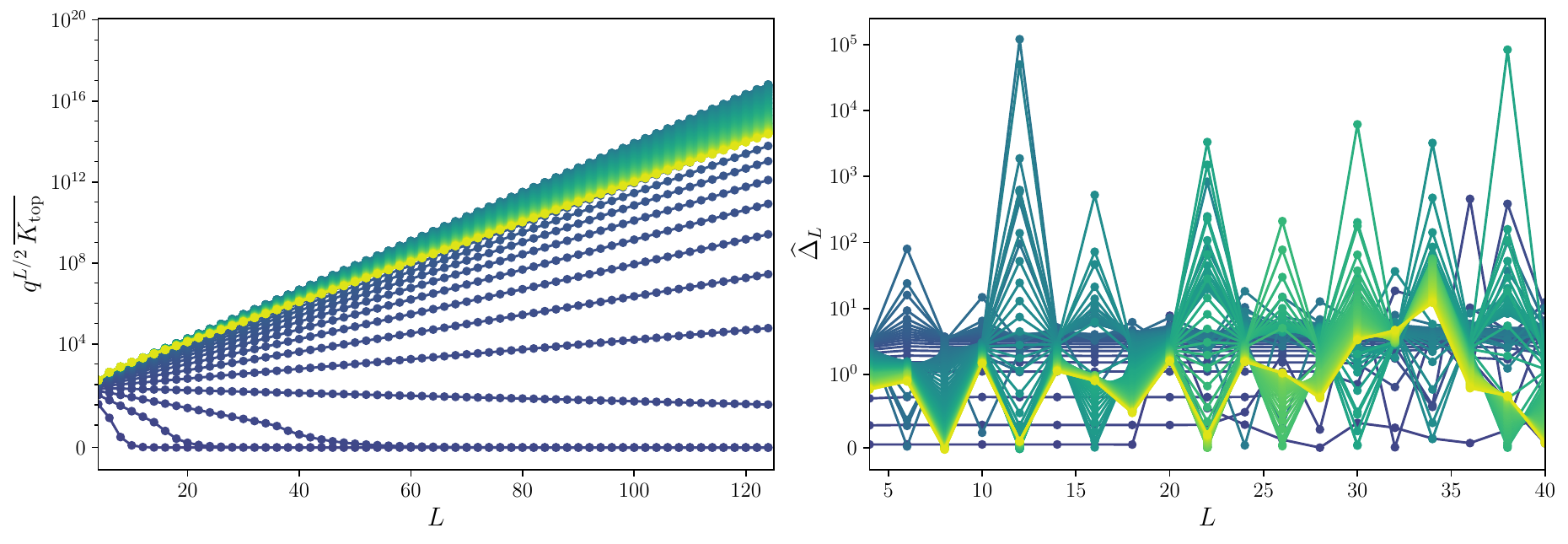}
    \vspace{0.5em}
    \includegraphics[width=0.49\linewidth]{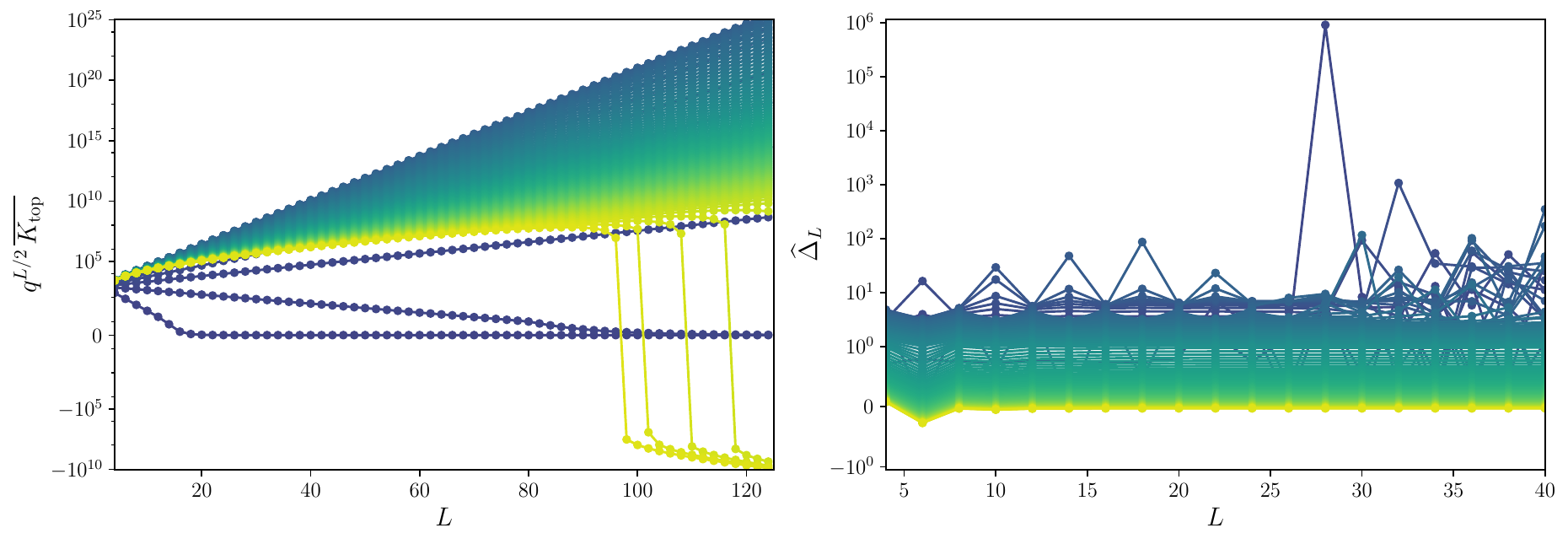}
    \vspace{0.5em}
    \includegraphics[width=0.49\linewidth]{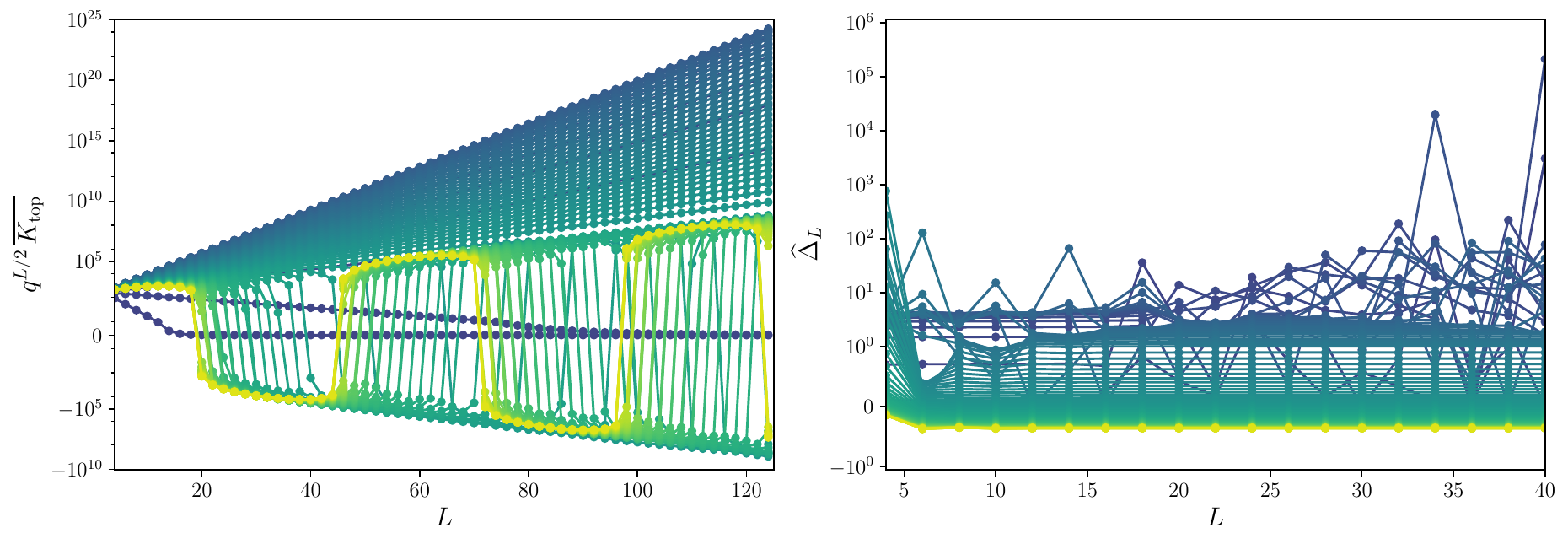}
    \vspace{0.5em}
    \includegraphics[width=0.49\linewidth]{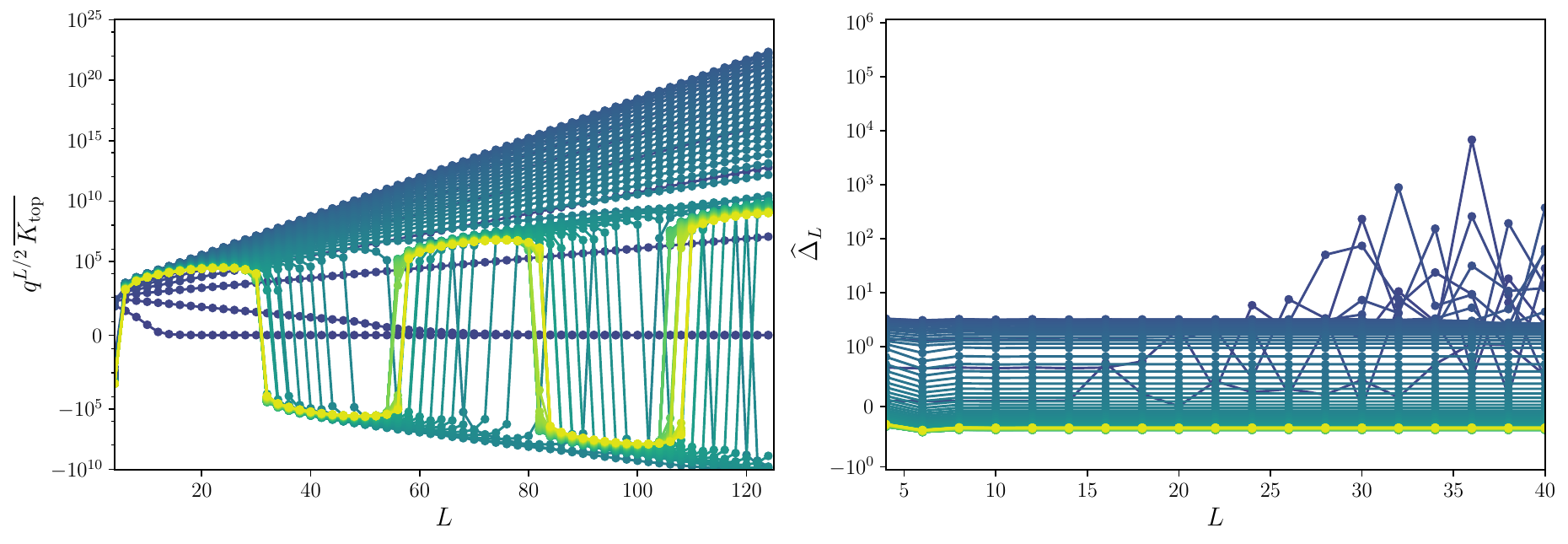}
 \caption{
\textbf{Finite-$q$ TopSFF data and Cayley-Hamilton extraction.}
We compare the TopSFF with the corresponding locally extracted discriminant
\(
\widehat{\Delta}_L=\widehat{t}_L^{\,2}-4\widehat{d}_L
\).
The top two rows show $t=3$ in the $(2\pi/t,-2\pi/t,0)$ sector for
$q=3,5,6,7$, while the bottom two rows show $t=4$ in the $(\pi,\pi,0)$ sector for
$q=3,4,5,6$.
For each $q$, the TopSFF is shown in the left panel and $\widehat{\Delta}_L$ in the adjacent right panel.
Within each set of two rows, $q$ increases from left to right and then from top to bottom.
Curves are coloured by interaction strength $\epsilon$, from blue at weaker interaction to yellow at stronger interaction.
The instability of $\widehat{\Delta}_L$ at small $q$, particularly for $q=3$, is consistent with the Cayley-Hamilton stability analysis in Fig.~\ref{fig:CH_validity}.
}
    \label{fig:sidebyside}
\end{figure}
These extractions of \(\widehat{\Trb}_{\ell}\) and \(\widehat{\Detb}_{\ell}\) are illustrated in Fig.~\ref{fig:sidebyside}, where we compare the exact TopSFF data with the locally extracted CH discriminant
\(
\widehat{\Delta}_{\ell}
\)
in the \(k=(\pi,\pi,0)\) sector. The left panels show \(\overline{K_{\mathrm{top}}}\) as a function of \(L\), while the right panels show the corresponding \(\widehat{\Delta}_{\ell}\) obtained from the local CH extraction. For even \(t\) with \(q\geq4\) and odd \(t\) with \(q \geq 6\), \(\widehat{\Delta}_{\Leff}\) is stable over an extended range of \((\epsilon,L)\), and its sign separates the oscillatory and non-oscillatory regimes of the TopSFF. 
For smaller \(q\), the local discriminant varies in space, showing that the effective two-mode CH description does not apply.

To diagnose the stability of the CH extraction in \(\ell\), we compute the relative variation
\begin{equation}
\label{app_eq:error_ch}
    \mathcal{E}_{\CH}(\ell,\epsilon)
    :=
    \max\left[
    \frac{
    |\widehat{\Trb}_{\ell+\delta \ell}(\epsilon)-\widehat{\Trb}_{\ell}(\epsilon)|
    }{
    \frac12\left(
    |\widehat{\Trb}_{\ell+\delta \ell}(\epsilon)|+|\widehat{\Trb}_{\ell}(\epsilon)|
    \right)+\eta
    },
    \,
    \frac{
    |\widehat{\Detb}_{\ell+\delta \ell}(\epsilon)-\widehat{\Detb}_{\ell}(\epsilon)|
    }{
    \frac12\left(
    |\widehat{\Detb}_{\ell+\delta \ell}(\epsilon)|+|\widehat{\Detb}_{\ell}(\epsilon)|
    \right)+\eta
    }
    \right] ,
\end{equation}
where \(\eta=10^{-12}\) is a small regularizer that prevents singular denominators. Small values of \(\mathcal{E}_{\CH}\) indicate that the locally extracted \(\widehat{\Trb}_{\ell}\) and \(\widehat{\Detb}_{\ell}\) are stable upon varying \(\ell\), and hence that the effective two-dimensional Cayley-Hamilton description is self-consistent. In this stable regime, the sign change of \(\widehat{\Delta}_{\ell}\)
provides a reliable finite-\(q\) diagnostic of the \(\mathcal{PT}\) transition.
\begin{figure}[ht!]
    \centering
    \includegraphics[width=0.32\linewidth]{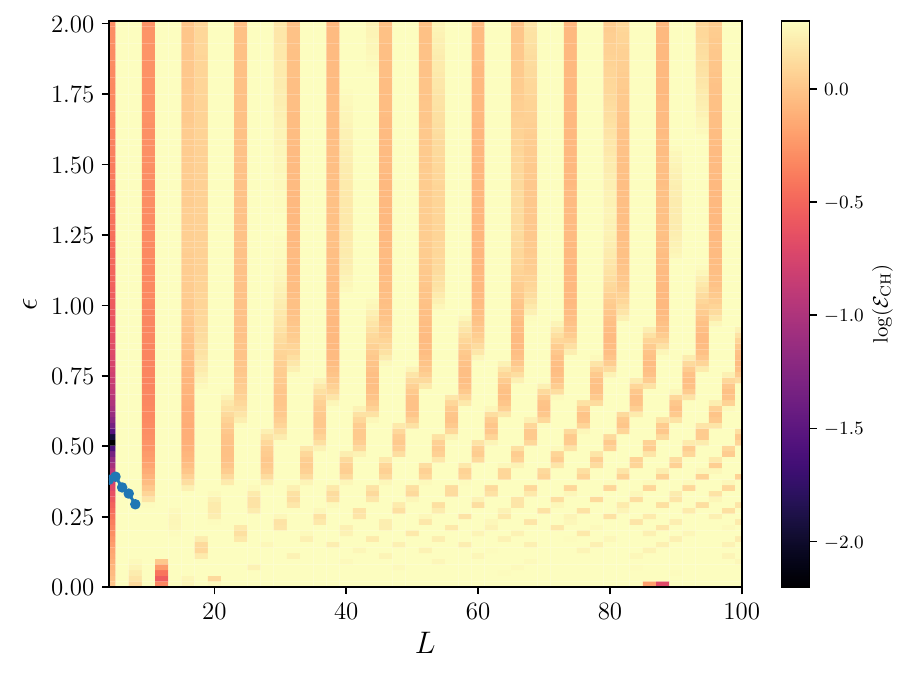}
    \includegraphics[width=0.32\linewidth]{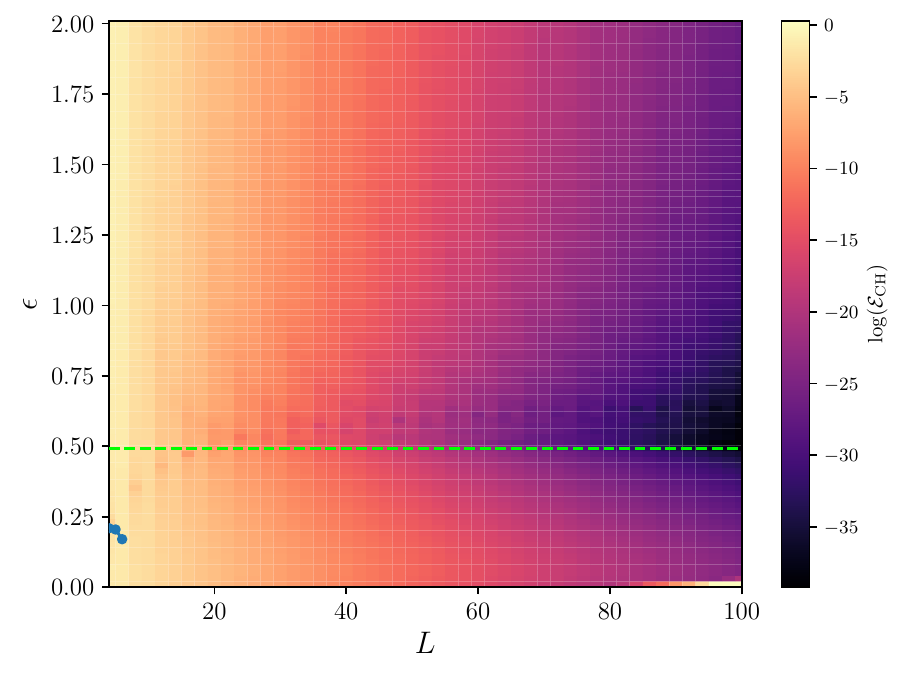}
    \includegraphics[width=0.32\linewidth]{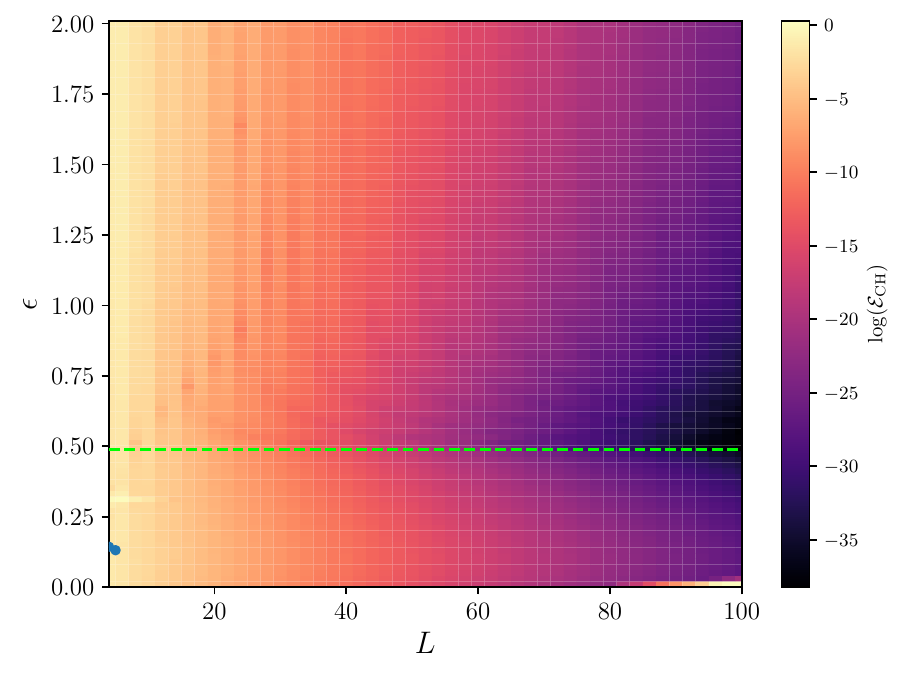}

   \vspace{0.25em}

    \includegraphics[width=0.32\textwidth]{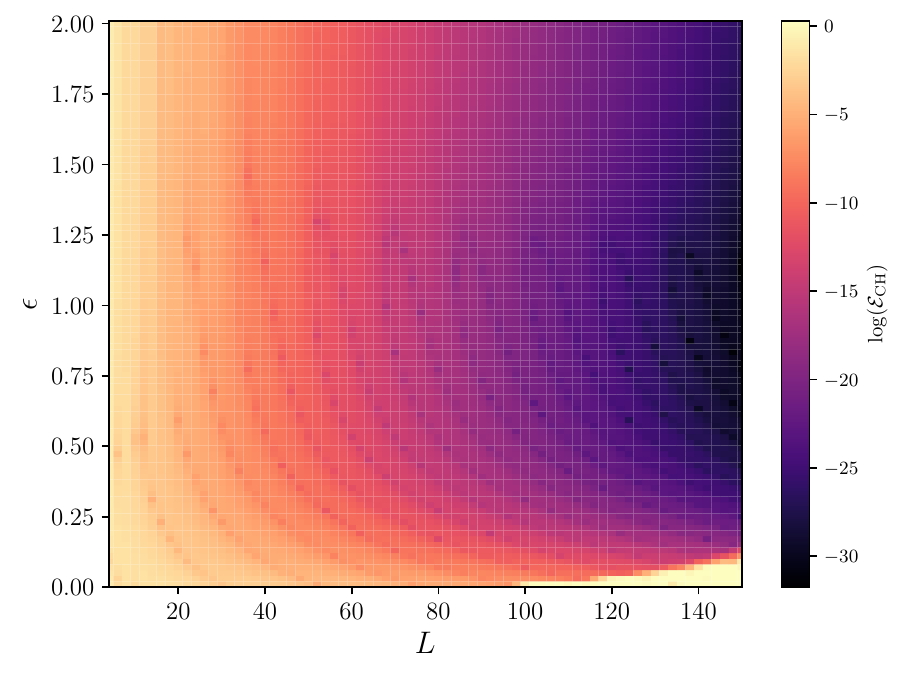}
    \hfill
    \includegraphics[width=0.32\textwidth]{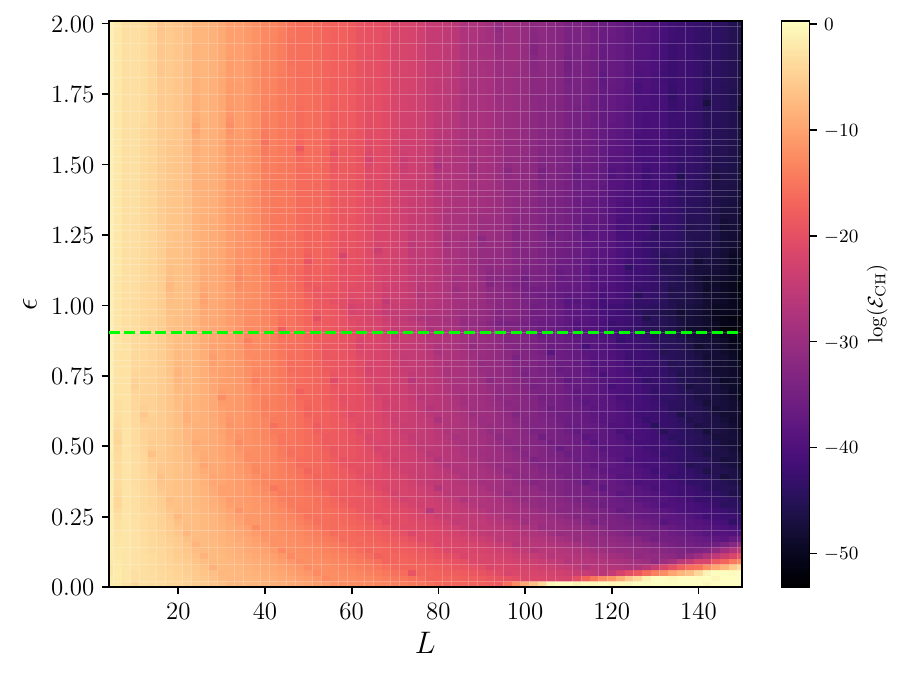}
    \hfill
    \includegraphics[width=0.32\linewidth]{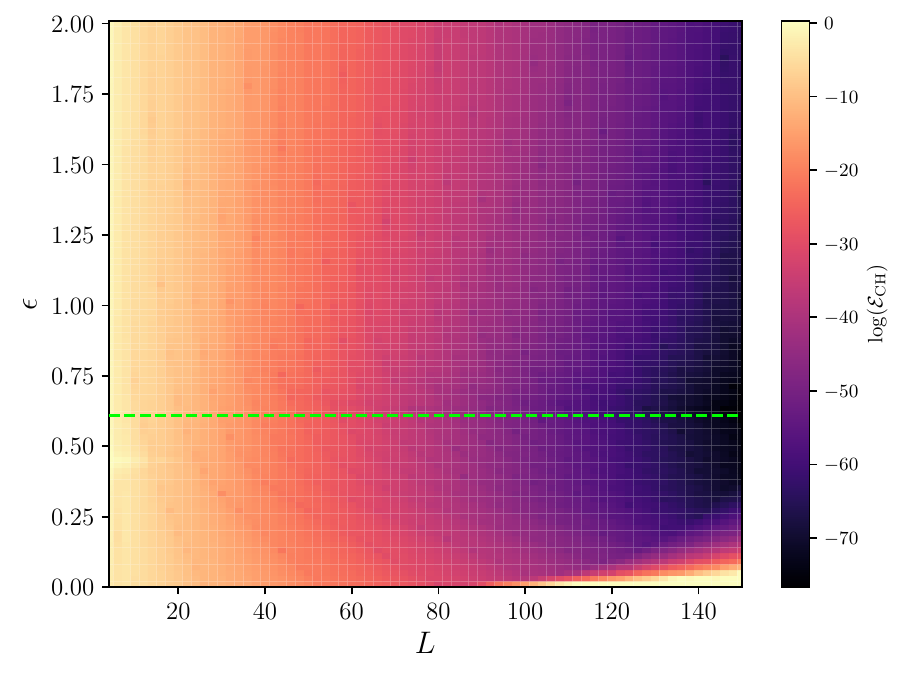}

    \includegraphics[width=0.24\textwidth]{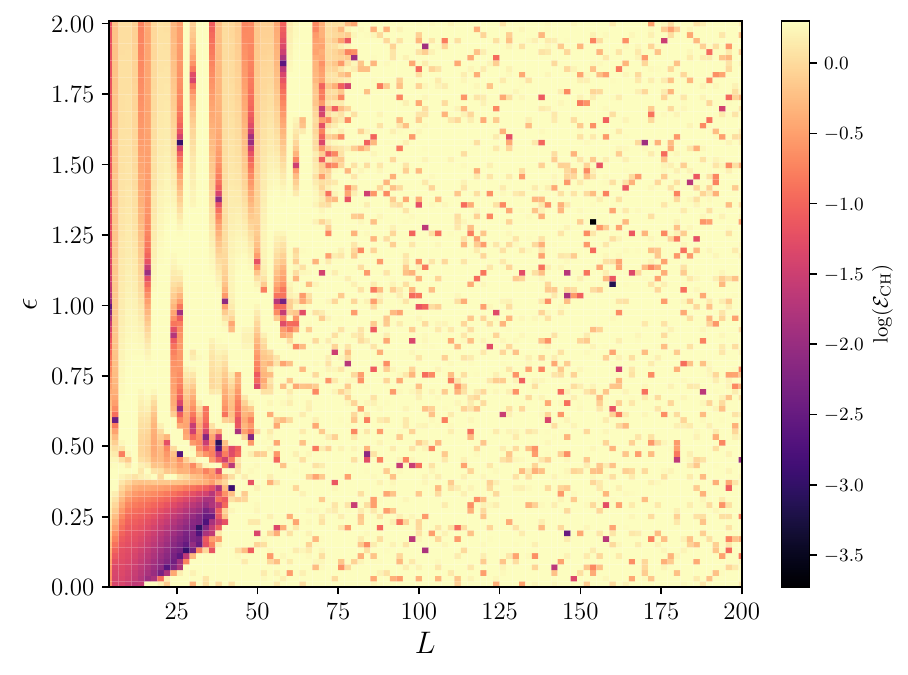}
    \hfill
    \includegraphics[width=0.24\textwidth]{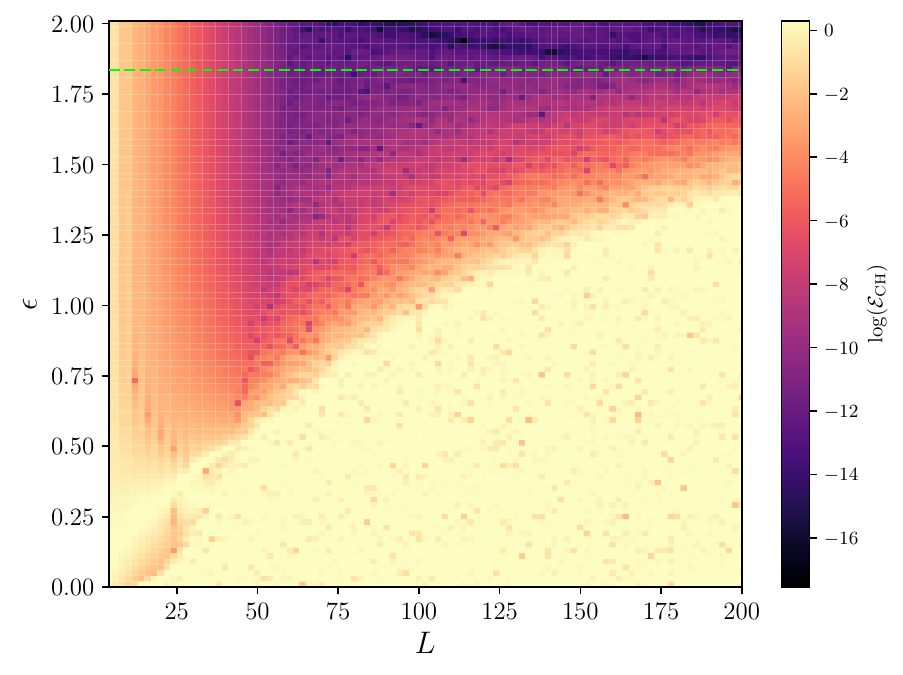}
    \hfill
    \includegraphics[width=0.24\linewidth]{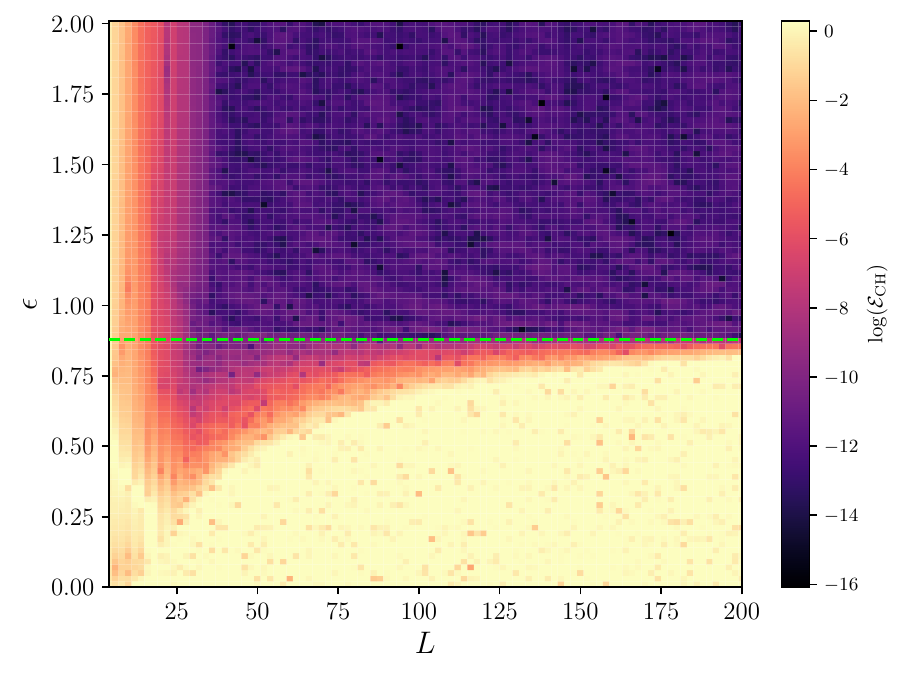}
    \hfill
        \includegraphics[width=0.24\linewidth]{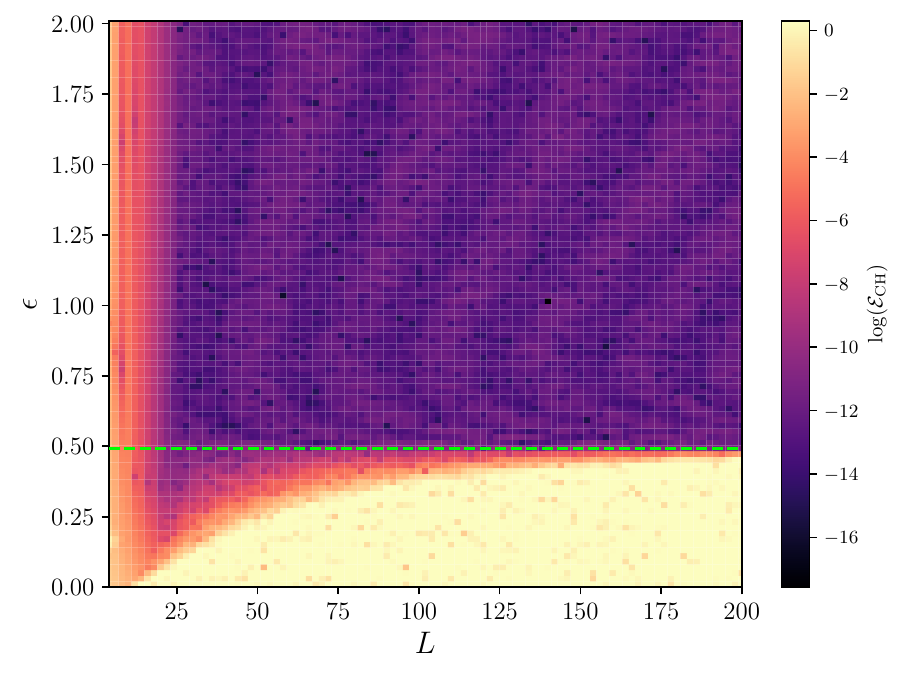}

    \caption{\textbf{Heatmap of $\mathcal{E}_{\CH}$ in the $(\epsilon,L)$ plane.} 
    Darker and lighter shades indicate lower and higher $\mathcal{E}_{\CH}$ values, respectively.
    We quantify the effectiveness of the two-mode Cayley-Hamilton description using $\mathcal{E}_{\mathrm{CH}}$, defined in Eq.~\eqref{app_eq:error_ch}. The $q$'s are (left to right) $q=3,4,5$ for $t=2$ (top row), $q=5,6,10$ for $t=3$ (middle row), and $q=3,4,5$ for $t=4$ (bottom row). Cases for larger $q$ display qualitatively similar behaviour to the cases of the largest $q$ displayed for each $t$. The dashed green lines indicate the $\mathcal{PT}$ transition line, while the dark blue lines at $t=2$ indicate the boundary between chaotic and non-chaotic regimes as probed by spacing ratio $\overline{r}$ in Fig.~\ref{fig:eps_star_parity}. Note that the two-mode Cayley-Hamilton description applies well into the non-chaotic regime. For $t=4$ at small $\epsilon$ and large $L$, $\mathcal{E}_{\mathrm{CH}}=\mathcal{O}(1)$ because exponential suppression of the subleading eigenmode renders the two-mode extraction numerically ill-conditioned.
    Note that an independent model BWM, already at $q=2$, exhibits momentum-dependent TopSFF behaviour qualitatively consistent with the large-$q$ RPM predictions; see Sec.~\ref{app:bwm_topsff}.
    }
    \label{fig:CH_validity}
\end{figure}

Fig.~\ref{fig:CH_validity} shows plots of \(\mathcal{E}_{\CH}\) in the \((\epsilon,L)\) plane for $t=2,3,4$ and various $q$. 
For $t=2,4$, the $q=3$ data do not exhibit an extended stable region, consistent with the instability shown in Fig.~\ref{fig:sidebyside} for $\widehat{\Delta}_{\ell}$. In contrast, for $q\geq4$, a broad region with small \(\mathcal{E}_{\CH}\) appears, showing that the effective CH description is valid over an extended part of the \((\epsilon,L)\) plane. The dashed green line marks the extracted \(\mathcal{PT}\) transition, while the blue curve in the plots for $t=2$ shows $\epsilon^*$ against $L$. The blue curve clearly deviates from the boundary where the effective CH description holds, which is indicated by the separation between the dark (low $\mathcal{E}_{\CH}$) and the light region (high $\mathcal{E}_{\CH}$) region. This suggests that the $\mathcal{PT}$ transition is inherently different to the chaotic-to-non-chaotic crossover.

For $t=4$, $\mathcal{E}_{\mathrm{CH}}$ becomes of order one at sufficiently large $L$ in the $\mathcal{PT}$-unbroken regime, $\epsilon<\epsilon_{\mathrm{EP}}$, indicating that the two-mode extraction becomes unreliable in this region.
This does not necessarily imply a breakdown of the underlying two-mode description.
Rather, in the $\mathcal{PT}$-unbroken phase the two leading eigenvalues have unequal moduli, so that the contribution of the subleading eigenvalue is exponentially suppressed with increasing $L$.
At sufficiently large $L$, the TopSFF is therefore effectively dominated by a single eigenmode, rendering the extraction of the second eigenvalue numerically ill-conditioned and leading to an enhanced $\mathcal{E}_{\mathrm{CH}}$.
This effect is particularly pronounced for $t=4$, for which the finite-$q$ transfer matrix is substantially larger than for $t=2$ and $t=3$, making the isolation of the two leading modes more numerically demanding.
In contrast, in the $\mathcal{PT}$-broken regime the dominant eigenvalues form a complex-conjugate pair with equal modulus, so that both contributions remain visible at large $L$ and the two-mode Cayley-Hamilton description remains resolved.

Having established the region of numerical stability, we define a truncation length \(L_{\rm trunc}\) in the following. For each fixed \((q,\epsilon)\), we scan the local estimates in increasing \(L\) and retain the initial stable region satisfying
\begin{equation}
    \mathcal{E}_{\CH}(\ell)
    \leq
    \mathcal{E}_{\CH}^{\rm cut},
    \qquad
    \mathcal{E}_{\CH}^{\rm cut}=1 .
\end{equation}
\(L_{\rm trunc}\) is the largest retained value of \(\ell\) before this stability criterion fails. Thus \(L_{\rm trunc}\) gives the maximum effective length up to which the locally extracted Cayley-Hamilton coefficients can be treated as stable for a given \((q,\epsilon)\). 
In the stable regime, we obtain the final effective coefficients by averaging the local estimates \( \{ \widehat{\Trb}_\ell\} \) and \( \{ \widehat{\Detb}_\ell\} \) over the retained window,
\[
    \widehat{\Trb}
    =
    \frac{1}{|\mathcal{W}|}
    \sum_{\ell\in\mathcal{W}}
    \widehat{\Trb}_{\ell},
    \qquad
    \widehat{\Detb}
    =
    \frac{1}{|\mathcal{W}|}
    \sum_{\ell\in\mathcal{W}}
    \widehat{\Detb}_{\ell},
\qquad 
    \widehat{\Delta}
    :=
    \widehat{\Trb}^2
    - 4 \widehat{\Detb} ,
\]
where \(\mathcal{W}\) denotes the retained CH stable range, \(L\leq L_{\rm trunc}\). These averaged coefficients are used to form the effective discriminant and eigenvalues below.

\subsubsection{Cayley-Hamilton extraction and \(\PT\) transitions}

\begin{figure}
    \centering
    \includegraphics[width=0.32\linewidth]{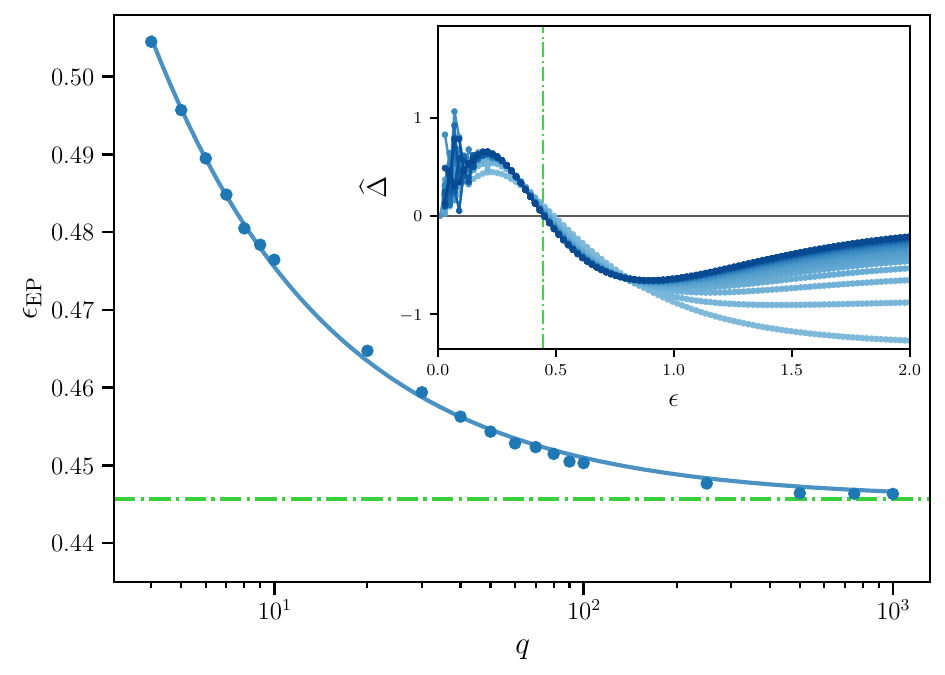}\hfill
    \includegraphics[width=0.32\linewidth]{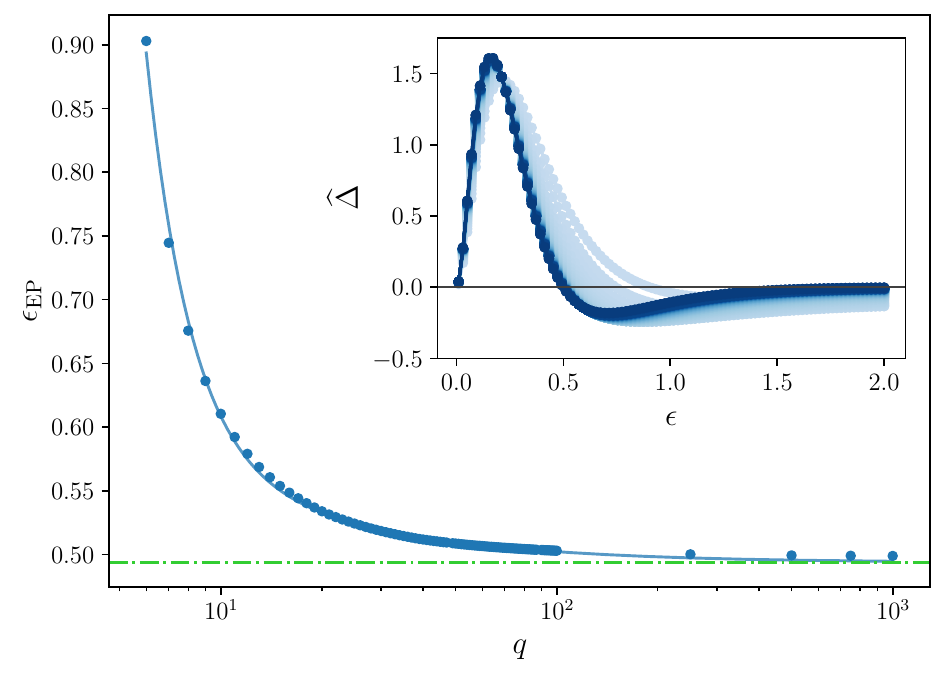}\hfill
    \includegraphics[width=0.32\linewidth]{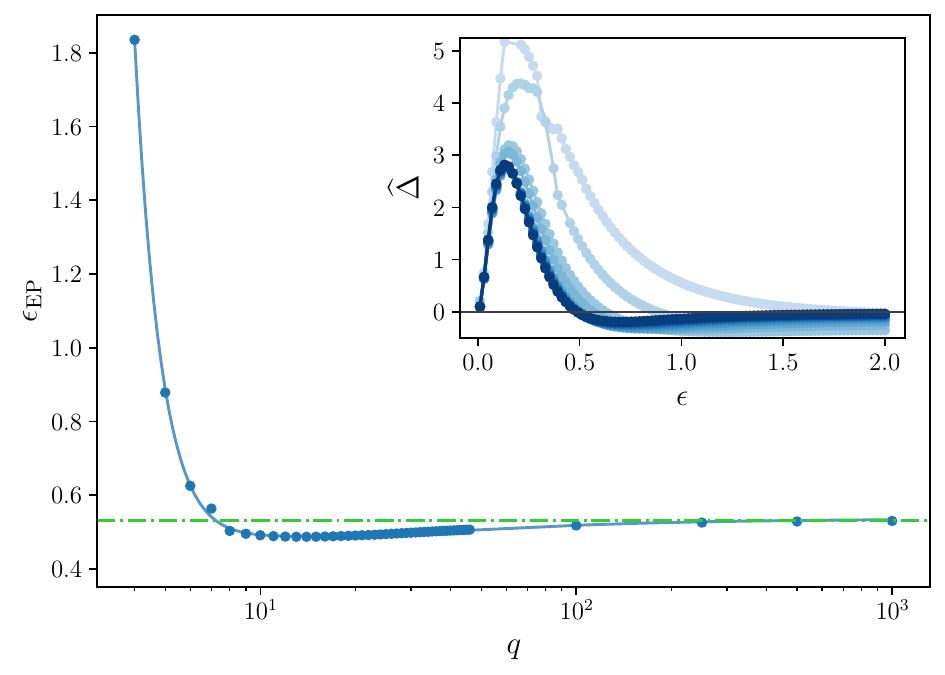}\hfill
    \caption{
    \textbf{Finite-$q$ $\mathcal{PT}$ transition points $\epsilon_\mathrm{EP}$ agree with large-$q$ predictions.}
    Results are shown for $t=2,3,4$, from left to right. Main plots: the dependence of $\epsilon_\mathrm{EP}$ on $q$. The blue markers denote numerically computed values, while the solid blue curves are polynomial fits.
The green dash-dotted lines indicate the corresponding large-$q$ predictions.
As $t$ increases, the finite-$q$ results converge to the large-$q$ limits.
Insets: dependence of effective discriminant $\widehat{\Delta}$ on $\epsilon$ for various values of $q$, with lighter (darker) shades of blue corresponding to smaller (larger) $q$.
For each $q$, the zero crossing $\widehat{\Delta}=0$ determines $\epsilon_{\mathrm{EP}}$, yielding the corresponding point shown in the main panel.}
    \label{fig:app_I_finite_q_ep}
\end{figure}

\begin{figure}
    \centering
    \includegraphics[width=0.32\linewidth]{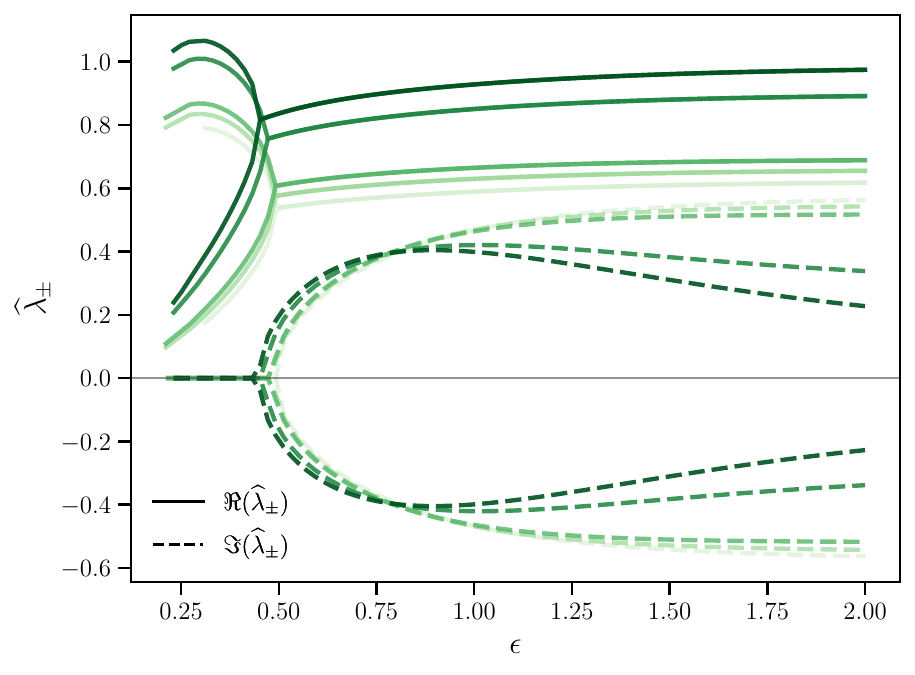}
    \includegraphics[width=0.32\linewidth]{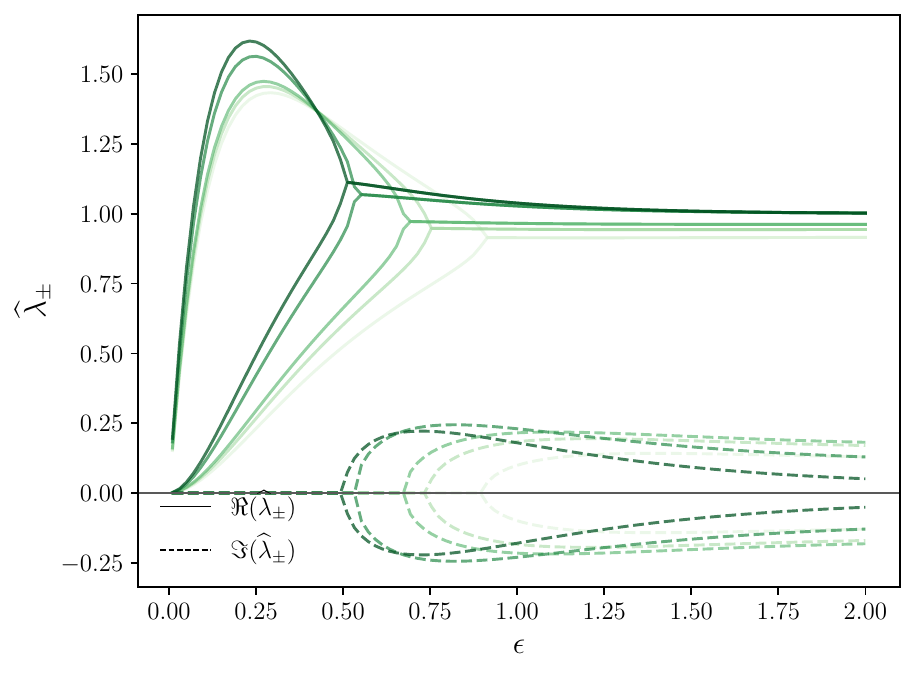}
    \includegraphics[width=0.32\linewidth]{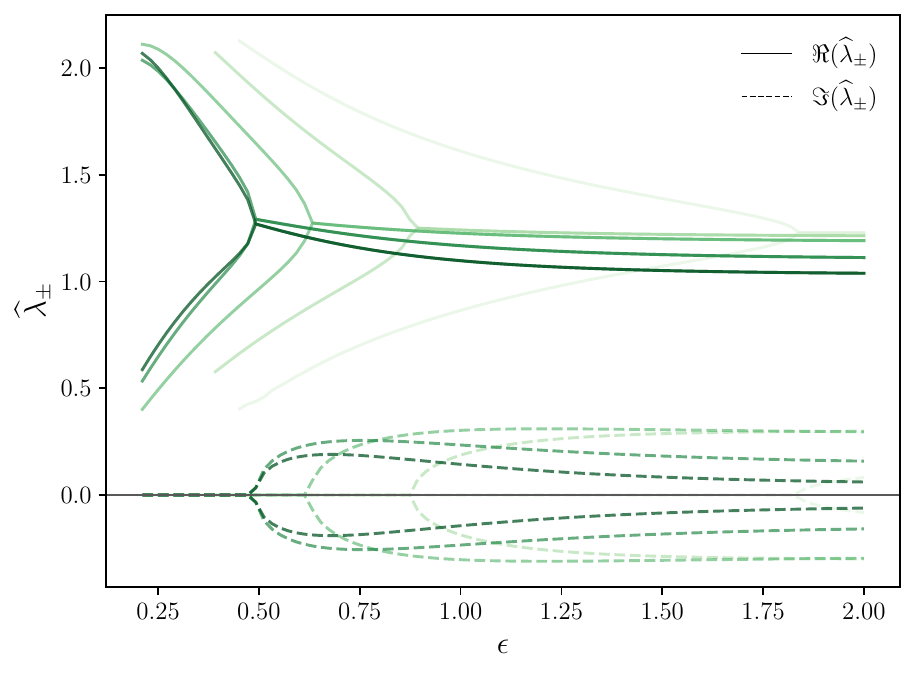}
    \caption{
\textbf{Finite-$q$ effective eigenvalues.}
Results are shown for $t=2,3,4$, from left to right.
The solid and dashed curves denote
$\operatorname{Re}(\widehat{\lambda}_{\pm})$ and
$\operatorname{Im}(\widehat{\lambda}_{\pm})$, respectively, as functions of $\epsilon$.
From light to dark green, the values of $q$ are
$q=4,5,6,20,100$ for $t=2$,
$q=6,7,8,20,100$ for $t=3$, and
$q=4,5,6$ for $t=4$.
For all $t$ and $q$ shown, the two real eigenvalue branches coalesce at an exceptional point, beyond which they form a complex-conjugate pair with non-zero imaginary parts.
As $q$ increases, the EP persists and the eigenvalue trajectories progressively approach the spectra predicted by the large-$q$ analysis in Figs.~\ref{fig:grid_4_sectors}, \ref{fig:grid_all_sectors}, and \ref{fig:grid_all_sectors_t3}. 
    }
    \label{fig:eff_eigs_manyq}    
\end{figure}

The corresponding effective discriminant and eigenvalues are
\begin{equation}
    \widehat{\Delta}
    =
    \widehat{\Trb}^{\,2}
    -
    4\widehat{\Detb},
    \qquad
    \widehat{\lambda}_{\pm}
    =
    \frac{\widehat{\Trb}\pm\sqrt{\widehat{\Delta}}}{2}.
    \label{eq:CH_discriminant_extraction}
\end{equation}
Since \(\widehat{\Trb}\) and \(\widehat{\Detb}\) are real, the sign of \(\widehat{\Delta}\) diagnoses the \(\mathcal{PT}\) transition: \(\widehat{\Delta}>0\) gives two real effective eigenvalues, while \(\widehat{\Delta}<0\) gives a complex-conjugate pair, \(\widehat{\lambda}_{-}=\widehat{\lambda}_{+}^{*}\). Thus, the zero of \(\widehat{\Delta}\) locates the finite-\(q\) exceptional point (EP).
We plot this finite-$q$ EP analysis in Fig.~\ref{fig:app_I_finite_q_ep} for $t=2,3,4$ (left to right) and increasing $q$. Across all $t$, the results show that as $q$ increases, the transition interaction strength $\epsilon_{\mathrm{EP}}$ approaches the large-$q$ value predicted in the main text. These results strongly support the correctness of the finite-$q$ analysis, and they show that the EP found in the large-$q$ theory persists at finite-$q$.

In Fig.~\ref{fig:eff_eigs_manyq}, we plot the effective eigenvalues in the \(k=(\pi,\pi,0)\) sector for increasing \(q\) and $t=2,3,4$. Across the transition, the two real branches coalesce and split into a complex-conjugate pair, providing a direct finite-\(q\) signature of \(\mathcal{PT}\) symmetry breaking. The extracted spectra show qualitative agreement with the large-\(q\) eigenvalues in the same momentum sector, see figures in the main text and Fig.~\ref{fig:grid_all_sectors}, and approach them as \(q\) is increased.
In particular, in the case of $t=2$, the eigenvalue moduli in the \(\mathcal{PT}\) unbroken phase increase with \(q\) and eventually exceed unity, so that the TopSFF changes from exponential decay to exponential growth with system size in the $\mathcal{PT}$ unbroken phase. At $t=3$, the finite-\(q\) analytics predicts both an exponential decrease (smaller $\epsilon$) and exponential growth (larger but still small $\epsilon$) of the TopSFF in the \(\mathcal{PT}\) unbroken phase. This is consistent with the findings in Fig.~\ref{fig:sidebyside}, in which the curves change from decay to growth (blue curves). For $t=4$, exponential growth is predicted throughout the pre-Heisenberg regime.

Although the TopSFF extraction provides access to effective eigenvalues rather than eigenvectors, the effective discriminant still serves as a useful proxy for eigenvector coalescence. We  define the discriminant based proxy  Petermann factor $P^{\mathrm{proxy}}_{\pm}$ and proxy eigenvector overlap  $O^{\mathrm{proxy}}_{\pm}$ as
\begin{equation}
    P^{\mathrm{proxy}}(\epsilon)
    =
    \frac{1}{|\widehat{\Delta}(\epsilon)|+\eta},
    \qquad
    O^{\mathrm{proxy}}(\epsilon)
    =
    \left|\sqrt{\widehat{\Delta}(\epsilon)}\right| ,
    \label{eq:petermann_proxy}
\end{equation}
where \(\eta=10^{-12}\) is a small regularizer that removes the numerical singularity at \(\widehat{\Delta}=0\). The first quantity mimics the divergence of the Petermann factor at an exceptional point, while the second tracks the square-root eigenvalue splitting and vanishes as the two effective eigenvalues coalesce. As shown in Fig.~\ref{fig:petermann_proxy}, these finite-\(q\) proxies display the same qualitative behaviour as the large-\(q\) Petermann factor and eigenvector overlap diagnostics in Figs.~\ref{fig:pipi0_petermann_largeq} and~\ref{fig:pipi0_evector_overlap_largeq}.

\begin{figure}
    \centering
    \includegraphics[width=0.4\linewidth]{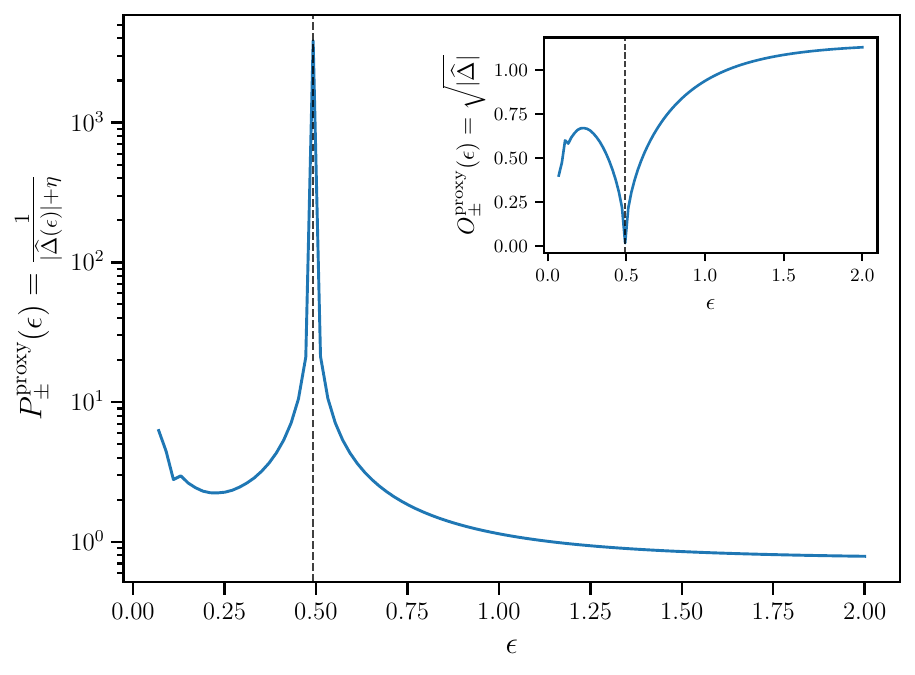}
    \caption{
\textbf{Finite-\(q\) Petermann factor proxy for the exceptional point.} For \(q=4\), we plot the Petermann factor proxy
\(P^{\rm proxy}_{\pm}(\epsilon):=1/(\sqrt{|\widehat{\Delta}(\epsilon)|}+\eta)\),
where \(\widehat{\Delta}\) is the discriminant of the two leading modes and \(\eta= 10^{-12}\) is a small regulator. 
The vertical dashed line marks \(\epsilon_{\rm EP}\simeq0.490495\). 
The peak in \(P^{\rm proxy}_{\pm}\), together with the collapse of \(\sqrt{|\widehat{\Delta}|}\) in the inset, provide proxy signatures of eigenvector coalescence at the EP. 
}
   \label{fig:petermann_proxy}
\end{figure}

At finite \(q\), at the extracted exceptional point, the boundary-resolved TopSFF combination
\(
\frac{
\overline{\Knormalized_{\spatdef}^{00}}
+
\overline{\Knormalized_{\spatdef}^{11}}
}
{
[\widehat{\Trb}/2]^{\Leff-1}
}
\)
exhibits linear scaling in \(L\), consistent with the Jordan-block contribution expected at an EP. Away from the transition, the curves bend away from this linear form, as expected when the effective spectrum consists either of two distinct real eigenvalues or of a complex-conjugate pair. Thus, the boundary-resolved TopSFF exposes the Jordan contribution, hidden in the unresolved TopSFF by the vanishing physical-boundary prefactor.
In the main text, we demonstrated this boundary-resolved scaling for \(q=4\) in the \(k=(\pi,\pi,0)\) sector. 
In Fig.~\ref{fig:jordan_diff_t}, we show the corresponding finite-$q$ results for $t=2,3,4$ at $q=4,6,4$, respectively. 
At the exceptional point, after dividing out the leading exponential envelope, the TopSFF exhibits the characteristic linear-in-$L$ scaling expected from the nontrivial Jordan block. The same behaviour is observed more generally for larger $q$.
The far right panel of Fig.~\ref{fig:jordan_diff_t} verifies that the EP extracted from the boundary-resolved TopSFF agrees with that obtained from the regular, boundary-unresolved TopSFF.  Specifically, we independently extract the transition points from
$\overline{\Knormalized_{\spatdef}^{00}}$ and
$\overline{\Knormalized_{\spatdef}^{11}}$,
and compare them with $\epsilon_{\EP}^{\mathrm{R}}$ obtained from the regular TopSFF.
The small discrepancies show that the boundary-resolved and regular observables locate the same finite-$q$ exceptional point within numerical extraction uncertainty.

\begin{figure}
    \centering
    \includegraphics[width=0.24\linewidth]{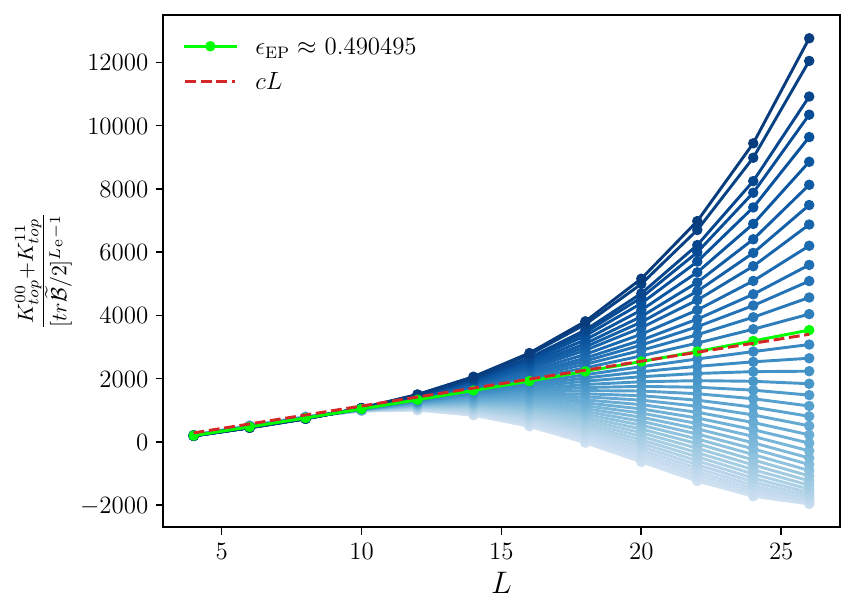}
    \includegraphics[width=0.24\linewidth]{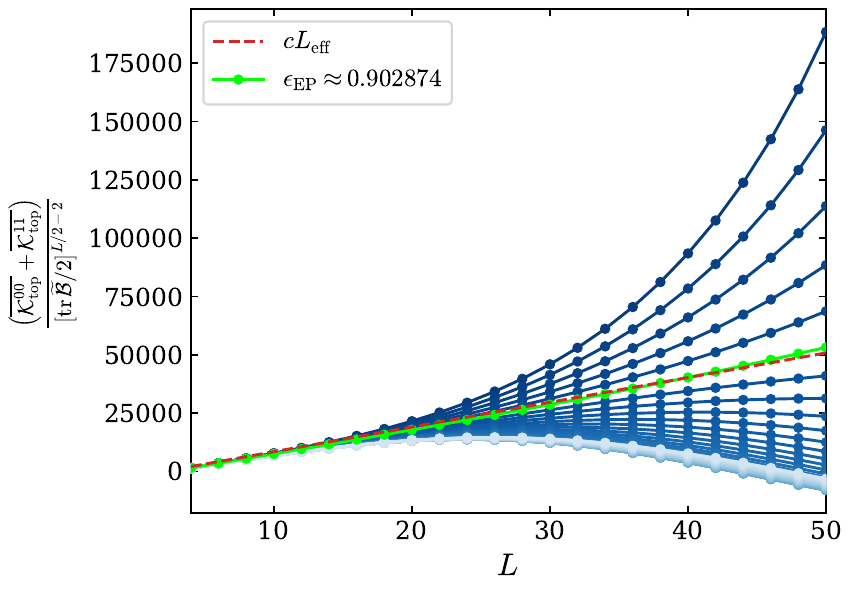}
    \includegraphics[width=0.24\linewidth]{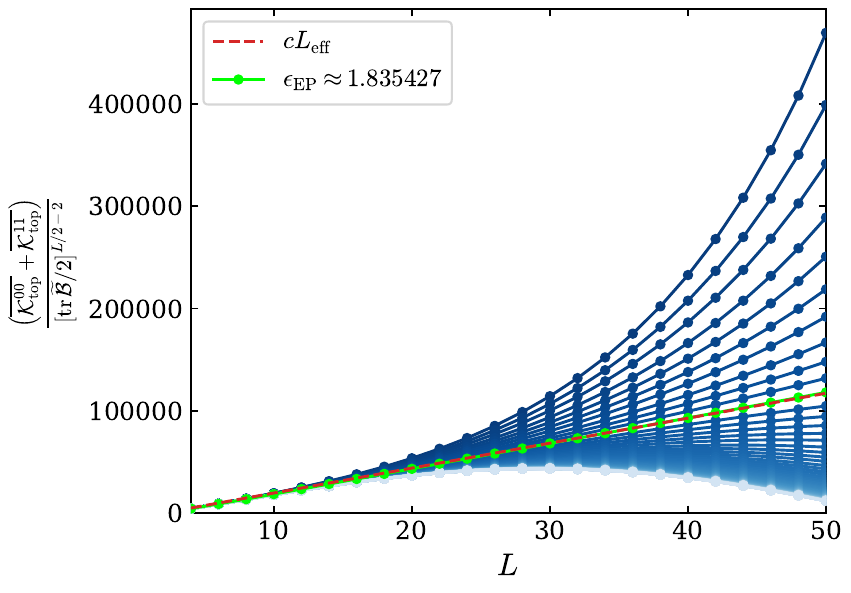}
    \includegraphics[width=0.24\linewidth]{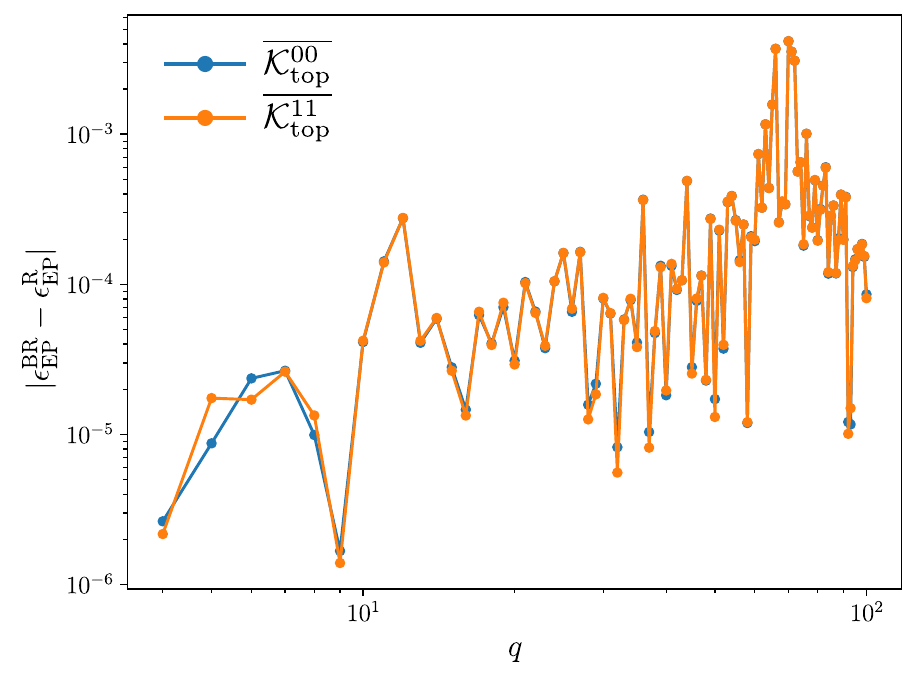}
    \caption{
\textbf{Jordan-block scaling of the TopSFF at the exceptional point.}
Results are shown for $t=2,3,4$ with $q=4,6,4$, respectively, from left to right.
The blue curves show the normalized TopSFF as a function of $L$ for values of $\epsilon$ in the vicinity of the exceptional point $\epsilon_{\mathrm{EP}}$, with lighter to darker shades corresponding to increasing $\epsilon$.
The green curve denotes the result at the numerically determined $\epsilon_{\mathrm{EP}}$, while the red dashed line shows the linear form $c L_{\mathrm{eff}}$.
At the EP, the two effective eigenvalues coalesce and form a non-trivial Jordan block, producing the characteristic linear prefactor $L_{\mathrm{eff}}$ after the leading exponential dependence has been removed.
In the far right panel, we show the difference \(|\epsilon_{\EP}^{\rm BR}-\epsilon_{\EP}^{\mathrm{R}}|\) between the EP extracted from 
boundary-resolved TopSFF and regular TopSFF for $t=2$ as an illustrative example. The consistently small discrepancy demonstrates that the two observables identify the same EP within numerical extraction errors.
 }
    \label{fig:jordan_diff_t}
\end{figure}

\subsubsection{Verification of large-$q$ RPM TopSFF predictions in the BWM at $q=2$}\label{app:bwm_topsff}

Simulations of a second independent model, the Brick Wall Model (BWM, Sec.~\ref{app:model_bwm}), which, unlike RPM, is composed of \textit{non-diagonal} gates, corroborate at $q=2$ the generality of the exact large-$q$ RPM results, by exhibiting the predicted momentum-sector-dependent oscillatory and monotonic TopSFF behaviours. 
We simulate BWM in an approach analogous to  the one described in Sec.~\ref{subsec:finite_q_symbolic_tsff}. Unlike the RPM, however, the simplest BWM that can be efficiently simulated in this manner does not possess a continuous tuning parameter analogous to $\epsilon$. Consequently, rather than tracking the $\mathcal{PT}$ transition as a function of a coupling parameter, we test its characteristic momentum-sector-dependent signatures at a fixed ensemble.

Fig.~\ref{fig:bwm_t2_q4_allmom} shows the BWM TopSFF at $t=2$ and $q=2$ in all momentum sectors.
Despite the microscopic differences between the BWM and RPM, the BWM reproduces the characteristic sector dependence predicted by the large-$q$ RPM analysis.
First, the results obey the expected symmetry relations between momentum sectors.
In particular, the TopSFF is invariant under $k_c\leftrightarrow k_d$, such that the corresponding sectors in the second and third columns of Fig.~\ref{fig:bwm_t2_q4_allmom} coincide, together with the associated relations under momentum inversion.
These relations agree with those obtained analytically for the RPM.

Furthermore, the TopSFF of BWM at different momentum sectors exhibit qualitatively distinct $L$ dependence, in agreement with the RPM  momentum-sector TopSFF results  at finite $q$ in Figs.~\ref{fig:t2_q4_allmom} and in large $q$ in Fig.~\ref{fig:grid_all_sectors}.
For sufficiently large $L$, the BWM TopSFF exhibits oscillatory exponential behaviour when dominated by a complex-conjugate pair of eigenvalues, and monotonic exponential behaviour when dominated by a real eigenvalue, as predicted by RPM.
In the momentum sectors with $k_b=\pi/2$ or $3\pi/2$, the BWM TopSFF exhibits oscillatory behaviour, consistent with the complex-conjugate eigenvalue pair predicted by the large-$q$ RPM analysis.
A particularly clean test is provided by the sectors $(\pi,\pi,\pi)$, $(\pi,0,0)$, $(0,\pi,0)$, and $(0,0,\pi)$, for which the large-$q$ RPM analysis predicts a dominant real negative eigenvalue. Consistently, the BWM TopSFF alternates in sign between consecutive values of $\Leff$. 
For $(\pi,\pi,0)$ and $(\pi,0,\pi)$, a tunable $\PT$ transition cannot be accessed in the simplest efficiently-simulatable  form of the BWM, since it lacks an interaction-strength tuning parameter.

The agreement between the BWM and RPM is significant because the two models have different microscopic constructions. It demonstrates that the momentum-dependent TopSFF structure obtained from the large-$q$ RPM is not an artifact of the random-phase model or of the large-$q$ limit, but persists in an independent many-body chaotic circuit at finite $q$.
Note however that the present BWM calculation does not resolve the exceptional point through a parameter sweep, due to the absence of a tunable parameter in this simplest formulation.

\begin{figure}
    \centering
    \includegraphics[width=\linewidth]
    {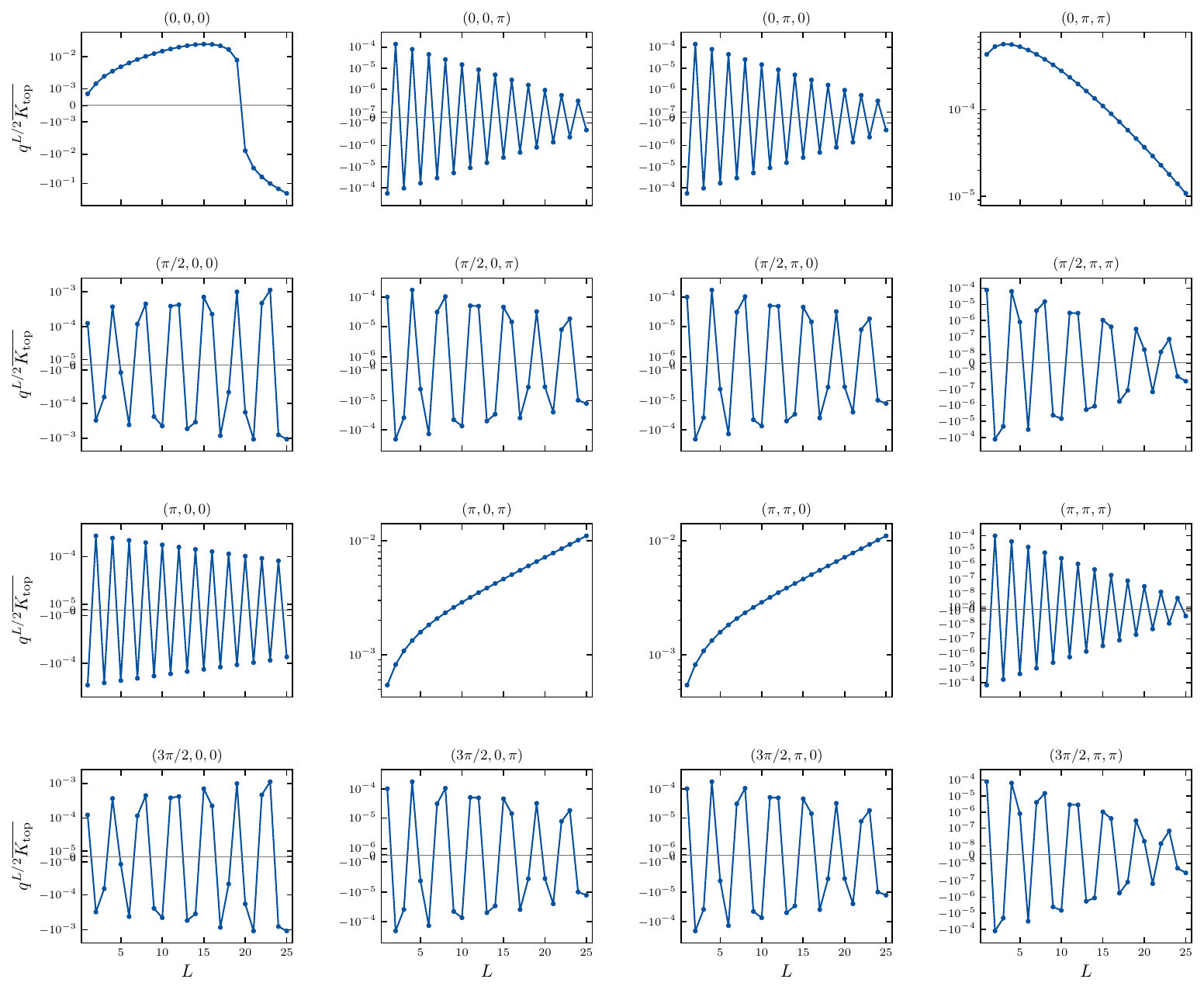}
\caption{
\textbf{BWM TopSFF at $t=2$ and $q=2$ in all momentum sectors.}
The normalized TopSFF, $q^{L/2}\overline{\mathcal{K}_{\mathrm{top}}}$, is plotted as a function of $L$ up to $L= 25$ for each momentum sector. The results exhibit the expected symmetry under $k_c\leftrightarrow k_d$, such that the second and third columns are identical, together with the corresponding relations under momentum inversion.
Several sectors, including $(\pi,0,0)$ and $(\pi,\pi,\pi)$, exhibit an alternating sign with $L$, consistent with a dominant negative real eigenvalue, whereas other sectors display non-oscillatory growth or decay. Importantly, these momentum-sector-resolved results of BWM at finite $q$ agree with the RPM  momentum-sector TopSFF results  at finite $q$ in Figs.~\ref{fig:t2_q4_allmom} and in large $q$ in Fig.~\ref{fig:grid_all_sectors}.
}   
    \label{fig:bwm_t2_q4_allmom}
\end{figure}

\subsection{Temporally extended topological defects}

We consider the TopSFF of a temporally extended global swap defect \(\hat{\mathcal D}=\hat S\otimes \hat\iden\) acting on \( U_{\mathrm{s}} \otimes  U_{\mathrm{s}}\) in the spatial direction.  The generic quantum many-body chaotic systems \(\hat U_{\mathrm{s}}(t,L):= \prod_{i=1}^L  \hat{V}(t,L)\) is generated by the space-time dual transfer matrices \(\hat{V}(t,L)\) [Fig.~\ref{fig:tdw_mickey_mouse}(a) purple] of RPM defined in \eqref{app_eq:global_trs_rpm} and RHM defined in \eqref{app_eq:rhm_gtrs}, satisfying the \([\hat{V}(t,L), \hat S]=0\), and therefore has time reversal symmetry (TRS).

We follow the scheme of Ref.~\cite{chan2021trans} to compute the SFF and TopSFF in the Thouless double-scaling limit
\begin{equation}
  \kappa(\ell)
  =
  \lim_{\substack{L,t\to\infty\\ \ell \equiv L/L_{\mathrm{Th}}(t)}}
   K(t,L),
\end{equation}
and analogously for the TopSFF, where \(L_{\mathrm{Th}}(t)\) is the Thouless length. We collapse the numerical data by estimating \(L_{\mathrm{Th}}(t)\) independently for each time \(t\). To do this, we choose a reference point \(\ell_0\) and set \(y_0=\kappa(\ell_0)\). For each fixed \(t\), we interpolate the numerical curve \(K(t,L)\) and find the first intersection with the horizontal line \(y=y_0\), which defines an intermediate scale \(L^\ast(t)\) through \(K\bigl(t,L^\ast(t)\bigr)=y_0\). Since \(L^\ast(t)/L_{\mathrm{Th}}(t)=\ell_0\) in the scaling regime, this yields the estimator \(L_{\mathrm{Th}}(t)=L^\ast(t)/\ell_0\). Finally, plotting \(K(t,L)\) against the rescaled variable \(L/L_{\mathrm{Th}}(t)\) produces the desired collapse.

In addition to the above procedure, we make two practical adjustments to improve the robustness of the collapse. First, there is residual freedom in the choice of the reference point \(\ell_0\). We therefore scan over a range of \(\ell_0\) values and quantify the quality of the resulting collapse by a sum of squared errors, defined as the squared distance between the collapsed data points and the scaling function \(\kappa(\ell)\). The optimal \(\ell_0\) is chosen by minimising this sum over the scan range.

Second, to improve the robustness of the interpolation and intersection-finding step in the presence of noisy data, we optionally apply Gaussian filtering. Writing the measured data as \(y_i=s_i+\eta_i\), with underlying signal \(s_i\) and noise \(\eta_i\), we define the filtered data by the discrete convolution
\begin{equation}
  y_i^{\mathrm{(filt)}}= \sum_{k} g_k y_{i-k},
  \qquad
  g_k =
  \frac{\exp\!\left[-k^2/(2\sigma_g^2)\right]}
  {\sum_{m}\exp\!\left[-m^2/(2\sigma_g^2)\right]},
  \qquad
  \sum_k g_k = 1 ,
\end{equation}
where \(\sigma_g\) controls the smoothing width. Assuming uncorrelated noise with variance \(\sigma_\eta^2\), the filtered noise variance is
\(
    \mathrm{Var}\!\left(\eta_i^{\mathrm{(filt)}}\right)
    =
    \sum_k g_k^2\,\mathrm{Var}(\eta_{i-k})
    =
    \sigma_\eta^2\sum_k g_k^2
    \leq
    \sigma_\eta^2 .
\)
Thus the filter suppresses fluctuations by forming a weighted average over neighbouring points, leading to a more stable estimate of \(L^\ast(t)\) and  \(L_{\mathrm{Th}}(t)\). We use this step only mildly, since excessive smoothing can wash out genuine fine structure in the data.

In Fig.~\ref{fig:scaling_collapse_and_Lth}, we apply this procedure to the TRS RPM and TRS RHM with a temporally extended global swap defect. The left and middle panels show the resulting collapse of the SFF and TopSFF, respectively, after rescaling the system size by the extracted Thouless length \(L_{\Th}(t)\). The data for different \(t\) and for both models collapse onto the large-\(q\) scaling functions, \(\kappa(\ell)=2\cosh\ell\) for the SFF and \(\kappa_{\topo}(\ell)=2\sinh\ell\) for the TopSFF. These results demonstrate the universality of the scaling forms of SFF and TopSFF with temporally extended global swap defect. The right panel shows the extracted \(L_{\Th}(t)\) for the two models, obtained independently from the SFF and TopSFF collapses. The SFF and TopSFF estimates agree within the numerical uncertainty, confirming that the two observables identify the same Thouless length.

In Fig.~\ref{fig:timetrans_collapse}, we apply the same collapse procedure to the Floquet RPM and Floquet RHM with a temporally extended time translation defect. After rescaling the system size by the extracted Thouless length \(L_{\Th}(t)\), the SFF and TopSFF data for different \(t\) and for both models collapse onto the large-\(q\) scaling functions \(\kappa(\ell)=e^{\ell}-\ell-1\) and \(\kappa_{\topo}(\ell)=e^{\ell}-1\), respectively. Again, this demonstrates the universality of the time translation defect scaling forms beyond the large-\(q\) limit. In the right panel, the independently extracted \(L_{\Th}(t)\) values from the SFF and TopSFF collapses agree within numerical uncertainty, confirming that both observables identify the same Thouless length.

\begin{figure}
    \centering
    \includegraphics[width=0.32\textwidth]{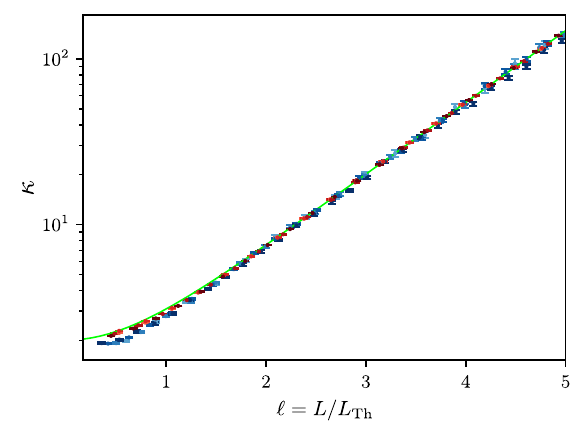}
    \hfill
    \includegraphics[width=0.32\textwidth]{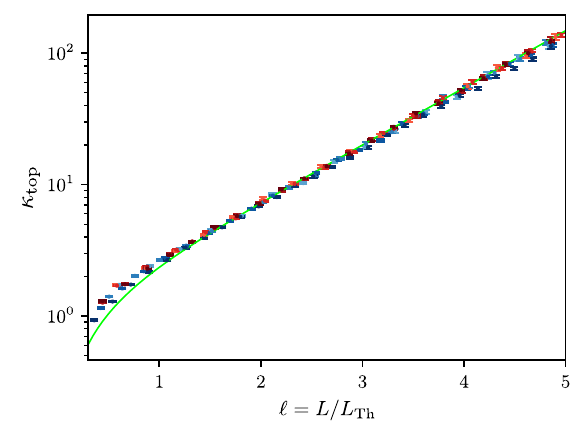}
    \hfill
    \includegraphics[width=0.32\textwidth]{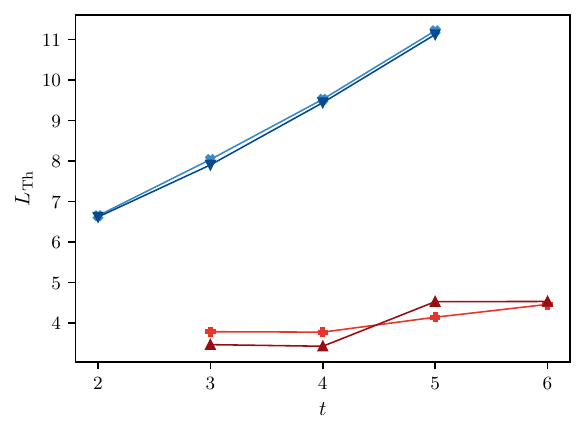}
    \caption{
    \textbf{Scaling collapse of TopSFF with temporally extended global swap defect.}
    Left: scaling collapse of the SFF for TRS RPM defined in \eqref{app_eq:global_trs_rpm} and TRS RHM defined in \eqref{app_eq:rhm_gtrs}. 
    Middle: corresponding scaling collapse of the TopSFF. 
    RPM (RHM) data are shown for timesteps \(t\in[2,6]\) (\(t\in[2,5]\)), with darker shades of red (blue) indicating larger \(t\). Results are averaged over \(7500\) realizations for RPM and \(5100\) realizations for RHM. In the left and middle panels, the data collapse onto the large-\(q\) exact scaling forms, \(\kappa=2\cosh(\ell)\) for the SFF  and \(\kappa_{\topo}=2\sinh(\ell)\) for the TopSFF (green lines). The RPM (RHM) data correspond to \(\epsilon=1.56\) (\(\epsilon=2.94\)). 
    Right: effective Thouless length \(L_{\mathrm{th}}(t)\) for RPM (red) and RHM (blue). For each model, the SFF and TopSFF are shown in lighter and darker shades, respectively.
    }
    \label{fig:scaling_collapse_and_Lth}
\end{figure}

\begin{figure}
    \centering
    \includegraphics[width=0.32\linewidth]{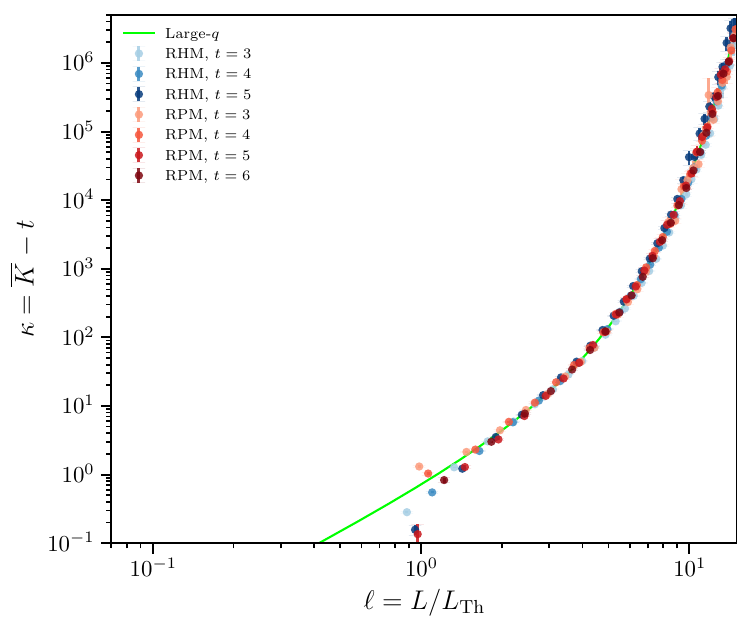}
    \hfill
    \includegraphics[width=0.32\linewidth]{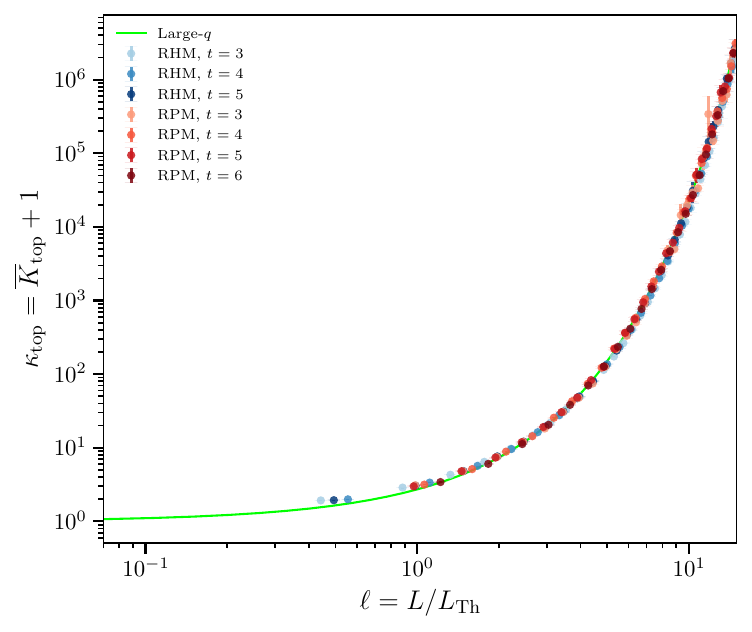}
    \hfill
    \includegraphics[width=0.32\linewidth]{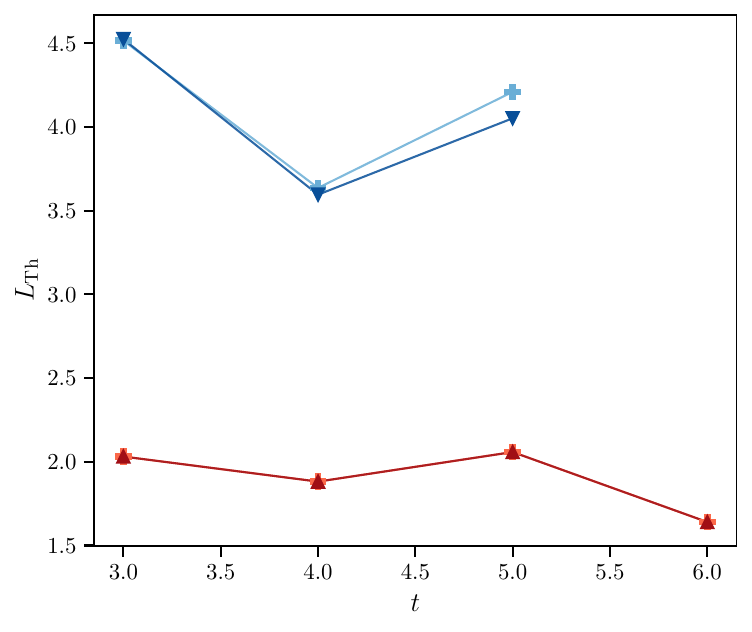}
    \caption{
    \textbf{Scaling collapse of TopSFF with temporally extended translation defect.}
Left: scaling collapse of the SFF for Floquet RPM defined in \eqref{app_eq:floqonly_rpm} and Floquet RHM defined in \eqref{app_eq:floqonly_rhm}. 
    Middle: corresponding scaling collapse of the TopSFF with a temporally extended time translation defect.
    RPM and RHM data are shown in red and blue, respectively, with darker shades indicating larger $t$.
    The data is averaged over $10000$ circuit realisations. The RPM data is evaluated at $\epsilon=2$, while the RHM data are evaluated at $\epsilon=1.85$.
    In the left panel the SFF collapses onto
    $\kappa=e^{\ell}-\ell-1$, while middle panel, the TopSFF collapses onto the large-$q$ scaling form
    $\kappa_{\topo}=e^{\ell}-1$, shown as green reference curves.
    Right: effective Thouless length \(L_{\mathrm{th}}(t)\) extracted from the SFF and TopSFF collapses for the RPM and RHM.
    For each model, the SFF and TopSFF are shown in lighter and darker shades, respectively.
    }
    \label{fig:timetrans_collapse}
\end{figure}

\end{document}